\documentclass[
 reprint,
 aps,
 pra
]{revtex4-2}

\usepackage{amsmath, amsthm, amssymb, amsfonts, bbm, mathtools}
\usepackage{booktabs}
\usepackage{supertabular}
\usepackage{braket}
\usepackage{graphicx}
\usepackage{subcaption}
\usepackage{tikz}
\usepackage{xcolor}
\usepackage[hidelinks]{hyperref}
\usepackage[linesnumbered, ruled, vlined, titlenotnumbered, noend]{algorithm2e}
\usepackage{placeins}
\usepackage{listings}
\usepackage{verbatim}
\usepackage{cleveref}
\crefname{algocf}{alg.}{algs.}
\Crefname{algocf}{Algorithm}{Algorithms}
\crefname{part}{Part}{Parts}
\Crefname{part}{Part}{Parts}
\crefname{figure}{Fig.}{Figs.}
\Crefname{figure}{Fig.}{Figs.}
\crefname{table}{Table}{Tables}
\Crefname{table}{Table}{Tables}
\crefname{section}{Sec.}{Secs.}
\Crefname{section}{Sec.}{Secs.}

\let\braketket\ket
\renewcommand{\ket}[1]{\ensuremath{\braketket{#1}}}

\newcommand{\F}{\mathbb{F}}
\newcommand{\Z}{\mathbb{Z}}
\newcommand{\N}{\mathbb{N}}
\newcommand{\Ocal}{\mathcal{O}}

\newcommand{\Xcal}{\mathcal{X}}

\newcommand{\ind}{\mathbbm{1}}
\newcommand{\sample}{\xleftarrow{\$}}

\newcommand{\add}{\mathrm{add}}
\newcommand{\sub}{\mathrm{sub}}
\newcommand{\cmpl}{\mathrm{cmpl}}

\newcommand{\mul}{\mathrm{mul}}

\newcommand{\divi}{\mathrm{div}}
\newcommand{\red}{\mathrm{red}}

\newcommand{\bitcoincurve}{\mathtt{secp256k1}}

\DeclareRobustCommand{\CCZ}{\ensuremath{\mathrm{CCZ}}}
\DeclareRobustCommand{\CNOT}{\ensuremath{\mathrm{CNOT}}}
\DeclareRobustCommand{\CZ}{\ensuremath{\mathrm{CZ}}}
\DeclareRobustCommand{\SWAP}{\ensuremath{\mathrm{SWAP}}}
\newcommand{\wcat}{w_{\mathrm{cat}}}

\newcommand{\pswap}{p_{\mathrm{swap}}}
\newcommand{\ploss}{p_{\mathrm{loss}}}
\newcommand{\pleak}{p_{\mathrm{leak}}}
\newcommand{\plog}{p_{\mathrm{log}}}

\newcommand{\AddSlice}{\mathsf{AddSlice}}
\newcommand{\AddBMultiple}{\mathsf{AddBMultiple}}
\newcommand{\AddCMultiple}{\mathsf{AddCMultiple}}
\newcommand{\PhaseGE}{\mathsf{PhaseGE}}

\newcommand{\measurementLayerDepthVariable}{D_{\mathrm{meas}}}
\newcommand{\measurementLayerDepthTotal}{75{,}251{,}329}
\newcommand{\toffoliInjectionsVariable}{N_{\mathrm{Tof}}}
\newcommand{\toffoliInjectionsTotal}{40{,}420{,}330}
\newcommand{\nonToffoliLogicalMeasurementsVariable}{N_{\mathrm{non\text{-}Tof}}}
\newcommand{\nonToffoliLogicalMeasurementsTotal}{369{,}882{,}166}
\newcommand{\peakMeasurementParallelismVariable}{k_{\mathrm{cat}}^{(\mathrm{peak})}}
\newcommand{\peakMeasurementParallelismTotal}{24}

\newtheorem{lemma}{Lemma}
\newtheorem{proposition}{Proposition}

\begin{document}

\title{Computing 256-bit elliptic curve discrete logarithms in 26 days on
\\a fault-tolerant trapped-ion quantum computer with 20,000 qubits}

\author{Thomas H{\"a}ner}
\thanks{These authors contributed equally to this work.}
\author{Felix Tripier}
\thanks{These authors contributed equally to this work.}
\author{Jacob Young}
\thanks{These authors contributed equally to this work.}
\author{Michael Naehrig}
\author{Andrii Maksymov}
\author{Safwan Alam}
\author{Dmitri Maslov}
\author{Matthew Parrott}
\author{Yvette de Sereville}
\author{Jordan Sullivan}
\author{Mark Webster}
\author{Nicolas Delfosse}
\author{John Gamble}
\author{Martin Roetteler}
\affiliation{IonQ Inc.}
\date{\today}

\begin{abstract}
    One of the strengths of our recently proposed \emph{Walking Cat} Architecture~\cite{tripier2026walkingcat} for a trapped-ion quantum computer is that it is straightforward to extend and optimize for a specific application.
    As a proof-of-concept, here we present such optimizations for solving the $256$-bit elliptic curve discrete logarithm problem (ECDLP)
    on $\bitcoincurve$, which is the elliptic curve used by blockchain technologies such as Bitcoin, using Shor's algorithm.
    We optimize the circuits from Schrottenloher's recent work~\cite{schrottenloher2026optimized} and arrive at a logical quantum circuit
    for solving the ECDLP using about $1450$ qubits and $40\cdot 10^6$ Toffoli gates,
    with a rigorous lower bound on the logical-level success probability that holds with confidence at least $1-2^{-128}$, as well as a heuristic estimate thereof.
    Using our compilation toolchain in combination with manual optimization of the logical layout and integrated routing, we produce estimates for the logical measurement depth and the required number of physical qubits
    by compiling all components to measurement schedules that obey the architectural constraints.
    A key ingredient is a fast \CCZ{} magic-state factory and a depth-one \CCZ{} state injection, reducing the execution time of \CCZ{} gates by a factor of 31.
    We increase the logical-measurement parallelism using non-overlapping cat-based measurements in parallel, and we leverage the recently proposed logical CliNR protocol to speed up Clifford operations.
    To reduce the qubit overhead, we introduce a more efficient loss correction protocol, design a layout that allows us to recycle the CliNR ancilla qubits, and provision reusable cat-state resources according to the circuit's peak measurement parallelism.
    All results and optimizations combined, we conclude that a trapped-ion quantum computer based on our architecture would be able to solve the ECDLP on $\bitcoincurve$ in approximately 25.7 days using 19,397 physical qubits with an estimated success probability of $63\%$.
\end{abstract}

\maketitle
\null\clearpage
\tableofcontents
\clearpage

\newpage

\part{Introduction and overview}
\label{part:introduction}
\section{Introduction}

Quantum computers have the potential to deliver substantial speedups over their classical counterparts
for certain computational problems in cryptography, chemistry, materials science, data science, differential equations, and
optimization~\cite{shor1997, von2021quantum, lee2021even, ZGAA+2026, liu2022prospects, beverland2022assessing, dalzell2025quantum, IRSG+2026, RJIG+2026, chen2025framework, ZARG+2026, babbush2025grand},
and are thus expected to accelerate the process of scientific discovery~\cite{alexeev2021quantum}.

Resource estimates for some of these applications have been used to measure progress of the field, because such estimates capture improvements
in quantum algorithm design, circuit optimization, fault-tolerant quantum computing architecture and quantum error correction.
For example, the number of physical qubits required for factoring 2048-bit numbers using Shor's algorithm~\cite{shor1997} has dropped by three to four orders of magnitude over the last 14 years~\cite{ekeraa2017quantum,may2019quantum,gidney2021factor,chevignard2025reducing}, from approximately $10^9$ qubits in 2012~\cite{jones2012layered, fowler2012surface} to
$10^6$ qubits with 2D grid connectivity~\cite{gidney2025factor} or $10^5$ qubits when assuming long-range connectivity~\cite{webster2026pinnacle}.

Similar progress has been made for solving the elliptic curve discrete logarithm problem (ECDLP).
Since the initial estimates by Proos and Zalka from 2003~\cite{proos2003shor} with a gate count of $6\cdot 10^9$, a series of works~\cite{roetteler2017quantum,haner2020improved,litinski2023compute} has
reduced the resource requirements (without amortizing over multiple instances~\cite{litinski2023compute}) to about $200 \cdot 10^6$ Toffolis~\cite{litinski2023compute} in 2023, and finally
to $1462$ logical qubits and $84\cdot 10^6$ Toffoli gates~\cite{schrottenloher2026optimized} in 2026.
In addition, recent work has also explored reducing the space requirements and managed to achieve logical qubit counts of about $1200$---using more space-efficient adders~\cite{schrottenloher2026optimized,babbush2026securing,Cuccaro2004RippleCarryAdder,takahashi2010addition} or a compression of the output~\cite{chevignard2026reducing}---and even below $850$~\cite{luo2026quantum} using location-controlled arithmetic, but the resulting Toffoli counts are (in the latter two cases significantly) larger.

Moreover, the 256-bit curve $\bitcoincurve$ admits additional optimizations, because its prime is pseudo-Mersenne. As a result, gate counts for solving the ECDLP on $\bitcoincurve$ can be lower than those for other elliptic curves~\cite{babbush2026securing,schrottenloher2026optimized}: about $1462$ qubits and $60\cdot 10^6$ Toffoli gates have been shown to be sufficient~\cite{schrottenloher2026optimized}.
Recent high-level resource estimates suggest that solving this problem might require under 20,000 physical qubits on a neutral-atom quantum computer~\cite{cain2026shor}, or fewer than 500,000 physical qubits on a superconducting quantum computer~\cite{babbush2026securing}.

In this work, we design a specialized architecture for a fault-tolerant trapped-ion quantum computer capable of 
solving the ECDLP on $\bitcoincurve$ in 26 days using approximately 20,000 qubits, whereas previous estimates for solving this problem with a trapped-ion architecture required between 1.2 million qubits~\cite{baek2025sdqc} and 9.4 million qubits~\cite{litinski2023compute}.
Our work goes beyond previous resource estimation studies~\cite{haner2020improved,babbush2026securing,litinski2023compute,schrottenloher2026optimized} and architecture proposals~\cite{gouzien2023performance,cain2026shor}
in several ways. In particular, we make the following contributions:
\begin{enumerate}
    \item We develop and describe a detailed trapped-ion architecture for the ECDLP, including a compiler and complete descriptions of all logical operations, quantum error-correcting codes, and error-correction circuits. This design is compatible with the Walking Cat architecture~\cite{tripier2026walkingcat}. To our knowledge, it is the first ECDLP architecture for trapped ions based on quantum LDPC codes.
    \item We reduce the Toffoli count of Schrottenloher's circuits~\cite{schrottenloher2026optimized} from $2^{25.78}\approx$ 58 million to 39 million Toffoli gates at 1457 logical qubits.
    \item We derive rigorous bounds on the success probability of the entire quantum algorithm, taking into account approximation errors (phase as well as computational basis errors) in arithmetic circuits.
    \item We reduce the execution time of Toffoli gates in the Walking Cat architecture by a factor of 31 by introducing a Toffoli state factory, a fast Toffoli injection circuit, and by increasing the cat-based logical-measurement parallelism.
    \item We quantify and take into account the overheads associated with the layout of arithmetic components and the allocation and transport of ancilla qubits consumed by logical operations.
    \item We implement a compiler to provide accurate depth estimates and physical qubit counts by compiling all high-level algorithmic components to executable programs for our architecture, replacing all naive Clifford decompositions with integrated logical routing and getting exact depth counts.
    \item We consider the impact of qubit transport, leakage, loss, and qubit reloading on the logical operation time and noise, which is significant for atomic qubits and often ignored in resource estimations~\cite{webster2026pinnacle, litinski2023compute, cain2026shor}.
    \item We derive conservative estimates for logical operation times based on a detailed structure of the syndrome extraction circuit, whereas some previous estimates assume a syndrome extraction time of 1ms for any code which may lead to over-optimistic time estimates~\cite{webster2026pinnacle, litinski2023compute, cain2026shor}.
    \item We account for all logical operations including Toffoli gates, Clifford gates, and logical measurements which constitute a significant fraction of the total runtime but are often omitted from other resource estimates~\cite{babbush2026securing, webster2026pinnacle, cain2026shor}. For example, half of the cost of a Gidney adder~\cite{gidney2018halving} lies in the measurement-based uncomputation of temporary ANDs, a cost not captured by its Toffoli count.
\end{enumerate}

Throughout this work, we prioritize simplicity over exhaustive optimization.
For example, increasing the amount of classical postprocessing or executing
more instances of the quantum algorithm, potentially in parallel, could reduce
the size of the quantum circuit required for each instance%
~\cite{ekeraa2019revisiting,chevignard2026reducing}.
In the compiler, we use a single adder family throughout the application rather
than switching between width- and depth-optimized adders according to the
available scratch space and runtime requirements. At the architecture level,
we similarly keep the device configuration fixed throughout the computation
rather than dynamically reallocating qubits between memory and gate resources.
Finally, we report a conservative upper bound on runtime by treating every logical
measurement as lasting as long as a Toffoli injection. We likewise find a lower bound on
the success probability by assuming that all 69 memory blocks remain
active throughout the computation, although the lifespan of memory is generally shorter.
These choices leave room for further optimization, but we forgo changes that
offer only marginal gains in favor of a simpler design that is easier to
analyze and verify.

\section{Overview}

\begin{figure}[t]
    \centering
    \captionsetup{hypcap=false}
    \includegraphics[width=0.82\linewidth]{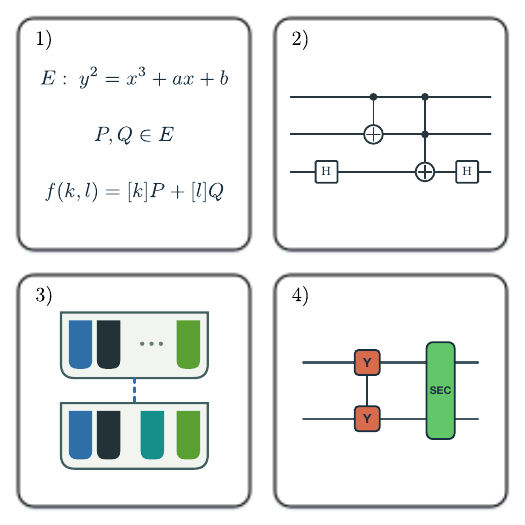}
    \caption{Overview of the paper in terms of an end-to-end pipeline of
    increasingly lowered representations of
    Shor's algorithm for the ECDLP. Stage~1 gives the high-level mathematical
    description of the algorithm (\cref{part:ecdlog}). Stage~2 
    describes it in terms of a high-level logical quantum circuit
    (\cref{part:point_addition}). Stage~3 compiles these circuits into executable,
    architecture-legal measurement schedules (\cref{part:compiler}). Panel~3
    illustrates the assignment of logical qubits to memory blocks across
    successive measurement layers: each box represents a memory block, each
    colored rectangle represents a logical qubit, and the dashed line denotes a
    logical-qubit move between memory blocks performed by integrated routing.
    Stage~4 lowers the schedules to native operations on the QEC architecture
    using its logical instruction set architecture (ISA)
    (\cref{part:architecture}).}
    \label{fig:paper-organization}
\end{figure}

\subsection{Architecture}
Our ECDLP-specialized architecture is derived from the Walking Cat Architecture~\cite{tripier2026walkingcat}, 
which is a general-purpose architecture for fault-tolerant quantum computing with trapped ions~\cite{cirac1995quantum, monroe1995demonstration, kielpinski2002architecture, monroe2014large, bruzewicz2019trapped}.
The Walking Cat Architecture relies on a 2D grid of trapped ions, each encoding a qubit, which can be 
transported across the grid, leveraging electronic qubit control~\cite{malinowski2023wire}, to establish long-range connections between qubits.
Transport enables the use of quantum LDPC codes~\cite{breuckmann2021quantum}, making the architecture more efficient than surface-code-based architectures~\cite{dennis2002topological, fowler2012surface}.
We assume two-qubit gates with a noise rate of $10^{-4}$, which has been demonstrated experimentally on small trapped-ion
devices~\cite{hughes2025trapped}, and single-qubit operations with a noise rate of $10^{-5}$, which is also within reach of
current trapped-ion devices~\cite{loschnauer2025scalable}.

A naive implementation of Shor's algorithm for solving the ECDLP on the general-purpose Walking Cat
architecture~\cite{tripier2026walkingcat} is possible but faces the following challenges. 
(1) The magic factories of \cite{tripier2026walkingcat} produce $T$ states (more precisely $H$ states) 
with a logical error rate of $7 \cdot 10^{-8}$ at best. However, solving the ECDLP on $\bitcoincurve$ using our circuits would require
39 million Toffoli gates
consuming a total of 273 million $T$ states. Here we assume the implementation~\cite{selinger2013quantum} for exact Toffoli (which uses seven $T$ gates and seven \CNOT{} gates) rather than the relative-phase-Toffoli-based four-$T$ constructions of Jones~\cite{jones2013low} and Maslov~\cite{maslov2016advantages}, since the latter
are more challenging to compare costs to as they require an additional logical ancilla. The $T$ state logical error rate is not sufficient 
to achieve a high probability of success for the algorithm.
(2) Each $T$ state injection consumes a logical measurement which takes an average of about 40ms 
and a \CNOT{} gate takes an average of about 90ms (assuming 3 logical measurements per \CNOT{} gate, 5 SECs per logical measurement and 6.7ms per SEC).
If they are implemented sequentially, this would result in a runtime of at least
910ms per Toffoli gate and 411 days in total.
One could increase the number of magic factories, but this only gives a small improvement
because the Toffoli-width of this algorithm is limited to one, which limits its $T$-width. Indeed, Babbush et al. assume sequential Toffoli gates
in their resource estimation~\cite{babbush2026securing}.
(3) Storing 1500 logical qubits requires a large number of physical qubits. 
For example, storing the logical qubits using a distance-9 surface code would consume 243,000 qubits, excluding 
the cost for logical operations and magic factories.
One could use the Walking Cat architecture code Q70 encoding 6 logical qubits into 70 data qubits, but it would
still require 55,000 qubits split into 250 memory blocks, ignoring the cost of cat factories and magic factories.
A reasonable option in terms of qubit count is the code Q102, which encodes 22 logical qubits into 102 data qubits 
(using 102 ancilla qubits for syndrome extraction and 102 beacon qubits for loss correction), 
consuming about 20,000 qubits to store 1500 logical qubits.
(4) Using a dense code like Q102, which encodes many logical qubits in each block, may reduce logical gate parallelism because
it is challenging to perform multiple logical operations at once in the same block given an arbitrary Clifford frame.
One can consider using the recently proposed logical CliNR scheme~\cite{webster2026fast} to simplify Clifford
frames, but making CliNR available
for each Q102 memory block would require two extra blocks of Q102 for each memory block, adding a total of 40,000 
qubits.

To address challenges (1) and (2), we design a two-level magic factory producing magic states with logical error
rate below~$10^{-9}$.
Our magic factory directly produces \CCZ{} states in a quantum LDPC code as opposed to manufacturing Toffoli gates from $T$ gates. This results in a dramatic reduction
in runtime of the quantum algorithm, which is dominated by Toffoli gates once layout and routing have been optimized.
To make both \CCZ{} state production and injection faster, we include multiple cat factories per block and we perform parallel
logical operations using disjoint logical operators.
This allows us to build a fast \CCZ{} factory and a depth-1 \CCZ{} injection circuit, which
we use to perform Toffoli gates as a result of the local Clifford equivalence of \CCZ{} and Toffoli gates.
Using fast Toffoli injections and Clifford frame tracking results in a Toffoli gate time of 29.5ms, which is faster than the naive implementation by a factor of 31.
Four Toffoli factories consuming a total of $1276 = 4 \cdot 319$ qubits are enough to produce Toffoli states with
an average production time under 28ms, providing a continuous flow of Toffoli gates.
The challenge (3) is addressed by optimizing the Walking Cat architecture to make use of the dense code Q102.
We use the logical CliNR scheme, as suggested in (4), but to avoid the extra cost of 40,000 qubits, we pipeline 
the implementation of the logical CliNR scheme, leveraging the structure of the quantum algorithm and its limited
Toffoli width. We find that 4 extra Q102 blocks are sufficient, a reduction of 39,200 qubits compared to the naive implementation of CliNR.
Moreover, we further reduce the qubit count by exploiting the limited Toffoli width of the quantum algorithm. 
Because only a relatively small number of Toffoli gates can be executed in parallel, many cat states remain idle for 
prolonged periods of time. Instead of filling all the cat state factories with qubits, we use a restricted number
of cat qubits which are transported to different cat factories as needed. We minimize the number of cat qubits
required through careful pipelining and routing.
Finally, we improve the underlying qubit loss model of the Walking Cat architecture and introduce a new leakage and loss correction unit.
Together, these improvements remove the need for beacon qubits, eliminating
about one third of the qubits in each memory block.

\subsection{Algorithm and implementation}

At the level of logical circuits, we present two optimizations for the quantum circuits by Schrottenloher~\cite{schrottenloher2026optimized}
that result in a reduction of the Toffoli count from $58\cdot 10^6$ to $39\cdot 10^6$ at a logical qubit count of $1457$.
Namely, we employ an alternative representation for the in-place multiplication~\cite{khattar2025verifiable,schrottenloher2026optimized} in which controlled additions become
conditionally-inverted additions---which are about half as costly in terms of non-Clifford gates~\cite{gidney2018halving}---and we build a modular squarer
for pseudo-Mersenne primes based on a single Karatsuba split~\cite{karatsuba1962multiplication} and
a squarer-specific optimization of Litinski's non-modular multiplier~\cite{litinski2024quantum}.

Because the resulting point addition circuits are approximate in the sense that they produce correct outputs and no $(-1)$-phase errors for most, but not all, inputs from $\F_p$, we derive bounds on the probability of success of individual components and the entire double-scalar multiplication.
To enable a tight translation of these component-level success probability bounds to the failure probability of the entire algorithm, we use random offset masks. Because we sample fresh masks for every run of the algorithm, we are able to prove an upper bound $p_f$ on the failure probability for each fixed input scalar pair.
In turn, this per-input bound can be turned into a lower bound on the probability of success of Shor's algorithm that is at most $4p_f$ lower than the success probability given an exact implementation of the double-scalar multiplication.
Previous work~\cite{proos2003shor,schrottenloher2026optimized,babbush2026securing,gidney2025factor} used heuristics and informal fidelity or shot-failure arguments, while bounds via state distance~\cite{chevignard2025reducing, bernstein1993quantum} lead to an $\Ocal(\sqrt{p_f})$-bound on the reduction in the probability of success.

We then implement and map these logical circuits to our architecture. We describe our internal compilation toolchain with its verification pipeline, which we use to obtain executable programs for each high-level component in the algorithm, producing rigorous space and depth estimates.
To reduce the space-time footprint of the components (adders, multipliers, etc.), we optimize the assignment
of logical qubits to the different memory blocks of our architecture.
In our final depth estimates, we take into account the overhead due to different such assignments.
Unlike prior gate-count resource estimates that abstract away Clifford and routing costs,
we compile every charged component to an executable, architecture-legal measurement
schedule with routing included.
Thanks to our \emph{integrated routing}, i.e., the process of merging cross-block routing operations between
two components with those same components, we are able to reduce the routing overhead to $5\%$ of the total runtime.

\subsection{Results and Outline}

We combine our accurate depth estimates with (1) the logical measurement times provided by our proposed architecture, (2) the rigorous lower bound on the success probability of the algorithm, and (3) the estimated probability of a logical error.
This allows us to conclude that our proposed architecture, with around 20,000 physical qubits, can solve the ECDLP on
$\bitcoincurve$ in approximately \(25.7\) days per attempt. Using Mosca's rigorous lower bound~\cite{mosca1999quantum}
as a starting point, the single-run success probability estimate is \(40.7\%\);
using Eker\r{a}'s heuristic success probability estimate~\cite{ekeraa2019revisiting} as the starting point gives \(63.3\%\).

\Cref{tab:component-accounting} lists the depth for each component, its contribution to the total of all
$28$ windowed point additions, including the semi-classical inverse quantum Fourier transforms (iQFT),
as well as the total runtime in hours and the device's physical-qubit footprint.

\begin{table*}[hbt]
  \caption{Algorithm component-level layer and runtime accounting for solving the ECDLP on $\bitcoincurve$
  using Shor's algorithm at $0.02950\,\mathrm{s}$ per measurement layer. Each ordinary
  lookup is charged as ${\approx}236{,}599$ base layers plus an average of $24{,}483$
  stochastic-penalty layers ($261{,}082$ in total), and each phase-fix lookup is
  charged as $36{,}857$ base layers plus an average of $11{,}984$ stochastic-penalty
  layers ($48{,}841$ in total). The initial $2^{18}$ lookup is charged as $940{,}028$
  base layers plus an average of $150{,}513$ stochastic-penalty layers ($1{,}090{,}541$
  in total). CliNR is described in \cref{part:architecture}.
  The total expected runtime is $616.555\,\mathrm{h}$, or
  \textbf{25.690 days}. The physical-qubit footprint is $19{,}397$; see \cref{part:architecture}
  for the architecture discussion and \cref{subsec:footprint} for the footprint derivation.
}
   
  \label{tab:component-accounting}
  \begin{ruledtabular}
  \begin{tabular}{lrrrrr}
  Component & Layers/instance & Instances/window & Layers/window & Layers/28 windows & Runtime (h) \\
  \colrule
  Unary lookup (including penalty)         & 261{,}082 & 3 & 783{,}246     & 21{,}930{,}888 & 179.686 \\
  Approximate mod. subtract                & 912       & 4 & 3{,}648       & 102{,}144      & 0.837   \\
  In-place multiplication                  & 820{,}994 & 2 & 1{,}641{,}988 & 45{,}975{,}664 & 376.691 \\
  Approximate mod. add                     & 674       & 3 & 2{,}022       & 56{,}616       & 0.464   \\
  Square-subtract                          & 123{,}530 & 1 & 123{,}530     & 3{,}458{,}840  & 28.339  \\
  Approximate mod. negation                & 179       & 1 & 179           & 5{,}012        & 0.041   \\
  Unary phase-fix (worst case)             & 48{,}841  & 1 & 48{,}841      & 1{,}367{,}548  & 11.205  \\
  Routing between components               & 20        & 1 & 20            & 560            & 0.005   \\
  CliNR at turnarounds                     & 7.73      & 5{,}783 & 44{,}703 & 1{,}251{,}684  & 10.255  \\
  iQFT (per-window average)                & 408       & 1 & 408           & 11{,}424       & 0.094   \\
  \colrule
  28-window subtotal                       &           &   & 2{,}648{,}585 & 74{,}160{,}380 & 607.616 \\
  \colrule
  Initial lookup (of size $2^{18}$)        & 1{,}090{,}541 & --- & ---      & 1{,}090{,}541  & 8.935 \\
  Initial iQFT outside the window loop     & 408 & --- & ---          & 408            & 0.003 \\
  \colrule
  Total expected                           &           &   &               & \measurementLayerDepthTotal & 616.555 \\
  \colrule
  \multicolumn{5}{l}{Physical-qubit footprint} & 19{,}397 \\
  \end{tabular}
  \end{ruledtabular}
\end{table*}

The paper is organized along the lines of the pipeline described in \cref{fig:paper-organization}. In \cref{part:ecdlog}, we introduce Shor's algorithm for the ECDLP and discuss the success probability of the algorithm when using approximate arithmetic. In \cref{part:point_addition}, we present the arithmetic circuits and optimizations that we use to implement the double-scalar multiplication. We describe our compiler, the logical layout of components, the integrated routing and our verification pipeline in \cref{part:compiler}.
In \cref{part:architecture}, we present a quantum charge-coupled device
(QCCD) model for trapped ions based on an extension of the Walking Cat architecture~\cite{tripier2026walkingcat}.
Relative to the baseline architecture, it achieves significant spacetime savings
through a logical ISA tailored to Shor's algorithm for the ECDLP and
optimized handling of resource states for logic.
Finally, we conclude the paper with a discussion of our results in \cref{part:conclusion}.

\newpage
\part{Elliptic curve discrete logarithms}
\label{part:ecdlog}
In this part, we introduce Shor's algorithm for the elliptic curve discrete logarithm problem (ECDLP) and give a bound on its success probability when an approximate version of the oracle is used instead of an ideal implementation.

\section{Shor's algorithm for the ECDLP}
\label{sec:shor-ecdlog}
Along with the factoring algorithm, Shor~\cite{shor1997} presented an algorithm to compute discrete logarithms. The variant for solving the ECDLP was worked out by Proos and Zalka~\cite{proos2003shor}. Subsequent works explored various optimizations and space-time tradeoffs~\cite{roetteler2017quantum,haner2020improved,litinski2023compute,gouzien2023performance,schrottenloher2026optimized,babbush2026securing}.

Our discussion of Shor's algorithm for the ECDLP is restricted to the case of elliptic curves defined over large characteristic finite fields, and we work in a prime order subgroup. Let $E: y^2 = x^3 + ax + b$ for $a,b \in \F_p$ be an elliptic curve defined over the prime field $\F_p$. The set $E(\F_p) = \{(x,y)\in \F_p \times \F_p\mid y^2 = x^3 + ax + b\} \cup \{\Ocal\}$ of $\F_p$-rational points on $E$ together with the point at infinity $\Ocal$ is an abelian group with neutral element $\Ocal$. For an integer $k \in \Z$ and a point $P$, we use the notation $[k]P$ to denote the $k$-fold sum $P + P + \dots + P$ (for negative $k$, $[k]P = -[-k]P$). Let $P\in E(\F_p)$ be a point of prime order $r$, i.e., $[r]P = \Ocal$ and 
the subgroup $\langle P \rangle = \{[k]P \mid k \in \Z\}$ generated by $P$ has order $r$. The \emph{elliptic curve discrete logarithm problem (ECDLP)} in $\langle P \rangle$ is: given $Q\in \langle P \rangle$, compute $d\in \Z_r$ such that $Q = [d]P$.

The ECDLP is a special case of period finding and the hidden subgroup problem and uses the quantum Fourier transform (QFT). Naturally, an $r$-dimensional QFT would be used here, but in practice a $2^n$-dimensional one is preferred, where $n$ is the bitlength of $r$, $2^{n-1}\le r < 2^n$. In that setting, the algorithm computes the following quantum state $\ket \Psi$, a superposition over all scalars $k, l\in [0,2^n)$,
\[
\ket{\Psi}=\frac 1{2^n}\sum_{k,l\in [0,2^n)}\ket{k}\ket{l}\ket{[k]P+[l]Q}.
\]
After applying the inverse QFT to the first two registers, measuring these registers provides a pair of classical values from which the discrete logarithm can be recovered with high probability by classical postprocessing. Using the semi-classical Fourier transform~\cite{griffiths1996semiclassical}, it is possible to interleave the QFT operations on the scalar qubit registers $\ket{k}\ket{l}$ with the point addition operations and thus reduce the number of qubits needed to a single one that can be reused.
In this paper, we use a windowed version of the semi-classical QFT with window size $w=16$~\cite{haner2020improved,litinski2023compute} and we drop the final three windows in favor of additional classical postprocessing~\cite{litinski2023compute,ekeraa2019revisiting}.
In addition, we replace the first windowed point addition by a direct table lookup~\cite{babbush2026securing}, resulting in a total of 28 windowed point additions.

The efficiency and success probability of the overall algorithm is mainly determined by the efficiency and accuracy of the most costly operation, namely computing the quantum operation $U_f$ that represents the double scalar multiplication $f(k,l) = [k]P + [l]Q$ given the fixed points $P$ and $Q$ that constitute the ECDLP:
\[
U_f: \ket{k}\ket{l}\ket{\Ocal} \mapsto \ket{k}\ket{l}\ket{[k]P + [l]Q} = \ket{k}\ket{l}\ket{f(k,l)}.
\]  

This paper focuses on quantum circuits for $U_f$ specifically for the elliptic curve $\bitcoincurve$ used in Bitcoin. The curve is given in short Weierstrass form as $E: y^2 = x^3 + 7$, defined over the prime field $\F_p$, where $p$ is the pseudo-Mersenne prime 
\[p = 2^{256}-c,\quad c = 2^{32} + 2^{10} - 2^6 + 2^4 + 1 = 2^{32} + 977.
\]
The curve has prime order, which means that the number $r = \#E(\F_p)$ of points on the curve is given by $r = p+1-t$, where $t$ is the trace of Frobenius (here: $t=432420386565659656852420866390673177327$). We pick this curve for various reasons. First of all, it is used in the Bitcoin system and its security is relied upon to secure vast amounts of cryptocurrency. Concrete resource estimates for this specific curve are of high value. Second, recent resource estimates used to discuss the security of cryptocurrencies against quantum attacks by Babbush et al.~\cite{babbush2026securing} and the state-of-the-art circuits by Schrottenloher~\cite{schrottenloher2026optimized} of lowest cost are based on $\bitcoincurve$ as well (due to the special pseudo-Mersenne shape of $p$).

\section{Single-shot success probability}

This section considers the success probability of a single run of Shor's algorithm for the ECDLP.
All state-of-the-art implementations of Shor's algorithm have to deal with superpositions that contain a number of failed computations due to the use of optimizations such as approximate arithmetic~\cite{schrottenloher2026optimized}.
It is thus important to analyze how such failures affect the overall success probability.

Here, we bound the reduction in success probability due to the use of an approximate implementation of the double-scalar multiplication that produces wrong outputs or residual phases on a small fraction of inputs.

\subsection{Approximate oracle and success probability}
Previous work~\cite{proos2003shor,schrottenloher2026optimized} has found that the cost of the double scalar multiplication can be further reduced by introducing approximations beyond relying on the generic point addition formula, ignoring the few exceptional cases (such as point doubling and adding $\Ocal$). Our circuits also make use of such approximations as well as approximations that introduce phase errors, meaning that we do not prepare the ideal state
\[
    \ket{\Psi}=\frac 1{2^n}\sum_{k,l}\ket{k}\ket{l}\ket{[k]P+[l]Q} = \frac 1{2^n}\sum_{k,l}\ket{k}\ket{l}\ket{f(k,l)},
\]
but instead only an approximate state before we measure the input registers in the Fourier basis. We note that as in previous work~\cite{haner2020improved,litinski2023compute}, our actual implementation uses a windowed implementation of the double scalar multiplication with the semi-classical quantum Fourier transform~\cite{griffiths1996semiclassical}. Since both variants are equivalent, we will pretend for the theoretical analysis that we keep both input registers for the scalars $k$ and $l$, and that the inverse quantum Fourier transform is applied at the very end of the algorithm before we measure the input registers.

When preparing the approximate state instead of the ideal state above, we also introduce several random masks that are captured by the parameter $\omega$. We prepare a state where the double scalar multiplication output is shifted by a (random) elliptic curve point that only depends on $\omega, P, Q$. For each run of our Shor implementation, we sample a non-zero $a_0\sample \F_r^*$ and initialize the accumulator register with the point $[a_0]P \neq \Ocal$; we also shift our lookup tables by fixed curve points such that none of the table points is the point at infinity $\Ocal$, leading to an additional constant shift by a point $[\mu_P]P + [\mu_Q]Q \in E(\F_p)$ independent of $k$ and $l$. Therefore, we compute
\[
\ket{k}\ket{l}\ket{\Ocal} \mapsto \ket{k}\ket{l}\ket{[a_0+\mu_P]P + [\mu_Q]Q + [k]P + [l]Q}.
\] 
This avoids special treatment for the first point addition and the addition of table lookup points, where having $\Ocal$ as one of the summands constitutes an exceptional case for the addition formulas. It also randomizes the occurrence of exceptional cases and arithmetic failures for different input scalars $(k,l)$.
While such a shift does not affect the success probability of Shor's algorithm~\cite{proos2003shor}, we still assume that it is removed by subtracting the point $[a_0+\mu_P]P + [\mu_Q]Q$ before we apply the quantum Fourier transform, so that we can ignore it for the analysis. However, since this hypothetical (and exact) removal of $[a_0+\mu_P]P + [\mu_Q]Q$ operation acts as the identity on the input register, we do not implement it in practice.

We write $f_\omega(k,l)$ for the ideal function that computes $[a_0+\mu_P]P + [\mu_Q]Q + [k]P + [l]Q$ including the random masks, and $\widetilde f_\omega(k,l)$ for the function computed by our circuits using the approximate arithmetic (that we will explain in detail below), which computes an approximation of $f_\omega(k,l)$. Integer-valued versions of the circuits to compute $\widetilde f_\omega(k,l)$ are represented in Algorithms~\ref{alg:randomized-window-setup}, \ref{alg:windowed-double-scalar}, and \ref{alg:point-add-dialog-binary-gcd}.

The state we prepare before applying the inverse quantum Fourier transform can therefore be written as
\[
   \ket{\tilde\Psi_\omega} := \frac 1{2^n}\sum_{k,l}s_\omega(k,l)\ket{k}\ket{l}\ket{\widetilde f_\omega(k,l)},
\]
where $s_\omega(k,l)\in \{-1,1\}$ absorbs $(-1)$-phase errors.
We define the failure set
\begin{align*}
B_\omega = \{(k,l)\in & [0,2^n)\times [0,2^n) \mid \\
& f_\omega(k,l) \neq \widetilde f_\omega(k,l) \vee s_\omega(k,l) = -1\},
\end{align*}
capturing failures due to elliptic curve exceptional cases, arithmetic approximations, and phase errors. Failure means there is a residual phase (e.g., due to approximate phase-fixes) or the output is wrong in the computational basis (e.g., the gcd computation produced a wrong output). We assume for our analysis that no other errors are introduced such that $f_\omega(k,l)=\widetilde f_\omega(k,l)$ and $s_\omega(k,l)=1$ for $(k,l) \notin B_\omega$.

The random offset masks $\omega$ used for our point addition circuits randomize the inputs $(k,l)$ for which a given point addition fails. This allows us to prove that there exists a constant $p_f>0$ such that for every fixed $(k,l)$, the circuit fails with probability at most $p_f$ over the random masks $\omega$, i.e., we will compute $p_f$ such that for all $(k,l)\in [0,2^n)\times [0,2^n)$,
\[
    \Pr_\omega[(k,l)\in B_\omega]\leq p_f.
\]
By introducing a projector $\Pi_S$ onto the set of measurement outcomes for which postprocessing succeeds pulled back through the inverse QFT, the success probabilities in the ideal and approximate cases are, respectively,
\[
    p_S = \|\Pi_S\ket{\Psi}\|^2\,\text{ and }\,\tilde p_S = \mathbb E_\omega[\|\Pi_S\ket{\tilde\Psi_\omega}\|^2].
\]
Given a lower bound $p_{S, lb}$ such that $p_{S, lb}\leq p_S$, we prove (see~\cref{proof:masked-oracle-success}) that the expected success probability over our random offset masks is at least
\[
    \widetilde p_{S}\geq \left(\max\{0,\sqrt{p_{S, lb}} - 2p_f\}\right)^2.
\]

\subsection{Concrete estimates of the success probability of our implementation}

In our implementation, we choose accuracy parameters such as the number of bits after which we truncate carry propagation in the approximate modular reduction~\cite{schrottenloher2026optimized} in a way that guarantees that with confidence level at least $1-2^{-128}$, we have that the probability of failure of the entire sequence of $28$ point additions (instead of $32$, as in~\cite{litinski2023compute}) is at most
\[
    p_f \leq 0.0335.
\]
We arrive at this bound by combining a Monte Carlo-based bound (for the in-place multiplications), which gives rise to the $1-2^{-128}$ confidence level, with exactly computed bounds for the remaining components, as detailed in \cref{sec:failureanalysis}.
The bounds below all inherit this confidence level, but since the noncoverage probability of at most $2^{-128}$ is very small compared to the uncertainty in other aspects that affect the final resource estimate, we omit it.

To arrive at a bound for the success probability of a single run of Shor's algorithm using an approximate oracle,
we may use the success probability lower bound by Mosca~\cite[p. 58]{mosca1999quantum} for an ideal oracle implementation and power-of-two QFTs of
\[
    p_S\geq p_S^{\text{(Mosca)}} = \frac{r-1}{r}\left(\frac{8}{\pi^2}\right)^2 \approx 0.657
\]
as a starting point to find that
\[
    \tilde p_S\geq 0.553
\]
is a lower bound on the probability of success assuming no logical errors.

While Mosca uses $2$ padding qubits for each of the two scalar registers, one extra qubit per register is sufficient, since the order $r<2^n$ and thus the required phase accuracy of $\frac{1}{2r}$~\cite{mosca1999quantum} is achieved using $M=2^{n+1}$.
To match this setting, we increase the size of the table lookup replacing the first windowed point addition to $2^{18}$ and change the subsequent window boundaries by one scalar bit. Enlarging the first lookup has a small effect on the total runtime (see \cref{tab:component-accounting}) that is included in our runtime estimate, and re-indexing does not lead to any change in gate or qubit counts.
Classical postprocessing~\cite{ekeraa2019revisiting,litinski2023compute} is then used to recover the $48$ bits corresponding to the final $3$ phase estimation windows that we have removed from the algorithm.
It would be possible to remove more than $3$ windowed point additions by increasing the amount of classical postprocessing or the number of (parallel) runs of the quantum algorithm~\cite{ekeraa2019revisiting,chevignard2026reducing}, but we do not make use of these optimizations here to facilitate comparison to prior work.

As an alternative to the lower bound by Mosca, we may use Eker\r{a}'s heuristic success probability estimate~\cite{ekeraa2019revisiting} as a starting point. Assuming an exact oracle and the postprocessing described in Ref.~\cite{ekeraa2019revisiting},
\[
    p_{S}^{\text{(Eker\r{a})}} = 0.99,
\]
which results in an estimate of the algorithmic success probability when using our approximate oracle of
\[
    \hat p_S \approx 0.861.
\]
Finally, we take into account the probability of failure due to a logical error, \(26.45\%\), from \cref{subsec:logical-failure-probability}.

\textbf{End-to-end success probability estimates:} Taking into account the probability of a logical error, our implementation thus succeeds in a single run with probability \(40.7\%\) when using Mosca's lower bound and with probability \(63.3\%\) using Eker\r{a}'s postprocessing and heuristic success probability estimate.

The runtime and success probabilities quoted above are per attempt. In general, one may repeat the algorithm multiple times to increase the overall success probability,
or simply run the algorithm on multiple devices in parallel.
Assuming independent failures and the \(63.3\%\) single-shot success probability estimate, running the application on five identical devices in parallel increases the probability
that at least one succeeds from \(63.3\%\) to \(1-(1-0.633)^5=99.3\%\). In other words, one
may run the algorithm reliably without increasing the wall-clock runtime by using multiple devices.

\newpage
\part{Point addition circuits}
\label{part:point_addition}
The most expensive component of Shor's algorithm for computing elliptic curve discrete logarithms on a curve $E$ is the quantum operation that implements double-scalar multiplication 
\[
    U_f: \ket{k}\ket{l}\ket{\mathcal{O}} \mapsto \ket{k}\ket{l}\ket{[k]P+[l]Q},
\]
where $k,l \in [0,2^n)$ are $n$-bit scalars, $\mathcal{O}$ is the point at infinity, and $P$ and $Q$ are the two elliptic curve points given by the discrete logarithm instance, $Q = [d]P$.

As in previous works~\cite{haner2020improved,litinski2023compute,gouzien2023performance,babbush2026securing,schrottenloher2026optimized}, we implement a scalar multiplication $[k]P$ 
by repeated point addition, decomposing the scalar $k$ into a windowed binary representation  
$k = \sum_{i=0}^{L-1} k_i 2^{i\cdot w}$ with window size $w$, where $L = \lceil n/w \rceil$ and $k_i \in \{0,1, \ldots, 2^w-1\}$, i.e.,
\[
    [k]P = \sum_{i=0}^{L-1} [k_i 2^{i\cdot w}]P.
\]
For each windowing index $i \in \{0,1,\dots,L-1\}$, we classically precompute a table of $2^w$ points $[k_i 2^{i\cdot w}]P$ for $k_i \in \{0,1,\dots,2^w-1\}$. In the quantum scalar multiplication at the $i$-th window iteration, we load the corresponding table into a quantum register, and then add it to the accumulator register using a point addition circuit. After the scalar multiple $[k]P$ is computed into the accumulator, we proceed in an analogous way with the  addition of $[l]Q$, using the corresponding precomputed tables for a windowed decomposition of $l$.

While the high-level structure of our circuits is identical, we introduce several random offset masks as already described above: we initialize the accumulator point with a classically sampled random non-zero point and we add a classically sampled random offset point to each precomputed table, which changes the entries of our tables to avoid the point at infinity (see \cref{sec:failureanalysis}). Both provide sources of randomness that allow us to prove a tighter bound on the success probability. 

In this part, we describe the quantum circuits that we use for adding two elliptic curve points. We start with the recent construction by Schrottenloher~\cite{schrottenloher2026optimized}, which combines the register sharing idea from Proos and Zalka~\cite{proos2003shor} with the Dialog representation from Khattar et al.~\cite{khattar2025verifiable}, and optimized approximate arithmetic.
We make two changes to the construction:

First, we use a different representation of the computation of the binary gcd, in which controlled (modular) additions become conditionally-inverted (modular) additions, which are about half as expensive in our setting~\cite{gidney2018halving,schrottenloher2026optimized}. This is similar in spirit to the optimized schoolbook multiplication circuits by Litinski~\cite{litinski2024quantum}.

Second, we implement the modular squaring step using a specialized version of the Litinski multiplier~\cite{litinski2024quantum} for non-modular squaring and explicit modular reduction by the pseudo-Mersenne prime leveraging a single Karatsuba split~\cite{karatsuba1962multiplication}.

In addition, we make minor optimizations such as using measurement-based uncomputation of temporary predicates. This introduces phase errors in addition to arithmetic errors, and we take this into account when deriving the success probability of the algorithm.

All optimizations combined, we arrive at a windowed elliptic curve point addition circuit that uses approximately $1.196\cdot 10^6$ Toffoli gates (without the $3$ table lookups of size $2^{16}$), when targeting a final width of at most that in~\cite{schrottenloher2026optimized}: 1462 qubits including the 16 windowing qubits.
To arrive at these gate counts, we follow the same convention as in previous work~\cite{babbush2026securing} to allow for a direct comparison, i.e., we only count actually executed gates and we estimate the Toffoli count of the entire algorithm using the same formula (see below). The Toffoli count above corresponds to the largest observed gate count from 100 runs.
With a total of $28$ point additions~\cite{litinski2023compute,babbush2026securing,ekeraa2019revisiting}, this yields a gate count estimate for an entire run of Shor's algorithm of approximately $28(1.196\cdot 10^6 + 3\cdot 2^{16})\approx 39.0\cdot 10^6$ Toffoli gates.

\section{High-level implementation of point addition}
\label{sec:point-addition-overview}

We start with an abstract description of the point addition circuit. We follow the same high-level construction as Schrottenloher~\cite{schrottenloher2026optimized} and implement windowed elliptic curve point addition using three forward table lookups~\cite{gouzien2023performance}, immediately freeing the output registers after consumption using measurement-based uncomputation.
We ignore exceptional cases for point addition (showing below that their contribution is negligible), except the case where the classical point to be added is the point at infinity. This case must be handled, especially for small windowing parameters $w$, since a $2^{-w}$ fraction of the input table is $\Ocal$. However, instead of adding extra controls to the circuit as done in~\cite{schrottenloher2026optimized}, we choose a single offset point for each classically precomputed table that is added to each point in the table such that none of these points is the point at infinity $\Ocal$. This results in a constant shift by the sum of all table offset points, which can be subtracted at the end. Since a constant shift does not change the outcome of the algorithm~\cite{proos2003shor}, in practice, we opt to not remove it. 

We consider a single point addition of the accumulator point $A=(x_1, y_1)$  (which we initialize to a random nonzero $A=[a_0]P$ at the beginning) with a table point $T = (x_2,y_2)$. After the table lookup, we compute the coordinate offsets $(\Delta x, \Delta y) = (x_1-x_2, y_1-y_2)$ in-place and clear the lookup registers. We compute $\Delta y\mapsto \lambda$, where $\lambda = (\Delta x)^{-1}\Delta y$ using in-place multiplication~\cite{schrottenloher2026optimized,khattar2025verifiable}, but using conditionally-inverted additions instead of the controlled additions (or subtractions) used in Schrottenloher's implementation.

We then look up $3x_2$, update $\Delta x\mapsto \Delta x+3x_2=x_1+2x_2$, and clear the lookup register. We map 
$x_1+2x_2 \mapsto (x_1+2x_2-\lambda^2) = (x_2-x_3)$ and $\lambda \mapsto \lambda (x_2-x_3)=(y_3+y_2)$ using the same implementation of in-place multiplication as above.

We complete the point addition using a third and final table lookup, fixing the coordinate offsets, and uncomputing the table lookup with measurement-based uncomputation, using as input the XOR of all measurement masks from previous \(X\)-basis measurements of looked-up coordinates. For a detailed integer-valued representation of our circuits, see Algorithms~\ref{alg:randomized-window-setup}, \ref{alg:windowed-double-scalar}, and \ref{alg:point-add-dialog-binary-gcd}.

\section{Optimized binary gcd iteration}

In this section, we describe an optimization that we applied to the implementation of the binary gcd circuit recently published by Schrottenloher~\cite{schrottenloher2026optimized}, which internally uses the Dialog representation from Khattar et al.~\cite{khattar2025verifiable} for the two in-place multiplication steps in the elliptic curve point addition circuit.

For an approximate implementation, one may replace the comparison in the binary gcd by a shorter comparison that only considers $n_\text{cmp} < n$ most-significant bits~\cite{schrottenloher2026optimized} for a failure probability decreasing exponentially with $n_\text{cmp}$.
As a result, the most expensive components of each iteration (at application scale) are the controlled adders and the controlled register swaps.

We start by recalling the observation that controlled operations are sometimes more expensive than conditional adjoint operations, i.e., $U^c=\ket{0}\bra{0}\otimes \mathbbm{1} + \ket{1}\bra{1}\otimes U$ can be more expensive than $U^{(\dagger)^c}=\ket{0}\bra{0}\otimes U + \ket{1}\bra{1}\otimes U^\dagger$~\cite{litinski2024quantum, sanders2020compilation, wecker2015solving}. This is also the case for adders~\cite{litinski2024quantum} and we would thus like to replace the controlled adders by conditionally-inverted adders.

To achieve this, we use an alternative representation of the register values: Instead of storing $(u,v)$ during dialog construction and $(r,s)$ during dialog replay, we store $(u,\tilde v)$ and $(r,\tilde s)$ with
\[
    \tilde v := v + (1-a) u\text{ and } \tilde s := s - (1-a) r,
\]
where $a:=v[0]$ is the parity of $v$. We also require that $u$ is the odd register and we keep a persistent orientation qubit to store the orientation.

In each iteration, we may compute the next parity bit $a_\text{new}$ from the current registers using that $\tilde v$ and $u$ are always odd by
\[
    a_\text{new} = v_\text{new}[0] = \frac{(\tilde v - u)}2 [0] = \tilde v[1] \oplus u[1],
\]
where the third equality holds because $\tilde v - u = v + (1-a)u - u = v - a u = 2v_\text{new}$ in the usual representation. A correct iteration must result in $u_\text{new}=u$ and
\[
    \tilde v_\text{new} = v_\text{new}+(1-a_\text{new})u_\text{new}.
\]
Since $v_\text{new} = \frac{(\tilde v - u)}2$, we may subtract $u$ (and then halve) if $a_\text{new}=1$ (meaning that $\tilde v_\text{new} = v_\text{new})$. If $a_\text{new}=0$, we need $\tilde v_\text{new} = v_\text{new}+u$, so we may add $u$ (and then halve). In summary,
\[
    \underbrace{v + (1-a)u}_{\tilde v} + (-1)^{a_\text{new}} u = \underbrace{v - au}_{2v_\text{new}} + \delta_{a_\text{new},0}2u,
\]
which can be rewritten as
\[
    \frac{\tilde v + (-1)^{a_\text{new}} u}2 = v_\text{new} + (1-a_\text{new})u = \tilde v_\text{new}. 
\]
A similar calculation shows that the update for the $r,s$ registers during dialog replay will be
\[
    \tilde s_\text{new} = \tilde s - (-1)^{a_\text{new}} r\;(\operatorname{mod} p).
\]
Therefore, by switching to this representation, we may replace the conditional (modular) adders by conditionally-inverted (modular) adders. For pseudo-Mersenne primes, conditionally-inverted modular adders may be implemented efficiently using conditional one's complement before and after the approximate modular adder described in~\cite{schrottenloher2026optimized}. One slight modification is necessary for the approximate modular adder that handles the $x+y=p$ case from~\cite{schrottenloher2026optimized}: In the addition branch, we conditionally swap the values $0\leftrightarrow p$ using approximate comparisons before invoking the conditionally-inverted adder.

To see why bitwise complement is sufficient for pseudo-Mersenne primes $p=2^n-c$, note that it acts as a correct modular negation followed by a non-modular shift by $c-1$, which makes a fraction of $\frac{c-1}p$ inputs non-canonical. For canonical inputs, the approximate modular adder works correctly with high probability in our setting (see \cref{sec:failureanalysis}), and the final bitwise complement removes the shift again.

After the conditionally-inverted adder of each iteration, we have to also update the orientation of the $u,\tilde v$ registers for the next iteration. If $v_\text{new}$ is odd, i.e., $a_\text{new} = 1$, we need to swap if $u > v$. Since $v=\tilde v$ in this case, this is equivalent to swapping the registers conditionally on $m_i:=a_\text{new} \land (u_\text{new}>\tilde v_\text{new})$.

\section{Modular squaring}
\label{sec:squaring}

For the modular squaring step in the elliptic curve point addition, i.e., implementing
\[
  \ket{x}\ket{y}\mapsto \ket{x}\ket{(y-x^2)\bmod p},
\]
we start with the non-modular, integer schoolbook multiplier from~\cite{litinski2024quantum}, which uses a sequence of conditionally-inverted adders (instead of controlled adders), and derive a squaring-specific version with roughly half the non-Clifford gate cost.

We then use a single Karatsuba split~\cite{karatsuba1962multiplication} together with explicit approximate modular reductions, making use of the special form of the pseudo-Mersenne prime $p$. For computing the squares of smaller $128/129$-bit integers, we use the optimized (non-modular) squarer.

\subsection{Optimized squaring}

Given an $n$-qubit input register $\ket x$ and a $2n$-qubit output register $\ket{z=0}$, the goal is to map $\ket{x}\ket{z=0}\mapsto\ket{x}\ket{x^2}$.

We start by writing the square of $x = \sum_{i=0}^{n-1}x_i2^i$ as
\begin{align*}
  x^2 & = \sum_{i=0}^{n-1} x_i 2^i\sum_{j=0}^{n-1} x_j 2^j\\ 
  & = \sum_i 2^{2i}x_i + \sum_{i=0}^{n-1}\sum_{j>i} 2^{i+j+1}x_ix_j\\
  & = \sum_i 2^{2i}x_i + \sum_{i=0}^{n-1}2x_i 2^{2i+1}\underbrace{\sum_{j>i} 2^{j-i-1}x_j}_{=:h_i}.
\end{align*}
Since we would like to use conditionally-inverted adders instead of controlled adders, we need to bring this into a form prefactored by $(2x_i-1)$ instead of $2x_i$, so
\[
  x^2 = \sum_i 2^{2i}x_i + \sum_{i=0}^{n-1}(2x_i-1) 2^{2i+1} h_i + \sum_{i=0}^{n-1} 2^{2i+1} h_i.
\]
The final sum in the expression above can be rewritten as
\[
  \sum_{i=0}^{n-1} \sum_{j>i} 2^{i+j}x_j = \sum_{j=0}^{n-1} x_j \sum_{i=0}^{j-1} 2^{i+j} = \sum_{j=0}^{n-1} x_j (2^{2j}-2^{j}).
\]
Substituting this back in and folding the result into the middle sum yields
\[
  x^2 = -x + \sum_{i=0}^{n-1}(2x_i-1) 2^{2i+1} (h_i+x_i),
\]
since $(2x_i-1)x_i=x_i$ for $x_i\in\{0,1\}$. From this expression, we see that squaring can be implemented using an initial $(n+1)$-bit subtraction of $x$ followed by a sequence of $n$ operations that add $h_i+1$ if $x_i=1$ and subtract $h_i$ if $x_i=0$. Starting with the least-significant bit allows us to keep these operations shorter than the full $2n$ output register. Instead, these operations modify a sliding window of length $n+1,n,...,2$ moving the higher end of the window by one position toward the most-significant bit in each step.

We implement these operations as conditionally-inverted additions, similar to~\cite{litinski2024quantum}, but instead of inverting the output, we bit-invert the input $h_i$ (padded to match the length of the active output window) conditionally on $x_i$, mapping 
\[
  h_i \mapsto 2^{n+1-i} - 1 - h_i
\]
if $x_i=1$. We then perform a subtraction unconditionally, before undoing the conditional inversion of the input register. As a result, the needed $+1$-offset for the $x_i=1$ case is handled automatically and does not need additional non-Clifford gates.

Before each of the first $(n-1)$ operations computing the sum above, we sign-extend the product register using a \CNOT{} gate (from index $n+i$ to $n+i+1$). The final step $i=(n-1)$ can be implemented directly using a single \CNOT{} gate.

Assuming we use the adder from~\cite{gidney2018halving}, which needs $w-1$ Toffoli gates to add two $w$-qubit numbers, the Toffoli count of this squarer is
\[
  T(n) = n + \sum_{i=0}^{n-2} (n-i) = \frac{1}{2}n(n+3)-1,
\]
saving approximately $50\%$ of the non-Clifford gates compared to a direct use of the multiplier from~\cite{litinski2024quantum}.

\subsection{Modular square subtraction}

To realize the operation $\ket{x}\ket{y}\mapsto \ket{x}\ket{(y-x^2)\operatorname{mod} p}$, our circuit uses one Karatsuba decomposition as follows. We treat input values representing $n$-bit registers as integers less than $2^n$ that are not necessarily canonical representatives less than $p$. Also intermediate results are represented by integer values of their respective bit sizes.

Let $n$ be even and $B=2^{n/2}$. We use the notation $x_{[t,t+\ell)}=\lfloor x/2^t\rfloor\bmod2^\ell$ to denote the $\ell$-bit slice of the non-negative integer $x$ starting with bit $t$. Write $x_L = x_{[0,n/2)}$ and $x_H = x_{[n/2,n)}$ for the integers representing the low bits and high bits of $x$. If $x<2^n$, then $x=x_L + Bx_H$, with $0\le x_L,x_H < B$. Since $B^2=2^n\equiv c\pmod p$, we have 
\begin{align*}
-x^2 & =-B^2x_H^2-2Bx_Hx_L-x_L^2\\
&\equiv(B-c)x_H^2+(B-1)x_L^2 \\
&\qquad -B(x_H+x_L)^2\pmod p.
\end{align*}
Using the squaring operation from the previous section, we then compute the exact integer square $x_L^2$ in an $n$-qubit auxiliary register, then add $Bx_L^2$ to $y$ modulo $p$, subtract $x_L^2$, and uncompute the register. This yields $y+(B-1)x_L^2\bmod p$ in the $y$-register. The analogous steps for $x_H$ add $(B-c)x_H^2$. After that, we compute the $(n/2+1)$-bit sum $x_L+x_H$ in place, square it into an $(n+2)$-qubit scratch space, and subtract $B(x_L+x_H)^2$ from $y$ modulo $p$. We then uncompute the square and restore $x_H$ by subtracting $x_L$. 

Algorithm~\ref{alg:karatsuba-square-sub} in Appendix~\ref{sec:algorithms} gives the integer-level operation order in detail. It uses the components $\AddSlice_{\sigma}^{q}$, which invokes the phase correction $\PhaseGE^{q}$, as well as $\AddBMultiple_{\pm}^{q}$ and $\AddCMultiple_{-}^{q}$. These operations are given in Algorithms~\ref{alg:add-slice}, \ref{alg:phasege}, \ref{alg:add-b-multiple}, and~\ref{alg:add-c-multiple} in Appendix~\ref{sec:algorithms}.

\newpage
\part{Compilation toolchain}
\label{part:compiler}
% Listing style for compiler code excerpts.
\definecolor{compilerCode}{HTML}{003B4F}
\definecolor{compilerFunction}{HTML}{4758AB}
\definecolor{compilerString}{HTML}{20794D}
\definecolor{compilerComment}{HTML}{5E5E5E}
\definecolor{compilerAccent}{HTML}{AD0000}
\definecolor{compilerBackground}{HTML}{F8F9FA}
\definecolor{compilerBorder}{HTML}{DEE2E6}

\lstdefinestyle{compilerpython}{
    language=Python,
    basicstyle=\ttfamily\footnotesize\color{compilerCode},
    keywordstyle=\color{compilerCode}\bfseries,
    keywordstyle=[2]\color{compilerAccent},
    keywordstyle=[3]\color{compilerFunction},
    commentstyle=\color{compilerComment},
    stringstyle=\color{compilerString},
    identifierstyle=\color{compilerCode},
    emphstyle=\color{compilerFunction},
    emph={define_ports,QubitPort,define_ancilla,CleanAncilla,max,tuple,list,
          visit_row,append,interleaved_cycle,LayoutPlan},
    morekeywords=[2]{True,False,None},
    morekeywords=[3]{staticmethod,dataclass,isinstance,range,dict},
    backgroundcolor=\color{compilerBackground},
    rulecolor=\color{compilerBorder},
    breaklines=true,
    breakatwhitespace=false,
    columns=fullflexible,
    keepspaces=true,
    frame=single,
    framerule=0.4pt,
    framesep=5pt,
    xleftmargin=0.4em,
    xrightmargin=0.4em,
    aboveskip=0.8\baselineskip,
    belowskip=0.8\baselineskip,
    captionpos=b,
    floatplacement=htbp,
    showstringspaces=false
}

\section{Compiler overview}

In this part, we describe our logical compiler, which, for a given quantum algorithm,
synthesizes quantum circuits for the target logical instruction set architecture (ISA).
Specifically, we compile the elliptic-curve circuits of \cref{part:point_addition}
to the specialized Walking Cat architecture of \cref{part:architecture}.

For our fault-tolerant architecture, the allowed instruction set consists of
logical Pauli measurements, tracked Clifford updates,
and magic-state consumption. From these instructions, the compiler produces an executable,
architecture-legal measurement schedule with exact resource accounting.
We separately generate measurement schedules for each algorithm-level component and
compose them through end-to-end integrated routing at $5\%$ overhead.
\Cref{tab:component-accounting} provides a total cost breakdown per component.
Lowering to native operations
on the system graph (transport, physical gates, and readout) happens below
the compiler, at runtime, using the library of ISA primitives.

Optimally compiling to such an ISA becomes intractable at the scale of hundreds of logical qubits, because
the target algorithms inflate to millions of operations, and implementation choices compound
while having to adhere to architectural constraints.
Compilers that flatten the algorithm into a gate-level netlist
discard the structure this optimization requires, leaving only local rewrites,
while high-level resource-estimation frameworks preserve the structure
but do not lower it to executable instructions.
We bridge this gap with a hierarchical approach built from reusable compilation primitives
that are selected, lowered, and scheduled with full knowledge of the target architecture.

The compiler uses the hierarchical instruction model of Ref.~\cite{tripier2026walkingcat}, similar to those of
Qualtran~\cite{harrigan2024qualtran,qualtran2026},
QREF~\cite{qref2024,qref_format}, and ProjectQ~\cite{steiger2018projectq}, as well as the classical hierarchical compilation framework MLIR~\cite{lattner2021mlir}. Within this model, the architecture-based
library of operations can be extended with higher-level logical components.
Each component has a \emph{spec} defining
its semantic inputs, outputs (its \emph{ports}), and parameters, and one or more \emph{defs}
describing decompositions into other smaller component specs. Low-level gates are components
that are legal on the architecture, and high-level algorithmic structures like square-subtract
are comprised of many abstraction levels of progressively smaller defs and specs.
Any component may be conditioned on a classical outcome. Defs declare their
own ancillae, so alternative implementations can expose different depth,
width, non-Clifford-cost, and layout tradeoffs without changing
the inputs and outputs of the component. This system lets the compiler select among reusable and
hand-optimized implementations
according to the architecture and the surrounding circuit context as it recursively
lowers the root component into a synthesis tree~\cite{tripier2026walkingcat}.

The rest of this part follows the following structure.
\Cref{sec:compiler-architecture} describes the architectural constraints we compile against,
and the scheduling routine that enforces them.
\Cref{sec:compiler-components-routing} describes layouts of algorithmic components 
and the integrated routing that connects them.
\Cref{sec:compiler-verification} describes how we prove defs realize their advertised
spec using independent ``contracts'', and how compiled schedules are checked against
resource and architecture constraints.

\section{Architectural constraints and scheduling}
\label{sec:compiler-architecture}

Specs and defs are stateless descriptions of the component hierarchy:
they record what a component does, how it can be decomposed, and its semantic action
on its qubits (like returning any declared ancilla to $\ket{0}$)
but not the assignments of logical qubits to algorithmic qubits, binding of magic states
to factories, or whether a def's child specs must be decomposed further on the architecture.
That state is tracked during circuit synthesis, where we bind qubits and resources,
enforce architectural constraints, and schedule operations.

\subsection{Architectural constraints}
Our architecture, described in \cref{part:architecture}, provides the following capabilities
and constraints to our compiler:
\begin{enumerate}
    \item Memory blocks use the Q102 code and provide 22 logical qubits.
    \item Entangling Cliffords within a memory block may be executed via frame tracking.
    Empirically validated for our structured arithmetic circuits and discussed in
    \cref{sec:depth-one-ccz-injection}.
    \item Each memory block can participate in up to three logical measurements meeting
    compatibility conditions. As above, validated and discussed in
    \cref{sec:parallel-cat-state-measurements,sec:depth-one-ccz-injection}.
    \item The processor can perform a single \CCZ{} injection per layer, with factory-side
    cleanup taking place the following layer. The lone exception, the \CCZ{} gates injected during
    lookup tables, have consequences discussed in \cref{sec:depth-one-ccz-injection}.
    \item The processor provides $24$ memory-side cat bundle pairs.
    \item The processor provides $69$ memory blocks.
    \item A structured logical CliNR sweep performs logical frame clearing, chasing after
    structured serial arithmetic (discussed further in \cref{sec:idle-frame-clearing}).
\end{enumerate}

\subsection{Scheduling}
As the compiler decomposes the root component into a tree, it ensures that we meet both
memory block capacity and leaf operation legality constraints. We then emit the decomposed
tree serially into the measurement scheduling routine, which greedily packs measurements
into layers, with each layer followed by frame-tracked Clifford updates. The scheduler enforces
per-block parallelism limits, schedule-time measurement compatibility conditions, \CCZ{} injection capacity,
and global memory-side cat bundle parallelism over the course of the complex scheduling
procedure we describe below.

\begin{algorithm}[H]
\SetAlgoCaptionLayout{raggedright}
\DontPrintSemicolon
\SetAlFnt{\footnotesize}
\SetArgSty{textnormal}
\caption{Measurement scheduling routine. \textbf{Notation:} For a measurement
$M$, \texttt{floor} is $M$'s causal lower bound; \texttt{seed} is its viable
starting placement; $i$ is the layer under examination for placement;
\texttt{best} is the earliest viable placement layer found; and
$\operatorname{CanPlace}(M,i)$ tests architecture constraints: injection count,
cat bundle parallelism, block capacity, and disjoint measurement compatibility.}
\label{alg:measurement-scheduling}
\ForEach{leaf emitted from tree}{
  \uIf{(leaf is a measurement $M$)}{
    Find \texttt{floor} and \texttt{seed} for $M$\;
    $(\texttt{best}, i) \gets (\texttt{seed}, \texttt{seed})$\;

    \While{($i > \texttt{floor}$)}{

      Propagate and transform $M$ through prior Cliffords\;
      \If{(a conditional Clifford changes the basis of $M$)}{
        \textbf{break}\;
      }
      \ElseIf{($M$ crosses an anticommuting conditional Pauli)}{
        XOR that condition to its virtual meaning\;
      }

      \If{($M$ anticommutes with a measurement in layer $i$)}{
        \textbf{break}\;
      }

      \If{(parallel injection measurements stop being compatible)}{
        \textbf{break}\;
      }

      \If{($\operatorname{CanPlace}(M,i)$)}{
        $\texttt{best} \gets i$\;
      }

      $i \gets i - 1$\;

    }
    Place $M$ in \texttt{best}\;
  }
  \ElseIf{(leaf is a Clifford)}{
    \If{(it is conditioned on virtual measurements)}{
      Resolve virtual measurements\;
    }
    Place it\;
  }
}
\end{algorithm}

\section{Algorithmic components and integrated routing}
\label{sec:compiler-components-routing}

We leverage layout plans (\cref{sec:compiler-layout-plans}) to provide high-level strategic
reasoning to the compiler. The layout plan suggests optimized mappings of algorithmic to
logical qubits and may optionally force the decomposition of a spec into a chosen def.
Each algorithmic component is compiled and optimized for its own authored layout plan
designed to reduce the measurement-layer depth on the target architecture. The algorithm, however, may have to execute algorithmic components
with incompatible layouts in sequence, and indeed most algorithmic components have consecutive
child components with incompatible layouts. Connecting them efficiently is the
integrated-routing problem addressed in this section.

\subsection{Routing an adder}

\begin{figure}
      \centering
      \includegraphics[width=\linewidth]{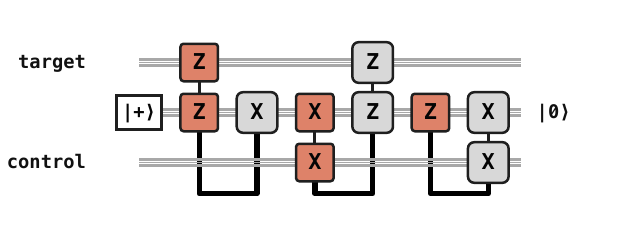}
      \caption{An arbitrary Pauli-controlled-Pauli (PcP) gate can be decomposed using
      three measurement layers and an ancilla, as shown for a \CNOT{} gate. The first
      measurement is akin to copying the first qubit onto a clean ancilla, providing the
      intuition for our measurement-based integrated routing operations.}
      \label{fig:cnot-3layer-decomp}
\end{figure}

A component's layout maps the component's ports to qubits in one or multiple memory blocks.
We can minimize the measurement depth required to execute a component by leveraging
entangling Clifford tracking within memory blocks and exploiting parallelism
between different blocks. To show why this is powerful, consider the traditional
decomposition of an entangling Clifford, or a Pauli-controlled-Pauli gate (see \cref{fig:cnot-3layer-decomp}).
This protocol requires three measurements and an ancilla to form the link between qubits. Through careful
routing, we manage to avoid this decomposition for every single \CNOT{} or \CZ{} gate in our algorithm.
\Cref{fig:gidney-adder-overview} shows an adder with its natural layout to demonstrate how our architecture can best
support a simple structured def.

\begin{figure*}[t]
      \centering
      \begin{subfigure}{\textwidth}
            \centering
            \includegraphics[width=0.8\textwidth]{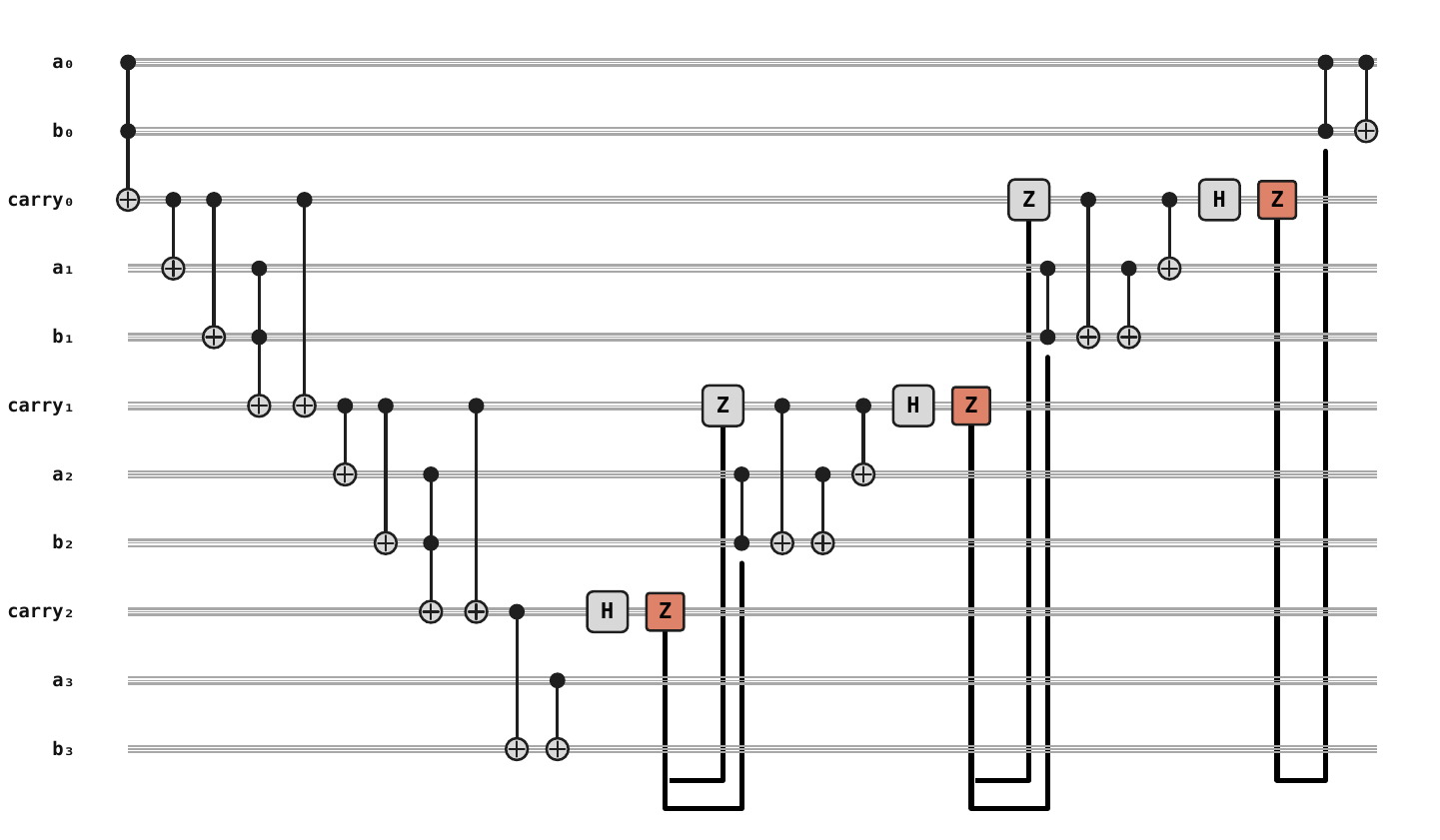}
            \caption{Quantum circuit of a 4-bit Gidney adder with measurement-based uncomputes.
            Note that entangling Cliffords either act on pairs of qubits $a_i$ and $b_i$
            or link $c_i$ to $a_{i+1}$ and $b_{i+1}$, and Toffoli gates similarly act on triples of
            $a_i$, $b_i$, and $c_i$.}
            \label{fig:gidney-adder-4bit-circuit}
      \end{subfigure}

      \medskip

      \begin{subfigure}{\textwidth}
            \centering
            \includegraphics[width=\linewidth]{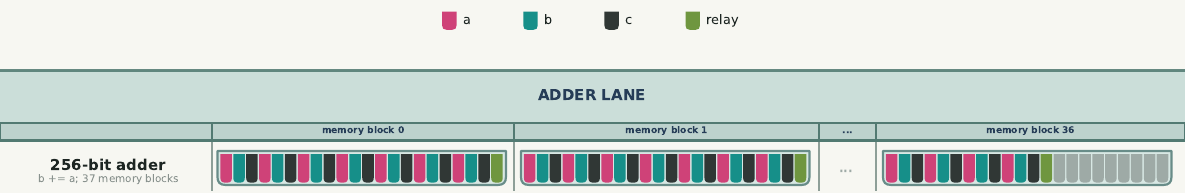}
            \caption{Layout plan for the Gidney adder. We designate a run of successive
            memory blocks as an adder lane, where we interleave $a_i$, $b_i$, and $carry_i$.
            This fits as seven triples per memory block, with the 22nd qubit serving as a
            hole for relaying data cheaply when a carry $c_i$ is needed for \CNOT{} gates
            targeting qubits in the next block.}
            \label{fig:gidney-adder}
      \end{subfigure}
      \caption{Gidney adder layout and circuit.}
      \label{fig:gidney-adder-overview}
\end{figure*}

The addition of two quantum registers is at the core of the arithmetic of
\cref{part:point_addition}. The adder specification declares two $n$-qubit semantic registers
$\ket{x}$ and $\ket{y}$ and adds the first into the second
modulo $2^n$.
Distinct defs realize that specification with different resource trade-offs: a ripple adder in
the style of Cuccaro et al.~\cite{Cuccaro2004RippleCarryAdder}, a carry-lookahead
adder~\cite{DraperKutinRainsSvore2004QCLA}, a Gidney temporary-AND
adder~\cite{gidney2018halving}, an ancilla-free Takahashi
adder~\cite{takahashi2010addition}, a logarithmic-depth Remaud--Vandaele
adder~\cite{remaud2025ancillafree}, and hybrids between them. Controlled and
constant-operand adders live alongside as their own families with different semantic specs.

For simplicity, routing tricks, and consistency of our logical frame analysis, we leverage Gidney-style
ripple-carry adders, even for constant adders (often at little or no cost, as in
\cref{fig:phase-approx-mod-adder}). An optimized layout of a Gidney adder is shown in
\cref{fig:gidney-adder}: packing each $x_i$, $y_i$, and $carry_i$ in successive qubits
places the classically controlled \CZ{} gates from both \CCZ{} injections and the measurement-based uncomputation
(MBU) of the carries during the rippling uncompute section of the adder within memory
blocks, allowing them to be executed in software via frame tracking.
Similarly, nearly all \CNOT{} gates in the adder land within a single memory block and are naturally
executable via frame tracking. The remaining cross-block \CNOT{} gates are mitigated by integrated routing as follows:
whenever a given $carry_i$ controls \CNOT{} gates that would cross memory
blocks, we leverage the next memory block's 22nd qubit to quickly copy $carry_i$ forward.
The surprisingly lossless nature of this additional measurement is discussed further at
the end of the next section.

\subsection{Integrated routing}

While an individual low-level component may be mapped to memory blocks in a natural way
to maximally exploit entangling Clifford tracking, satisfying the natural layouts of
consecutive components without stalling the serial execution of the overall algorithm
is highly nontrivial. We develop integrated routing to address this challenge---in particular,
we leverage the serial rippling nature of the arithmetic and architectural-level measurement
parallelism to hide the cost of logically moving data between memory blocks.

There are a few basic data movement operations, each decomposed into a low number of measurements:
\begin{itemize}
      \item \textbf{Copy}: Create an entangled logical basis copy of a qubit with a single cross-block ZZ joint measurement.
      \item \textbf{Fanout}: Create many copies of a qubit across blocks by creating bell pairs with a layer of
      cross-block XX joint measurements and entangling them with the original with a layer of ZZ joint
      measurements, forming a two-layer measurement brickwork.
      \item \textbf{MBU}: Measurement-based uncomputation clears a qubit by measuring it in the \(X\) basis and applying
      a conditional \(Z\) Pauli correction to the qubit's source.
      \item \textbf{Move}: Move a logical basis qubit from one block to another by making a
      logical basis copy and MBUing the original; cf. Fig.~12 of Ref.~\cite{haener2022spacetime}.
      \item \textbf{Swap}: Swap two qubits through one or more temp qubits by moving qubits in a sequence.
      If we use temps that are already local, a swap may be thought of as two parallel moves
      plus locally-tracked Clifford swaps.
\end{itemize}
First, because the compiler can propagate conditional Paulis through measurements with
classical bookkeeping, the cost of an MBU from a \textbf{Move} is often hidden and the cost of
our movement operations is typically just the number of serial cross-block measurement layers.
Second, because we may perform up to three measurements per block in parallel, we can reconfigure
memory blocks for the next component as the rippling arithmetic of the prior component is
collapsing back to the low qubits. Third, logical CliNR (see
\cref{sec:idle-frame-clearing}) is performing this
same rippling sweep and may be leveraged to both incorporate local swaps and guarantee
that the moves are cheaply done with empty logical frames.
Fourth, logical CliNR incurs short stalls of approximately \(7.73\) layers at certain critical points
between components, which we already account for in our costs reported in \cref{tab:component-accounting}
and which give additional time for routing moves to complete. Altogether, integrated routing
allows us to push the layer depth of an algorithmic component very close to the serial
layer depth of the rippling arithmetic (typically, the number of Toffolis and their
associated rippling measurement-based uncomputations).

\subsection{Routing the approximate modular adder}

Here, we show a worked example of a modular adder, demonstrating how we leverage integrated
routing to approach the serial lower bound. We leverage an approximate modular adder
similar to that of~\cite{schrottenloher2026optimized} to compute
$\ket{x, y} \rightarrow \ket{x, y + x \bmod p}$, whose def has the following sequence of operations:

\begin{algorithm}[H]
\SetAlgoCaptionLayout{raggedright}
\DontPrintSemicolon
\caption{Phase-approximate modular adder}
\label{alg:phase-approx-mod-adder}
Compute $\ket{x, y} \rightarrow \ket{x, x + y}$ with a 256-bit Gidney-style carry out adder.\;
Unload the low 65 bits of $x$ into the work qubits in the middle of the adder.\;
Fan the carry out qubit onto the low-region relay qubits and use it as the
control for a controlled load of the pseudo-Mersenne modular correction $c$.\;
MBU the block-local copy of each carry out qubit to make room for the Gidney
adder's block-local relay qubit.\;
Use a 65-bit Gidney-style register adder to compute \ket{c, x + y} $\rightarrow$
\ket{c, x + y + c \text{ mod }2^{65}}.\;
MBU the local copy of $c$.\;
Return the low bits of $x$ from the packed middle region.\;
Run a phase comparator on the high 32 bits of $x$ and $y$ to correct the phase of all MBUs.\;
\end{algorithm}

The exact details of this sequence are shown in \cref{fig:phase-approx-mod-adder}.
Now, the serial lower bound of the rippling arithmetic of the adder is given by twice the
ripple depth of the 256-bit adder, plus twice the ripple depth of the 65-bit adder, plus
the ripple depth of the phase comparator. Note that we parallelize half of the phase comparator
against the closing ripple of the low-width constant adder. That gives a total lower bound of approximately
$510 + 128 + 31 = 669$, and our phase approximate mod adder compiles to a layer depth of
$674$. The logical CliNR penalty adds approximately \(7.73\) layers to the seam between the 256-bit adder and
the low-width constant adder (as discussed in \cref{sec:idle-frame-clearing}), causing enough delay
to fully hide the five-layer routing overhead.
This demonstrates the power of integrated routing and why leveraging Gidney-style register
arithmetic even for constant adders is nigh lossless in measurement depth. Our measurement
scheduler also demonstrates its capabilities in this example---the scheduling routine
achieves this depth, even while needing to relay data between memory blocks, by parallelizing
every relay alongside the very first \CCZ{} injection into a later block. While the missing
relay data in a given block prevents us from interpreting the results of the injection
measurements, applying tracked conditional \CZ{} gates, and then injecting the next \CCZ{} gate, we may still
free up the factory resource. Then only the relay between the final two memory blocks increases
the serial depth of the adder. Our scheduler finds this optimization automatically.

\begin{figure*}[t]
      \centering
      \includegraphics[width=\textwidth]{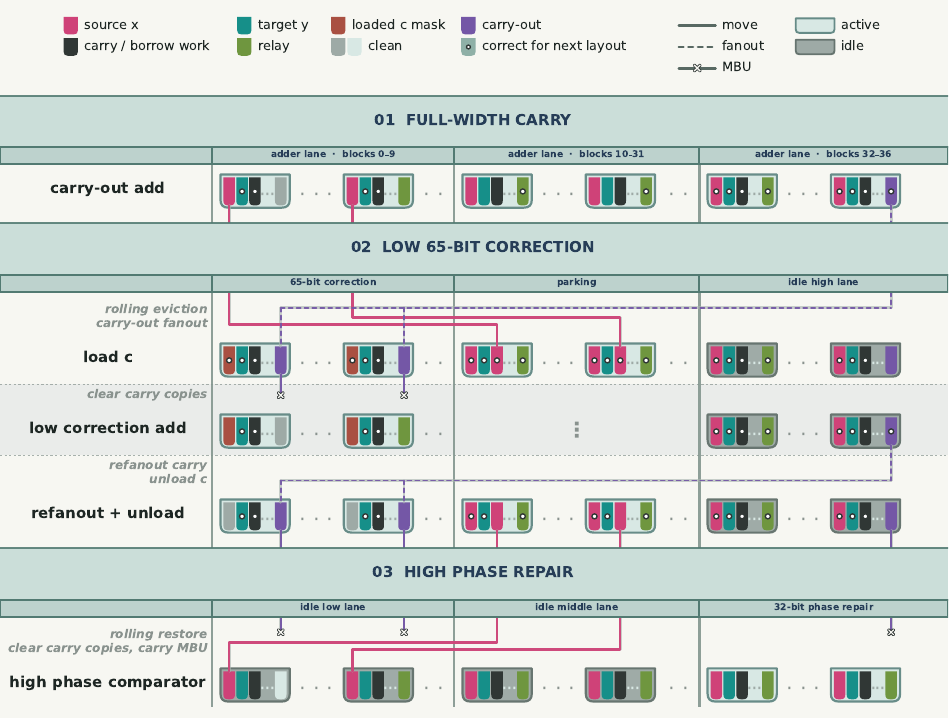}
      \caption{Routing of the phase-approximate modular adder as described in
      \cref{alg:phase-approx-mod-adder}.}
      \label{fig:phase-approx-mod-adder}
\end{figure*}

\subsection{End-to-end integrated routing for algorithmic components}

Here we give a brief description of the integrated routing strategy internal to each
algorithmic component documented in \cref{tab:component-accounting}, excluding the
phase-approximate modular adder.

\smallskip
\noindent\textbf{Lookup.}\enspace
Lookups consist of a depth-first traversal over the index bits, using scratch qubits to
help with partial answers. Putting scratch qubits and the lowest index qubit in the same
block keeps all traversal \CNOT{} gates block-local. To do the same for \CCZ{} injections, we cycle copies of
relevant index qubits into and out of the block with \textbf{copy} and \textbf{MBU} operations.
We pack the controlled mask-load \CNOT{} gates into blocks by leveraging \textbf{fanout} and \textbf{MBU}
to relay each control bit to the load target blocks in turn.
We further split the $k=16$ lookup into two parallel $k=15$ lookups,
following the zipper construction of Ref.~\cite{haener2022spacetime}, with copies of the index bits
and a second relay qubit in each target block.

\smallskip
\noindent\textbf{Approximate modular subtraction.}\enspace
This works similarly to the modular adders, except we perform a modular complement before
and after the adder. Modular complement is done by bitwise inversion using frame-tracked
$X$ gates and unconditional low-width constant addition of $2^{65} - c$, where $c$ is the
pseudo-Mersenne modular correction. We again pack and restore the low 65 bits of $x$
using the middle of the adder around each constant correction.

\smallskip
\noindent\textbf{In-place multiplication.}\enspace
Our implementation of the in-place multiplication~\cite{khattar2025verifiable,schrottenloher2026optimized}
uses Gidney adders~\cite{gidney2018halving} for the arithmetic on the $u$ and $v$
registers, which we place in the adder lane, and computes history bits nearby.
The adder lane shrinks as we reduce $u$ and $v$, freeing up more qubits that we
slowly load $y$ into. We compress each set of six local history qubits into five~\cite{schrottenloher2026optimized} using
a dedicated compute zone, and we pack compressed history qubits into long-term storage.
When $u$ and $v$ are reduced away, we unpack the rest of the $y$ register into the adder lane,
play back the history to perform in-place multiplication or division, and then repack $y$ while
working back up the $u-v$ sequence to uncompute the history. This is the widest portion of
the algorithm---at its peak, we have $68$ active memory blocks, plus one for the algorithm's
16 idling index bits.

\smallskip
\noindent\textbf{Square-subtract.}\enspace
Our square-subtract routine leverages a single-level Karatsuba split, as described
in \cref{sec:squaring}.
The routine accumulates squares of parts of a register $y$ into a product register, then
adds different alignments of the squared terms into the real accumulator $x$ before uncomputing each square.
We again leverage Gidney-based arithmetic---our central squaring routine
repeatedly shifts and performs conditionally-inverted addition. After each square of a 128-bit portion of $y$ (or the sum of the
high and low halves of $y$), the product register has been shifted nearly 128 times, bringing it
nearly in alignment to represent the shifted value $B\cdot\text{product}$. We may then exchange $y$ and $x$ to
accumulate $B\cdot\text{product}$ and any constant multiples $c\cdot\text{product}$, realized as shifted
adds or subtractions of $\text{product}$ at different offsets.

\smallskip
\noindent\textbf{Approximate modular negation.}\enspace
Approximate modular negation is nearly the usual approximate modular complement, with one correction---we
have a rippling chain of Toffolis, with a similar structure to an adder, used to detect the edge
case $x = p$ and replace $p$ with $0$~\cite{schrottenloher2026optimized}. Though the basic
routing details are covered in prior components, we mention that we use this routine to avoid
the traditional exact modular arithmetic used to finalize a round of windowed point addition.
This preserves our maximum circuit width while still allowing us to leverage Gidney-style
register addition throughout the entire circuit.

\smallskip
\noindent\textbf{Integrated Routing for windowed point addition.}\enspace
We compose integrated routing for an entire window of the point addition routine---the
total integrated routing penalty per window is nearly \textbf{zero}. We compile each of
the major algorithmic components separately and add their costs together, adding a
conservative cost of 20 layers per window. Moreover, layer depth reductions
from compressing the algorithmic components across their seams would far outweigh any
penalty from integrated routing.
Our routing strategy is best described visually in \cref{fig:point-add-route} within the Appendix.

Of note, because we unconditionally execute all
phase comparators, we have no conditionally-executed measurements to account for, outside of our
iQFT implementation (\cref{sec:t-state-factory-conversion}). Because we parallelize the
first half of each phase comparator against other rippling serial arithmetic elements
with the strategy shown in \cref{fig:phase-approx-mod-adder}, we achieve the same average cost as conditional execution.

\section{Verification}
\label{sec:compiler-verification}

Trusting the definition and compilation of fault-tolerant-scale circuits
is a hard problem. We address it by splitting the problem in two halves: first, prove
that the top-level component is decomposed into a correct tree, and second, prove
that the scheduler converts the tree to a valid program executable on our architecture.

\subsection{Hierarchical verification of the tree}

The only operation that synthesis performs on the component tree is replacing
some spec with one of its defs. This leads to a powerful insight---by verifying
that every single def implements its spec, we can trust every operation the tree
performs, and thereby trust the final tree's correctness. To do this, we leverage
another key insight: many components may be described without introducing genuine
quantum behavior like superposition or entanglement. In particular, for our circuits,
nearly every component is a quantum implementation of a reversible classical circuit,
enabling classical emulation as in Refs.~\cite{haener2016emulation,steiger2018projectq}.

For every spec, we specify a \emph{contract}: a machine-checkable reference for the semantic
action on the spec's ports. Implementation details like ancillae do not appear in the contract.
These contracts take the form of functions defined on some declared domain and
 come in three types of escalating complexity:
\begin{enumerate}
    \item \textbf{Basis contract}: the spec maps computational basis states to computational
    basis states. Reversible classical circuits fall into this category, including \CNOT{} gates
    and Toffolis. Contracts of this type are nearly trivial to verify classically.
    \item \textbf{Basis-plus-phase contract}: a two-dimensional variant of a basis contract.
    Now, alongside the basis mapping, the spec may also apply a phase function to any given
    basis state. The common occurrence in our circuits is MBUs and their associated phase
    correction components.
    \item \textbf{Operator contract}: a full operator, for components that create superpositions or
    entanglement. All measurement-based protocols eventually fall into this category, but
    the overall action of a small component relying on measurements is almost always in the
    prior two categories.
\end{enumerate}

For each def, the verifier expands its immediate children,
including its structured control flow, and composes their contracts in
execution order. It maps each comparison input, together with the
implementation ancillae, into the def's internal register layout, evaluates
the child composition, projects the result back onto the semantic ports, and
compares it with the spec contract. Clean ancillae must be returned to
$\ket{0}$. Child calls are also required to satisfy their own declared input domains.
If a spec has a basis contract but a given def has children with a basis-plus-phase contract,
verification proves that the phase functions cancel, reducing to a basis contract with
an identity phase function.

For the reversible arithmetic that dominates the ECC pipeline,
the contract is simply a function on basis values. The adder specification's
contract, for example, is $(x, y) \mapsto (x,\, x + y \bmod 2^n)$.
By representing all measurement-based protocols with simple
basis contracts (e.g., \textbf{move}, \textbf{swap}, \textbf{copy}, etc.), and MBUs and \CCZ{} injections
as basis-plus-phase contracts, we avoid opening defs whose children have operator contracts
directly. We independently certify that each small measurement-based primitive performs the
correct action.

The verification routine runs either exhaustively for components with a small input domain, or randomly
sampled over components with a large input domain, including approximate components.
For a two-bit adder def, for instance, the verifier sweeps all sixteen $(x,y)$ basis inputs
with the def's scratch at $\lvert 0 \rangle$, evaluates the composition of
the children's own contract functions, and requires the output
$(x,\, x+y \bmod 4)$ with the scratch returned to $\lvert 0 \rangle$.
In contrast, at the massive scale of our application, we randomly sample input states.
The approximate modular adder for the secp256k1 field, whose modulus is
$p=2^{256}-2^{32}-977$, is checked on 100{,}000 reproducibly sampled inputs. Such sampling
is both strong evidence of our implementation's accuracy over all ${\sim}2^{512}$
input pairs as well as affirmation that our approximate arithmetic holds at scale.
We verify both large- and small-scale versions of all components.

\subsection{Scheduling verification}

We prove our schedules accurately preserve the tree's semantics with a
combination of different checks and unit tests, and we perform these checks
at multiple levels of abstraction and record their results with the compiled data.
\begin{enumerate}
  \item We have a variety of tests for the pipeline itself, including but not limited
  to compiling Shor's integer factoring algorithm for 3- and 5-bit numbers to executable
  programs and simulating the resulting programs exactly.
  \item We verify that the compiler's schedule is \emph{physical}, that is,
  conforms to the set of constraints imposed by the architecture, including cat bundle
  consumption and magic-state injection limits.
  \item We bound resource figures with analytic
  envelopes derived from the algorithm; in a Gidney-style adder, for instance, we expect two
  measurement layers (\cref{sec:compiler-circuit-synthesis}) per Toffoli (to be exact, given
  our block-aware decomposition, $2n - 1$ layers, where $n$ is the number of bits in the adder).
  The envelope imposes a floor as well as a ceiling, both to assert that
  information-theoretic bounds are respected and to ensure we do not unknowingly
  introduce a performance regression.
  \item We ensure our routing applies correctly and we avoid naive decomposition of
  any cross-block entangling Cliffords.
  \item We run small-scale component-level
  simulations of the schedule, ensuring that our final programs preserve circuit semantics.
  This includes simulations of the 674-layer schedule of the 256-bit approximate
  modular adder discussed in \cref{sec:compiler-components-routing} with computational basis
  state inputs, as well as phase-aware simulations with superpositions of inputs at reduced
  scale.
\end{enumerate}

Every reported figure is produced by a reproducible pipeline. A run is
declared in a suite file that fixes all inputs: the entry-point component and
its parameters, the architecture, the layout plan, and the verification
checks to apply. Executing the declaration writes a run directory holding
the compiled result, its measurement schedule, the recorded outcomes of the
checks described above, and a manifest of the resolved arguments, including
the serialized architecture, the git commit of every repository involved,
and the execution environment. Runs are identified by a content hash over
the declared inputs: two executions of the same declaration on the same code
hash identically, so any reported number can be traced to the configuration
that produced it. Because the measurement schedule is retained alongside the
result, a reported depth can be re-examined without recompiling.

\newpage
\part{Architecture}
\label{part:architecture}
% Architecture section assembled from focused subsection files.
\section{Architecture overview}
\label{sec:architecture-overview}

\subsection{The \texttt{secp256k1} Device}
\label{sec:enhanced-architecture}

We consider the \(\bitcoincurve\) device, based on a modified version of the Walking Cat
architecture optimized for solving a \(\bitcoincurve\) discrete-logarithm instance.
Because the elliptic curve point addition circuits developed
in this work account for most of the algorithm's logical operations, we tailor
the architecture to their resource requirements rather than to those of other
stages, such as the quantum Fourier transform. The resulting
architecture enables execution of an end-to-end
\(\bitcoincurve\) discrete-logarithm instance with the resource totals
summarized in \cref{tab:ecc-architecture-summary}.

\begin{table}[ht]
    \centering
    \caption{Resource totals for solving \(\bitcoincurve\) with
    the \(\bitcoincurve\) device.}
    \label{tab:ecc-architecture-summary}
    \begin{tabular}{@{}lr@{}}
        \toprule
        Metric & Value \\
        \midrule
        Physical-qubit footprint & \(19{,}397\) \\
        Runtime per attempt & \(\sim 26\) days \\
        Logical failure probability per attempt & \(\sim 26\%\) \\
        \bottomrule
    \end{tabular}
\end{table}

The device contains 69 Q102 memory blocks, 24 cat-state bundle
pairs, 12 Bell-state bundles for LM2 measurements, four logical CliNR blocks,
and four magic factories, each with two local six-qubit C6 Bell-state sources.
Including local and global reloading reservoirs, these components
give the \(19{,}397\)-physical-qubit footprint summarized in
\cref{tab:ecc-architecture-components}.  Details and justification for runtime,
logical failure probability, and the physical-footprint total are provided in
\cref{subsec:runtime}, \cref{subsec:logical-failure-probability}, and
\cref{subsec:footprint}, respectively.

\begin{figure*}[t]
    \centering
    \includegraphics[width=0.82\textwidth]{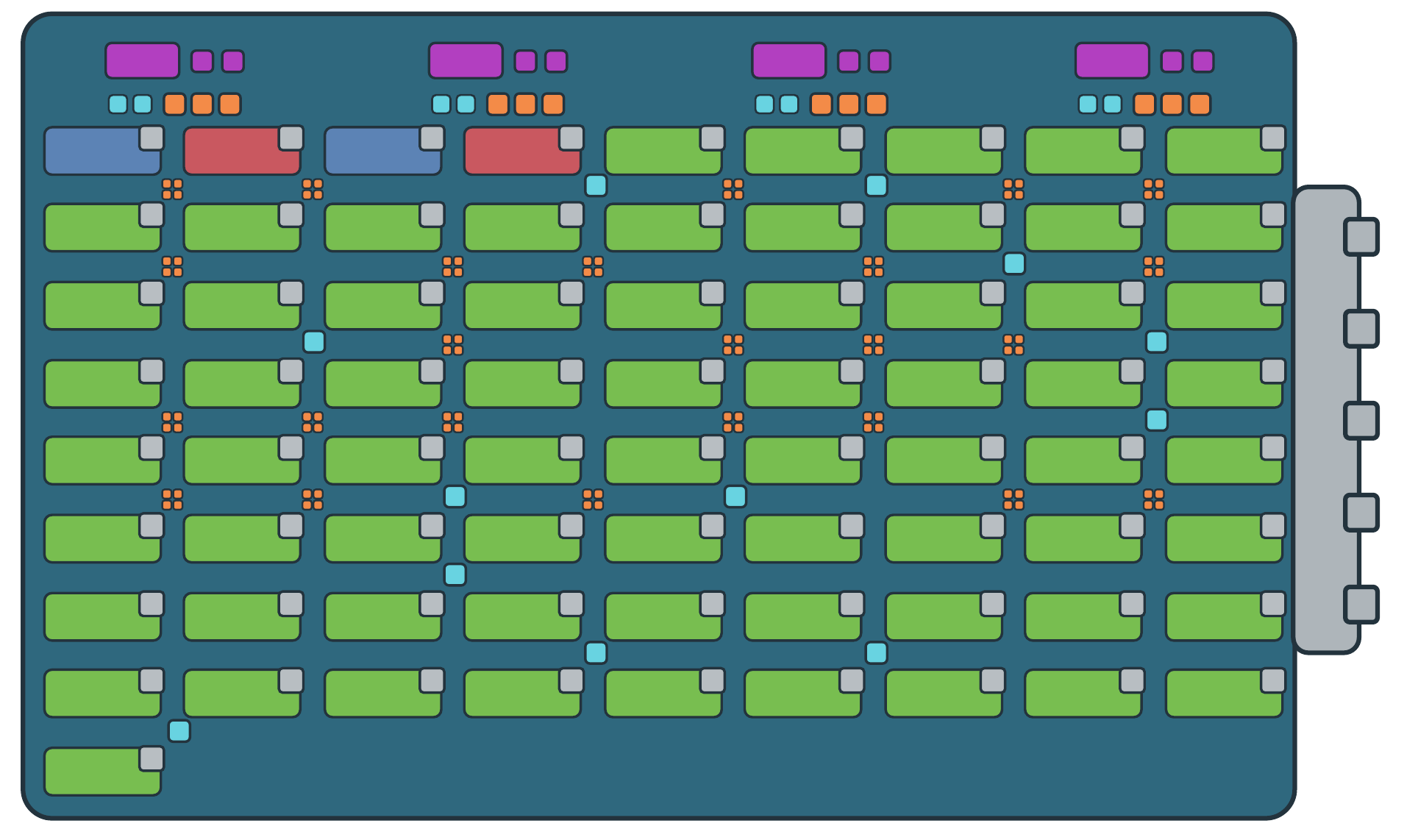}
    \caption{Full-scale \(\bitcoincurve\) device. Purple blocks are magic-state
    factories required for \CCZ{} state generation, green blocks are memory blocks,
    orange blocks are cat-state factories, and blue blocks are Bell-state
    bundles for LM2 measurements. Each group of four orange blocks
    forms a mobile cat-state bundle pair. The two cyan sites inside each
    magic-factory group are local six-qubit C6 Bell-state sources.
    The cobalt and crimson blocks are reserved for two Logical CliNR instances
    used for Clifford-frame clearing.
    The small gray box attached to each code block is its local ion reservoir,
    while the full loading reservoir runs along the right edge of the device.}
    \label{fig:full-scale-quantum-computing-architecture}
\end{figure*}

\begin{table}[ht]
    \centering
    \caption{Component allocations and physical-qubit footprint of the
    \(\bitcoincurve\) device.}
    \label{tab:ecc-architecture-components}
    \footnotesize
    \setlength{\tabcolsep}{3pt}
    \begin{tabular}{@{}lrrr@{}}
        \toprule
        Component & Count & \shortstack{Qubits per\\component}
                  & \shortstack{Qubits in\\allocation} \\
        \midrule
        Memory blocks
            & 69 & 207 & 14{,}283 \\
        \shortstack[l]{Cat-state bundle pairs}
            & 24 & 120 & 2{,}880 \\
        \shortstack[l]{Bell-state bundles}
            & 12 & 8 & 96 \\
        Logical CliNR blocks
            & 4 & 207 & 828 \\
        \shortstack[l]{Magic factories}
            & 4 & 319 & 1{,}276 \\
        Reloading-reservoir qubits
            & 34 & 1 & 34 \\
        \midrule
        \multicolumn{3}{l}{Physical-qubit footprint}
            & \(19{,}397\) \\
        \bottomrule
    \end{tabular}
\end{table}

The architecture uses the three-ring QEC codes summarized below to support its
memory and magic-state factory blocks. The Q102 code is from the
Walking Cat Architecture paper~\cite{tripier2026walkingcat}, Q66 is
a new generalized bicycle code introduced in this work, and C6 is Knill's
\([\![6,2,2]\!]\) code~\cite{knill2005quantum}. C6 can be expressed as a generalized
bicycle code, making it a three-ring code, and therefore compatible
with the Walking Cat architecture.

\begin{table}[ht]
    \centering
    \caption{Three-ring QEC codes used in the \(\bitcoincurve\) device.
    Logical error rates are quoted per SEC for Q102 and Q66, and per accepted
    encoded \(\ket{H}^{\otimes 2}\) pair for C6. Logical error rates are estimated 
    for the moving qubits model with physical error rate $p=10^{-4}$ including transport
    noise, leakage and loss.}
    \label{tab:architecture-qec-codes}
    \begin{tabular}{@{}lccc@{}}
        \toprule
        Code & Parameters & SEC depth (POCs) & Logical error rate \\
        \midrule
        Q102 & \([\![102,22,9]\!]\) & 27.0 & \(9.34\times10^{-12}\) \\
        Q66 & \([\![66,4,10]\!]\) & 17.0 & \(2.63\times10^{-11}\) \\
        C6 & \([\![6,2,2]\!]\) & 5.4 & \(1.30\times10^{-8}\) \\
        \bottomrule
    \end{tabular}
\end{table}

The \(\bitcoincurve\) device specializes the Walking Cat architecture around the
resource demands of the elliptic curve point addition circuits. In particular,
we prioritize
reducing its physical footprint while retaining enough measurement and
magic-state throughput to avoid introducing new runtime bottlenecks. This
specialization rests on four changes to the baseline architecture.

First, we introduce the CLAW ISA, an extension of the Walking Cat architecture
ISA~\cite{tripier2026walkingcat} that addresses bottlenecks in elliptic curve
point addition circuits, as detailed in \cref{sec:claw-isa}. By studying the
compiled logical
measurements, we find that each has an accessible physical representative in
the Q102 code even when every Clifford gate within a block is tracked
virtually. Consequently, Clifford gates account for less than \(10\%\) of our
final logical-measurement count. Because Toffoli gates account for the bulk of
the remaining measurements, we optimize their execution with a new parallel
logical-measurement operation, an idle-frame-clearing scheme, and a depth-one \CCZ{}
gate.

We introduce the lightweight Eastinthillation factory in \cref{sec:ccz-factory}, which supports
the depth-one \CCZ{} gate.

Then, we adapt cat-state distribution to the elliptic curve point addition
circuits' limited global logical-measurement parallelism. Rather than
permanently provisioning every code block
with both sets of cat states and verification ancillas, the device moves the
cat states required for logical measurements across the device as detailed in
\cref{sec:reusable-cat-state-factories}.

Finally, in \cref{sec:moving-qubit-extension}, we replace the baseline
architecture's qubit-loss model~\cite{tripier2026walkingcat} with the swap-loss
model, in which an interaction can move a loss but cannot spread it. This model
relaxes the hardware requirements associated with qubit loss: an implementation
of the \(\bitcoincurve\) device need neither guarantee which output retains the
surviving qubit nor ensure that an attempted interaction with a missing qubit
ejects or disables its partner. Preventing a single loss from branching also
permits lower-overhead loss-detection protocols than the beacon protocol
developed in~\cite{tripier2026walkingcat}. For a code block with \(n\) data
qubits, these new protocols reduce the physical-qubit count from \(3n\) under
the beacon protocol to \(2n\), thereby reducing the footprint of every memory
and factory block in the device.

\subsection{Swap-Loss Model}
\label{sec:moving-qubit-extension}

We replace the loss-propagation rule in the moving-qubit noise model of the
Walking Cat architecture~\cite{tripier2026walkingcat} with the
\emph{swap-loss model}. In this model, a two-qubit gate involving exactly one
lost qubit can exchange the locations of the loss and the surviving qubit but
cannot spread the loss. Because qubit loss no longer spreads, the number
of lost qubits arising from a single qubit-loss event no longer grows
exponentially with circuit depth. Consequently, loss need not be checked using
the beacon protocol after every scheduled layer; instead, we introduce a
modified \emph{leakage-detection unit} (LDU) that detects both leakage and loss
under the swap-loss model. The new loss-propagation rule is
specified in \cref{sec:swap-loss-model}; \cref{sec:swap-loss-ldu-sec} gives the
LDU and corresponding SEC procedure; and
\cref{sec:swap-loss-memory-performance} reports the memory performance of the
Q102 and Q66 code blocks in this model.

\subsubsection{Swap-Loss Model}
\label{sec:swap-loss-model}

The moving-qubit noise model of the Walking Cat
architecture~\cite{tripier2026walkingcat} supplements circuit-level Pauli noise
with independent qubit-loss and leakage processes. Let \(\ploss\) and
\(\pleak\) denote their respective rates per physical operation
cycle (POC). After a parallel round of
depth \(\Delta\) POCs, each non-lost qubit is marked lost with probability
\(\min\{1,\Delta \ploss\}\); among the remaining qubits, each qubit in
the computational subspace is marked leaked with probability
\(\min\{1,\Delta \pleak\}\). 
We retain the Walking Cat leakage rules: a single- or two-qubit gate involving
a leaked qubit is replaced with the identity operation while its associated
physical noise is still applied, and a subsequent loss changes the qubit's
status from leaked to lost.

In the original moving-qubit noise model, any non-measurement operation
involving a lost qubit is replaced with the identity operation, and any
two-qubit gate involving a lost qubit propagates the loss to the other qubit
involved in the gate.
Consequently, if every lost qubit participates in a two-qubit gate, the number
of lost qubits can double from one two-qubit-gate layer to the next.

The \emph{swap-loss model} changes only this loss-propagation rule. Suppose a
two-qubit gate acts on locations \(i\) and \(j\), exactly one of which is marked
lost. The intended two-qubit gate is replaced with the identity operation. With
probability \(\pswap\), however, the contents of the two locations,
including their tracked loss and leakage statuses and tableau entries, are
exchanged; with probability \(1-\pswap\), they remain unchanged.
With no loss
of generality, we set \(\pswap=0.5\) in our simulations; sensitivity analysis
of memory performance with the swap-resilient LDU shows that the logical error rate is
insensitive to the precise value of
\(\pswap\), as shown in
\cref{fig:q102-swap-loss-sensitivity}.

In simulation, this update is implemented by swapping the lost
and leaked status flags and applying a \SWAP{} gate to the corresponding
entries in the underlying tableau state.

The original Walking Cat architecture's loss propagation rule models a QCCD merge/split
operation between an occupied and an empty potential well, which can strongly
heat and ultimately eject the surviving ion. The swap-loss model instead assumes potential wells sufficiently deep to retain
the heated ion.

\subsubsection{Swap-Resilient Loss Detection and Syndrome Extraction}
\label{sec:swap-loss-ldu-sec}

Our goal is to construct a circuit that detects loss or leakage of a data qubit
with probability close to one, even when a loss can swap between the data and
ancilla wires. In the absence of loss or leakage, the circuit should preserve
the state of the data qubit under observation.

The teleportation-based leakage-detection unit (LDU) used in the Walking Cat
architecture~\cite{tripier2026walkingcat} assumes that the data and ancilla
are swapped/teleported to definite locations. The \SWAP-based and teleportation-based
syndrome-extraction schemes of Ref.~\cite{baranes2026leveraging} make the same
assumption. However, under the swap-loss model, a loss can move to an unmeasured location
while a surviving qubit moves to the measured location, so the loss may be missed.

We instead modify Preskill's leakage-detection
circuit~\cite{preskill1998reliable}. The issue with the original LDU is as follows. If the loss swaps
with the surviving ancilla qubit at both entangling gates, that qubit returns
to its original location in state \(\ket{1}\) (which normally indicates a healthy outcome), producing a
false negative. At \(\pswap=1/2\), this double-swap trajectory has
probability \(\pswap^2=1/4\), which is too large to neglect.

The \emph{swap-resilient LDU} shown in \cref{fig:swap-loss-ldu} adds an \(X_A\)
gate between the two entangling gates.
A computational ancilla outcome \(0\) accepts the data qubit, whereas
outcome \(1\), `lost', or `leaked' triggers recovery.
We assume throughout that augmented readout perfectly distinguishes
computational, leaked, and lost states.

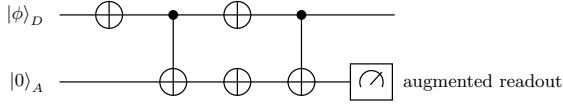
\begin{figure}[ht]
    \centering
    \resizebox{0.88\linewidth}{!}{%
        \begin{tikzpicture}[
    x=1cm,
    y=1cm,
    wire/.style={line width=0.6pt},
    target/.style={
        circle,
        draw,
        fill=white,
        line width=0.6pt,
        minimum size=0.46cm,
        inner sep=0pt
    },
    control/.style={circle, fill=black, inner sep=1.7pt},
    meas/.style={
        draw,
        fill=white,
        line width=0.6pt,
        minimum width=0.72cm,
        minimum height=0.62cm,
        inner sep=0pt
    }
]
\newcommand{\ldutarget}[2]{%
    \node[target] (#1) at #2 {};
    \draw[wire] (#1.north) -- (#1.south);
    \draw[wire] (#1.west) -- (#1.east);
}

\draw[wire] (0,0) -- (5.75,0);
\draw[wire] (0,-1.15) -- (4.99,-1.15);

\node[anchor=east] at (-0.15,0) {$\ket{\phi}_{D}$};
\node[anchor=east] at (-0.15,-1.15) {$\ket{0}_{A}$};

\ldutarget{data-x-1}{(0.85,0)}

\node[control] at (1.95,0) {};
\draw[wire] (1.95,0) -- (1.95,-1.15);
\ldutarget{ancilla-cx-1}{(1.95,-1.15)}

\ldutarget{data-x-2}{(3.05,0)}
\ldutarget{ancilla-x}{(3.05,-1.15)}

\node[control] at (4.15,0) {};
\draw[wire] (4.15,0) -- (4.15,-1.15);
\ldutarget{ancilla-cx-2}{(4.15,-1.15)}

\node[meas] (ancilla-meas) at (5.35,-1.15) {};
\draw[wire] ([xshift=-0.18cm,yshift=-0.09cm]ancilla-meas.center)
    arc[start angle=205,end angle=-25,radius=0.20cm];
\draw[wire] ([xshift=-0.02cm,yshift=-0.05cm]ancilla-meas.center) --
    ([xshift=0.18cm,yshift=0.17cm]ancilla-meas.center);
\node[anchor=west] at (5.78,-1.15) {augmented readout};
\end{tikzpicture}%
    }
    \caption{Swap-resilient LDU for data qubit \(D\) and aligned ancilla qubit \(A\).
    Its final measurement is augmented with loss and leakage
    readout.}
    \label{fig:swap-loss-ldu}
\end{figure}

\begin{proposition}
\label{prop:swap-resilient-ldu}
In the absence of additional faults within the gadget, the swap-resilient LDU
preserves and accepts every computational input state and flags every leaked or
lost input for all possible loss-swap patterns.
\end{proposition}

\begin{proof}
For \(\ket{\phi}=\alpha\ket{0}+\beta\ket{1}\), the action of the circuit is
\begin{align*}
\ket{\phi,0}
&\xrightarrow{X_D}
  \alpha\ket{1,0}+\beta\ket{0,0} \\
&\xrightarrow{\CNOT_{D\rightarrow A}}
  \alpha\ket{1,1}+\beta\ket{0,0} \\
&\xrightarrow{X_DX_A}
  \alpha\ket{0,0}+\beta\ket{1,1} \\
&\xrightarrow{\CNOT_{D\rightarrow A}}
  \ket{\phi,0},
\end{align*}
where each ordered ket lists \(D\) before \(A\). Notably, the data state is
preserved and the ancilla outcome is \(0\).

Write \(\ell\) and \(\varnothing\) for leaked and lost statuses. For a leaked
input data qubit, all gates involving \(D\) are suppressed, leaving only
\[
    (\ell,0)_{DA}\xrightarrow{X_A}(\ell,1)_{DA}.
\]
The same argument gives ancilla outcome \(1\) for a lost input with no swaps.
With exactly one swap, the surviving ancilla replaces the lost data qubit on
\(D\), while the loss moves to \(A\), so augmented readout invariably reports
`lost'.

In the two-swap case, each entangling gate is suppressed and its
loss-swap event exchanges the two locations. The intermediate \(X_A\) is
suppressed while the loss occupies \(A\), whereas \(X_D\) flips the surviving
qubit:
\[
    (\varnothing,0)_{DA}
    \xrightarrow{\mathrm{swap}_1}(0,\varnothing)_{DA}
    \xrightarrow{X_D}(1,\varnothing)_{DA}
    \xrightarrow{\mathrm{swap}_2}(\varnothing,1)_{DA}.
\]
Thus the surviving qubit returns to the ancilla location in state \(1\).
Hence the gadget flags every non-computational input, as summarized in
\cref{tab:swap-loss-ldu-outcomes}.
\end{proof}

\begin{table}[ht]
    \centering
    \caption{Ideal swap-resilient LDU outcomes, excluding additional faults within
    the gadget.}
    \label{tab:swap-loss-ldu-outcomes}
    \begin{tabular}{@{}lcc@{}}
        \toprule
        Input status of \(D\) & Swap events & Ancilla readout \\
        \midrule
        Computational & 0 & \(0\) \\
        Leaked & --- & \(1\) \\
        Lost & 0 & \(1\) \\
        Lost & 1 & `lost' \\
        Lost & 2 & \(1\) \\
        \bottomrule
    \end{tabular}
\end{table}

The gadget has a nominal depth of three POCs: two entangling-gate layers and
one augmented measurement/reset layer. The single-qubit \(X\) operations are
implemented virtually by physical Pauli-frame tracking, and therefore do not contribute
to the POC depth. A flagged leakage requires one
additional POC for leakage reset; a flagged loss additionally requires a fresh
ion from the local reservoir. However, since these occur rarely (at rates $\approx$ \(\pleak \cdot n \cdot 2\) or \(\ploss \cdot n \cdot 2\)), they do not affect the nominal depth of the gadget significantly.

We modify the SEC circuit of Ref.~\cite{tripier2026walkingcat} to use the
swap-resilient LDU
in place of the beacon protocol and teleportation-based
LDU, as detailed in \cref{alg:swap-loss-sec}.

\begin{figure*}[tbp]
\refstepcounter{algocf}
\label{alg:swap-loss-sec}
\addcontentsline{loa}{algocf}{\protect\numberline{\thealgocf}{Syndrome extraction with the loss- and leakage-detection gadget only}}
\noindent\rule{\textwidth}{0.4pt}
\begin{center}
\textbf{Algorithm \thealgocf:} Syndrome extraction with the loss- and leakage-detection gadget only
\end{center}
\vspace{-0.8em}
\noindent\rule{\textwidth}{0.4pt}
\noindent\textbf{Data:} A data register $D$, an ancilla register $A$, and a valid maximally parallel schedule $\Sigma = \big((T_X^{(1)},T_Z^{(1)}), \dots, (T_X^{(w_{\mathrm{check}})},T_Z^{(w_{\mathrm{check}})})\big)$ for a three-ring code with $a=2$

\smallskip
\begingroup
\newcommand{\algnum}[1]{\footnotesize\bfseries #1}
\newcommand{\algi}{\hspace*{1.2em}}
\newcommand{\algii}{\hspace*{2.4em}}
\newcommand{\algline}[2]{%
  \noindent\makebox[1.7em][r]{\algnum{#1}}\hspace{0.7em}%
  \parbox[t]{\dimexpr\linewidth-2.4em\relax}{\raggedright #2\strut}%
  \par\vspace{0.25ex}}
\noindent\begin{minipage}[t]{0.48\textwidth}
\vspace{0pt}
\algline{1}{Prepare the ancilla qubits in $\ket{0}$;}
\algline{2}{Apply the loss- and leakage-detection gadget of \cref{fig:swap-loss-ldu} in parallel between each data qubit and its aligned ancilla qubit;}
\algline{3}{Measure the gadget ancillas using augmented loss/leakage readout;}
\algline{4}{Let $\mathcal{K}_1$ be the set of indices whose gadget ancilla returns computational outcome $1$;}
\algline{5}{\textbf{for} $i \in \mathcal{K}_1$ \textbf{do}}
\algline{6}{\algi Check data qubit $D_i$ to determine whether it is leaked or lost;}
\algline{7}{\algi \textbf{if} $D_i$ is lost \textbf{then}}
\algline{8}{\algii Replace $D_i$ by a reloaded qubit;}
\algline{9}{\algi Reset $D_i$ to the maximally mixed state $I/2$;}
\algline{10}{Let $\mathcal{K}_{\mathrm{lost}}$ be the set of indices whose gadget ancilla readout reports `lost';}
\algline{11}{\textbf{for} $i \in \mathcal{K}_{\mathrm{lost}}$ \textbf{do}}
\algline{12}{\algi Check data qubit $D_i$ for loss;}
\algline{13}{\algi Reload every qubit in $\{A_i,D_i\}$ that is lost;}
\algline{14}{\algi Reset $D_i$ to the maximally mixed state $I/2$;}
\end{minipage}\hfill
\begin{minipage}[t]{0.48\textwidth}
\vspace{0pt}
\algline{15}{Prepare all ancilla qubits in $\ket{+}$;}
\algline{16}{Relabel the ancilla in software so that the initial alignment matches the first schedule pair in $\Sigma$;}
\algline{17}{Apply the scheduled parallel gates for $(T_X^{(1)},T_Z^{(1)})$: \CNOT{} gates from each aligned \(X\)-check ancilla to its data qubit, and \CZ{} gates between each aligned \(Z\)-check ancilla and its data qubit;}
\algline{18}{\textbf{for} $\tau \in \{2,\dots,w_{\mathrm{check}}\}$ \textbf{do}}
\algline{19}{\algi Compute the common long-ring shift $r$ from the family change between rounds $\tau-1$ and $\tau$, together with $(s_X,t_X)$ from $T_X^{(\tau-1)}$ and $T_X^{(\tau)}$ and $(s_Z,t_Z)$ from $T_Z^{(\tau-1)}$ and $T_Z^{(\tau)}$;}
\algline{20}{\algi Apply the transport step consisting of a long-ring cycle by $r$ on all ancilla qubits together with the medium/short shifts $(s_X,t_X)$ on $\mathcal{A}_X$ and $(s_Z,t_Z)$ on $\mathcal{A}_Z$;}
\algline{21}{\algi Apply the scheduled parallel gates for $(T_X^{(\tau)},T_Z^{(\tau)})$: \CNOT{} gates from each aligned \(X\)-check ancilla to its data qubit, and \CZ{} gates between each aligned \(Z\)-check ancilla and its data qubit;}
\algline{22}{Measure all ancilla qubits in the $X$ basis using augmented loss/leakage readout;}
\algline{23}{Reset any ancilla qubit whose final readout reports `leaked' or `lost' before the next syndrome extraction round;}
\end{minipage}
\endgroup
\noindent\rule{\textwidth}{0.4pt}
\end{figure*}

\subsubsection{Memory Performance of Code Blocks}
\label{sec:swap-loss-memory-performance}
We simulate the swap-loss model and append the
swap-resilient LDU to every SEC.  The LDU contributes three POCs of depth. The LDU is simulated with the same Pauli gate noise as in the moving-qubit model,
together with the corresponding idle-noise, leakage, and loss channels on idle qubits. Unless stated otherwise,
we set
\begin{equation}
    \ploss=p/1000,\qquad
    \pleak=p/10,\qquad
    \pswap=0.5.
\end{equation}
For a physical justification of these parameters, see Ref.~\cite{tripier2026walkingcat}.

The generalized-bicycle code instances and SEC schedules used in the memory
simulations are specified in \cref{tab:memory-code-schedule-instances}.

\begin{table*}[t]
    \centering
    \caption{Generalized-bicycle code instances and syndrome-extraction
    schedules used in the memory simulations. The symbols \(A_i\) and
    \(B_i\) index the monomials of \(a(x)\) and \(b(x)\), respectively, in
    the order shown; \({}^{\mathsf{T}}\) denotes a transposed interaction.}
    \label{tab:memory-code-schedule-instances}
    \small
    \setlength{\tabcolsep}{4pt}
    \renewcommand{\arraystretch}{1.15}
    \begin{tabular}{@{}p{0.08\textwidth}p{0.35\textwidth}p{0.52\textwidth}@{}}
        \toprule
        Name & Code specification & Schedule permutation \(\Sigma\) \\
        \midrule
        Q102
        & \(\begin{gathered}
            \left[\!\left[102,22,9\right]\!\right],\quad (l,m)=(51,1),\\
            a(x)=x^{22}+x^{26}+x^{37}+x^{50},\\
            b(x)=x^{19}+x^{28}+x^{29}+x^{35}
          \end{gathered}\)
        & \(\begin{gathered}
            \bigl((A_1,A_2^{\mathsf{T}}),(B_1,B_4^{\mathsf{T}}),
            (B_2,B_3^{\mathsf{T}}),(A_3,A_4^{\mathsf{T}}),\\
            (A_4,A_3^{\mathsf{T}}),(B_3,B_2^{\mathsf{T}}),
            (B_4,B_1^{\mathsf{T}}),(A_2,A_1^{\mathsf{T}})\bigr)
          \end{gathered}\) \\
        \midrule
        Q66
        & \(\begin{gathered}
            \left[\!\left[66,4,10\right]\!\right],\quad (l,m)=(33,1),\\
            a(x)=1+x+x^3+x^{10},\\
            b(x)=1+x+x^8
          \end{gathered}\)
        & \(\begin{gathered}
            \bigl((B_2,B_1^{\mathsf{T}}),(A_2,A_1^{\mathsf{T}}),
            (A_3,A_4^{\mathsf{T}}),(B_3,B_3^{\mathsf{T}}),\\
            (A_4,A_3^{\mathsf{T}}),(A_1,A_2^{\mathsf{T}}),
            (B_1,B_2^{\mathsf{T}})\bigr)
          \end{gathered}\) \\
        \bottomrule
    \end{tabular}
\end{table*}

We evaluate the performance of our primary code blocks, Q102 and Q66, under
the swap-loss model at the assumed physical noise rate \(p=10^{-4}\). Direct simulation at this
noise rate would require prohibitively many shots, so we extrapolate the
swap-loss data from
\(p\in\{5\times10^{-4},8\times10^{-4},10^{-3},2\times10^{-3}\}\), and
the Pauli-noise-only baselines from
\(p\in\{10^{-3},2\times10^{-3},3\times10^{-3}\}\), using the
three-parameter ansatz
\(\plog(p)=p^5\exp(a+bp+cp^2)\). We simulate the Pauli-noise-only
model with Stim~\cite{gidney2021stim} and the loss- and leakage-enhanced model
with the bespoke simulator described in the original Walking Cat architecture paper~\cite{tripier2026walkingcat}. We
decode both data sets using the beam-search decoder~\cite{ye2026beam}
with beam width 32, at most 10 search rounds,
40 initial iterations, and 30 iterations per subsequent round. For Q102,
the fitted ansatz is

\begin{equation*}
    \plog(p)
    =p^5\exp\!\left(20.4344+2262.86p-5.13283\times10^5p^2\right),
\end{equation*}

which gives a logical error rate of \(9.34\times10^{-12}\) per SEC at
\(p=10^{-4}\), as shown in
\cref{fig:q102-q66-swap-loss-logical-error-rate}.

For Q66, the fitted swap-loss ansatz is

\begin{equation*}
    \plog(p)
    =p^5\exp\!\left(21.6356+574.572p-9.79513\times10^4p^2\right).
\end{equation*}

It gives a logical error rate of \(2.63\times10^{-11}\) per SEC at
\(p=10^{-4}\); the corresponding extrapolated Pauli-noise-only rate is
\(2.47\times10^{-12}\) per SEC.  The sampled data and both fits are shown in
\cref{fig:q102-q66-swap-loss-logical-error-rate}.

\begin{figure*}[tbp]
    \centering
    \includegraphics[width=0.43\textwidth]{%
        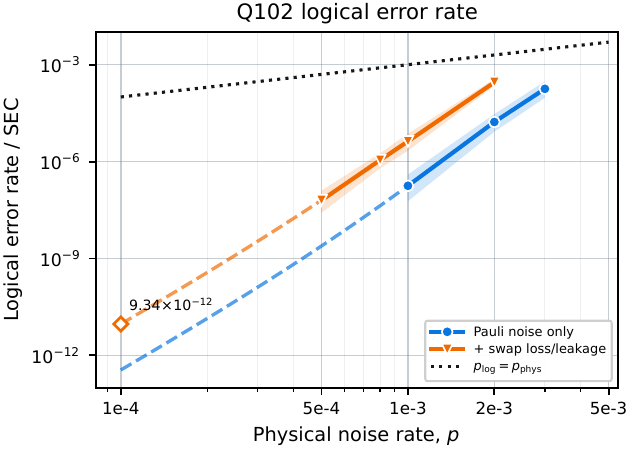}%
    \hfill
    \includegraphics[width=0.43\textwidth]{%
        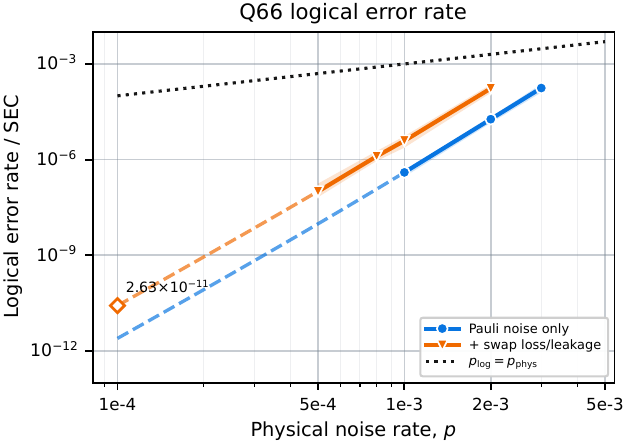}
    \caption{Memory performance of the \([[102,22,9]]\) Q102 code (left) and
    the \([[66,4,10]]\) Q66 code under the swap-loss
    model. For each code, the Pauli-noise-only baseline is compared with the
    corresponding swap-loss model, with
    \(\pleak=p/10\), \(\ploss=p/1000\), and
    \(\pswap=0.5\). Dashed curves extrapolate each series to
    \(p=10^{-4}\) using the three-parameter fifth-order ansatz
    \(\plog(p)=p^5\exp(a+bp+cp^2)\).}
    \label{fig:q102-q66-swap-loss-logical-error-rate}
\end{figure*}

\subsection{Reusable Cat State Factories with Parallel Layout}
\label{sec:reusable-cat-state-factories}
We modify the cat-state-factory layout to support the three parallel
measurements required for depth-one \CCZ{} injection. The resulting layout
provides this parallelism with fewer physical qubits than would be required
under the default Walking Cat architecture.

\subsubsection{Edge-Packed Cat-State Factories for Parallel Measurements}
\label{sec:stacked-cat-state-factories}
Depth-one injection of the \CCZ{} gates described in
\cref{sec:claw-isa} requires three parallel logical measurements. The
simplest layout required to accomplish this places three cat-state factories beside each code block, but
this arrangement relaxes the Walking Cat architecture's assumption that each
cat-state factory is at fixed distance from its associated code block. Longer
transport distances may require additional cat-state copies in flight to
pipeline logical measurements without increasing gate latency. We therefore
place the factories along the code-block edges and determine the number of
cat-state copies required by the resulting transport geometry.

In the Walking Cat architecture, measured cat-state qubits circulate back to
their factories through a device-wide outer loop so that their carriers can be
reused~\cite{tripier2026walkingcat}. We introduce a new loop design that minimizes
the transport distance: every code
block and its assigned factories share a local loop confined to their
region.

Let \(\wcat\in\mathbb N\) denote the cat-state weight, namely the
number of physical qubits comprising the cat state. We call a reusable allocation of
\(\wcat\) qubits that holds one weight-\(\wcat\) cat
state a \emph{cat-state register}. A register cycles through
cat-state preparation, delivery, measurement, and then returns. Verification
ancillas are factory resources and are not part of the register.

Consider an \(n\)-qubit code block on which \(k\) logical measurements are
performed simultaneously. We index the measurements by
\(j\in\{1,\ldots,k\}\) and assign factory \(\mathcal F_j\) to measurement
\(j\). Each factory supplies weight-\(\wcat\) cat states at the
cadence of one
state per syndrome extraction cycle (SEC), whose duration is
\(T_{\mathrm{SEC}}(n)\) POCs. Let \(r\) be the number of
weight-\(\wcat\)
cat-state registers reserved for each of the \(k\) parallel logical
measurements, let \(m\) be the number of verification rounds, and let
\(\eta=\tau_{\mathrm{tr}}/\tau_{\mathrm{POC}}\) denote the duration of one
transport step in POC units.

For even \(\wcat\), the Walking Cat protocol of
Ref.~\cite{tripier2026walkingcat} prepares and verifies a cat state in
\begin{equation}
    \begin{aligned}
        T_{\mathrm{prep}}(\wcat,m)
        ={}& \underbrace{\left\lceil\log_2 \wcat\right\rceil+3m+3}
            _{\text{preparation and verification}} \\
        &+\underbrace{\eta\left(\frac{\wcat}{2}+m-1\right)}
            _{\text{factory-internal transport}}
    \end{aligned}
    \label{eq:cat-preparation-time}
\end{equation}
POCs.

We pack factories along the long edges of the code block, as shown in
\cref{fig:stacked-cat-state-factories}. For placement, a
weight-\(\wcat\) factory is assigned a width of
\(\wcat\) carrier sites. One long edge can
therefore accommodate
\begin{equation}
    s(n,\wcat)=\left\lfloor\frac{n}{\wcat}\right\rfloor
    \label{eq:cat-factories-per-edge}
\end{equation}
factories. We place factories
\(\mathcal F_1,\ldots,\mathcal F_{\min\{k,s\}}\) consecutively along the top
edge and, when \(k>s\), continue with \(\mathcal F_{s+1}\) along the bottom
edge. Two long edges suffice whenever \(k\leq 2s(n,\wcat)\), without
occupying either short edge on the sides of the code blocks. The memory and cat-factory components of
Ref.~\cite{tripier2026walkingcat} each span four physical qubit rows. At the
granularity of the present transport model, an \(n\)-qubit code block therefore
occupies a \(4\times n\) rectangle and the loop immediately surrounding it has
length
\begin{equation}
    P_{\mathrm{loop}}(n)=2n+8
    \label{eq:local-cat-loop-perimeter}
\end{equation}
transport steps.

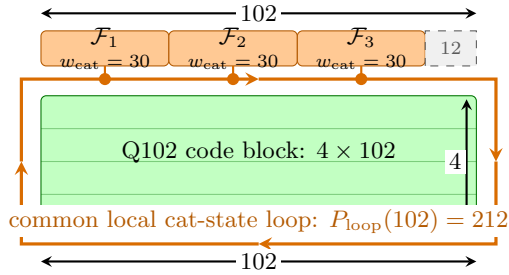
\begin{figure}[ht]
    \centering
    \begin{tikzpicture}[
    x=0.72cm,
    y=0.58cm,
    font=\small,
    >=stealth,
    factory/.style={
        draw=orange!70!black,
        fill=orange!42,
        rounded corners=2pt
    },
    loop/.style={very thick, orange!85!black, ->},
    port/.style={thick, orange!85!black}
]
    % Four-row code-block footprint.
    \draw[fill=green!24, draw=green!45!black, rounded corners=2pt]
        (0,0) rectangle (8,3);
    \foreach \y in {0.75,1.50,2.25}
        \draw[green!45!black, opacity=0.36] (0,\y) -- (8,\y);
    \node at (4,1.72) {Q102 code block: \(4\times102\)};

    % Three width-30 factories occupy 90 of the 102 top-edge sites.
    \draw[factory] (0.00,3.68) rectangle (2.35,4.48);
    \draw[factory] (2.35,3.68) rectangle (4.71,4.48);
    \draw[factory] (4.71,3.68) rectangle (7.06,4.48);
    \draw[draw=gray!70!black, fill=gray!12, dashed]
        (7.06,3.68) rectangle (8.00,4.48);

    \node[align=center] at (1.175,4.08)
        {\(\mathcal F_1\)\\[-2pt]\scriptsize \(\wcat=30\)};
    \node[align=center] at (3.530,4.08)
        {\(\mathcal F_2\)\\[-2pt]\scriptsize \(\wcat=30\)};
    \node[align=center] at (5.885,4.08)
        {\(\mathcal F_3\)\\[-2pt]\scriptsize \(\wcat=30\)};
    \node[align=center, gray!70!black] at (7.53,4.08)
        {\scriptsize 12};

    \foreach \portx in {1.175,3.530,5.885}
        \draw[port] (\portx,3.68) -- (\portx,3.38)
            node[pos=1, circle, fill=orange!85!black, inner sep=1.6pt] {};

    % Clockwise local loop. Any port-to-interface delivery and complementary
    % return together traverse this full perimeter.
    \draw[loop] (-0.36,1.50) -- (-0.36,3.38) -- (4.00,3.38);
    \draw[loop] (4.00,3.38) -- (8.36,3.38) -- (8.36,1.50);
    \draw[loop] (8.36,1.50) -- (8.36,-0.38) -- (4.00,-0.38);
    \draw[loop] (4.00,-0.38) -- (-0.36,-0.38) -- (-0.36,1.50);

    % Q102 dimensions and common circulation length.
    \draw[<->, thick] (0,4.88) -- (8,4.88)
        node[midway, fill=white, inner sep=1.5pt] {\(102\)};
    \draw[<->, thick] (0,-0.78) -- (8,-0.78)
        node[midway, fill=white, inner sep=1.5pt] {\(102\)};
    \draw[<->, thick] (7.82,0.08) -- (7.82,2.92)
        node[midway, left, fill=white, inner sep=1.5pt] {\(4\)};
    \node[orange!70!black, fill=white, inner sep=1.5pt]
        at (4,0.15) {common local cat-state loop: \(P_{\mathrm{loop}}(102)=212\)};
\end{tikzpicture}
    \caption{Local cat-state circulation for Q102. Three
    weight-30 factories occupy 90 of the 102 sites along the top edge, leaving
    12 sites unused; the bottom edge is not needed. In general, factories are
    packed along the top edge up to the capacity in
    \cref{eq:cat-factories-per-edge}, then continued along the bottom edge.
    Note that every factory has the
    same circulation distance \(2n+8\).}
    \label{fig:stacked-cat-state-factories}
\end{figure}

This routing policy also accommodates representatives whose supports are
interleaved across the data qubits, because the interchangeable qubits within
each cat state can be staged in support order. We now determine
how many registers are needed to sustain this delivery schedule.

\begin{proposition}
\label{prop:cat-state-rotation}
Assume that a successful cat-state preparation is available after
\(T_{\mathrm{prep}}(\wcat,m)\) POCs and that the local loop can transport the cat
states required for all \(k\) parallel logical measurements without
contention. Assume also that the physical representatives for the parallel
measurements have pairwise-disjoint supports. Measure time from the instant at
which the cat state for the current round is consumed. Assume that the next cat
state is not required until the next cat-based measurement, one SEC later, and
that the register can return, be prepared and verified, and attempt
redelivery concurrently with that intervening SEC.
The smallest register count for each measurement that sustains one
weight-\(\wcat\) cat state per SEC is
\begin{equation}
    r
    =\left\lceil
        \frac{T_{\mathrm{prep}}(\wcat,m)+\eta P_{\mathrm{loop}}(n)}
             {T_{\mathrm{SEC}}(n)}
      \right\rceil .
\end{equation}
\end{proposition}

\begin{proof}
Suppose register \(A\) supplies the cat state for the current measurement
round. From the moment that state is consumed, the register is unavailable for
the turnaround time
\[
    T_{\mathrm{prep}}(\wcat,m)+\eta P_{\mathrm{loop}}(n)
\]
while it returns to its original cat-state factory around the local
loop, undergoes preparation and verification, and then travels along the loop
to the next measurement. The return followed by the next delivery has total
length \(P_{\mathrm{loop}}(n)\). Let
\[
    q=\left\lceil
        \frac{T_{\mathrm{prep}}(\wcat,m)+\eta P_{\mathrm{loop}}(n)}
             {T_{\mathrm{SEC}}(n)}
    \right\rceil .
\]
During this turnaround, \(q-1\) other registers supply the intervening demands
at times \(T_{\mathrm{SEC}}(n),\ldots,(q-1)T_{\mathrm{SEC}}(n)\), and register
\(A\) returns for the demand at time \(qT_{\mathrm{SEC}}(n)\). Thus \(r=q\)
registers suffice. For example, when \(q=2\), register \(B\) is consumed one
SEC after \(A\), and \(A\) is consumed again after two SECs.
With fewer registers, \(A\) would be reused before it is ready.
\end{proof}

Every factory requires only one register if
\begin{equation}
    T_{\mathrm{prep}}(\wcat,m)
    +\eta(2n+8)
    \leq T_{\mathrm{SEC}}(n).
    \label{eq:cat-single-register-condition}
\end{equation}
Under this condition, \cref{prop:cat-state-rotation} gives \(r=1\) for every
measurement, for \(k\) cat-state registers in total.

The concrete Q102 and Q66 register counts used in the component footprints are
calculated in \cref{subsec:footprint}.

\subsubsection{Reusing Cat States}
\label{sec:reusing-cat-states}

The elliptic-curve circuit's limited global logical-measurement parallelism
leaves many cat-state resources idle for long periods of time under the default
Walking Cat architecture's per-code-block cat-state allocation policy. This motivates provisioning only
enough cat states to meet the circuit's peak measurement parallelism,
then transporting and reusing them across the device as needed. The basic unit
of reuse is the complete set of resources needed to pipeline one logical
measurement around the local loop: \(r\) weight-\(\wcat\) cat-state registers,
with \(r\) given by \cref{prop:cat-state-rotation}, together with one bank of
\(\wcat\) verification ancillas for successive cat-state preparations. We call
this a \emph{cat-state bundle}.

For the global routing analysis, we arrange the targets of each of the \(k\)
parallel logical measurements into an ordered stream. Each stream is served by
a pool of reusable cat-state bundles that may travel between code blocks. We
ask how many bundles each stream needs to avoid delaying logical measurements.

The registers within each bundle allow one cat state to be prepared and
verified while another is consumed, as in
\cref{prop:cat-state-rotation}; the ancilla bank travels with the registers so
that the bundle can operate beside any memory block. Upon arrival at the destination code block,
the bundle docks at one of the cat-state factories along the target block's edge
and undergoes the preparation and verification procedure of
\cref{sec:stacked-cat-state-factories}. At least one bundle remains beside the
active block throughout its logical measurement, but the bundle that begins
the protocol need not be the one that finishes it. The bundles assigned to each
measurement stream follow a handoff schedule, illustrated for two bundles in
\cref{fig:global-cat-bundle-leapfrog}. The lookahead bundle begins the
measurement at the current target while the trailing bundle travels to that
same block and is prepared and verified there. The trailing bundle then takes
over the ongoing measurement, releasing the lookahead bundle to advance to the
following target. Repeating this handoff moves the bundles through the ordered
stream of logical measurements.

\begin{figure*}[t]
    \centering
    \providecolor{clinrCobalt}{HTML}{4477AA}
\providecolor{clinrCrimson}{HTML}{C94B5A}

\begin{tikzpicture}[
    x=1.00cm,
    y=0.72cm,
    font=\scriptsize,
    >=stealth,
    chip/.style={
        draw=gray!55,
        fill=gray!7,
        rounded corners=6pt,
        thick
    },
    chip route/.style={thick, gray!55, ->},
    memory/.style={
        draw=green!45!black,
        fill=green!24,
        rounded corners=2pt,
        thick
    },
    active memory A/.style={
        memory,
        draw=clinrCobalt!85!black,
        very thick
    },
    active memory B/.style={
        memory,
        draw=clinrCrimson!85!black,
        very thick
    },
    factory/.style={
        draw=orange!70!black,
        fill=orange!42,
        rounded corners=1pt
    },
    active factory/.style={
        factory,
        draw=clinrCobalt!85!black,
        very thick
    },
    bundle A/.style={
        circle,
        draw=clinrCobalt!85!black,
        fill=clinrCobalt!35,
        thick,
        inner sep=1.7pt,
        font=\bfseries
    },
    bundle B/.style={
        circle,
        draw=clinrCrimson!85!black,
        fill=clinrCrimson!35,
        thick,
        inner sep=1.7pt,
        font=\bfseries
    },
    move A/.style={very thick, clinrCobalt!85!black, ->},
    move B/.style={very thick, clinrCrimson!85!black, ->}
]
    % (a) Lookahead bundle B starts the measurement at M6 while trailing
    % bundle A traverses the chip toward that same block.
    \begin{scope}
        \draw[chip] (0,0) rectangle (8,4.45);
        \node[anchor=west, font=\small] at (0.24,4.12)
            {\textbf{(a)} Transport};
        \draw[chip route] (0.40,0.42) -- (0.40,3.72) -- (7.60,3.72);
        \draw[chip route] (7.60,3.72) -- (7.60,0.42) -- (0.40,0.42);
        \foreach \x/\y/\name in {
            0.65/2.35/1,3.10/2.35/2,5.55/2.35/3,
            0.65/0.65/4,3.10/0.65/5,5.55/0.65/6} {
            \draw[memory] (\x,\y) rectangle +(1.80,0.72);
            \node at (\x+0.90,\y+0.36) {\(M_{\name}\)};
            \foreach \dx in {0.12,0.64,1.16}
                \draw[factory] (\x+\dx,\y+0.80) rectangle +(0.42,0.27);
        }
        \draw[active memory B] (5.55,0.65) rectangle +(1.80,0.72);
        \node at (6.23,1.01) {\(M_6\)};
        \node[bundle B] at (6.93,1.01) {B};
        \draw[move A, rounded corners=3pt]
            (2.70,3.72) -- (7.60,3.72) -- (7.60,1.82);
        \node[bundle A] at (5.85,3.72) {A};
    \end{scope}

    % (b) A reaches M6's edge-packed local factory while B remains active there.
    \begin{scope}[xshift=8.55cm]
        \draw[chip] (0,0) rectangle (8,4.45);
        \node[anchor=west, font=\small] at (0.24,4.12)
            {\textbf{(b)} Arrival at local factory};
        \draw[chip route] (0.40,0.42) -- (0.40,3.72) -- (7.60,3.72);
        \draw[chip route] (7.60,3.72) -- (7.60,0.42) -- (0.40,0.42);
        \foreach \x/\y/\name in {
            0.65/2.35/1,3.10/2.35/2,5.55/2.35/3,
            0.65/0.65/4,3.10/0.65/5,5.55/0.65/6} {
            \draw[memory] (\x,\y) rectangle +(1.80,0.72);
            \node at (\x+0.90,\y+0.36) {\(M_{\name}\)};
            \foreach \dx in {0.12,0.64,1.16}
                \draw[factory] (\x+\dx,\y+0.80) rectangle +(0.42,0.27);
        }
        \draw[active memory B] (5.55,0.65) rectangle +(1.80,0.72);
        \node at (6.23,1.01) {\(M_6\)};
        \node[bundle B] at (6.93,1.01) {B};
        \draw[move A] (7.58,2.40) -- (6.45,1.80);
        \node[bundle A] at (6.45,1.73) {A};
    \end{scope}

    % (c) A is prepared and verified in its assigned edge-packed factory at M6.
    % The operation remains hidden beneath B's logical measurement at M6.
    \begin{scope}[yshift=-4.95cm]
        \draw[chip] (0,0) rectangle (8,4.45);
        \node[anchor=west, font=\small] at (0.24,4.12)
            {\textbf{(c)} Local preparation and verification};
        \draw[chip route] (0.40,0.42) -- (0.40,3.72) -- (7.60,3.72);
        \draw[chip route] (7.60,3.72) -- (7.60,0.42) -- (0.40,0.42);
        \foreach \x/\y/\name in {
            0.65/2.35/1,3.10/2.35/2,5.55/2.35/3,
            0.65/0.65/4,3.10/0.65/5,5.55/0.65/6} {
            \draw[memory] (\x,\y) rectangle +(1.80,0.72);
            \node at (\x+0.90,\y+0.36) {\(M_{\name}\)};
            \foreach \dx in {0.12,0.64,1.16}
                \draw[factory] (\x+\dx,\y+0.80) rectangle +(0.42,0.27);
        }
        \draw[active memory B] (5.55,0.65) rectangle +(1.80,0.72);
        \node at (6.23,1.01) {\(M_6\)};
        \node[bundle B] at (6.93,1.01) {B};
        \draw[active factory] (6.19,1.45) rectangle +(0.42,0.27);
        \node[bundle A] at (6.40,1.585) {A};
    \end{scope}

    % (d) A takes over at M6; B starts the same sequence toward a later target.
    \begin{scope}[xshift=8.55cm,yshift=-4.95cm]
        \draw[chip] (0,0) rectangle (8,4.45);
        \node[anchor=west, font=\small] at (0.24,4.12)
            {\textbf{(d)} Handoff and repeat};
        \draw[chip route] (0.40,0.42) -- (0.40,3.72) -- (7.60,3.72);
        \draw[chip route] (7.60,3.72) -- (7.60,0.42) -- (0.40,0.42);
        \foreach \x/\y/\name in {
            0.65/2.35/1,3.10/2.35/2,5.55/2.35/3,
            0.65/0.65/4,3.10/0.65/5,5.55/0.65/6} {
            \draw[memory] (\x,\y) rectangle +(1.80,0.72);
            \node at (\x+0.90,\y+0.36) {\(M_{\name}\)};
            \foreach \dx in {0.12,0.64,1.16}
                \draw[factory] (\x+\dx,\y+0.80) rectangle +(0.42,0.27);
        }
        \draw[active memory A] (5.55,0.65) rectangle +(1.80,0.72);
        \node at (6.23,1.01) {\(M_6\)};
        \node[bundle A] at (6.93,1.01) {A};
        \draw[move B]
            (6.45,1.42) -- (7.58,1.92) -- (7.58,3.72) -- (4.30,3.72);
        \node[bundle B] at (5.40,3.72) {B};
    \end{scope}
\end{tikzpicture}
    \caption{Device-level handoff schedule for reusing two cat-state bundles
    in a measurement stream. Each memory block has three cat-state factories,
    shown in orange. Lookahead bundle B begins the logical measurement
    at \(M_6\) while bundle A follows the highlighted route toward that same
    block (a). A arrives at one of \(M_6\)'s local factories (b) and is prepared
    and verified as B continues measuring (c). A then takes over the ongoing
    measurement at \(M_6\), releasing B to advance toward the following target (d).}
    \label{fig:global-cat-bundle-leapfrog}
\end{figure*}
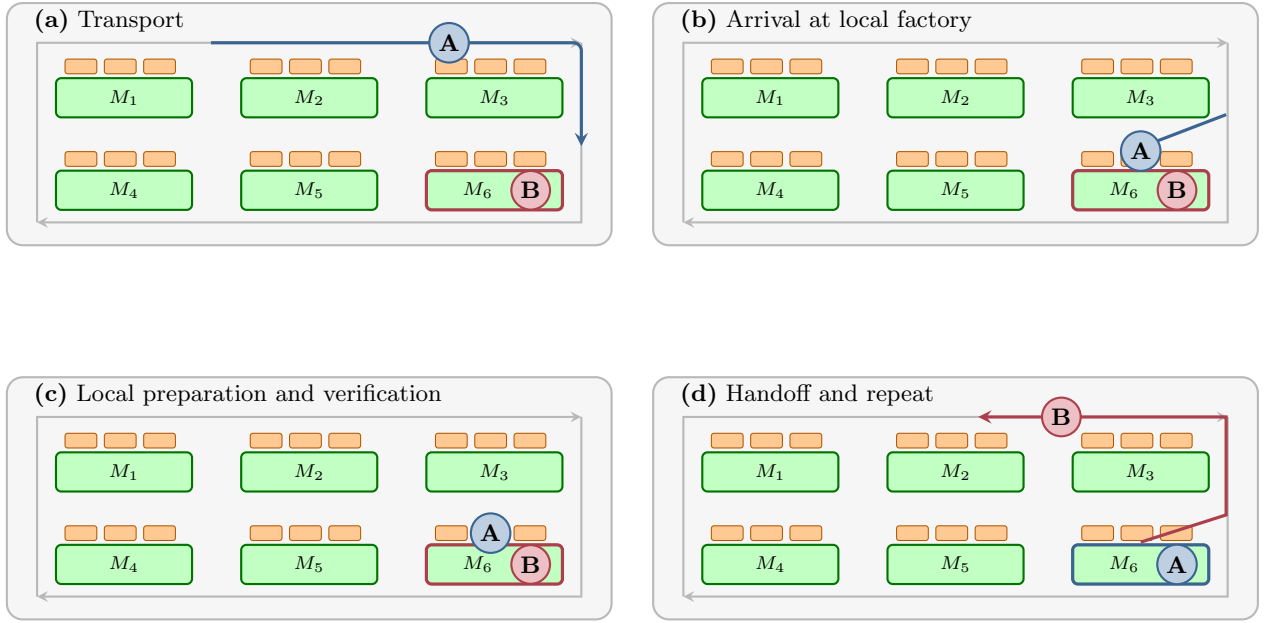

Consider \(M\) memory blocks packed in an \(R\)-by-\(C\) rectangle, where
\(M\leq RC\) and unused positions, if any, are left empty. As above, each block
occupies a \(4\)-by-\(n\) rectangle. We leave one
additional horizontal routing row between neighboring block rows. The maximum
Manhattan distance between corresponding ports of two memory blocks is
therefore
\begin{equation}
    D_{\max}(R,C,n)
    = n(C-1)+5(R-1)
    \label{eq:global-cat-manhattan-diameter}
\end{equation}
transport steps. Let \(T_{\mathrm{LM}}\) be the duration in POCs of the shortest
logical measurement in the workload.

\begin{proposition}
\label{prop:global-cat-state-reuse}
Let
\(
L=\eta D_{\max}(R,C,n)+T_{\mathrm{prep}}(\wcat,m)
\)
and assume \(L\leq T_{\mathrm{LM}}\). Assume also that, under the ouroboros
block layout introduced for the idle-frame-clearing schedule in \cref{sec:idle-frame-clearing},
the transport routes serving
simultaneous logical measurements do not cross one another. For a computation
with at most \(k\) simultaneous logical measurements, the handoff schedule can
execute any sequence of measurement locations without transport delays using
at most
\begin{equation}
    \begin{aligned}
        &N_{\mathrm{bundle}}(k,R,C,n,\wcat,m) \\
        &\quad = k\left(
            1+\left\lceil
                \frac{2\left[\eta D_{\max}(R,C,n)
                    +T_{\mathrm{prep}}(\wcat,m)\right]}
                     {T_{\mathrm{LM}}}
            \right\rceil
          \right).
    \end{aligned}
    \label{eq:global-cat-bundle-count}
\end{equation}
In particular, two bundles per measurement stream suffice whenever
\begin{equation}
    2\left[\eta D_{\max}(R,C,n)+T_{\mathrm{prep}}(\wcat,m)\right]
    \leq T_{\mathrm{LM}}.
    \label{eq:two-global-cat-bundle-condition}
\end{equation}
\end{proposition}

\begin{proof}
Partition the bundles among \(k\) independent measurement streams. In each stream,
use one trailing bundle and
\(
p=\lceil 2L/T_{\mathrm{LM}}\rceil
\)
lookahead bundles. At the start of a logical measurement, the
next lookahead bundle is already prepared at the target block and begins the
measurement. The trailing bundle, released when the preceding logical
measurement ended, travels to this same target and is prepared and verified in
at most \(L\) POCs. Because \(L\leq T_{\mathrm{LM}}\), it can take over before
the current measurement ends.

This handoff releases the active lookahead bundle at most \(L\) POCs after the
measurement began. That bundle is next required \(p\) measurements later, so
it has at least \(pT_{\mathrm{LM}}-L\geq L\) POCs to travel to and prepare at
its next assigned target. It is therefore ready before that measurement begins.
Repeating the handoff gives a delay-free pipeline with \(p+1\) bundles per lane,
which yields \cref{eq:global-cat-bundle-count}. When \(2L\leq
T_{\mathrm{LM}}\), \(p=1\): the leading bundle starts each target, the trailing
bundle takes over, and the released leading bundle advances again.
\end{proof}

The \(\bitcoincurve\) device packs its 69 Q102 memory blocks and four Q102
Logical CliNR blocks in a nine-by-nine grid, as
shown in \cref{fig:full-scale-quantum-computing-architecture}. Consequently,
\begin{equation}
    D_{\max}(9,9,102)
    =102(9-1)+5(9-1)
    =856
\end{equation}
transport steps, or \(\eta D_{\max}=42.8\) POCs for \(\eta=1/20\). The
subsequent local preparation and verification takes
\(T_{\mathrm{prep}}(30,2)=14.8\) POCs, for a total lookahead time of
\(57.6\) POCs. The two successive pipeline legs---bringing the trailing bundle
to the active block and then sending the released lookahead bundle onward---take
at most \(2(57.6)=115.2\) POCs. At the physical error rate \(p=10^{-4}\), a
Viterbi measurement with cat-state weight \(\wcat\geq16\) targeting a
logical-measurement error rate of \(\epsilon=10^{-10}\) requires at least five
SECs
\cite{tripier2026walkingcat}, hence
\(T_{\mathrm{LM}}\geq 5T_{\mathrm{SEC}}=135\) POCs for Q102. Therefore
\begin{equation}
    \begin{aligned}
        N_{\mathrm{bundle}}(k,9,9,102,30,2)
        &=k\left(1+\left\lceil
            \frac{2(42.8+14.8)}{135}
        \right\rceil\right) \\
        &=2k.
    \end{aligned}
\end{equation}
Thus, two bundles suffice for each Q102 measurement stream. We call this
stream-level allocation a \emph{cat-state bundle pair}. One set of \(k\) bundles
can start the measurements at the current target
blocks while the other set arrives, is prepared and verified, and takes over;
the released set then moves ahead to the next \(k\) targets. For the \(k=3\)
measurements used by depth-one \CCZ{} injection, this requires three mobile
cat-state bundle pairs (six individual bundles) rather than the \(73k=219\)
block-local bundles of the default
Walking Cat architecture, a factor-36.5 reduction in spatial
overhead without increasing logical-measurement latency.

The lookahead-based pipelining policy described above assumes that the target block of each lane is known at
least \(L\) POCs before its measurement begins. Feed-forward can violate this
assumption when a measurement outcome determines which previously inactive
code block is measured next: the released lookahead bundle has no unique
destination until that outcome is available. A measurement on the selected
block can then incur a stall of up to
\(\eta D_{\max}+T_{\mathrm{prep}}(\wcat,m)=L\) POCs while the leading bundle is
transported, prepared, and verified. So, the delay-free guarantee of
\cref{prop:global-cat-state-reuse} does not extend across such a conditional
change of target block without provisioning bundles for the alternative
destinations.

In the Walking Cat architecture, frame tracking does not cross code-block
boundaries. Consequently, a tracked conditional Clifford gate may change a
subsequent logical measurement by conjugating its observable, but it cannot
change its target block. And, otherwise, we observe that no conditional operations
in the actual compiled elliptic curve point addition
circuits select among alternative destination blocks.

Let \(w_{\max}\) be the largest weight among the possible physical
representatives of the conjugated logical observable. A weight-\(w_{\max}\)
cat state can therefore be prepared at the known target before the frame is
resolved. If the selected representative has lower weight, we pad it with
identity factors to weight \(w_{\max}\). The destination and cat-state size are
therefore
independent of the conditional outcome, so the handoff schedule remains
applicable, with its preparation-time and bundle-count bounds evaluated at
\(\wcat=w_{\max}\). 

Under these constraints, \(k\) cat-state bundle pairs, comprising \(2k\) individual bundles,
suffice to pipeline the \(k\) parallel logical measurements of depth-one \CCZ{}
injection.

\subsection{The CLAW ISA and Depth-One \CCZ{} Injection}
\label{sec:claw-isa}

The \emph{Clifford-frame and Logical-operator Acceleration for Workloads}
(CLAW) ISA extends the logical instruction set of the Walking Cat
architecture~\cite{tripier2026walkingcat} with operations tailored to elliptic
curve point addition circuits. It targets two principal bottlenecks in the
compiled circuits. The first is the cost of implementing Clifford gates. Under
the default Walking Cat assumptions, only pure-\(X\) and pure-\(Z\) logical
measurements are guaranteed to admit physical representatives of weight at most
30 and hence to be accessible using a weight-30 cat state. Tracking a Clifford
gate such as a conditional \CZ{} can conjugate these into more general Pauli
measurements, so the baseline architecture cannot
guarantee that every Clifford correction can be implemented by frame tracking.
CLAW instead checks accessibility for the specific logical measurements required
in the compiled elliptic curve point addition circuits. We
compile each component assuming a clean input frame and use Monte Carlo sampling
to check the conjugated logical measurements along its conditional Clifford
trajectories. Every measurement encountered has an accessible Q102 physical
representative of weight at most 30. \emph{Idle frame clearing}, described in
\cref{sec:idle-frame-clearing}, maintains the clean-input-frame invariant without
introducing additional runtime latency. This extension dramatically reduces the
number of logical measurements required to implement Clifford gates in the
elliptic-curve point-addition circuit, leaving Toffoli gates as the dominant
runtime bottleneck.

To address this remaining bottleneck, we extend the CLAW ISA to include an
\emph{LMP} operation that measures in parallel
any set of logical operators admitting pairwise-disjoint physical
representatives; the corresponding measurement protocol is described in
\cref{sec:parallel-cat-state-measurements}. In particular, the LMP operation allows the
three disjoint logical measurements
required for \CCZ{} state injection to be performed simultaneously, yielding an
optimal depth-one implementation of the Toffoli gate as described in
\cref{sec:depth-one-ccz-injection}.

\subsubsection{Parallel Cat-State Measurements}
\label{sec:parallel-cat-state-measurements}

We extend the cat-state measurement protocol of the Walking Cat
architecture~\cite{tripier2026walkingcat} to measure several logical Pauli
operators in parallel. Let \(P_1,\ldots,P_k\) be logical operators with
physical representatives \(\widetilde P_1,\ldots,\widetilde P_k\) whose
supports are pairwise disjoint. We insist upon disjointness here to enable
a simple analysis of logical measurement fault rate--representatives with
overlap may introduce correlated error when measured together. Additionally,
for simplicity, suppose without loss of generality that every representative has weight
\(\wcat\). We prepare and verify \(k\) distinct
weight-\(\wcat\) cat
states, one for each representative. The \(a\)th qubit of cat state \(C_j\)
is coupled to the \(a\)th data qubit in the support of \(\widetilde P_j\).
Because the data supports are disjoint, all \(k\wcat\) data--cat interactions can
be scheduled in the same physical gate layer.

The \(k\) verified cat states are placed side by side and transported as a
single bundle. The bundle is slid into alignment with the
union of the \(k\) supports, the data--cat gates are applied simultaneously,
and all cat qubits are then measured in the \(X\) basis in the same readout
round. Taking the parity of the outcomes within \(C_j\) yields the measurement
bit for \(P_j\), as illustrated for two parallel logical measurements in
\cref{fig:parallel-cat-state-measurements}. This protocol then
allows us to make \(k\) logical measurements in a single physical measurement layer.

\begin{figure}[t]
    \centering
    \includegraphics[width=0.70\linewidth]{%
        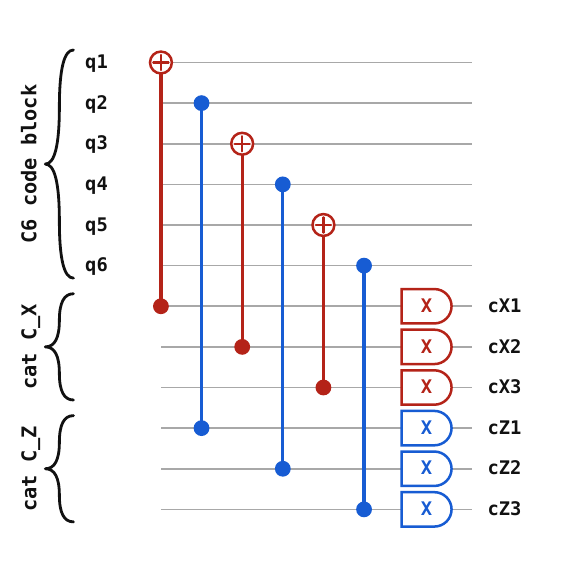}
    \caption{Parallel cat-state measurement of two disjoint logical
    representatives in the \([\![6,2,2]\!]\) "C6" code. The prepared cat states \(C_X\) and
    \(C_Z\) control data--cat gates on the disjoint representatives
    \(\overline X_0=\texttt{XIXIXI}\) and
    \(\overline Z_1=\texttt{IZIZIZ}\), respectively. Red controlled-\(X\)
    gates belong to the logical-\(X\) measurement, while blue
    controlled-\(Z\) gates belong to the logical-\(Z\) measurement.}
    \label{fig:parallel-cat-state-measurements}
\end{figure}

Because the physical representatives have disjoint supports, the corresponding
measurements commute. Therefore, in the noiseless setting, executing the cat-state
measurements in parallel is equivalent to executing them sequentially in any
order. This guarantee persists under the swap-loss model because its faults
are local. Each logical measurement uses a distinct cat state and a disjoint
set of data--cat interactions, and faults are sampled locally for each physical
operation. Consequently, no fault arising within one cat-state
measurement---including a hook error or a swap induced by a lost ion---can
propagate into another gadget. Parallel execution therefore leaves the outcome
distribution and error probability of each logical measurement unchanged
relative to executing that gadget alone in the same physical layer.

\subsubsection{Idle Frame Clearing}
\label{sec:idle-frame-clearing}

As in the Walking Cat architecture~\cite{tripier2026walkingcat}, we implement
conditional Clifford corrections whenever possible by Clifford frame-tracking
rather than applying them physically. The compiler updates the block's
Clifford frame \(U\) and replaces each subsequent target logical Pauli
operator \(P\) by \(U^\dagger P U\). When a subsequent non-Clifford operation
requires a trivial Clifford frame, the tracked Clifford gate can be implemented
using logical CliNR~\cite{webster2026fast}, thereby teleporting the encoded
state into a block whose Clifford frame is reset to \(U=I\), as summarized in
\cref{fig:logical-clinr-circuit-step}. We refer to this reset as
\emph{clearing} the Clifford frame. Because logical CliNR can be applied only
once a block's Clifford frame is known, frame clearing is generally bottlenecked
by resolving the classical conditions associated with conditional Clifford
gates. That generally means it's only safe to apply logical CliNR
once the block becomes idle, or we must actively pause execution to clear the frame.
Logical CliNR uses two auxiliary memory
blocks, \(A\) and \(B\), as resource blocks; both are encoded using the same
code as the data block whose frame is being cleared. At completion, the logical
state of the cleaned block has been teleported to \(B\), whose Clifford frame is trivial, while the original
data block and \(A\) can be reset and reused.

\begin{figure*}[t]
    \centering
    \resizebox{0.98\textwidth}{!}{%
        \input{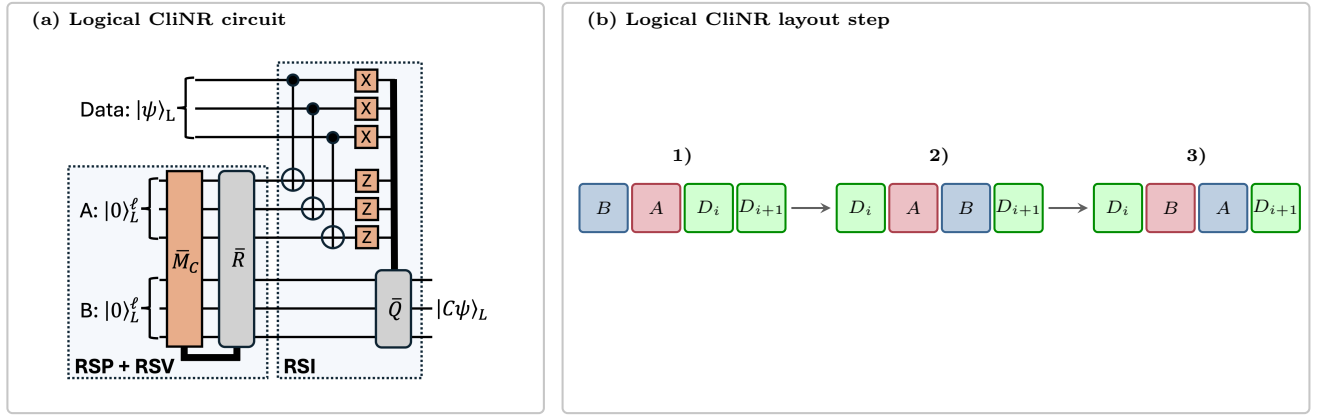}%
    }
    \caption{Logical CliNR and its compiler-level motion primitive. (a) The
    logical CliNR circuit of Webster and
    Delfosse~\cite{webster2026fast}. Joint stabilizer measurements and a Pauli
    correction prepare and verify a logical resource state across auxiliary
    blocks \(A\) and \(B\). A transversal \CNOT{} gate from the data block to \(A\),
    followed by destructive \(X\)- and \(Z\)-basis measurements and the Pauli
    correction \(\overline Q\), teleports \(C\lvert\psi\rangle_L\) into block
    \(B\). (b) One forward step begins with adjacent ordering
    \([B,A,D_i,D_{i+1}]\). Exchanging the \(B\) carrier with \(D_i\) gives
    \([D_i,A,B,D_{i+1}]\). The physical carriers are then kept fixed while the
    semantic roles \(A\) and \(B\) are freely exchanged, yielding
    \([D_i,B,A,D_{i+1}]\). The renewed triple \([B,A,D_{i+1}]\) is therefore
    ready for the next CliNR step.}
    \label{fig:logical-clinr-circuit-step}
\end{figure*}

Optimizing the CLAW ISA for elliptic curve point addition circuits requires
identifying the precise set of logical Pauli operators actually measured by
each compiled component.
To make this tractable, we
compile components independently under an invariant: each memory block
in a component's support has a trivial Clifford frame immediately before its
first non-Clifford use. To enforce this invariant,
we must use frame clearing on each block in a component's support.

Frame clearing can run in parallel with logical measurements on other blocks.
Because most components of the elliptic-curve point-addition circuit are
serial, each block is idle long enough to clear its frame without extending the
critical path. More generally, the elliptic-curve circuits consist of a sequence
of components that act serially on ordered sets of logical blocks: each
component proceeds down its blocks and then uncomputes in reverse order, with
occasional shallow layers of parallel gates between components; see
\cref{fig:serial-component-structure}. We call these components \emph{serial components}.

\begin{figure*}[t]
    \centering
    \resizebox{0.96\textwidth}{!}{%
        \begin{tikzpicture}[
    x=0.92cm,
    y=0.70cm,
    font=\scriptsize,
    >=stealth,
    serial down/.style={red!75!black, very thick, ->},
    serial up/.style={blue!70!black, very thick, ->},
    phase label/.style={fill=white, inner sep=1.5pt, font=\scriptsize\bfseries}
]
    % Logical-qubit lanes and axes.
    \foreach \y in {0.6,1.6,...,7.6}
        \draw[gray!25] (0.25,\y) -- (17.05,\y);
    \draw[->, thick] (0.25,0.15) -- (17.25,0.15)
        node[below] {time};
    \draw[->, thick] (0.25,0.15) -- (0.25,8.15);
    \node[rotate=90] at (-0.65,4.15) {logical qubit index};

    % A wide serial component spanning the full register.
    \draw[serial down]
        (0.90,7.60) -- (1.30,7.60) -- (1.30,6.60)
        -- (1.70,6.60) -- (1.70,5.60)
        -- (2.10,5.60) -- (2.10,4.60)
        -- (2.50,4.60) -- (2.50,3.60)
        -- (2.90,3.60) -- (2.90,2.60)
        -- (3.30,2.60) -- (3.30,1.60)
        -- (3.70,1.60) -- (3.70,0.60);
    \draw[serial up]
        (3.70,0.60) -- (4.10,0.60) -- (4.10,1.60)
        -- (4.50,1.60) -- (4.50,2.60)
        -- (4.90,2.60) -- (4.90,3.60)
        -- (5.30,3.60) -- (5.30,4.60)
        -- (5.70,4.60) -- (5.70,5.60)
        -- (6.10,5.60) -- (6.10,6.60)
        -- (6.50,6.60) -- (6.50,7.60);
    \node[phase label, text=red!75!black] at (1.75,6.35) {serial down};
    \node[phase label, text=blue!70!black] at (5.65,6.35) {serial up};

    % Successive components supported on the lower and upper register halves.
    \draw[serial down]
        (6.90,3.60) -- (7.20,3.60) -- (7.20,2.60)
        -- (7.50,2.60) -- (7.50,1.60)
        -- (7.80,1.60) -- (7.80,0.60);
    \draw[serial up]
        (7.80,0.60) -- (8.10,0.60) -- (8.10,1.60)
        -- (8.40,1.60) -- (8.40,2.60)
        -- (8.70,2.60) -- (8.70,3.60);
    \draw[serial down]
        (9.10,7.60) -- (9.40,7.60) -- (9.40,6.60)
        -- (9.70,6.60) -- (9.70,5.60)
        -- (10.00,5.60) -- (10.00,4.60);
    \draw[serial up]
        (10.00,4.60) -- (10.30,4.60) -- (10.30,5.60)
        -- (10.60,5.60) -- (10.60,6.60)
        -- (10.90,6.60) -- (10.90,7.60);
    % One logical layer acting on every qubit in parallel.
    \draw[blue!70!black, very thick] (11.70,0.60) -- (11.70,7.60);
    \node[phase label, text=blue!70!black] at (11.70,8.15) {parallel layer};

    % A final serial component spanning the full register.
    \draw[serial down]
        (12.50,7.60) -- (12.80,7.60) -- (12.80,6.60)
        -- (13.10,6.60) -- (13.10,5.60)
        -- (13.40,5.60) -- (13.40,4.60)
        -- (13.70,4.60) -- (13.70,3.60)
        -- (14.00,3.60) -- (14.00,2.60)
        -- (14.30,2.60) -- (14.30,1.60)
        -- (14.60,1.60) -- (14.60,0.60);
    \draw[serial up]
        (14.60,0.60) -- (14.90,0.60) -- (14.90,1.60)
        -- (15.20,1.60) -- (15.20,2.60)
        -- (15.50,2.60) -- (15.50,3.60)
        -- (15.80,3.60) -- (15.80,4.60)
        -- (16.10,4.60) -- (16.10,5.60)
        -- (16.40,5.60) -- (16.40,6.60)
        -- (16.70,6.60) -- (16.70,7.60);
\end{tikzpicture}%
    }
    \caption{Schematic temporal structure of the compiled elliptic-curve
    point-addition circuits. Time runs from left to right, and each row denotes
    a logical-qubit index. A serial component (e.g., an adder) first visits its support in
    descending index order (red) and then returns in ascending order (blue).
    In between components, we may have shallow layers of operations ran
    all in parallel (e.g., large controlled multi-target gates for conditional inversions),
    and/or dense layers of operations on a single block (e.g., a table lookup's depth-first selector traversal).}
    \label{fig:serial-component-structure}
\end{figure*}
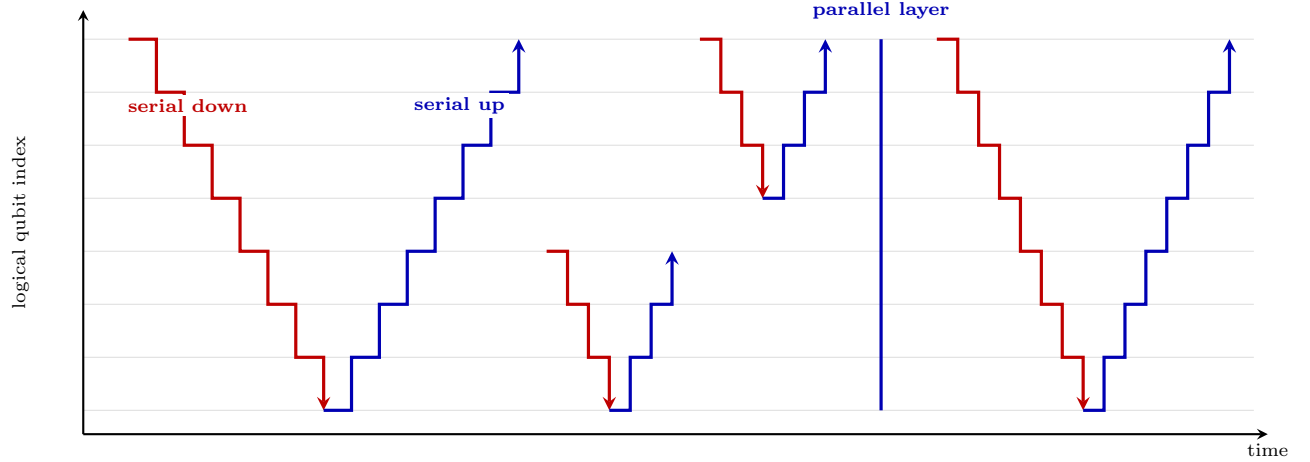

We therefore pipeline frame clearing ahead of each serial component's execution,
as illustrated in \cref{fig:idle-frame-clearing-pipeline}.

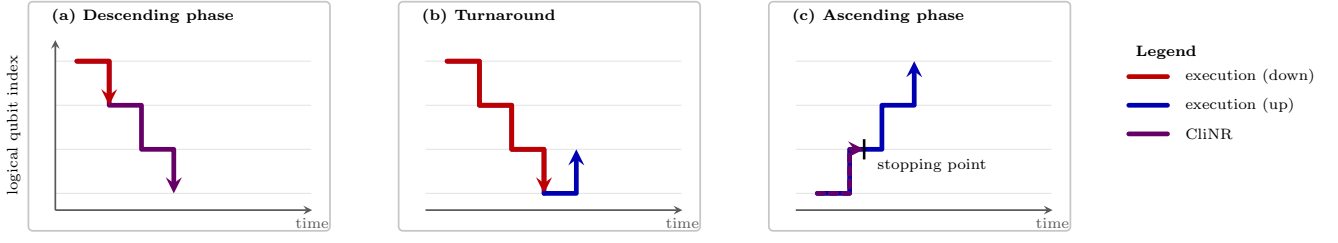
\begin{figure*}[t]
    \centering
    \resizebox{0.98\textwidth}{!}{%
        \begin{tikzpicture}[
    x=0.88cm,
    y=0.72cm,
    font=\scriptsize,
    >=stealth,
    panel/.style={draw=gray!45, rounded corners=2pt, thick},
    lane/.style={gray!22},
    axis/.style={gray!70!black, thick, ->},
    serial down/.style={
        red!75!black,
        line width=2.2pt,
        line cap=round,
        line join=round,
        ->
    },
    serial up/.style={
        blue!70!black,
        line width=2.2pt,
        line cap=round,
        line join=round,
        ->
    },
    clinr/.style={
        violet!80!black,
        line width=2.2pt,
        line cap=round,
        line join=round,
        ->
    },
    clinr trace/.style={
        violet!80!black,
        line width=2.2pt,
        line cap=round,
        line join=round
    },
    legend label/.style={anchor=west, font=\scriptsize}
]
    % (a) Forward CliNR stays ahead of the descending serial frontier.
    \begin{scope}
        \draw[panel] (-0.15,-0.15) rectangle (5.45,5.05);
        \node[anchor=west, font=\scriptsize\bfseries] at (0.15,4.72)
            {(a) Descending phase};
        \foreach \y in {0.70,1.70,2.70,3.70}
            \draw[lane] (0.35,\y) -- (5.10,\y);
        \draw[axis] (0.35,0.32) -- (5.12,0.32) node[below] {time};
        \draw[axis] (0.35,0.32) -- (0.35,4.18);
        \node[rotate=90] at (-0.42,2.25) {logical qubit index};

        % The red execution trace overlays the completed portion of the
        % violet cleaning trace. The exposed violet tail is the cleaning lead.
        \draw[clinr]
            (0.75,3.70) -- (1.35,3.70) -- (1.35,2.70)
            -- (1.95,2.70) -- (1.95,1.70)
            -- (2.55,1.70) -- (2.55,0.70);
        \draw[serial down]
            (0.75,3.70) -- (1.35,3.70) -- (1.35,2.70);

    \end{scope}

    % (b) At the bottom of the V, reverse cleaning waits for the frame
    % generated by the component to become known.
    \begin{scope}[xshift=6.05cm]
        \draw[panel] (-0.15,-0.15) rectangle (5.45,5.05);
        \node[anchor=west, font=\scriptsize\bfseries] at (0.15,4.72)
            {(b) Turnaround};
        \foreach \y in {0.70,1.70,2.70,3.70}
            \draw[lane] (0.35,\y) -- (5.10,\y);
        \draw[axis] (0.35,0.32) -- (5.12,0.32) node[below] {time};

        % Descent has completed. The blue return begins while the violet
        % cleaner remains parked at the turnaround.
        \draw[clinr trace]
            (0.75,3.70) -- (1.35,3.70) -- (1.35,2.70)
            -- (1.95,2.70) -- (1.95,1.70)
            -- (2.55,1.70) -- (2.55,0.70);
        \draw[serial down]
            (0.75,3.70) -- (1.35,3.70) -- (1.35,2.70)
            -- (1.95,2.70) -- (1.95,1.70)
            -- (2.55,1.70) -- (2.55,0.70);
        \draw[serial up]
            (2.55,0.70) -- (3.15,0.70) -- (3.15,1.70);

    \end{scope}

    % (c) A reverse preparation pass follows the ascending frontier only as
    % far as needed to align the resources for the next serial component.
    \begin{scope}[xshift=12.10cm]
        \draw[panel] (-0.15,-0.15) rectangle (5.45,5.05);
        \node[anchor=west, font=\scriptsize\bfseries] at (0.15,4.72)
            {(c) Ascending phase};
        \foreach \y in {0.70,1.70,2.70,3.70}
            \draw[lane] (0.35,\y) -- (5.10,\y);
        \draw[axis] (0.35,0.32) -- (5.12,0.32) node[below] {time};

        % Blue advances first.  The dashed violet path uses CliNR where a frame
        % needs cleaning and full-block swaps elsewhere, and may stop before
        % reaching the execution head.
        \draw[serial up]
            (0.75,0.70) -- (1.35,0.70) -- (1.35,1.70)
            -- (1.95,1.70) -- (1.95,2.70)
            -- (2.55,2.70) -- (2.55,3.70);
        \draw[clinr, dashed]
            (0.75,0.70) -- (1.35,0.70) -- (1.35,1.70)
            -- (1.62,1.70);
        \draw[black, line width=1.2pt] (1.62,1.47) -- (1.62,1.93);
        \node[anchor=west, font=\scriptsize] at (1.74,1.34)
            {stopping point};

    \end{scope}

    % Compact legend outside the panels keeps the traces themselves uncluttered.
    \begin{scope}[xshift=17.85cm]
        \node[anchor=west, font=\scriptsize\bfseries] at (0.00,3.90) {Legend};
        \draw[serial down, -] (0.00,3.35) -- (0.72,3.35);
        \node[legend label] at (0.92,3.35) {execution (down)};
        \draw[serial up, -] (0.00,2.70) -- (0.72,2.70);
        \node[legend label] at (0.92,2.70) {execution (up)};
        \draw[clinr trace] (0.00,2.05) -- (0.72,2.05);
        \node[legend label] at (0.92,2.05) {CliNR};
    \end{scope}
\end{tikzpicture}%
    }
    \caption{Idle frame clearing around a serial component. Time runs from
    left to right, and vertical position denotes logical-qubit index. (a)
    CliNR cleans the frames of code blocks that are pending
    execution in a serial component. (b) Execution in a serial
    component reaches the end of the code blocks, and must now uncompute them,
    iterating in reverse order. (c) CliNR cleans the frames of the uncomputed blocks
    up to the first block of the next serial component.}
    \label{fig:idle-frame-clearing-pipeline}
\end{figure*}

However, in order to keep up with the serial execution, we must ensure that
logical CliNR can clear a block in less time than the serial execution takes to process all measurements in
a code block.
To achieve this, we optimize logical CliNR by measuring MECM stages in parallel when
they are compatible under the rules of parallel cat-state measurements.
Within each MECM stage, we connect two outputs in a compatibility graph when
their physical representatives can be chosen with disjoint supports on both
CliNR resource blocks and with weight at most 30, then find a
maximum matching; unmatched stages are measured individually.
By sampling uniformly from the Clifford frames encountered in a simulation of
the compiled elliptic-curve point-addition circuit, we find that the optimized
MECM portion of a frame clear takes \(41.242\pm0.276\) Q102 SEC layers. The
transversal \CNOT{} gate is followed by one additional Q102 SEC for decoding, so a
complete frame clear takes \(42.242\pm0.276\) Q102 SECs. Using the joint \CCZ{}
injection duration of approximately \(5.46\) Q102 SECs derived in
\cref{eq:parallel-ccz-viterbi-duration}, this corresponds to approximately \(7.73\)
logical-measurement intervals to clear a block.  The
serial components expose a new block after approximately \(7.5\) such
intervals; it therefore requires at least two instances
of logical CliNR to keep up with the serial execution.

Because we apply logical CliNR only to serial components, the compiler can lay
out the data blocks used by those components as an ordered register following
their serial traversal order. The timing estimate above shows that two logical
CliNR instances are required to ensure that frames are always cleared ahead
of serial execution. We divide the register between them by alternating
blocks: one instance clears the blocks in even positions, and the other clears
those in odd positions. Thus, each instance clears only every other block.
Since, typically, execution of serial components in the elliptic curve point addition
circuits reaches a new block every \(7.5\)
logical-measurement intervals, each CliNR instance
has \(2\times7.5=15\) intervals before it is needed again, enough time to clear
a block in approximately \(7.73\) intervals.

The transversal \CNOT{} gate in logical CliNR requires data qubits to be
transported close to each other to perform pairwise physical \CNOT{} gates.
We therefore keep the \(A\) resource block in the logical CliNR protocol
adjacent to the data block during injection. Each data qubit then moves by at
most the width of one Q102 block to interact with its corresponding resource
qubit, minimizing errors due to idling, loss, and leakage during transport.

The \(\bitcoincurve\) device introduced in \cref{sec:enhanced-architecture} has
a 69-data-block register. Assigning alternating blocks to the two logical
CliNR instances divides this register into groups of 35 and 34 data blocks.
Adding one \(A/B\) resource pair, or two blocks, to each group gives 37 and 36
blocks. Each group fits on a closed nearest-neighbor path through a
\(4\times10\) rectangle, which provides 40 cells. The two rectangles therefore
leave three and four cells empty, respectively. Stacking them gives
an \(8\times10\) layout containing the 69 data blocks, four resource blocks,
and seven empty cells.
We call this assignment of the alternating block subsequences and their
resource pairs to closed nearest-neighbor paths the \emph{ouroboros block
layout}.
As the application processes further serial components, the compiler can append
each new traversal to the end of the current execution path, which wraps in a
circular pattern around the data blocks---hence the name ``ouroboros.''

Under the ouroboros block layout, it remains to show that each logical CliNR
instance can follow the path of serial execution on its data blocks while
retaining nearest-neighbor locality with the next target
(and, with this condition, execute a transversal \CNOT{} gate with minimal overhead).
A single application of logical CliNR results in a swap of the \(B\) and
data blocks, leaving \(A\) adjacent to the next data block.
If \(D_i\) has no frame to clear, one can simply cyclically shift the blocks, implemented by swapping
corresponding physical qubits between
the blocks. Induction over these local permutations shows that the
resource pair remains adjacent to its next target for the entire sequence of clearing
a serial application component ahead of execution.
\Cref{fig:clinr-perimeter-routing} illustrates the ouroboros block layout for
cleaning three partially overlapping serial components.

\begin{figure*}[t]
    \centering
    \resizebox{0.90\textwidth}{!}{%
        \input{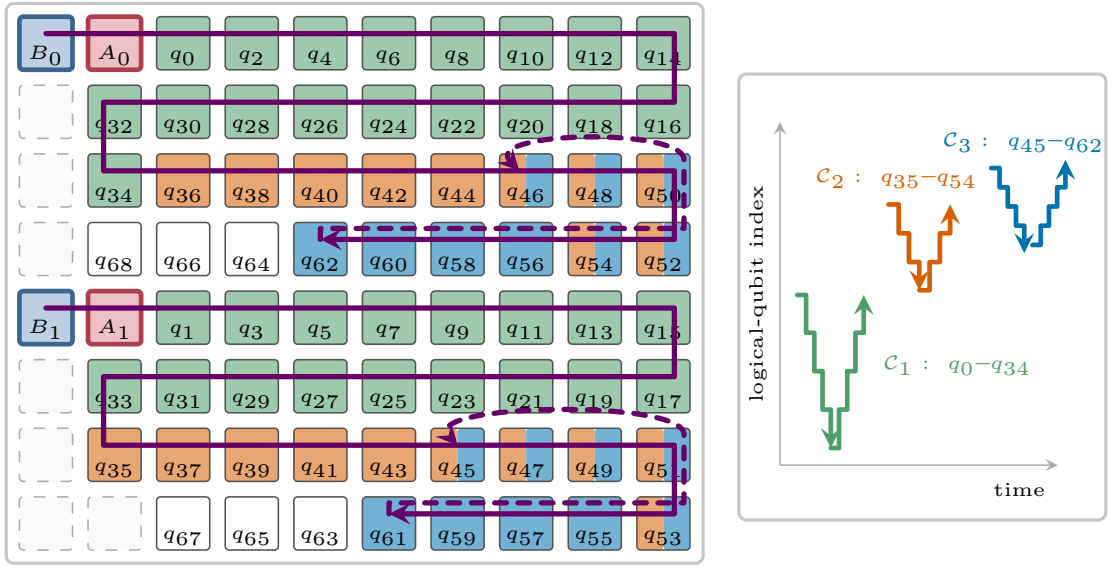}%
    }
    \caption{Logical CliNR for three serial components with partially
    overlapping supports. Each pair \((B_i,A_i)\) provides the two resource
    blocks for one CliNR instance. The diagram shows all 69 data blocks,
    \(q_0,\ldots,q_{68}\); the white blocks \(q_{63},\ldots,q_{68}\) are unused.
    The diagram on the right shows the three serial components executing in
    sequence in application order. The first component acts on
    \(q_0,\ldots,q_{34}\). The second acts on \(q_{35},\ldots,q_{54}\), and the
    third acts on \(q_{45},\ldots,q_{62}\), overlapping the second component on
    \(q_{45},\ldots,q_{54}\). The two CliNR instances begin at
    the first even- and odd-indexed blocks of each support, then clear frames
    ahead of execution along the solid purple paths. Because the third component
    overlaps the second, the resource pairs must backtrack to the start of the
    overlap before they can begin clearing the third component; the dashed
    purple paths show this backtracking.}
    \label{fig:clinr-perimeter-routing}
\end{figure*}

The ouroboros block layout also accommodates a component that starts within
the support of its predecessor. In this case,
the logical CliNR instance follows the predecessor's execution across the ``ascending''
execution of logical measurements,
cleaning the frame of blocks with a nontrivial frame
and simply making full-block swaps to the intervening blocks whose frames are already clean.
They may stop as soon as they reach the first target of the later component;
see \cref{fig:clinr-component-transition}. In the worst case, the last block
released by the earlier component is also the first block required by the
later one. Its clearing cannot be hidden behind earlier execution, so the
transition stalls for one logical CliNR operation, or approximately \(7.73\)
logical-measurement intervals.

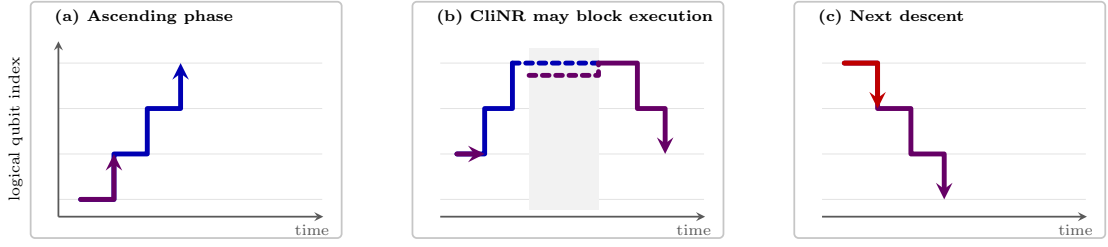
\begin{figure*}[t]
    \centering
    \resizebox{0.82\textwidth}{!}{%
        \begin{tikzpicture}[
    x=0.88cm,
    y=0.72cm,
    font=\scriptsize,
    >=stealth,
    panel/.style={draw=gray!45, rounded corners=2pt, thick},
    lane/.style={gray!22},
    axis/.style={gray!70!black, thick, ->},
    serial down/.style={
        red!75!black,
        line width=2.2pt,
        line cap=round,
        line join=round,
        ->
    },
    serial up/.style={
        blue!70!black,
        line width=2.2pt,
        line cap=round,
        line join=round,
        ->
    },
    clinr/.style={
        violet!80!black,
        line width=2.2pt,
        line cap=round,
        line join=round,
        ->
    }
]
    % (a) Reverse preparation follows the ascending frame-resolution frontier
    % only where needed for a later overlapping component.
    \begin{scope}
        \draw[panel] (-0.15,-0.15) rectangle (5.45,5.05);
        \node[anchor=west, font=\scriptsize\bfseries] at (0.15,4.72)
            {(a) Ascending phase};
        \foreach \y in {0.70,1.70,2.70,3.70}
            \draw[lane] (0.35,\y) -- (5.10,\y);
        \draw[axis] (0.35,0.32) -- (5.12,0.32)
            node[below, xshift=-0.18cm] {time};
        \draw[axis] (0.35,0.32) -- (0.35,4.18);
        \node[rotate=90] at (-0.42,2.25) {logical qubit index};

        \draw[serial up]
            (0.75,0.70) -- (1.35,0.70) -- (1.35,1.70)
            -- (1.95,1.70) -- (1.95,2.70)
            -- (2.55,2.70) -- (2.55,3.70);
        \draw[clinr]
            (0.75,0.70) -- (1.35,0.70) -- (1.35,1.70);
    \end{scope}

    % (b) A final required cleanup may extend beyond the component boundary.
    \begin{scope}[xshift=6.05cm]
        \draw[panel] (-0.15,-0.15) rectangle (5.45,5.05);
        \node[anchor=west, font=\scriptsize\bfseries] at (0.15,4.72)
            {(b) CliNR may block execution};
        \foreach \y in {0.70,1.70,2.70,3.70}
            \draw[lane] (0.35,\y) -- (5.10,\y);
        \draw[axis] (0.35,0.32) -- (5.12,0.32)
            node[below, xshift=-0.18cm] {time};

        % The incoming ascending execution reaches the component boundary.
        \draw[serial up, -]
            (0.65,1.70) -- (1.15,1.70) -- (1.15,2.70)
            -- (1.65,2.70) -- (1.65,3.70);
        \draw[clinr]
            (0.65,1.70) -- (1.15,1.70);

        % The shaded tail is needed only when the final required cleanup cannot
        % complete before its frame is released at the component boundary.
        \fill[gray!10] (1.95,0.48) rectangle (3.20,4.03);
        \draw[serial up, dashed, -]
            (1.65,3.70) -- (3.20,3.70);
        \draw[violet!80!black, line width=2.2pt, line cap=round,
              line join=round, dashed]
            (1.95,3.43) -- (3.20,3.43) -- (3.20,3.70);

        % Preparation continues just in time for the next descent.
        \draw[clinr]
            (3.20,3.70) -- (3.90,3.70) -- (3.90,2.70)
            -- (4.40,2.70) -- (4.40,1.70);
    \end{scope}

    % (c) The next descending execution begins once its first target is clean.
    \begin{scope}[xshift=12.10cm]
        \draw[panel] (-0.15,-0.15) rectangle (5.45,5.05);
        \node[anchor=west, font=\scriptsize\bfseries] at (0.15,4.72)
            {(c) Next descent};
        \foreach \y in {0.70,1.70,2.70,3.70}
            \draw[lane] (0.35,\y) -- (5.10,\y);
        \draw[axis] (0.35,0.32) -- (5.12,0.32)
            node[below, xshift=-0.18cm] {time};

        \draw[clinr]
            (0.75,3.70) -- (1.35,3.70) -- (1.35,2.70)
            -- (1.95,2.70) -- (1.95,1.70)
            -- (2.55,1.70) -- (2.55,0.70);
        \draw[serial down]
            (0.75,3.70) -- (1.35,3.70) -- (1.35,2.70);
    \end{scope}

\end{tikzpicture}%
    }
    \caption{Logical CliNR on consecutive serial components with overlapping
    support. Time runs from left to right, and vertical position denotes
    logical-qubit index. (a) Execution of a serial component begins uncomputing results,
    iterating across blocks in reversed order. (b) If the last block
    of a serial component aligns with the start of the next one, the logical CliNR
    instance working on the first component must block execution. Transitions to disjoint or already
    prepared supports avoid this. (c) The next execution begins
    once its first target is clean, while logical CliNR continues just in time
    ahead of it.}
    \label{fig:clinr-component-transition}
\end{figure*}

\subsubsection{Depth-One \CCZ{} Injection}
\label{sec:depth-one-ccz-injection}

\begin{figure*}[t]
    \centering
    \includegraphics[width=\linewidth]{%
        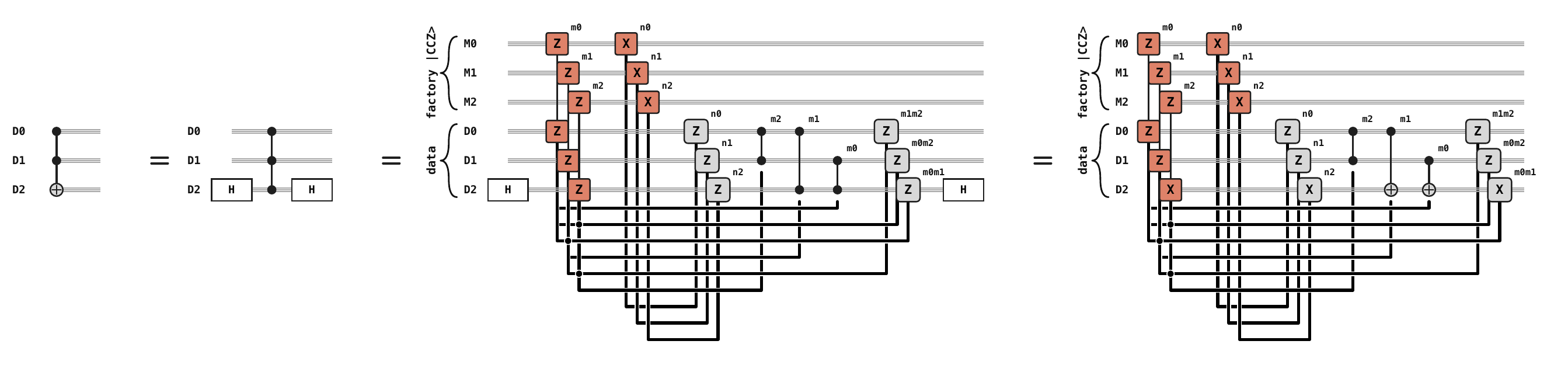}
    \caption{Circuit equivalence between a Toffoli gate, a \CCZ{} gate
    conjugated by Hadamard gates on the target qubit, a \CCZ{} injection circuit
    conjugated by Hadamard gates on the target qubit, and the same circuit with
    Hadamards implemented via conjugation with the logical measurements of the
    \CCZ{} injection circuit.}
    \label{fig:toffoli-ccz-injection-equivalence}
\end{figure*}

As shown in \cref{fig:toffoli-ccz-injection-equivalence}, a Toffoli gate is a
\CCZ{} gate conjugated by Hadamard gates on its target qubit. The \CCZ{} gate can be
injected via one-bit gate teleportation~\cite{shor1996faulttolerant,zhou2000methodology}
by measuring a joint logical-
\(Z\) operator between each qubit of the \CCZ{} resource state and the
corresponding data qubit, followed by destructive \(X\)-basis measurements of
the resource qubits and outcome-dependent corrections. The three joint logical-
\(Z\) measurements have disjoint supports and can therefore be performed in
parallel, so the injection requires only one logical-measurement layer. The
subsequent \(X\)-basis outcomes induce only Pauli corrections, which can be
tracked in a Pauli frame and need not be available until the following gate
has completed. Their readout is generally faster than the \(5.46\) Q102 SECs
required for the joint logical measurements, and hence does not add another
layer to the critical path. Lastly, the
compiler ensures that both operands of every conditional \CZ{} correction
reside in the same memory block, as described in
\cref{sec:compiler-components-routing}, so each correction can be absorbed into
the block's Clifford frame without adding a logical-measurement layer.

We quantify the measurement depth after compilation using Monte Carlo
simulation. For each independently compiled circuit component, we draw
\(10{,}000\) samples of random outcomes for conditional Clifford
corrections. Each sample begins with a
trivial Clifford frame on every input block, as required by the compilation
invariant and guaranteed during execution by the idle-frame-clearing procedure
of \cref{sec:idle-frame-clearing}. After propagating the sampled Clifford
frames and scheduling the resulting logical measurements, we find that
every sampled conjugated logical operator has an accessible physical representative.
The empirical frequency resolution is therefore \(10^{-4}\) per component;
under independent sampling, observing no inaccessible trajectory gives the
one-sided \(95\%\) upper confidence bound
\(3\times10^{-4}\) on the probability that a trajectory
contains an inaccessible logical operator, for each component.
Through this simulation of all the application components, we find that \(93.6\%\) of
Toffoli gates have measurement depth one. The remaining gates require depth
two, giving a mean measurement depth of \(1.07\).

\subsection{Eastinthillation \CCZ{} Factory}
\label{sec:ccz-factory}

We design a dedicated ``Eastinthillation'' factory (Eastin + synthillation) that produces \CCZ{} states directly in a three-ring
code block. We note that the alternative name ``East Mode'' may be used in less formal contexts---
for instance, in fan-made songs about this paper.
The Eastinthillation factory's resource totals and performance estimates are summarized below.

\begin{table}[ht]
    \centering
    \caption{Resource and performance summary for one Eastinthillation factory.}
    \label{tab:ccz-factory-summary}
    \begin{tabular}{@{}lr@{}}
        \toprule
        Metric & Value \\
        \midrule
        \CCZ{} injection depth & 147.48 POCs \\
        Per-attempt restart rate & \(\lesssim 15.94\%\) \\
        Average depth per accepted state & 555.75 POCs \\
        Number of factories for a continuous stream & 4 \\
        Physical-qubit footprint per factory & 319 \\
        Logical error rate & \(\lesssim 9.36\times10^{-10}\) \\
        \bottomrule
    \end{tabular}
\end{table}

The factory applies the Eastin Toffoli-state synthillation
protocol~\cite{eastin2013distilling,campbell2017unified} to eight \(T\) states
produced by zero-level distillation in the \([\![6,2,2]\!]\) code
(henceforth referred to as the C6 code).
This \(T\)-state distillation procedure is detailed in
\cref{sec:zero-level-t-state-distillation}. On acceptance, the protocol outputs
a Toffoli state. Applying a Hadamard to its target qubit converts it into a \CCZ{}
resource state.
The Eastinthillation factory circuit and its physical measurement schedule are described in
\cref{sec:eastin-factory}, while \cref{sec:ccz-factory-analysis} analyzes the
factory's logical error and rejection rates, depth, and physical-qubit
footprint. Finally we describe how to adapt the Eastinthillation factory to produce the \(T\) states required for the QFT
stage of Shor's algorithm in
\cref{sec:t-state-factory-conversion}.

\subsubsection{0-Level \texorpdfstring{\(T\)}{T}-State Distillation}
\label{sec:zero-level-t-state-distillation}

We prepare each pair of \(\ket{H}\) states using the
protocol of Dasu et al.~\cite{dasu2025breaking}. The protocol uses Goto's
arbitrary-state encoder~\cite{goto2014step} to encode the states in the C6 code
before joint logical-Hadamard verification and syndrome
extraction. The two verified \(\ket{H}\) states remain
encoded and are consumed through the dual-\(T\)-state injection circuit
of the Walking Cat architecture~\cite{tripier2026walkingcat}.
The joint verification measures
\(H_{\mathrm{L},0}H_{\mathrm{L},1}\) using a two-qubit Bell state.  A flagged
readout followed by C6 syndrome extraction makes this measurement
fault-tolerant.

The preparation circuit for two logical \(\ket{H}\) states in the C6 code is
shown in
\cref{fig:h-622-h2-magic-state-prep-with-sec}. As a three-ring code,
we may implement the C6 code's SEC circuit with cyclic shifts, but the
encoding circuit's connectivity cannot be achieved with just cyclic shifts.
We instead propose using a short sequence of column swaps to implement the encoding circuit.
Separate loss and leakage checks are unnecessary for this error-detecting code block because every
qubit is destructively measured with loss-and-leakage detection during either state
preparation or consumption.

\begin{figure*}[t]
    \centering
    \includegraphics[width=\textwidth]{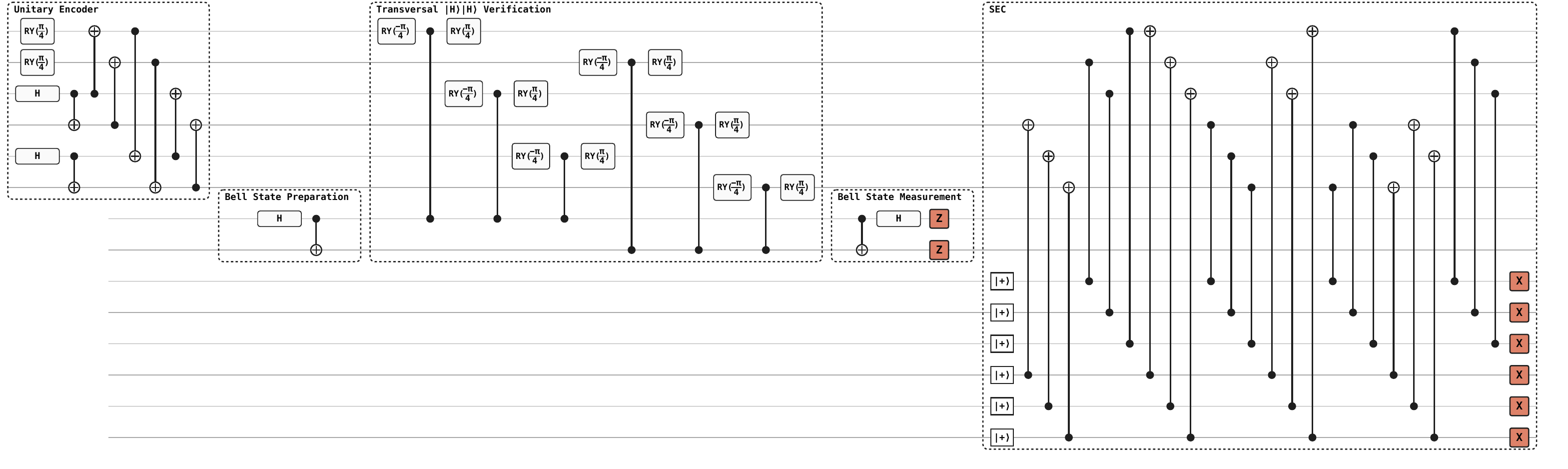}
    \caption{Preparation of two logical \(\ket{H}\) states in the C6 code.}
    \label{fig:h-622-h2-magic-state-prep-with-sec}
\end{figure*}

We simulated the state-preparation, verification, and SEC circuit using the
statevector method in Qiskit Aer. The logical error and rejection rates in
\cref{tab:zero-level-t-state-summary} are Monte Carlo estimates with
\(10^9\) shots using the swap loss noise model at \(p=10^{-4}\).
The simulator appends ideal, destructive measurements of the two
logical Hadamard operators \(H_{\mathrm{L},0}\) and
\(H_{\mathrm{L},1}\) to estimate logical fidelity. The results of this simulation
are included in \cref{tab:zero-level-t-state-summary}, which also reports the
depth and physical-qubit footprint of the circuit.

\begin{table}[ht]
    \centering
    \caption{Performance estimates and resource totals for one zero-level preparation
    of an encoded \(\ket{H}^{\otimes 2}\) pair using the swap-loss noise model at \(p=10^{-4}\).}
    \label{tab:zero-level-t-state-summary}
    \renewcommand{\arraystretch}{1.2}
    \begin{tabular}{@{}lcr@{}}
        \toprule
        Quantity & Symbol & Value \\
        \midrule
        Logical error rate per accepted pair
            & \(p_{\mathrm{L}}^{(0)}\) & \(1.30\times 10^{-8}\) \\
        Rejection rate
            & \(r_{\mathrm{rej}}^{(0)}\) & \(0.358\%\) \\
        C6 SEC rejection rate
            & \(r_{\mathrm{SEC}}^{\mathrm{C6}}\) & \(0.323\%\) \\
        Pair preparation and verification depth
            & \(T_{H^2}^{(0)}\) & \(6.6\) POCs \\
        SEC depth
            & \(T_{\mathrm{SEC}}^{(0)}\) & \(5.4\) POCs \\
        Total depth
            & \(T_{\mathrm{dist}}^{(0)}\) & \(12.0\) POCs \\
        Physical-qubit footprint
            & \(n_{\mathrm{phys}}^{(0)}\) & 12 \\
        \bottomrule
    \end{tabular}
\end{table}

\subsubsection{Eastinthillation factory circuit}
\label{sec:eastin-factory}

The Eastinthillation factory circuit, based on the Eastin Toffoli-state synthillation
protocol~\cite{eastin2013distilling}, is shown in
\cref{fig:eastin-synthillation-circuit}. The protocol uses eight \(\ket{H}\)
states from the zero-level factories based on the C6 code to implement two
Margolus--Toffoli
gates with shared controls and distinct targets. Let \(p_{Y,\mathrm{inj}}\)
denote the stochastic logical-\(Y\) fault probability per \(\ket{H}\)-state
injection used to implement a logical \(R_Y(\pi/4)\) rotation. A final
``acceptance'' parity check on the targets
detects any single stochastic-\(Y\) input fault and yields a Toffoli state with
error probability of order \(p_{Y,\mathrm{inj}}^2\).
We apply a Hadamard to the target qubit of the accepted Toffoli state. This
final Clifford converts it into a \CCZ{} resource state compatible with the
depth-one \CCZ{} gate of \cref{sec:claw-isa}.

\begin{figure}[ht]
    \centering
    \includegraphics[width=0.9\linewidth]{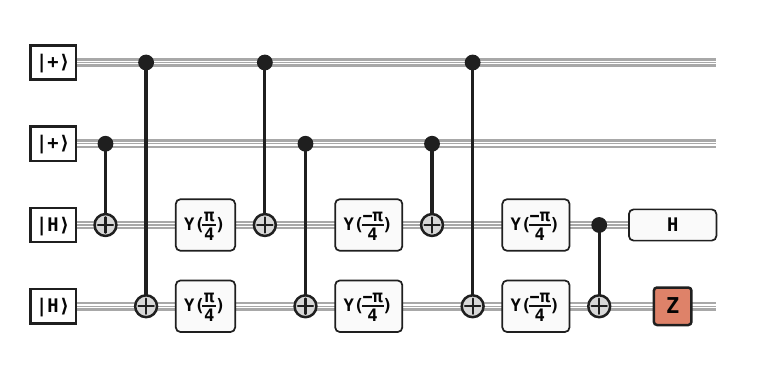}
    \caption{The Eastinthillation factory circuit.}
    \label{fig:eastin-synthillation-circuit}
\end{figure}

We implement the Eastinthillation factory circuit in a single Q66 code block using the measurement-layer
sequence shown in \cref{fig:eastin-synthillation-factory-measurements}. The
measurements use the LM1, LM2, and DM operations from the Walking Cat instruction
set~\cite{tripier2026walkingcat}, mediated by EDM (error-detecting measurement) and Viterbi measurement schemes.
We extend these schemes to include the Multi-Pauli error-corrected
measurement operation (MECM) introduced by Webster and Delfosse~\cite{webster2026fast}.

Every cat-state-mediated measurement layer in
\cref{fig:eastin-synthillation-factory-measurements} admits physical
representatives whose supports on the GB66 block are pairwise disjoint, with
each representative having weight at most 16. The complete representative
certificates are given in
\cref{app:gb66-representative-certificates}. The relatively
small width of these representatives allows us to use only
3 rounds of EDM for each double-\(R_Y(\pi/4)\) injection to achieve
a \(p_H\) of order \(10^{-8}\) at physical error rate \(p=10^{-4}\). Similarly,
we require only 5 rounds of Viterbi measurements to achieve logical error rates
of order \(10^{-13}\). Each double-\(R_Y(\pi/4)\) injection ends in a
DMX -- a destructive readout of all data qubits in the H6 block,
an operation which takes only 1 POC to execute regardless of width.

\paragraph*{\normalfont\bfseries \(R_Y(\pi/4)\) injection.}
Each pair of
\(R_Y(\pi/4)\) rotations is implemented with the double-\(R_Y(\pi/4)\)
injection gadget of the Walking Cat architecture~\cite{tripier2026walkingcat}, which involves a joint \(YY\) measurement
followed by a DMX operation on the C6 block.
To avoid correlated input errors, each Q66 block has two dedicated C6 source
blocks.  A double-\(R_Y(\pi/4)\) layer takes one logical qubit, prepared in
\(\ket{H}\), from each C6 block rather than taking both inputs from one
encoded \(\ket{H}^{\otimes 2}\) pair.  The terminal DMX consumes both C6
blocks, including the unused logical qubit in each block.  The two blocks are
then reinitialized and their next encoded \(\ket{H}^{\otimes 2}\) pairs are
prepared and verified in parallel for the same Q66 block.
We use dedicated C6 blocks rather than sharing them among multiple Q66 blocks.
This choice simplifies the factory analysis by eliminating contention for shared resources.

Within either source C6 block, the two logical-\(Y\) operators have
representatives
\begin{equation}
    \overline Y_0=\texttt{YIYIYI},
    \qquad
    \overline Y_1=\texttt{IYIYIY}.
\end{equation}
These representatives have disjoint supports and weight three. Nevertheless,
either operator can be measured using a two-qubit cat-state segment (a Bell
state) by reusing one cat qubit for two data interactions. This reuse does not
reduce the measurement distance. A fault on the reused cat qubit can propagate
to two C6 data qubits. We choose the data-qubit support of each cat-state qubit
so that no resulting hook error is a nontrivial C6 logical operator. Every such
error has a nonzero syndrome and is rejected by the following SEC.

Combining the two-qubit C6 segment with a Q66 representative of weight at most
16 gives a weight-18 cat state for the joint
\(Y_{\mathrm{C6}}Y_{\mathrm{Q66}}\) observable in the double-\(R_Y(\pi/4)\)
injection.

Each Q66 cat-state factory contains two cat-state registers and one
verification bank, all provisioned at width 16.  Increasing one factory's width to
18 therefore requires \(3(18-16)=6\) additional qubits.  We keep a six-qubit
\emph{C6 Bell-state source} next to each C6 block.  For C6 state verification, two
of these qubits form the Bell state described above.  During the joint
\(Y_{\mathrm{C6}}Y_{\mathrm{Q66}}\) measurement, however, all six qubits are
loaned to the weight 16 Q66 cat-state factory so it may create
a weight 18 cat state.  These two operations do not overlap, so
the same C6 Bell-state source can support both.

The logical measurement still contains 19 data--cat
interactions, so the logical-measurement error estimate below uses effective
width \(w_{\mathrm{eff}}=19\).

At \(w_{\mathrm{eff}}=19\) and \(p=10^{-4}\), the fitted EDM model of
Ref.~\cite{tripier2026walkingcat} gives the single-round outcome-flip
probability
\begin{equation}
    p_{\mathrm{flip}}\simeq 2.1w_{\mathrm{eff}}p
    =3.99\times10^{-3}.
\end{equation}
It takes three EDM rounds to suppress the measurement-induced error to the same
order of magnitude as the state preparation error (\(p_{H^2}\)),
since \(p_{\mathrm{flip}}^2=1.59\times10^{-5}\) whereas
\(p_{\mathrm{flip}}^3=6.35\times10^{-8}\). We therefore use EDM-3 for each
\(R_Y(\pi/4)\) injection. Its logical-outcome error rate is
\begin{equation}
    p_{\mathrm{EDM3}}
    \simeq p_{\mathrm{flip}}^3
    =6.35\times10^{-8}.
\end{equation}
An EDM-3 outcome sequence is accepted when its three outcomes are all correct
or all flipped.  The two inputs to a double injection come from the two
dedicated C6 blocks of the same Eastinthillation factory.  Each input undergoes three
C6 SECs, and the factory accepts only if both C6 blocks accept all three.
Using the common C6 SEC rejection rate
\(r_{\mathrm{SEC}}^{\mathrm{C6}}=0.323\%\), the per-injection acceptance
factor is therefore
\begin{equation}
    a_{R_Y}
    =\bigl[(1-p_{\mathrm{flip}})^3+p_{\mathrm{flip}}^3\bigr]
     (1-r_{\mathrm{SEC}}^{\mathrm{C6}})^3
    =0.978534.
    \label{eq:eastin-ry-injection-acceptance}
\end{equation}
The corresponding rejection rate is
\(r_{R_Y}=1-a_{R_Y}=2.147\%\). The factor \(a_{R_Y}^8\) accounts for the
eight EDM-3 outcome sequences and the 24 C6 SECs required by one \CCZ{} preparation
attempt.

The parity check at the end of the Eastinthillation factory can only detect single-logical-qubit \(Y\) or \(X\) errors; and all errors
in the synthillation protocol can effectively be traced to \(R_Y(\pi/4)\) injections. So we would like to ensure
that \(R_Y(\pi/4)\) injection errors are only either \(Y\)- or \(X\)-type errors. We rely on the injected-\(Y\) noise model of the Walking Cat
architecture~\cite{tripier2026walkingcat}: conditioned on the EDM and SEC
acceptance tests, an injection fault is represented by a stochastic logical
\(Y\) on the corresponding Q66 qubit. This follows from analyzing the
conditional corrections of the logical measurements involved in the \(R_Y(\pi/4)\) injection.
A wrong terminal DMX result erroneously applies a
\(Y\) correction. A wrong EDM result applies a
\(Z_{\mathrm{C6}}R_Y(\pi/2)_{\mathrm{Q66}}\) error; the extra
\(Z_{\mathrm{C6}}\) triggers the subsequent C6 \(X\) readout and therefore
also applies a \(Y\) error via DMX \(Y\) correction (so, at the very least, is always correlated with a \(Y\) error).
The residual Q66 operation is
\(R_Y(\pi/2)Y\doteq R_Y(-\pi/2)\), or the inverse rotation.
The source \(\ket{H}\) state has zero \(Y\) expectation. The two outcomes of
the joint \(Y_{\mathrm{C6}}Y_{\mathrm{Q66}}\) measurement, and hence the two
residual rotations, are equiprobable. Averaging the corresponding channels
gives
\begin{equation}
    \frac{1}{2}\mathcal{R}_{Y,\pi/2}
    +\frac{1}{2}\mathcal{R}_{Y,-\pi/2}
    =\frac{1}{2}\mathcal{I}+\frac{1}{2}\mathcal{Y},
    \label{eq:eastin-injection-y-channel}
\end{equation}
where \(\mathcal{R}_{Y,\theta}\) is the channel induced by \(R_Y(\theta)\) and
\(\mathcal{Y}\) is conjugation by logical \(Y\). We therefore bound the
stochastic logical-\(Y\) fault rate by the full EDM-3 logical-outcome error
probability.

The two input states in each injection layer have independent preparation
errors because they come from distinct C6 blocks. Because
faults at these locations are sampled independently, the
preparation errors are also independent. Consequently, for a
single injection, the marginal stochastic logical-\(Y\) probability is bounded
by the sum of the preparation and EDM-3 error probabilities:
\begin{equation}
    p_{Y,\mathrm{inj}}
    \lesssim p_{H^2}+p_{\mathrm{EDM3}}
    =7.65\times10^{-8}.
    \label{eq:eastin-effective-injection-y-error}
\end{equation}

\paragraph*{\normalfont\bfseries Acceptance Measurement.}
The acceptance measurement uses a LM1 operation mediated by Viterbi
measurement. Its representative has weight 16, so
\(\wcat=16\), as shown in
\cref{tab:gb66-acceptance-certificates}. At \(p=10^{-4}\), the single-round
bit-flip model of Ref.~\cite{tripier2026walkingcat} gives
\begin{equation}
    p_{\mathrm{flip},16}=2.1\wcat p
    =2.1\times16\times10^{-4}
    =3.36\times10^{-3}.
\end{equation}
For target error \(\varepsilon=10^{-11}\), the Viterbi stopping condition
requires vote margin
\begin{equation}
    d_{16}=\left\lceil
        \frac{\log[(1-\varepsilon)/\varepsilon]}
             {\log[(1-p_{\mathrm{flip},16})/p_{\mathrm{flip},16}]}
      \right\rceil
    =5.
\end{equation}
The logical error is therefore
\begin{equation}
    p_{\mathrm{Vit},16}
    =\left[1+\left(\frac{1-p_{\mathrm{flip},16}}
    {p_{\mathrm{flip},16}}\right)^5\right]^{-1}
    =4.36\times10^{-13}.
\end{equation}
The cat-state rejection model of Ref.~\cite{tripier2026walkingcat} gives
\(p_{\mathrm{miss},16}=5\wcat p=8.0\times10^{-3}\). Writing
\(\rho_{16}=p_{\mathrm{flip},16}/(1-p_{\mathrm{flip},16})\), the average
duration is
\begin{equation}
    \tau_{\mathrm{Vit},16}^{\mathrm{avg}}
    =\frac{1}{1-p_{\mathrm{miss},16}}
      \frac{d_{16}}{1-2p_{\mathrm{flip},16}}
      \frac{1-\rho_{16}^{d_{16}}}{1+\rho_{16}^{d_{16}}}
    =5.07~\mathrm{SEC}.
\end{equation}

\paragraph*{\normalfont\bfseries \CCZ{} Injection.}
The three \CCZ{} injection measurements follow the protocol of
\cref{sec:claw-isa}. The largest joint representative has weight 36,
with weight 20 from the Q102 logical \(Z\) operator and weight 16 from the Q66
logical \(Z\) operator. We therefore use \(\wcat=36\) for all three
measurements to deduce an upper bound on the logical error of the measurements.
At \(p=10^{-4}\), the single-round bit-flip probability is
\begin{equation}
    p_{\mathrm{flip},36}
    =2.1\times36\times10^{-4}
    =7.56\times10^{-3}.
\end{equation}
For target error \(\varepsilon=10^{-10}\), the Viterbi stopping condition
requires vote margin
\begin{equation}
    d_{36}=\left\lceil
        \frac{\log[(1-\varepsilon)/\varepsilon]}
             {\log[(1-p_{\mathrm{flip},36})/p_{\mathrm{flip},36}]}
      \right\rceil
    =5.
\end{equation}
The logical error per measurement is therefore
\begin{equation}
    p_{\mathrm{Vit},36}
    =\left[1+\left(\frac{1-p_{\mathrm{flip},36}}
    {p_{\mathrm{flip},36}}\right)^5\right]^{-1}
    =2.57\times10^{-11}.
\end{equation}
The cat-state rejection model of Ref.~\cite{tripier2026walkingcat} gives
\(p_{\mathrm{miss},36}=5\wcat p=1.8\times10^{-2}\). Writing
\(\rho_{36}=p_{\mathrm{flip},36}/(1-p_{\mathrm{flip},36})\), the average
duration of one Viterbi measurement is
\begin{equation}
\begin{aligned}
    \tau_{\mathrm{Vit},36}^{\mathrm{avg}}
    & =\frac{1}{1-p_{\mathrm{miss},36}}
      \frac{d_{36}}{1-2p_{\mathrm{flip},36}}
      \\
    &\quad\times\frac{1-\rho_{36}^{d_{36}}}{1+\rho_{36}^{d_{36}}}
      =5.16982~\mathrm{SEC}.
\end{aligned}
    \label{eq:single-weight-36-viterbi-duration}
\end{equation}
The three measurements run in parallel, so their duration is set by the
slowest measurement. Let \(F(t)\) be the probability that one weight-36
Viterbi measurement has halted by SEC \(t\). We obtain \(F(t)\) by iterating
the same vote-margin process used above. A cat-state attempt is rejected with
probability \(p_{\mathrm{miss},36}\); it then yields no measurement outcome,
so the vote margin is unchanged in that SEC. Since the three measurements use
independent cat states, the probability that all three have halted by SEC
\(t\) is \(F(t)^3\). Summing the probability that at least one measurement is
still running after SEC \(t\) gives
\begin{equation}
    \tau_{\CCZ}^{\mathrm{avg}}
    =\sum_{t=0}^{\infty}\left[1-F(t)^3\right]
      =5.46218~\mathrm{SEC}.
\label{eq:parallel-ccz-viterbi-duration}
\end{equation}
This is only \(0.29~\mathrm{SEC}\) longer than the mean duration of one
measurement. A union bound gives a logical-error contribution of at most
\(3p_{\mathrm{Vit},36}=7.70\times10^{-11}\) from the three measurements.

\paragraph*{\normalfont\bfseries Phase Fixups.}
After the three \CCZ{} injection measurements, we measure the three commuting
resource-block observables \(P_0,P_1,P_2\) obtained by conjugating the
logical-\(X\) operators by the tracked Clifford frame. This frame includes one
conditional Clifford correction from each of the eight \(R_Y(\pi/4)\)
injections, so we analyze all \(2^8=256\) frame branches.

We call a destructive readout that measures each physical data qubit in an
independently chosen Pauli basis a \emph{generalized destructive measurement}
(DMG). A single DMG measures logical Pauli observables with pairwise-disjoint
physical representatives in parallel: on the support of each representative,
the physical qubits are measured in the corresponding single-qubit Pauli
bases. A DMG takes only one POC--less than a full SEC.

Pairwise-disjoint representatives of all three \(P_i\) are not available on
every branch. When they are available, a single DMG measures the three \(P_i\)
in parallel. If exactly one pair overlaps, we measure one member of that pair
using a single Viterbi measurement.
We then measure the other two observables, whose
representatives are disjoint, in parallel with a single DMG. If every
pair overlaps, we measure products of the \(P_i\) that admit
disjoint representatives. In both cases requiring nondestructive logical measurement layers,
only one such layer precedes the terminal DMG. The phase fixups therefore
require at most one measurement layer.

We describe the MECM and Viterbi cases using the
scheduler-matrix convention of Ref.~\cite{webster2026fast}: for
\(G\in\F_2^{3\times3}\), column \(j\) specifies the observable
\begin{equation}
    Q_j=\prod_{i=0}^{2}P_i^{G_{ij}}.
    \label{eq:gb66-scheduler-observables}
\end{equation}
If the eigenvalues of \(P_i\) and \(Q_j\) are \((-1)^{u_i}\) and
\((-1)^{v_j}\), respectively, then
\begin{equation}
    v=uG,
    \qquad
    u=vG^{-1}.
    \label{eq:gb66-scheduler-inversion}
\end{equation}
An invertible \(G\) therefore replaces the \(P_i\) with an equivalent set of
logical observables \(Q_j\). Measuring the \(Q_j\) gives the same logical
information because \(u\) can be recovered from \(v\), and it requires no
additional measurement layer.

Only two scheduler matrices are required:
\begin{equation}
    G_{\mathrm{id}}
    =\begin{pmatrix}
        1&0&0\\
        0&1&0\\
        0&0&1
    \end{pmatrix},
    \qquad
    G_{\mathrm{prod}}
    =\begin{pmatrix}
        1&0&1\\
        0&1&0\\
        0&0&1
    \end{pmatrix}.
    \label{eq:gb66-scheduler-matrices}
\end{equation}
\(G_{\mathrm{id}}\) measures \(Q_i=P_i\) -- it is no different from an ordinary Viterbi measurement.
On the 32 branches whose original generators ordinarily require three
measurement layers, \(G_{\mathrm{prod}}\) replaces \(P_2\) with the equivalent
generator \(P_0P_2\), so that it measures
\begin{equation}
    (Q_0,Q_1,Q_2)=(P_0,P_1,P_0P_2).
    \label{eq:gb66-product-generators}
\end{equation}
Since \(G_{\mathrm{prod}}^{-1}=G_{\mathrm{prod}}\), the original outcomes are
\begin{equation}
    u_0=v_0,
    \qquad
    u_1=v_1,
    \qquad
    u_2=v_0+v_2
    \pmod 2.
    \label{eq:gb66-product-outcomes}
\end{equation}.

\Cref{tab:gb66-final-x-schedulers} gives the resulting schedule for all 256
frame branches. On 32 branches, the three \(Q_j\) share a physical
product basis and are obtained in parallel with a single DMG. On each of the
other 224 branches, one \(Q_j\) is measured by a Viterbi measurement before the
remaining pair is obtained in parallel with a single DMG. On the branches using
\(Q_2=P_0P_2\), \(Q_1\) and \(Q_2\) have compatible, disjoint physical representatives
even though the three original generators cannot be partitioned into fewer
than three compatible layers. All \(Q_j\) commute, so the Viterbi measurement layer
can be measured before the DMG layer. Physical representatives
for the MECM and DMG layers are given in
\cref{tab:gb66-final-singleton-certificates,tab:gb66-final-dmg-triple-certificates,tab:gb66-final-dmg-pair-certificates}.

\begin{table}[ht]
    \centering
    \caption{Final-\(X\) scheduler and chronological layer partition for all
    256 frame branches. A dash denotes an absent MECM layer.}
    \label{tab:gb66-final-x-schedulers}
    \begin{tabular}{@{}lrlcc@{}}
        \toprule
        Branch class & Count & Scheduler & Viterbi layer & DMG layer \\
        \midrule
        One layer & 32 & \(G_{\mathrm{id}}\) & --- & \(Q_0,Q_1,Q_2\) \\
        Pair \(01\) & 76 & \(G_{\mathrm{id}}\) & \(Q_2\) & \(Q_0,Q_1\) \\
        Pair \(02\) & 68 & \(G_{\mathrm{id}}\) & \(Q_1\) & \(Q_0,Q_2\) \\
        Pair \(12\) & 48 & \(G_{\mathrm{id}}\) & \(Q_0\) & \(Q_1,Q_2\) \\
        \(Q_2=P_0P_2\) & 32 & \(G_{\mathrm{prod}}\) & \(Q_0\) & \(Q_1,Q_2\) \\
        \bottomrule
    \end{tabular}
\end{table}

Because only one \(Q_j\) is measured via MECM on each two-layer branch, its decoder
reduces to the single-bit Viterbi stopping rule. As established in
\cref{app:gb66-representative-certificates}, the physical representative has
weight at most 16, so its outcome-flip probability is
\(p_{\mathrm{flip},16}=3.36\times10^{-3}\).
For target error \(\varepsilon=10^{-10}\), the required vote margin is
\begin{equation}
    d_{\mathrm{MECM}}
    =\left\lceil
        \frac{\log[(1-\varepsilon)/\varepsilon]}
             {\log[(1-p_{\mathrm{flip},16})/p_{\mathrm{flip},16}]}
      \right\rceil
    =5.
\end{equation}
The logical-outcome error contribution from the MECM is
\begin{equation}
    p_{\mathrm{MECM},16}
    =\left[1+\left(\frac{1-p_{\mathrm{flip},16}}
    {p_{\mathrm{flip},16}}\right)^5\right]^{-1}
    =4.36\times10^{-13}.
\end{equation}
Using \(p_{\mathrm{miss},16}=8.0\times10^{-3}\) and
\(\rho_{16}=p_{\mathrm{flip},16}/(1-p_{\mathrm{flip},16})\), the average
number of SECs on a branch requiring MECM is
\begin{equation}
    N_{\mathrm{SEC}}^{\mathrm{MECM}}
    =\frac{1}{1-p_{\mathrm{miss},16}}
      \frac{d_{\mathrm{MECM}}}{1-2p_{\mathrm{flip},16}}
      \frac{1-\rho_{16}^{d_{\mathrm{MECM}}}}
           {1+\rho_{16}^{d_{\mathrm{MECM}}}}
    =5.07.
\end{equation}
Averaging over the frame branches gives
\begin{equation}
    \overline N_{\mathrm{SEC}}^{\mathrm{phase}}
    =\frac{256-32}{256}N_{\mathrm{SEC}}^{\mathrm{MECM}}
    =4.44.
\end{equation}
The corresponding branch-averaged MECM error contribution is
\((224/256)p_{\mathrm{MECM},16}=3.81\times10^{-13}\).

\begin{figure*}[t]
    \centering
    \includegraphics[width=\textwidth]{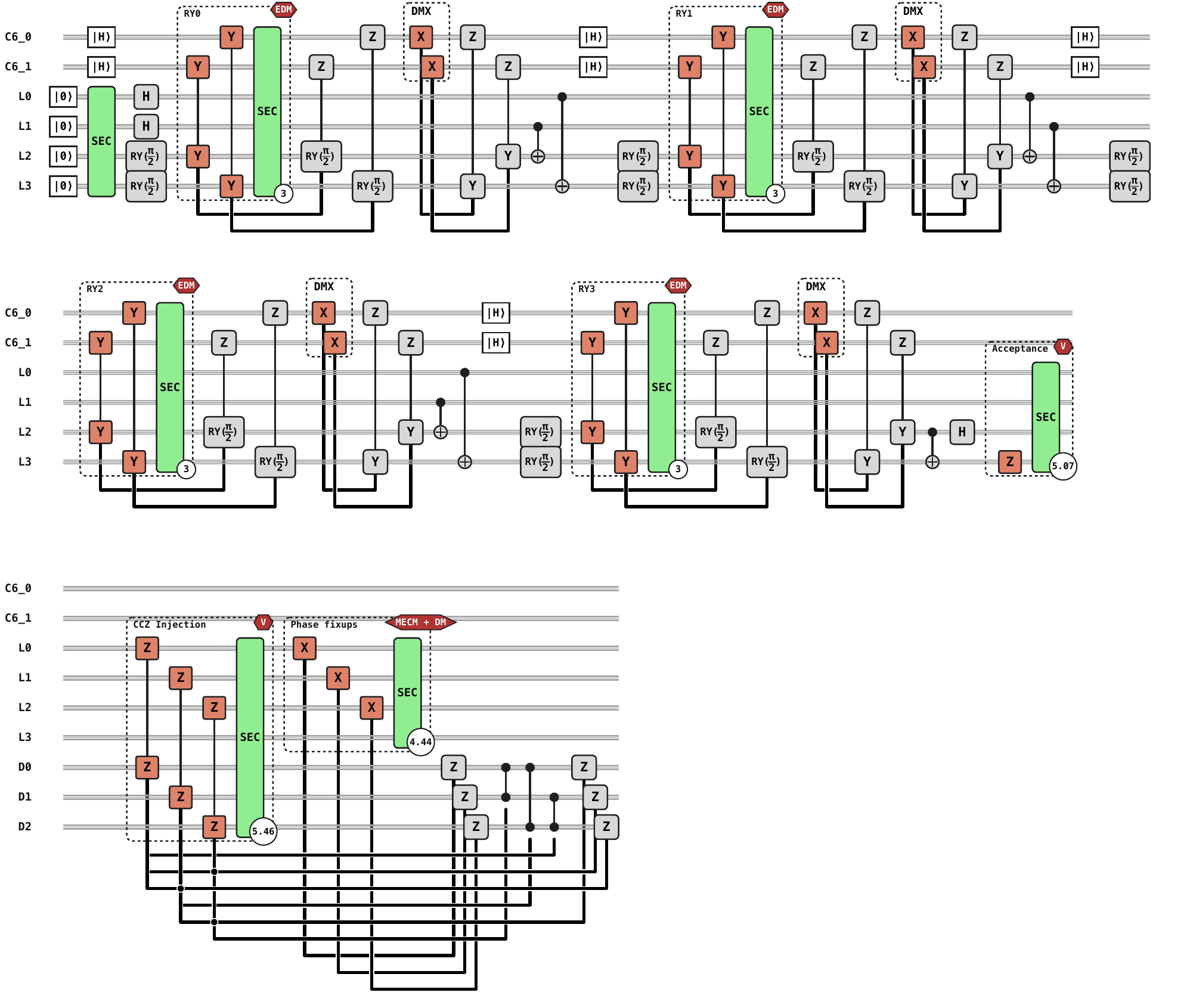}
    \caption{Physical measurement schedule for the Eastinthillation factory.
    \(\mathrm{C6}_0\) and \(\mathrm{C6}_1\) are logical
    qubits drawn from the two distinct C6 code blocks dedicated to this
    factory's Q66 block.}
    \label{fig:eastin-synthillation-factory-measurements}
\end{figure*}

\subsubsection{Analysis}
\label{sec:ccz-factory-analysis}

\paragraph*{\normalfont\bfseries Per-attempt failure rate.}

A \CCZ{} preparation attempt accepts only if its eight \(R_Y(\pi/4)\) injections,
and final input-parity check all accept. From
\cref{eq:eastin-ry-injection-acceptance}, the rejection probability of one
injection is \(r_{R_Y}=1-a_{R_Y}=2.147\%\). The final parity check rejects an
odd number of injection faults, with probability
\begin{align}
    r_{\mathrm{par}}
    &=\sum_{j\ \mathrm{odd}}\binom{8}{j}p_{Y,\mathrm{inj}}^j
      (1-p_{Y,\mathrm{inj}})^{8-j} \\
    &=\frac{1-(1-2p_{Y,\mathrm{inj}})^8}{2}
      \lesssim 6.12\times10^{-7}.
\end{align}
The per-attempt acceptance and failure probabilities therefore satisfy
\begin{align}
    1-p_{\mathrm{fail}}^{\mathrm{Eastin}}
    &\gtrsim
      (1-r_{R_Y})^8
      (1-r_{\mathrm{par}})
    =84.06\%, \notag\\
    p_{\mathrm{fail}}^{\mathrm{Eastin}}
    &\lesssim 15.94\%.
    \label{eq:eastin-factory-failure-rate}
\end{align}

\paragraph*{\normalfont\bfseries Factory depth.}
We first compute \(T_{\mathrm{prep}}^{(0)}\), the POC depth required to
prepare one \CCZ{} state when the attempt detects no failure and the factory does
not restart. The superscript \((0)\) denotes this no-restart case. We obtain
this depth by counting the Q66 SECs in the preparation schedule and converting
that count to POCs. Preparing the Q66 resource block requires one logical zero
preparation (LZ), as defined in the Walking Cat instruction
set~\cite{tripier2026walkingcat}, four sequential pairs of \(R_Y(\pi/4)\)
injections, and the final acceptance measurement. LZ takes one Q66 SEC. Each
pair of rotations is performed in parallel by EDM-3 and therefore takes three Q66 SECs.
Once the final C6 SEC finishes, we prepare the input states for the next pair
of \(R_Y(\pi/4)\) injections in parallel with the final Q66 SEC. This
preparation takes only 12.0 POCs and therefore finishes before the next Q66 SEC
begins.
\(\tau_{\mathrm{Vit},16}^{\mathrm{avg}}=5.07\) SECs of the acceptance
measurement, the Q66 block is occupied for an average of 18.07 Q66 SECs
during preparation. Finally, as a Q66 SEC has POC depth
\(T_{\mathrm{SEC}}^{\mathrm{Q66}}=17\):
\begin{align}
    N_{\mathrm{SEC}}^{\mathrm{prep}}
    &=1+4\times3+5.07=18.07, \notag\\
    T_{\mathrm{prep}}^{(0)}
    &=N_{\mathrm{SEC}}^{\mathrm{prep}}
      T_{\mathrm{SEC}}^{\mathrm{Q66}} \notag\\
    &=18.07~\mathrm{SEC}\times
      17~\frac{\mathrm{POCs}}{\mathrm{SEC}} \notag\\
    &=307.19~\mathrm{POCs}.
    \label{eq:eastin-preparation-depth-no-restart}
\end{align}

The three \CCZ{} injection measurements are performed in parallel. Their joint
average duration is \(\tau_{\CCZ}^{\mathrm{avg}}\approx5.46\) Q102 SECs from
\cref{eq:parallel-ccz-viterbi-duration}.
\begin{equation}
    T_{\mathrm{inj}}
    =\tau_{\CCZ}^{\mathrm{avg}}
       T_{\mathrm{SEC}}^{\mathrm{Q102}}
    =147.48~\mathrm{POCs}.
    \label{eq:eastin-ccz-injection-depth}
\end{equation}
The phase-fixup schedule uses an average of 4.44 Q66 SECs, followed by the
one-POC terminal DMG. Its average depth is
\begin{equation}
    T_{\mathrm{phase}}
    =\overline N_{\mathrm{SEC}}^{\mathrm{phase}}
       T_{\mathrm{SEC}}^{\mathrm{Q66}}+1
    =4.44\times17+1
    =76.48~\mathrm{POCs}.
    \label{eq:eastin-phase-fixup-depth}
\end{equation}
The depth from LZ through release of the Q66 block--assuming no restarts--is therefore
\begin{equation}
    T_{\mathrm{factory}}^{(0)}
    =T_{\mathrm{prep}}^{(0)}+T_{\mathrm{inj}}+T_{\mathrm{phase}}
    =531.15~\mathrm{POCs}.
    \label{eq:eastin-factory-depth-no-restart}
\end{equation}

Next, we account for the depth added from restarting the protocol in the event of
a detected error. Each of the four
\(R_Y(\pi/4)\)-injection layers performs two EDM-3 measurements in parallel.
An error may be detected after each of the three associated SECs or after the
final acceptance measurement. Let \(k\in\{1,2,3,4\}\) index the injection
layer and let \(j\in\{1,2,3\}\) index its SEC. The accumulated depth after SEC
\(j\) of layer \(k\) is
\begin{equation}
    d_{k,j}=1+3(k-1)+j.
    \label{eq:eastin-injection-checkpoint-depth}
\end{equation}
The three injection checkpoints therefore occur at depths
\((2,3,4)\), \((5,6,7)\), \((8,9,10)\), and \((11,12,13)\) Q66 SECs for
layers one through four. The final acceptance measurement reports a failure
at depth
\begin{equation}
    d_{\mathrm{par}}
    =1+4\times3+\tau_{\mathrm{Vit},16}^{\mathrm{avg}}
    =18.07~\mathrm{SECs}.
    \label{eq:eastin-parity-checkpoint-depth}
\end{equation}
Next, we calculate the error-detection probability after each SEC of the two
parallel EDM-3 measurements in a \(R_Y(\pi/4)\)-injection layer. After the
first SEC, a restart occurs only if either C6 SEC reports an error.
After the second and third SECs, a restart occurs if either C6 SEC reports an
error or if the EDM outcomes have a parity mismatch. We write the latter two
cases as disjoint events:
\begin{align}
    \Pr(\text{restart at SEC})
    &=\Pr(\text{C6 SEC error}) \notag\\
    &\quad+\Pr(\text{C6 SECs accept}) \notag\\
    &\qquad\times\Pr(\text{EDM mismatch}).
    \label{eq:eastin-sec-restart-events}
\end{align}
An EDM mismatch cannot be detected after the first SEC because each EDM-3
sequence has only one outcome at that point. The second and third SECs can
detect a disagreement among the accumulated outcomes.

Let \(f=p_{\mathrm{flip}}=3.99\times10^{-3}\). One EDM-3 sequence agrees through
\(m\) SECs when all \(m\) outcomes are correct or all \(m\) outcomes are
flipped. The probability that both parallel EDM-3 sequences agree through
\(m\in\{2,3\}\) SECs is
\begin{align}
    u_m&=\bigl[(1-f)^m+f^m\bigr]^2, \notag\\
    u_2&=0.984167,\qquad u_3=0.976298.
    \label{eq:eastin-parallel-edm-agreement}
\end{align}
The probability that both C6 SECs in one round accept is
\begin{equation}
    c=(1-r_{\mathrm{SEC}}^{\mathrm{C6}})^2=0.993550,
    \qquad
    1-c=0.00644957.
    \label{eq:eastin-parallel-c6-acceptance}
\end{equation}
Here \(1-c\) is the probability that at least one C6 SEC reports an error.
Conditional on reaching a given SEC, the restart probabilities after the
first, second, and third SECs are
\begin{align}
    e_1&=1-c, \notag\\
    e_2&=(1-c)+c(1-u_2)=1-cu_2, \notag\\
    e_3&=(1-c)+c\left(1-\frac{u_3}{u_2}\right)
         =1-c\frac{u_3}{u_2}.
    \label{eq:eastin-injection-sec-rejection}
\end{align}
The first term in each expression is the C6 SEC rejection probability. The
second term in \(e_2\) and \(e_3\) is the probability that both C6 SECs accept
but the EDM outcomes disagree.

Each injection layer requires two zero-level input-pair preparations. Both
preparations accept with probability
\(h=(1-r_{\mathrm{rej}}^{(0)})^2\). If either preparation rejects, the
parallel preparation batch is repeated. The first batch has depth
\(T_{\mathrm{dist}}^{(0)}=12\) POCs and is hidden inside the 17-POC Q66
scheduling window. Each additional batch adds \(12/17\) Q66 SECs. For the
geometrically distributed number of batches, the mean additional
critical-path depth per reached injection layer is
\begin{equation}
    \delta_{H^2}=\frac{12}{17}\left(\frac{1}{h}-1\right)
    =0.00508~\text{Q66 SECs}.
    \label{eq:eastin-input-preparation-retry-depth}
\end{equation}
So, a C6 rejection delays the injection layer but does not
require a full-fledged factory restart. Once both inputs are available, a parallel injection layer completes
with probability
\begin{equation}
    a=c^3u_3
     =a_{R_Y}^2
     =0.957529.
    \label{eq:eastin-paired-layer-acceptance}
\end{equation}
Conditional on starting the parallel EDM-3 measurements, the probabilities
that the attempt terminates after the first, second, or third SEC are
\begin{align}
    g_1&=e_1, \notag\\
    g_2&=ce_2, \notag\\
    g_3&=c^2u_2e_3.
    \label{eq:eastin-injection-checkpoint-failure}
\end{align}
These probabilities satisfy \(\sum_{j=1}^3g_j=1-a\). Layer \(k\) is reached
with probability \(a^{k-1}\), so the unconditional restart probability at
checkpoint \(j\) of that layer is \(a^{k-1}g_j\).

The acceptance measurement is reached with probability \(a^4\) and detects
an injection fault with probability \(r_{\mathrm{par}}\). Its unconditional
restart probability is
\begin{equation}
    q_{\mathrm{par}}=a^4r_{\mathrm{par}}=5.14\times10^{-7}.
\end{equation}
An attempt succeeds with probability
\begin{equation}
    A=a^4(1-r_{\mathrm{par}})=0.840636.
\end{equation}
The expected number of failures at each checkpoint per accepted state is the
unconditional checkpoint probability divided by \(A\). The average
preparation depth is therefore
\begin{align}
    \overline N_{\mathrm{SEC}}^{\mathrm{prep}}
    &=N_{\mathrm{SEC}}^{\mathrm{prep}}
      +\frac{1}{A}\Biggl[
        \sum_{k=1}^4a^{k-1}
        \left(\delta_{H^2}+\sum_{j=1}^3g_jd_{k,j}\right) \notag\\
    &\qquad
        +q_{\mathrm{par}}d_{\mathrm{par}}
      \Biggr] \notag\\
    &=18.07+1.44730=19.5173~\mathrm{SECs}, \notag\\
    \overline T_{\mathrm{prep}}
    &=\overline N_{\mathrm{SEC}}^{\mathrm{prep}}
      T_{\mathrm{SEC}}^{\mathrm{Q66}}
    =331.79~\mathrm{POCs}.
    \label{eq:eastin-preparation-depth-with-restarts}
\end{align}
The \CCZ{} injection and phase-fixup stages are executed only after the parity
check accepts. Their Viterbi repetition costs are already included in
\cref{eq:eastin-ccz-injection-depth,eq:eastin-phase-fixup-depth}. The complete
average depth per accepted and consumed \CCZ{} state is consequently
\begin{equation}
    \overline T_{\mathrm{factory}}
    =\overline T_{\mathrm{prep}}+T_{\mathrm{inj}}+T_{\mathrm{phase}}
    =555.75~\mathrm{POCs}.
    \label{eq:eastin-factory-mean-depth}
\end{equation}

\paragraph*{\normalfont\bfseries Factory footprint.}

We say that a number of factories is sufficient to produce a continuous stream
of \CCZ{} states if, collectively, their average time to produce a magic state is
shorter than the average time per state injection.
The point-addition circuit executes at most one Toffoli gate at a time. A Q102
memory block occupies the \CCZ{} injection interface for
\(\tau_{\CCZ}^{\mathrm{avg}}\approx5.46\) Q102 SECs and can therefore consume
at most one \CCZ{} state every \(T_{\mathrm{inj}}=147.48\) POCs. Since one factory supplies an accepted state
every \(555.75\) POCs on average, the number of factories required to maintain
this consumption rate is
\begin{equation}
    N_{\CCZ}
    =\left\lceil
       \frac{\overline T_{\mathrm{factory}}}
            {\tau_{\CCZ}^{\mathrm{avg}}
             T_{\mathrm{SEC}}^{\mathrm{Q102}}}
      \right\rceil
    =\left\lceil\frac{555.75}{147.48}\right\rceil
    =4.
    \label{eq:eastin-number-of-factories}
\end{equation}

Each Eastinthillation factory uses one Q66 block, three weight-16 cat-state
factories, two dedicated zero-level C6 source blocks, and one six-qubit C6
Bell-state source per C6 block.  As derived in \cref{subsec:footprint},
this gives 319 physical qubits per factory and 1,276 physical qubits for all
four factories.
The resulting layout is shown in
\cref{fig:ccz-factory-footprint}.

\begin{figure}[t]
    \centering
    \includegraphics[width=\linewidth]{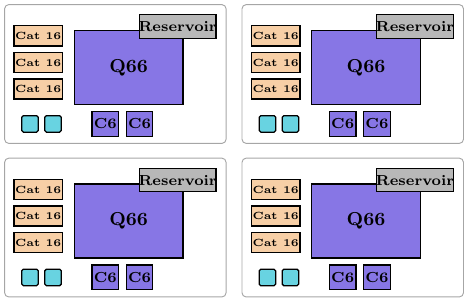}
    \caption{Layout of the four Eastinthillation factories required to sustain the
    application consumption rate. Each Q66 block is adjacent to three
    weight-16 cat-state factories and has two dedicated zero-level C6
    \(T\)-state factories. The two cyan sites at the lower left of each factory
    are its local six-qubit C6 Bell-state sources.}
    \label{fig:ccz-factory-footprint}
\end{figure}

\paragraph*{\normalfont\bfseries Logical Error Rate.}
Eastin's exact output-error formula (assuming only
\(Y\)- or \(X\)-type errors from \(R_Y(\pi/4)\) injections) gives \(28p_{Y,\mathrm{inj}}^2\). Using
\cref{eq:eastin-effective-injection-y-error} gives
\begin{equation}
    p_{\mathrm{L,in}}^{\mathrm{Eastin}}
    \lesssim 28p_{Y,\mathrm{inj}}^2
    \lesssim 28\bigl(7.65\times10^{-8}\bigr)^2
    =1.64\times10^{-13}.
    \label{eq:eastin-factory-input-error}
\end{equation}

The Q66 logical error rate is
\(p_{\mathrm{L,SEC}}^{\mathrm{Q66}}=2.63\times10^{-11}\) per SEC. The Q66
block is exposed for \(\overline N_{\mathrm{SEC}}^{\mathrm{prep}}\) Q66 SECs
during preparation. During \CCZ{} injection, one Q102 SEC lasts
\(T_{\mathrm{SEC}}^{\mathrm{Q102}}/T_{\mathrm{SEC}}^{\mathrm{Q66}}
=27/17\) Q66 SECs, and the phase-fixup schedule adds
\(\overline N_{\mathrm{SEC}}^{\mathrm{phase}}\) Q66 SECs. A union
bound over the input-error, acceptance-measurement, \CCZ{} injection, phase-fixup,
and Q66 SEC contributions gives
\begin{align}
    p_{\mathrm{L}}^{\mathrm{Eastin}}
    &\lesssim p_{\mathrm{L,in}}^{\mathrm{Eastin}}
       +p_{\mathrm{Vit},16}
       +3p_{\mathrm{Vit},36}
       +\frac{224}{256}p_{\mathrm{MECM},16} \notag\\
    &\quad
       +\left(\overline N_{\mathrm{SEC}}^{\mathrm{prep}}
       +\tau_{\CCZ}^{\mathrm{avg}}
        \frac{T_{\mathrm{SEC}}^{\mathrm{Q102}}}
             {T_{\mathrm{SEC}}^{\mathrm{Q66}}}
       +\overline N_{\mathrm{SEC}}^{\mathrm{phase}}\right)
        p_{\mathrm{L,SEC}}^{\mathrm{Q66}} \\
    &=7.80\times10^{-11} \notag\\
    &\quad
      +\left(19.5173+\tau_{\CCZ}^{\mathrm{avg}}\frac{27}{17}\right. \notag\\
    &\qquad\left.{}+4.44\right)(2.63\times10^{-11}) \notag\\
    &=9.36\times10^{-10}.
\end{align}

Under the error model above, the logical error rate satisfies
\begin{equation}
    p_{\mathrm{L}}^{\mathrm{Eastin}}
    \lesssim 9.36\times10^{-10}.
    \label{eq:eastin-factory-logical-error-rate}
\end{equation}

Q66 memory errors give the largest contribution to this bound by two orders of magnitude; the scheme is otherwise
capable of producing \CCZ{} gates with logical error rate of order 1e-13. This floor
can be met by higher-distance code or a decoder that
uses loss and leakage information, e.g. ~\cite{liu2026paulienvelope}. Or, one may exploit
soft information in the decoder: the factory can reject and
restart low-confidence attempts before their
outputs are used~\cite{bombin2024faulttolerantpostselection,smith2024mitigating}.

\subsubsection{T-state factory conversion}
\label{sec:t-state-factory-conversion}

The windowed semiclassical iQFT requires arbitrary-angle $Z$-rotations,
which can be synthesized approximately with \(T\) gates using the mixing+fallback method~\cite{kliuchnikov2023shorter}.
We only need 16 $Z$-rotations per window if we execute them sequentially,
feedforward measurement results and calculate the angles from the accumulated measurement outcomes.
The total number of $Z$-rotations across all iQFTs is 464, with $0.57 \log_2(1/\varepsilon) + 8.83$ \(T\) gates per rotation,
where $\varepsilon$ is the approximation error per rotation.

However, the Eastinthillation factory can prepare only \CCZ{} states. We can
repurpose its components to produce \(T\) states instead.
The local C6 \(T\)-state factories inject the required
\(T\) states into the Q66 blocks, where they can be safely stored in memory.
The states are then teleported to the Q102
blocks as required. Using the measured logical error rate of
\(1.30\times10^{-8}\) per accepted pair, and minimizing for the total diamond-norm distance,
we obtain the optimal $\varepsilon=5.30\times10^{-9}$ and the average \(T\)-sequence length 24.5.
This results in an expected total of 5,684 required \(T\) pairs.

This gives the total iQFT-stage failure probability, by the union bound, at most
\begin{align}
    p_{\mathrm{iQFT}}
    \leq 5{,}684\left(1.30\times10^{-8}\right) + 464\left(5.30\times10^{-9}\right) =\notag\\
    =7.64\times10^{-5}.
    \label{eq:iqft-t-state-failure}
\end{align}

\part{Conclusion}
\label{part:conclusion}

We presented optimized quantum circuits for solving the ECDLP on $\bitcoincurve$ using Shor's algorithm, derived rigorous bounds on the success probability of
the implementation---taking into account phase- and output-errors due to approximate arithmetic---and compiled our implementation to our application-specific extension of the Walking Cat architecture~\cite{tripier2026walkingcat} using our compilation toolchain, allowing us to accurately account for the overheads due to layout and routing.
By combining quantum-circuit, compilation, and architecture improvements, we showed that a trapped-ion quantum computer with $20{,}000$ qubits should be able to solve the ECDLP on $\bitcoincurve$ within 26 days.
Because we intentionally prioritize simple, uniform implementations over exhaustive co-optimization, we expect future work to improve the runtime, physical-qubit footprint, and success probability reported here.

In our resource estimates, we have taken into account hardware details that have typically been ignored in the literature.
This enables us to be more confident about our runtime and footprint estimates.
For instance, one of the major issues with trapped-ion architectures (and one we prominently addressed in the original Walking Cat architecture) is a loss cascade.
When an ion is lost, if its well is merged with that of a second ion to perform a two-qubit gate, the electrode configuration will not be correct for optimal trapping.
This can result in substantial heating of the remaining ion, which can lead to ejection from the trap. The leakage and loss reduction units (LLRUs) developed in that architecture assume that this ejection happens $100\%$ of the time.
However, for sufficiently deep traps, this might not be the case. Here, we use a different loss propagation model in which the ion is not ejected but its final position after a split is randomized, and we develop a new LLRU to address this model of loss propagation. We expect that the physical situation is more nuanced, with heating and ejection considered using a detailed microarchitecture, but we leave this to future work.

Several of our techniques could also reduce the resource requirements of other fault-tolerant quantum algorithms, e.g., for chemistry, condensed matter physics, and optimization~\cite{babbush2025grand,dalzell2025quantum}.
In particular, the Clifford frame clearing using logical CliNR and our integrated routing can be included in the general-purpose Walking Cat architecture without any additional hardware requirements.
The increased logical-measurement parallelism obtained using non-overlapping cat-based measurements is straightforward to include in the Walking Cat architecture at the price of adjusting the block allocation to allow for increased cat state throughput in the vicinity of the target blocks.
Lastly, our fast CCZ implementation provides a $31\times$ speedup over a CCZ compiled into Clifford and $T$ gates.
We leave it to future work to quantify the impact of these optimizations on quantum computing applications beyond the ECDLP.

\section*{Acknowledgments}

The authors thank Chris Ballance, Tom Harty, and Matt Keesan for discussions and support. MR thanks Christof Zalka for friendship and inspiration and Franziska Burkhalter for
discussions.

The authors acknowledge the use of generative AI tools to assist with code and figure generation, to explore different proof strategies, and to review and proofread the manuscript. The manuscript was written entirely by the authors, who take full responsibility for its content.

\bibliography{bibliography}

@article{dennis2002topological,
  title={Topological quantum memory},
  author={Dennis, Eric and Kitaev, Alexei and Landahl, Andrew and Preskill, John},
  journal={Journal of Mathematical Physics},
  volume={43},
  number={9},
  pages={4452--4505},
  year={2002},
  publisher={American Institute of Physics}
}

@article{beverland2022assessing,
  title={Assessing requirements to scale to practical quantum advantage},
  author={Beverland, Michael E and Murali, Prakash and Troyer, Matthias and Svore, Krysta M and Hoefler, Torsten and Kliuchnikov, Vadym and Low, Guang Hao and Soeken, Mathias and Sundaram, Aarthi and Vaschillo, Alexander},
  journal={arXiv preprint arXiv:2211.07629},
  year={2022}
}

@book{dalzell2025quantum,
  title={Quantum algorithms: A survey of applications and end-to-end complexities},
  author={Dalzell, Alexander M and McArdle, Sam and Berta, Mario and Bienias, Przemyslaw and Chen, Chi-Fang and Gily{\'e}n, Andr{\'a}s and Hann, Connor T and Kastoryano, Michael J and Khabiboulline, Emil T and Kubica, Aleksander and Salton, Grant and Wang, Samson and Brand{\~a}o, Fernando G S L},
  publisher={Cambridge University Press},
  year={2025},
  doi={10.1017/9781009639651}
}

@article{jones2012layered,
  title={Layered architecture for quantum computing},
  author={Jones, N Cody and Van Meter, Rodney and Fowler, Austin G and McMahon, Peter L and Kim, Jungsang and Ladd, Thaddeus D and Yamamoto, Yoshihisa},
  journal={Physical Review X},
  volume={2},
  number={3},
  pages={031007},
  year={2012},
  publisher={APS}
}

@article{selinger2013quantum,
  title={Quantum circuits of T-depth one},
  author={Selinger, Peter},
  journal={Physical Review A—Atomic, Molecular, and Optical Physics},
  volume={87},
  number={4},
  pages={042302},
  year={2013},
  publisher={APS}
}

@article{fowler2012surface,
  title={Surface codes: Towards practical large-scale quantum computation},
  author={Fowler, Austin G and Mariantoni, Matteo and Martinis, John M and Cleland, Andrew N},
  journal={Physical Review A},
  volume={86},
  number={3},
  pages={032324},
  year={2012},
  publisher={APS}
}

@article{gidney2021factor,
  title={How to factor 2048 bit RSA integers in 8 hours using 20 million noisy qubits},
  author={Gidney, Craig and Eker{\aa}, Martin},
  journal={Quantum},
  volume={5},
  pages={433},
  year={2021},
  publisher={Verein zur F{\"o}rderung des Open Access Publizierens in den Quantenwissenschaften}
}

@article{gidney2025factor,
  title={How to factor 2048 bit {RSA} integers with less than a million noisy qubits},
  author={Gidney, Craig},
  journal={arXiv preprint arXiv:2505.15917},
  year={2025}
}

@article{webster2026pinnacle,
  title={The Pinnacle Architecture: Reducing the cost of breaking {RSA}-2048 to 100 000 physical qubits using quantum {LDPC} codes},
  author={Webster, Paul and Berent, Lucas and Chandra, Omprakash and Hockings, Evan T and Baspin, Nou{\'e}dyn and Thomsen, Felix and Smith, Samuel C and Cohen, Lawrence Z},
  journal={arXiv preprint arXiv:2602.11457},
  year={2026}
}

@article{cain2026shor,
  title={Shor's algorithm is possible with as few as 10,000 reconfigurable atomic qubits},
  author={Cain, Madelyn and Xu, Qian and King, Robbie and Picard, Lewis RB and Levine, Harry and Endres, Manuel and Preskill, John and Huang, Hsin-Yuan and Bluvstein, Dolev},
  journal={arXiv preprint arXiv:2603.28627},
  year={2026}
}

@article{hughes2025trapped,
  title={Trapped-ion two-qubit gates with> 99.99\% fidelity without ground-state cooling},
  author={Hughes, Amy C and Srinivas, R and L{\"o}schnauer, CM and Knaack, HM and Matt, R and Ballance, CJ and Malinowski, M and Harty, TP and Sutherland, RT},
  journal={arXiv preprint arXiv:2510.17286},
  year={2025}
}

@article{malinowski2023wire,
  title={How to wire a 1000-qubit trapped-ion quantum computer},
  author={Malinowski, M and Allcock, DTC and Ballance, CJ},
  journal={PRX quantum},
  volume={4},
  number={4},
  pages={040313},
  year={2023},
  publisher={APS}
}

@article{cirac1995quantum,
  title={Quantum computations with cold trapped ions},
  author={Cirac, Juan I and Zoller, Peter},
  journal={Physical review letters},
  volume={74},
  number={20},
  pages={4091},
  year={1995},
  publisher={APS}
}

@article{monroe1995demonstration,
  title={Demonstration of a fundamental quantum logic gate},
  author={Monroe, Chris and Meekhof, David M and King, Barry E and Itano, Wayne M and Wineland, David J},
  journal={Physical review letters},
  volume={75},
  number={25},
  pages={4714},
  year={1995},
  publisher={APS}
}

@article{monroe2014large,
  title={Large-scale modular quantum-computer architecture with atomic memory and photonic interconnects},
  author={Monroe, Christopher and Raussendorf, Robert and Ruthven, Alex and Brown, Kenneth R and Maunz, Peter and Duan, L-M and Kim, Jungsang},
  journal={Physical Review A},
  volume={89},
  number={2},
  pages={022317},
  year={2014},
  publisher={APS}
}

@article{breuckmann2021quantum,
  title={Quantum low-density parity-check codes},
  author={Breuckmann, Nikolas P and Eberhardt, Jens Niklas},
  journal={PRX quantum},
  volume={2},
  number={4},
  pages={040101},
  year={2021},
  publisher={APS}
}

@article{bruzewicz2019trapped,
  title={Trapped-ion quantum computing: Progress and challenges},
  author={Bruzewicz, Colin D and Chiaverini, John and McConnell, Robert and Sage, Jeremy M},
  journal={Applied physics reviews},
  volume={6},
  number={2},
  year={2019},
  publisher={AIP Publishing}
}

@article{loschnauer2025scalable,
  title={Scalable, high-fidelity all-electronic control of trapped-ion qubits},
  author={L{\"o}schnauer, CM and Mosca Toba, J and Hughes, AC and King, SA and Weber, MA and Srinivas, R and Matt, R and Nourshargh, R and Allcock, DTC and Ballance, CJ and others},
  journal={PRX Quantum},
  volume={6},
  number={4},
  pages={040313},
  year={2025},
  publisher={APS}
}

@article{baek2025sdqc,
  title={SDQC: Distributed Quantum Computing Architecture Utilizing Entangled Ion Qubit Shuttling},
  author={Baek, Seunghyun and Lee, Seok-Hyung and Min, Dongmoon and Kim, Junki},
  journal={arXiv preprint arXiv:2512.02890},
  year={2025}
}

@article{baranes2026leveraging,
  title={Leveraging Qubit Loss Detection in Fault Tolerant Quantum Algorithms},
  author={Baranes, Gefen and Cain, Madelyn and Bonilla Ataides, J. Pablo and Bluvstein, Dolev and Sinclair, Josiah and Vuletic, Vladan and Zhou, Hengyun and Lukin, Mikhail D.},
  journal={Physical Review X},
  volume={16},
  number={1},
  pages={011002},
  year={2026},
  doi={10.1103/ycwc-3myc}
}

@article{shor1997,
  title={Polynomial-Time Algorithms for Prime Factorization and Discrete Logarithms on a Quantum Computer},
  author={Shor, Peter W.},
  journal={SIAM Journal on Computing},
  volume={26},
  number={5},
  pages={1484--1509},
  year={1997},
  doi={10.1137/S0097539795293172}
}

@inproceedings{shor1996faulttolerant,
  title={Fault-Tolerant Quantum Computation},
  author={Shor, Peter W.},
  booktitle={Proceedings of the 37th Annual Symposium on Foundations of Computer Science},
  pages={56--65},
  year={1996},
  publisher={IEEE},
  doi={10.1109/SFCS.1996.548464},
  eprint={quant-ph/9605011},
  archivePrefix={arXiv}
}

@article{zhou2000methodology,
  title={Methodology for Quantum Logic Gate Construction},
  author={Zhou, Xinlan and Leung, Debbie W. and Chuang, Isaac L.},
  journal={Physical Review A},
  volume={62},
  number={5},
  pages={052316},
  year={2000},
  doi={10.1103/PhysRevA.62.052316},
  eprint={quant-ph/0002039},
  archivePrefix={arXiv}
}

@article{schrottenloher2026optimized,
  title={Optimized Point Addition Circuits for Elliptic Curve Discrete Logarithms},
  author={Schrottenloher, Andr{\'e}},
  journal={arXiv preprint arXiv:2606.02235},
  year={2026}
}

@article{tripier2026walkingcat,
  title={Fault-Tolerant Quantum Computing with Trapped Ions: The Walking Cat Architecture},
  author={Tripier, Felix and Chung, Woo Chang and Young, Jacob and Alam, Safwan and Bjork, Bryce and Brodutch, Aharon and Buessen, Finn Lasse and Coble, Nolan J. and Dellaert, Thomas and Maslov, Dmitri and Roetteler, Martin and Tham, Edwin and Webster, Mark and Ye, Min and Gamble, John and Maksymov, Andrii and Marceaux, J. P. and Delfosse, Nicolas},
  journal={arXiv preprint arXiv:2604.19481},
  year={2026}
}

@article{knill2005quantum,
  title={Quantum Computing with Realistically Noisy Devices},
  author={Knill, E.},
  journal={Nature},
  volume={434},
  number={7029},
  pages={39--44},
  year={2005},
  doi={10.1038/nature03350},
  eprint={quant-ph/0410199},
  archivePrefix={arXiv}
}

@article{webster2026fast,
  title={Fast Logical Operations in Quantum {LDPC} Codes Using Simple Resource States},
  author={Webster, Mark and Delfosse, Nicolas},
  journal={arXiv preprint arXiv:2607.16166},
  year={2026}
}

@article{wall2026barium,
  title={Barium autoionization for efficient ion trap loading},
  author={Wall, Zachary J. and Piel, Justin D. and Vizvary, Samuel R. and
          Bareian, Michael and Diaz, Steven and Mossman, Elijah and
          Ransford, Anthony and Greene, Chris H. and Hudson, Eric R. and
          Campbell, Wesley C.},
  journal={Physical Review A},
  volume={114},
  number={1},
  pages={013102},
  year={2026},
  publisher={American Physical Society},
  doi={10.1103/c5m6-kt93}
}

@article{chiu2025continuous,
  title={Continuous operation of a coherent 3,000-qubit system},
  author={Chiu, Neng-Chun and Trapp, Elias C. and Guo, Jinen and
          Abobeih, Mohamed H. and Stewart, Luke M. and Hollerith, Simon and
          Stroganov, Pavel L. and Kalinowski, Marcin and Geim, Alexandra A. and
          Evered, Simon J. and Li, Sophie H. and Lyu, Xingjian and
          Peters, Lisa M. and Bluvstein, Dolev and Wang, Tout T. and
          Greiner, Markus and Vuleti{\'c}, Vladan and Lukin, Mikhail D.},
  journal={Nature},
  volume={646},
  number={8087},
  pages={1075--1080},
  year={2025},
  doi={10.1038/s41586-025-09596-6}
}

@article{gidney2021stim,
  title={Stim: a fast stabilizer circuit simulator},
  author={Gidney, Craig},
  journal={Quantum},
  volume={5},
  pages={497},
  year={2021},
  doi={10.22331/q-2021-07-06-497}
}

@article{ye2026beam,
  title={Beam Search Decoder for Quantum Low-Density Parity-Check Codes},
  author={Ye, Min and Wecker, Dave and Delfosse, Nicolas},
  journal={PRX Quantum},
  volume={7},
  number={3},
  pages={033002},
  year={2026},
  publisher={APS}
}

@article{proos2003shor,
  title={Shor's discrete logarithm quantum algorithm for elliptic curves},
  author={Proos, John and Zalka, Christof},
  journal={arXiv preprint quant-ph/0301141},
  year={2003}
}

@article{luo2026quantum,
  title={Quantum Algorithm for Elliptic Curve Discrete Logarithms with Space-Efficient Point Addition},
  author={Luo, Han and Yang, Ziyi and Luo, Jingquan and Wang, Ziruo and Su, Yuexin and Sun, Xiaoming and Li, Lvzhou and Li, Tongyang},
  journal={arXiv preprint arXiv:2607.13816},
  year={2026}
}

@article{khattar2025verifiable,
  title={Verifiable quantum advantage via optimized {DQI} circuits},
  author={Khattar, Tanuj and Shutty, Noah and Gidney, Craig and Zalcman, Adam and Yosri, Noureldin and Maslov, Dmitri and Babbush, Ryan and Jordan, Stephen P},
  journal={arXiv preprint arXiv:2510.10967},
  year={2025}
}

@inproceedings{haner2020improved,
  title={Improved quantum circuits for elliptic curve discrete logarithms},
  author={H{\"a}ner, Thomas and Jaques, Samuel and Naehrig, Michael and Roetteler, Martin and Soeken, Mathias},
  booktitle={PQCrypto 2020},
  pages={425--444},
  year={2020},
  organization={Springer}
}

@inproceedings{roetteler2017quantum,
  title={Quantum resource estimates for computing elliptic curve discrete logarithms},
  author={Roetteler, Martin and Naehrig, Michael and Svore, Krysta M and Lauter, Kristin},
  booktitle={Asiacrypt 2017},
  pages={241--270},
  year={2017},
  organization={Springer}
}

@article{babbush2026securing,
  title={Securing Elliptic Curve Cryptocurrencies against Quantum Vulnerabilities: Resource Estimates and Mitigations},
  author={Babbush, Ryan and Zalcman, Adam and Gidney, Craig and Broughton, Michael and Khattar, Tanuj and Neven, Hartmut and Bergamaschi, Thiago and Drake, Justin and Boneh, Dan},
  journal={arXiv preprint arXiv:2603.28846},
  year={2026}
}

@article{gouzien2023performance,
  title={Performance analysis of a repetition cat code architecture: Computing 256-bit elliptic curve logarithm in 9 hours with 126 133 cat qubits},
  author={Gouzien, {\'E}lie and Ruiz, Diego and Le R{\'e}gent, Francois-Marie and Guillaud, J{\'e}r{\'e}mie and Sangouard, Nicolas},
  journal={Physical Review Letters},
  volume={131},
  number={4},
  pages={040602},
  year={2023},
  publisher={APS}
}

@article{litinski2023compute,
  title={How to compute a 256-bit elliptic curve private key with only 50 million {T}offoli gates},
  author={Litinski, Daniel},
  journal={arXiv preprint arXiv:2306.08585},
  year={2023}
}

@article{gidney2018halving,
  title={Halving the cost of quantum addition},
  author={Gidney, Craig},
  journal={Quantum},
  volume={2},
  pages={74},
  year={2018},
  publisher={Verein zur F{\"o}rderung des Open Access Publizierens in den Quantenwissenschaften}
}

@article{litinski2024quantum,
  title={Quantum schoolbook multiplication with fewer {T}offoli gates},
  author={Litinski, Daniel},
  journal={arXiv preprint arXiv:2410.00899},
  year={2024}
}

@article{wecker2015solving,
  title={Solving strongly correlated electron models on a quantum computer},
  author={Wecker, Dave and Hastings, Matthew B and Wiebe, Nathan and Clark, Bryan K and Nayak, Chetan and Troyer, Matthias},
  journal={Physical Review A},
  volume={92},
  number={6},
  pages={062318},
  year={2015},
  publisher={APS}
}

@article{sanders2020compilation,
  title={Compilation of fault-tolerant quantum heuristics for combinatorial optimization},
  author={Sanders, Yuval R and Berry, Dominic W and Costa, Pedro CS and Tessler, Louis W and Wiebe, Nathan and Gidney, Craig and Neven, Hartmut and Babbush, Ryan},
  journal={PRX Quantum},
  volume={1},
  number={2},
  pages={020312},
  year={2020},
  publisher={APS}
}

@article{may2019quantum,
  title={Quantum period finding is compression robust},
  author={May, Alexander and Schlieper, Lars},
  journal={arXiv preprint arXiv:1905.10074},
  year={2019}
}

@inproceedings{ekeraa2017quantum,
  title={Quantum algorithms for computing short discrete logarithms and factoring RSA integers},
  author={Eker{\aa}, Martin and H{\aa}stad, Johan},
  booktitle={International Workshop on Post-Quantum Cryptography},
  pages={347--363},
  year={2017},
  organization={Springer}
}

@article{ekeraa2019revisiting,
  title={Revisiting {S}hor's quantum algorithm for computing general discrete logarithms},
  author={Eker{\aa}, Martin},
  journal={arXiv preprint arXiv:1905.09084},
  year={2019}
}

@article{preskill1998reliable,
  title={Reliable Quantum Computers},
  author={Preskill, John},
  journal={Proceedings of the Royal Society of London. Series A: Mathematical, Physical and Engineering Sciences},
  volume={454},
  number={1969},
  pages={385--410},
  year={1998},
  doi={10.1098/rspa.1998.0167}
}

@article{dasu2025breaking,
  title={Breaking Even with Magic: Demonstration of a High-Fidelity Logical Non-{C}lifford Gate},
  author={Dasu, Shival and Burton, Simon and Mayer, Karl and Amaro, David and Gerber, Justin A. and Gilmore, Kevin and Gresh, Dan and DelVento, Davide and Potter, Andrew C. and Hayes, David},
  journal={arXiv preprint arXiv:2506.14688},
  year={2025}
}

@article{goto2014step,
  title={Step-by-Step Magic State Encoding for Efficient Fault-Tolerant Quantum Computation},
  author={Goto, Hayato},
  journal={Scientific Reports},
  volume={4},
  number={1},
  pages={7501},
  year={2014},
  doi={10.1038/srep07501}
}

@article{eastin2013distilling,
  title={Distilling One-Qubit Magic States into {T}offoli States},
  author={Eastin, Bryan},
  journal={Physical Review A},
  volume={87},
  number={3},
  pages={032321},
  year={2013},
  doi={10.1103/PhysRevA.87.032321}
}

@article{campbell2017unified,
  title={Unified Framework for Magic State Distillation and Multiqubit Gate Synthesis with Reduced Resource Cost},
  author={Campbell, Earl T. and Howard, Mark},
  journal={Physical Review A},
  volume={95},
  number={2},
  pages={022316},
  year={2017},
  doi={10.1103/PhysRevA.95.022316}
}

@article{liu2026paulienvelope,
  title={Achieving Optimal-Distance Atom-Loss Correction via {Pauli} Envelope},
  author={Liu, Pengyu and Tan, Shi Jie Samuel and Huang, Eric and Acar, Umut A. and Zhou, Hengyun and Zhao, Chen},
  journal={arXiv preprint arXiv:2603.04156},
  year={2026}
}

@article{bombin2024faulttolerantpostselection,
  title={Fault-Tolerant Postselection for Low-Overhead Magic State Preparation},
  author={Bomb{\'i}n, H{\'e}ctor and Pant, Mihir and Roberts, Sam and Seetharam, Karthik I.},
  journal={PRX Quantum},
  volume={5},
  number={1},
  pages={010302},
  year={2024},
  doi={10.1103/PRXQuantum.5.010302}
}

@article{smith2024mitigating,
  title={Mitigating Errors in Logical Qubits},
  author={Smith, Samuel C. and Brown, Benjamin J. and Bartlett, Stephen D.},
  journal={Communications Physics},
  volume={7},
  pages={386},
  year={2024},
  doi={10.1038/s42005-024-01883-4}
}

@inproceedings{karatsuba1962multiplication,
  title={Multiplication of many-digital numbers by automatic computers},
  author={Karatsuba, Anatolii Alekseevich and Ofman, Yu P},
  booktitle={Doklady Akademii Nauk},
  volume={145},
  pages={293--294},
  year={1962},
  organization={Russian Academy of Sciences}
}

@article{griffiths1996semiclassical,
  title={Semiclassical {F}ourier transform for quantum computation},
  author={Griffiths, Robert B and Niu, Chi-Sheng},
  journal={Physical Review Letters},
  volume={76},
  number={17},
  pages={3228},
  year={1996},
  publisher={APS}
}

@phdthesis{mosca1999quantum,
  title={Quantum computer algorithms},
  author={Mosca, Michele},
  year={1999},
  school={University of Oxford. 1999.}
}

@article{clopper1934use,
  title={The use of confidence or fiducial limits illustrated in the case of the binomial},
  author={Clopper, Charles J and Pearson, Egon S},
  journal={Biometrika},
  volume={26},
  number={4},
  pages={404--413},
  year={1934},
  publisher={JSTOR}
}

@inproceedings{chevignard2025reducing,
  title={Reducing the number of qubits in quantum factoring},
  author={Chevignard, Cl{\'e}mence and Fouque, Pierre-Alain and Schrottenloher, Andr{\'e}},
  booktitle={Annual International Cryptology Conference},
  pages={384--415},
  year={2025},
  organization={Springer}
}

@inproceedings{chevignard2026reducing,
  title={Reducing the number of qubits in quantum discrete logarithms on elliptic curves},
  author={Chevignard, Cl{\'e}mence and Fouque, Pierre-Alain and Schrottenloher, Andr{\'e}},
  booktitle={Annual International Conference on the Theory and Applications of Cryptographic Techniques},
  pages={371--401},
  year={2026},
  organization={Springer}
}

@misc{harrigan2024qualtran,
  title        = {Expressing and Analyzing Quantum Algorithms with {Q}ualtran},
  author       = {Matthew P. Harrigan and Tanuj Khattar and Charles Yuan and
                  Anurudh Peduri and Noureldin Yosri and Fionn D. Malone and
                  Ryan Babbush and Nicholas C. Rubin},
  year         = {2024},
  eprint       = {2409.04643},
  archivePrefix= {arXiv},
  primaryClass = {quant-ph},
  doi          = {10.48550/arXiv.2409.04643},
  url          = {https://arxiv.org/abs/2409.04643}
}

@misc{qualtran2026,
  author  = {Harrigan, Matthew P. and Khattar, Tanuj and Yuan, Charles and
             Peduri, Anurudh and Yosri, Noureldin and Malone, Fionn D. and
             Babbush, Ryan and Rubin, Nicholas C.},
  title   = {Qualtran},
  year    = {2026},
  version = {v0.7.0},
  doi     = {10.5281/zenodo.18343990},
  url     = {https://github.com/quantumlib/Qualtran}
}

@misc{DraperKutinRainsSvore2004QCLA,
  author        = {Thomas G. Draper and Samuel A. Kutin and Eric M. Rains and
                   Krysta M. Svore},
  title         = {A Logarithmic-Depth Quantum Carry-Lookahead Adder},
  year          = {2004},
  eprint        = {quant-ph/0406142},
  archivePrefix = {arXiv}
}

@misc{qref2024,
  title   = {Quantum Resource Estimation Format ({QREF})},
  author  = {{PsiQuantum}},
  year    = {2024},
  url     = {https://github.com/PsiQ/qref},
  note    = {Open format for representing quantum algorithms for resource
             estimation}
}

@misc{qref_format,
  title        = {{QREF}: Quantum Resource Estimation Format},
  author       = {{PsiQuantum}},
  year         = {2024},
  howpublished = {\url{https://psiq.github.io/qref/}},
  note         = {{JSON}-based domain-specific format for representing quantum
                  algorithms and resource estimation workflows}
}

@inproceedings{Cuccaro2004RippleCarryAdder,
  author        = {Steven A. Cuccaro and Thomas G. Draper and Samuel A. Kutin
                   and David Petrie Moulton},
  title         = {A New Quantum Ripple-Carry Addition Circuit},
  year          = {2004},
  eprint        = {quant-ph/0410184},
  archivePrefix = {arXiv}
}

@article{takahashi2010addition,
  title   = {Quantum addition circuits and unbounded fan-out},
  author  = {Takahashi, Yasuhiro and Tani, Seiichiro and Kunihiro, Noboru},
  journal = {Quantum Information and Computation},
  volume  = {10},
  number  = {9--10},
  pages   = {872--890},
  year    = {2010},
  eprint  = {arXiv:0910.2530}
}

@inproceedings{remaud2025ancillafree,
  title     = {Ancilla-Free Quantum Adder with Sublinear Depth},
  author    = {Remaud, Maxime and Vandaele, Vivien},
  booktitle = {Reversible Computation (RC 2025)},
  series    = {Lecture Notes in Computer Science},
  volume    = {15716},
  pages     = {137--154},
  publisher = {Springer},
  year      = {2025},
  eprint    = {arXiv:2501.16802}
}

@article{lee2021even,
  title={Even more efficient quantum computations of chemistry through tensor hypercontraction},
  author={Lee, Joonho and Berry, Dominic W and Gidney, Craig and Huggins, William J and McClean, Jarrod R and Wiebe, Nathan and Babbush, Ryan},
  journal={PRX Quantum},
  volume={2},
  number={3},
  pages={030305},
  year={2021},
  publisher={APS}
}

@article{von2021quantum,
  title={Quantum computing enhanced computational catalysis},
  author={von Burg, Vera and Low, Guang Hao and H{\"a}ner, Thomas and Steiger, Damian S and Reiher, Markus and Roetteler, Martin and Troyer, Matthias},
  journal={Physical Review Research},
  volume={3},
  number={3},
  pages={033055},
  year={2021},
  publisher={APS}
}

@article{liu2022prospects,
  title={Prospects of quantum computing for molecular sciences},
  author={Liu, Hongbin and Low, Guang Hao and Steiger, Damian S and H{\"a}ner, Thomas and Reiher, Markus and Troyer, Matthias},
  journal={Materials Theory},
  volume={6},
  number={1},
  pages={11},
  year={2022},
  publisher={Springer}
}

@article{jones2013low,
  title={Low-overhead constructions for the fault-tolerant {T}offoli gate},
  author={Jones, Cody},
  journal={Physical Review A},
  volume={87},
  number={2},
  pages={022328},
  year={2013},
  publisher={APS}
}

@article{maslov2016advantages,
  title={On the advantages of using relative phase {T}offolis with an application to multiple control {T}offoli optimization},
  author={Maslov, Dmitri},
  journal={Physical Review A},
  volume={93},
  number={2},
  pages={022311},
  year={2016},
  publisher={APS},
  doi={10.1103/PhysRevA.93.022311}
}

@inproceedings{bernstein1993quantum,
  title={Quantum complexity theory},
  author={Bernstein, Ethan and Vazirani, Umesh},
  booktitle={Proceedings of the twenty-fifth annual ACM symposium on Theory of computing},
  pages={11--20},
  year={1993}
}

@article{chen2025framework,
  title={A framework for robust quantum speedups in practical correlated electronic structure and dynamics},
  author={Chen, Jielun and Chan, Garnet Kin},
  journal={arXiv preprint arXiv:2508.15765},
  year={2025}
}

@inproceedings{ZARG+2026, 
      title={End-to-end performance of quantum-accelerated large-scale linear algebra workflows}, 
      author={Daiwei Zhu and Miguel Angel Lopez-Ruiz and François-Henry Rouet and Claudio Girotto and Willie Aboumrad and Robert Lucas and Ananth Kaushik and Martin Roetteler},
      year={2026},
      booktitle={Proceedings IEEE Quantum Week (QCE'26)},
      note={Accepted for publication. See also arxiv.org preprint https://arxiv.org/abs/2603.15515}
}

@article{ZGAA+2026,
  title     = {Quantum-classical {Auxiliary Field Quantum Monte Carlo} with matchgate shadows on trapped ion quantum computers},
  author    = {Zhao, Luning and Goings, Joshua J. and Aboumrad, Willie and Arrasmith, Andrew and Calderin, Lazaro and Churchill, Spencer and Gabay, Dor and Harvey-Brown, Thea and Hiles, Melanie and Kaja, Magda and Keesan, Matthew and Kulesz, Karolina and Maksymov, Andrii and Maruo, Mei and Mu{\~n}oz, Mauricio and Nijholt, Bas and Schiller, Rebekah and de Sereville, Yvette and Smidutz, Amy and Tripier, Felix and Yao, Grace and Zaveri, Trishal and Collins, Coleman and Roetteler, Martin and Epifanovsky, Evgeny and Kovyrshin, Arseny and Tornberg, Lars and Broo, Anders and Hammond, Jeff R. and Chandani, Zohim and Khalate, Pradnya and Kyoseva, Elica and Chen, Yi-Ting and Kessler, Eric M. and Lin, Cedric Yen-Yu and Ramu, Gandhi and Shaffer, Ryan and Brett, Michael and Huang, Benchen and Hugues, Maxime R. and Takeshita, Tyler Y.},
  journal   = {Physical Review Research},
  volume    = {8},
  pages     = {033061},
  year      = {2026},
  note    = {See also arxiv.org preprint https://arxiv.org/abs/2506.22408}
}

@inproceedings{RJIG+2026,
  title     = {{Quantum Lattice Boltzmann} Solutions for Transport under {3D} Spatially Varying Advection on Trapped Ion Hardware},
  author    = {Ray, Sayonee and Jojo, Jezer and Iaconis, Jason and Gnanasekaran, Abeynaya and Tiwari, Apurva and Roetteler, Martin and Hill, Chris and Pathak, Jay},
  year      = {2026},
  booktitle={Proceedings IEEE Quantum Week (QCE'26)},
  note={Accepted for publication. See also arxiv.org preprint https://arxiv.org/abs/2604.28121}
  }

@misc{IRSG+2026,
  title     = {Practical {Quantum Topological Data Analysis} with Applications to High-Dimensional Feature Extraction and Time Series Analysis},
  author    = {Iaconis, Jason and Ray, Sayonee and Sekwao, Samwel and Girotto, Claudio and Roetteler, Martin},
  year      = {2026},
  eprint    = {2607.27206},
  archivePrefix = {arXiv}
}

@article{kliuchnikov2023shorter,
  title={Shorter quantum circuits via single-qubit gate approximation},
  author={Kliuchnikov, Vadym and Lauter, Kristin and Minko, Romy and Paetznick, Adam and Petit, Christophe},
  journal={Quantum},
  volume={7},
  pages={1208},
  year={2023},
  doi={10.22331/q-2023-12-18-1208}
}

@inproceedings{lattner2021mlir,
  title={{MLIR}: Scaling Compiler Infrastructure for Domain Specific Computation},
  author={Lattner, Chris and Amini, Mehdi and Bondhugula, Uday and Cohen, Albert and Davis, Andy and Pienaar, Jacques and Riddle, River and Shpeisman, Tatiana and Vasilache, Nicolas and Zinenko, Oleksandr},
  booktitle={2021 IEEE/ACM International Symposium on Code Generation and Optimization (CGO)},
  pages={2--14},
  year={2021},
  organization={IEEE},
  eprint={arXiv:2002.11054}
}

@article{steiger2018projectq,
  title={{ProjectQ}: An Open Source Software Framework for Quantum Computing},
  author={Steiger, Damian S and H{\"a}ner, Thomas and Troyer, Matthias},
  journal={Quantum},
  volume={2},
  pages={49},
  year={2018},
  doi={10.22331/q-2018-01-31-49}
}

@article{kielpinski2002architecture,
  title={Architecture for a large-scale ion-trap quantum computer},
  author={Kielpinski, David and Monroe, Chris and Wineland, David J},
  journal={Nature},
  volume={417},
  number={6890},
  pages={709--711},
  year={2002},
  publisher={Nature Publishing Group UK London}
}

@article{babbush2025grand,
  title={The Grand Challenge of Quantum Applications},
  author={Babbush, Ryan and King, Robbie and Boixo, Sergio and Huggins, William and Khattar, Tanuj and Low, Guang Hao and McClean, Jarrod R. and O'Brien, Thomas and Rubin, Nicholas C.},
  journal={arXiv preprint arXiv:2511.09124},
  year={2025}
}

@article{alexeev2021quantum,
  title={Quantum Computer Systems for Scientific Discovery},
  author={Alexeev, Yuri and Bacon, Dave and Brown, Kenneth R. and Calderbank, Robert and Carr, Lincoln D. and Chong, Frederic T. and DeMarco, Brian and Englund, Dirk and Farhi, Edward and
Fefferman, Bill and Gorshkov, Alexey V. and Houck, Andrew and Kim, Jungsang and Kimmel, Shelby and Lange, Michael and Lloyd, Seth and Lukin, Mikhail D. and Maslov, Dmitri and Maunz,
Peter and Monroe, Christopher and Preskill, John and Roetteler, Martin and Savage, Martin and Thompson, Jeff},
  journal={PRX Quantum},
  volume={2},
  number={1},
  pages={017001},
  year={2021},
  doi={10.1103/PRXQuantum.2.017001}
}

@misc{haener2022spacetime,
  title={Space-time optimized table lookup},
  author={H{\"a}ner, Thomas and Kliuchnikov, Vadym and Roetteler, Martin and Soeken, Mathias},
  year={2022},
  eprint={2211.01133},
  archivePrefix={arXiv},
  primaryClass={quant-ph}
}

@inproceedings{haener2016emulation,
  title={High Performance Emulation of Quantum Circuits},
  author={H{\"a}ner, Thomas and Steiger, Damian S. and Smelyanskiy, Mikhail and Troyer, Matthias},
  booktitle={Proceedings of the International Conference for High Performance Computing, Networking, Storage and Analysis},
  year={2016},
  doi={10.1109/SC.2016.73}
}

\clearpage
\appendix

\part{Appendix}
\section{Algorithms}
\label{sec:algorithms}

\subsection{Affine elliptic curve point addition}

This section contains integer-valued representations of the algorithms that comprise the affine elliptic curve point addition circuits in our implementation. Algorithm~\ref{alg:randomized-window-setup} shows the classical precomputation of the initial accumulator point and the lookup tables for each window in the point addition loop shown in Algorithm~\ref{alg:windowed-double-scalar}. There are $L=\lceil n/w \rceil$ windows of size $w$ with the top window possibly containing less than $w$ bits. After sampling the random initialization point $[a]P$ by sampling a uniform random $a\in \Z_r^*$, Algorithm~\ref{alg:randomized-window-setup} first computes the $L$ tables with multiples of the base point $P$, then the $L$ tables with multiples of the point $Q$ of which we seek to compute the discrete logarithm with respect to base $P$. It also computes and stores $3\cdot x$ for the $x$-coordinates of all points in the table.

\begin{algorithm}[hbp]
\SetAlgoCaptionLayout{raggedright}
\caption{Randomized window setup}
\label{alg:randomized-window-setup}
\KwIn{$P,Q\in E(\F_p)$ of prime order $r$; width $w\ge1$ with $2^w<r$}
\KwOut{Accumulator $A_0$ and classical augmented tables
       $(T_i)_{i=0}^{2L-1}$}
$a_0\sample\Z_r^\ast$; $A_0\gets[a_0]P$\;
\For{$i\gets0$ \KwTo $L-1$}{
  $\mu_i\sample
    \Z_r\setminus\{-j2^{iw}\bmod r:0\le j<2^w\}$\;
  \For{$j\gets0$ \KwTo $2^w-1$}{
    $(x_{i,j},y_{i,j})\gets[j2^{iw}+\mu_i]P$\;
    $T_i[j]\gets(x_{i,j},y_{i,j},3x_{i,j})$\;
  }
}
\For{$i\gets L$ \KwTo $2L-1$}{
  $\mu_i\sample
    \Z_r\setminus\{-j2^{(i-L)w}\bmod r:0\le j<2^w\}$\;
  \For{$j\gets0$ \KwTo $2^w-1$}{
    $(x_{i,j},y_{i,j})\gets[j2^{(i-L)w}+\mu_i]Q$\;
    $T_i[j]\gets(x_{i,j},y_{i,j},3x_{i,j})$\;
  }
}
\KwRet{$A_0,(T_i)_{i=0}^{2L-1}$}
\end{algorithm}

Algorithm~\ref{alg:windowed-double-scalar} is a simple loop over all $2L$ windows that just adds the respective table points. Note that the scalars $k$ and $l$ are decomposed into their windowed $2^w$-ary representations
\[
  k=\sum_{i=0}^{L-1}k_i2^{iw},
  \quad
  l=\sum_{i=0}^{L-1}l_i2^{iw},
\]
and the digits $k_i$ and $l_i$ are used to look up the correct table points that are passed to the affine addition $\mathrm{AffineAdd}$ depicted in Algorithm~\ref{alg:point-add-dialog-binary-gcd}. In practice, we use $w=16$ and we only perform $28$ of the $32$ point additions as described in~\cite{litinski2023compute}: We replace the first point addition by a table lookup and drop the final $3$ windowed point additions in favor of extended classical postprocessing~\cite{ekeraa2019revisiting}. We start with the $w$ most-significant bits of the first scalar and move toward the least-significant bit in each step, implementing the semi-classical quantum Fourier transform in blocks of size $w$~\cite{griffiths1996semiclassical}.

\begin{algorithm}[tbp]
\SetAlgoCaptionLayout{raggedright}
\caption{Windowed double-scalar multiplication}
\label{alg:windowed-double-scalar}
\KwIn{Registers $k,l$; accumulator $A=A_0$ and tables
      $(T_i)_{i=0}^{2L-1}$ from Algorithm~\ref{alg:randomized-window-setup}}
\KwOut{$k,l$ unchanged; updated accumulator $A$}
\For{$i\gets0$ \KwTo $L-1$}{
  $A\gets\mathsf{AffineAdd}(A;k_i,T_i)$\;
}
\For{$i\gets L$ \KwTo $2L-1$}{
  $A\gets\mathsf{AffineAdd}(A;l_{i-L},T_i)$\;
}
\KwRet{$k,l,A$}
\end{algorithm}

The affine addition algorithm $\mathsf{AffineAdd}$ uses the following subroutines.
$\mathsf{BinaryGCDRecord}$ consumes its input and retains the
compressed Dialog/add-sub record.
$\mathsf{BinaryGCDInvMul}$ and
$\mathsf{BinaryGCDMul}$ replay it to compute the in-place division or in-place multiplication, respectively, and
$\mathsf{BinaryGCDUnrecord}$ restores the input and clears all
record scratch. The two replays use modular halvings and doublings,
respectively.

The notation $\mathsf{Lookup}_{(x,y)}$ or $\mathsf{Lookup}_{3x}$ loads the indicated
coordinates, while $m\gets\mathsf{XClear}(u)$ measures and clears a lookup
register and records its \(X\)-basis outcome.  The final
$\mathsf{LookupPhaseFix}$ applies the deferred lookup correction.  Displayed
modular additions, subtractions, and the square-subtract use the approximate circuits. The square-subtract circuit is the one-split Karatsuba version specific
for $\bitcoincurve$. 
Note that we classically sample a uniform random mask $R\in \F_p$ and add it with an approximate quantum-classical adder to the $x$-register before the lookup of the $x$-coordinate multiple. This randomizes the accumulator register in the square-sub circuit. The mask is removed after the subtraction of the square by an approximate addition of $-R\bmod p$. The separate addition step $x\leftarrow (x+R)\bmod p$ can be omitted if a fresh offset $R_i$ is sampled for the $i$-th window in Algorithm~\ref{alg:randomized-window-setup} during the generation of the lookup tables. The table would then contain $3x_{i,j} + R_i$ instead of $3x_{i,j}$ and the masked value would be looked up, making the separate addition unnecessary.

\begin{algorithm}[tbp]
\SetAlgoCaptionLayout{raggedright}
\caption{$\mathsf{AffineAdd}$: affine addition with direct-value
binary-gcd arithmetic}
\label{alg:point-add-dialog-binary-gcd}
\KwIn{Accumulator $(x,y)=P_1$; selector $j$; table $T$ with
      $T[j]=(x_2,y_2,3x_2)$ and $P_2=(x_2,y_2)\ne\Ocal$}
\KwOut{On success, $(x,y)=P_1+P_2$; $j$ unchanged}
$(x_2,y_2)\gets\mathsf{Lookup}_{(x,y)}(T,j)$\;
$x\gets(x-x_2)\bmod p$\;
$m_x^{(1)}\gets\mathsf{XClear}(x_2)$\;
$y\gets(y-y_2)\bmod p$\;
$m_y^{(1)}\gets\mathsf{XClear}(y_2)$\;

$\rho\gets\mathsf{BinaryGCDRecord}(x)$\;
$y\gets\mathsf{BinaryGCDInvMul}(\rho,y)$\;
$x\gets\mathsf{BinaryGCDUnrecord}(\rho)$\;

$R\sample\F_p$\;
$x\gets(x+R)\bmod p$\;
$x_2\gets\mathsf{Lookup}_{3x}(T,j)$\;
$x\gets(x+x_2)\bmod p$\;
$m_x^{(2)}\gets\mathsf{XClear}(x_2)$\;
$x\gets(x-y^2)\bmod p$\;
$x\gets(x+(-R\bmod p))\bmod p$\;

$\rho\gets\mathsf{BinaryGCDRecord}(x)$\;
$y\gets\mathsf{BinaryGCDMul}(\rho,y)$\;
$x\gets\mathsf{BinaryGCDUnrecord}(\rho)$\;

$(x_2,y_2)\gets\mathsf{Lookup}_{(x,y)}(T,j)$\;
$y\gets(y-y_2)\bmod p$\;
$x\gets(x-x_2)\bmod p$\;
$x\gets-x\bmod p$\;
$m_x^{(3)}\gets\mathsf{XClear}(x_2)$\;
$m_y^{(3)}\gets\mathsf{XClear}(y_2)$\;
$(m_x,m_y)\gets
  (m_x^{(1)}\oplus m_x^{(3)},
   m_y^{(1)}\oplus m_y^{(3)})$\;
$\mathsf{LookupPhaseFix}(T,j;m_x,m_y,m_x^{(2)})$\;
\KwRet{$x,y,j$}
\end{algorithm}

\subsection{Modular subtraction of a square}
This section provides explicit integer-valued algorithm descriptions for the operation that subtracts a square modulo $p$, i.e., $\ket{x}\ket{y}$ to $\ket{x}\ket{(y-x^2)\bmod p}$. The main algorithm is Algorithm~\ref{alg:karatsuba-square-sub}. It uses $\AddBMultiple_\pm^q$ (Algorithm~\ref{alg:add-b-multiple}), $\AddCMultiple_-^q$ (Algorithm~\ref{alg:add-c-multiple}), and $\AddSlice_\sigma^q$ (Algorithm~\ref{alg:add-slice}), which in turn uses $\PhaseGE^q$ (Algorithm~\ref{alg:phasege}) and is itself a component of $\AddBMultiple$ and $\AddCMultiple$.

We treat input values as integers less than $2^n$ that are not necessarily canonical representatives less than $p$. Intermediate results are also integer values of their respective bit sizes. The phase $\phi$ is implicit in all algorithm signatures, except where it is displayed explicitly by $\PhaseGE$. An assignment to $\phi$ records the phase acquired by the active computational-basis branch.

For the specific parameters of the curve $\bitcoincurve$, $\AddBMultiple_{\pm}^{q}$ and $\AddCMultiple_{-}^{q}$ use the signed binary representation of $c$ in the form of the list
\begin{align*}
  T_c & =((32,+1),(10,+1),(6,-1),(4,+1),(0,+1)),\\
  c & =\sum_{(d,\epsilon)\in T_c}\epsilon2^d
   =2^{32}+2^{10}-2^6+2^4+1.
\end{align*}
Besides the given curve parameters $n$ and $c$, the carry-propagation width $\kappa$, the local carry-padding width $\rho$, and the phase-comparison width $\delta$ are positive integers that can be selected for different trade-offs. We assume that $\rho\ge 0$, $1\le\delta\le\frac n2$, $\kappa\ge 1$, $2c\le2^\kappa\le2^n$, $32+\frac n2+2+\rho<n$, $\kappa+\rho<n$ to ensure that all slice and window invocations in these operations are well-formed.

\begin{algorithm*}[tbp]
\caption{One-level Karatsuba square-subtract}
\label{alg:karatsuba-square-sub}
\KwIn{$x,y \in \Z$, $0\le x,y < p$, control $q\in\{0,1\}$;
      requested phase-comparison width $\delta\in\Z_{\ge1}$,
      $\delta\le n/2$}
\KwOut{$x,y\gets x,y'$, where $y'\in[0,p)$,
       $y'\equiv y-qx^2\pmod p$ (on failure-free execution)}
$x_L\equiv x_{[0,n/2)}$; $x_H\equiv x_{[n/2,n)}$
  \tcp*[r]{register views; $x=Bx_H+x_L$}
$v\gets x_L^2$\tcp*[r]{exact integer square in $n$-bit scratch}
$(v,y)\gets\AddBMultiple_{+}^{q}(v,y;n;\delta)$\tcp*[r]{add $Bx_L^2$}
$(v,y)\gets\AddSlice_{-}^{q}(v,y;n,0,n,0;\delta)$\tcp*[r]{add $-x_L^2$}
$v\gets0$\tcp*[r]{uncompute $x_L^2$}
$v\gets x_H^2$\tcp*[r]{exact integer square in $n$-bit scratch}
$(v,y)\gets\AddBMultiple_{+}^{q}(v,y;n;\delta)$\tcp*[r]{add $Bx_H^2$}
$(v,y)\gets\AddCMultiple_{-}^{q}(v,y;\delta)$\tcp*[r]{add $-cx_H^2$}
$v\gets0$\tcp*[r]{uncompute $x_H^2$}
$u\gets0$; $\widehat{x}_H\equiv x_H+uB$
  \tcp*[r]{append one clean bit; no copy of $x_H$}
$\widehat{x}_H\gets\widehat{x}_H+x_L$
  \tcp*[r]{overwrite the high half of $x$ in place}
$v\gets\widehat{x}_H^2$\tcp*[r]{exact integer square in $(n+2)$-bit scratch}
$(v,y)\gets\AddBMultiple_{-}^{q}(v,y;n+2;\delta)$
  \tcp*[r]{add $-B\widehat{x}_H^2$}
$v\gets0$\tcp*[r]{uncompute $\widehat{x}_H^2$}
$\widehat{x}_H\gets\widehat{x}_H-x_L$
  \tcp*[r]{restore $x_H$ and clear $u$}
release the clean bit $u$\;
\KwRet{$x,y$}\;
\end{algorithm*}

\begin{algorithm*}[tbp]
\caption{$\AddSlice_{\sigma}^{q}(s,v;n_s,o,w,t;\delta)$}
\label{alg:add-slice}
\KwIn{$s,v\in\Z$, $0\le s<2^{n_s}$, $n_s\in\Z_{\ge1}$,
      $0\le v<2^n$;\linebreak
      $0\le o<n_s$; $1\le w\le n_s-o$; $0\le t\le n-w$;
      sign $\sigma\in\{+1,-1\}$;
      requested phase-comparison width $\delta\in\Z_{\ge1}$;
      either $t+w+\rho<n$ or $t+w=n$ with $\delta\le w$;
      control $q\in\{0,1\}$}
\KwOut{$s,v\gets s,v'$, where $v'\in[0,2^n)$,
       $v'\equiv v+q\sigma2^ts_{[o,o+w)}\pmod p$
       (on failure-free execution)}
$a \equiv s_{[o,o+w)}$\;
\If{$\sigma=-1$}{
  $v_{[0,\kappa)}\gets v_{[0,\kappa)}+c\pmod{2^\kappa}$
    \tcp*[r]{add $c$ in the low correction window}
  $v\gets(2^n-1)-v$\tcp*[r]{enter complemented orientation}
}
\eIf{$t+w+\rho<n$}{
  $v_{[t,t+w+\rho)}\gets v_{[t,t+w+\rho)}+qa
     \pmod{2^{w+\rho}}$\tcp*[r]{zero-pad $a$ by $\rho$ bits}
}{
  $(v_{[t,n)},h)\gets v_{[t,n)}+q a$\tcp*[r]{here $t+w=n$}
  $v_{[0,\kappa)}\gets v_{[0,\kappa)}+hc\pmod{2^\kappa}$
    \tcp*[r]{fold $h2^n$ back as $hc$}
  $m\gets\mathsf{Measure}_X(h)$\tcp*[r]{measure and clear the carry}
  \If{$m=1$}{
    $\phi\gets(-1)^q\phi$\tcp*[r]{exact measured-carry phase}
    $\PhaseGE^{q}(v_{[n-\delta,n)},a_{[w-\delta,w)};\delta)$
      \tcp*[r]{repair the comparison phase}
  }
}
\If{$\sigma=-1$}{
  $v_{[0,\kappa)}\gets v_{[0,\kappa)}+c\pmod{2^\kappa}$
    \tcp*[r]{add $c$ in the low correction window}
  $v\gets(2^n-1)-v$\tcp*[r]{leave complemented orientation}
}
\KwRet{$s,v$}\;
\end{algorithm*}

\begin{algorithm}[H]
\caption{$\PhaseGE^{q}(x,y;\delta)$}
\label{alg:phasege}
\KwIn{$x,y\in\Z$, $0\le x,y<2^\delta$, $\delta\in\Z_{\ge1}$;
      control $q\in\{0,1\}$, phase $\phi\in\{+1,-1\}$}
\KwOut{$x,y,\phi\gets x,y,(-1)^{q\ind[x\ge y]}\phi$ \linebreak
       (on failure-free execution)}
$\chi\gets\ind[x\ge y]$ in a clean scratch qubit\;
$\phi\gets(-1)^{q\chi}\phi$\;
$\chi\gets0$\tcp*[r]{uncompute the exact comparison}
\KwRet{$x,y,\phi$}\;
\end{algorithm}

\begin{algorithm}[H]
\caption{$\AddBMultiple_{\sigma}^{q}(s,v;n_s;\delta)$}
\label{alg:add-b-multiple}
\KwIn{$s,v\in\Z$, $0\le s<2^{n_s}$, $n_s\in\{n,n+2\}$,
      $0\le v<2^n$; sign $\sigma\in\{+1,-1\}$; control $q\in\{0,1\}$;
      requested phase-comparison width $\delta\in\Z_{\ge1}$,
      $\delta\le n/2$}
\KwOut{$s,v\gets s,v'$, where $v'\in[0,2^n)$,
       $v'\equiv v+q\sigma Bs\pmod p$ \linebreak
       (on failure-free execution)}
$s_L\equiv s_{[0,n/2)}$; $s_H\equiv s_{[n/2,n_s)}$\\
  \tcp{low and high register views of $s$}
$w_h\gets n_s-n/2$\;
$(s,v)\gets\AddSlice_{\sigma}^{q}(s,v;n_s,0,n/2,n/2;\delta)$\\
  \tcp{add $\sigma Bs_L$; wrap}
\ForEach{$(d,\epsilon)\in T_c$}{
  $(s,v)\gets\AddSlice_{\sigma\epsilon}^{q}
       (s,v;n_s,n/2,w_h,d;\delta)$\\ 
       \tcp{add $\sigma\epsilon2^ds_H$; local}
}
\KwRet{$s,v$}\;
\end{algorithm}

\begin{algorithm}[H]
\caption{$\AddCMultiple_{\sigma}^{q}(s,v;\delta)$}
\label{alg:add-c-multiple}
\KwIn{$s,v\in\Z$, $0\le s,v<2^n$; sign $\sigma\in\{+1,-1\}$;
      control $q\in\{0,1\}$;
      requested phase-comparison width $\delta\in\Z_{\ge1}$,
      $\delta\le n/2$}
\KwOut{$s,v\gets s,v'$, where $v'\in[0,2^n)$,
       $v'\equiv v+q\sigma cs\pmod p$ \linebreak
       (on failure-free execution)}
\ForEach{$(d,\epsilon)\in T_c$}{
  $(s,v)\gets\AddSlice_{\sigma\epsilon}^{q}
       (s,v;n,0,n-d,d;\delta)$\\
       \tcp{direct low part; wrap}
}
$g\gets ch_s$ in a clean 65-bit register\;
$(g,v)\gets\AddSlice_{\sigma}^{q}(g,v;65,0,65,0;\delta)$\\
  \tcp{omitted high part; local}
$g\gets0$\tcp*[r]{uncompute $ch_s$}
\KwRet{$s,v$}\;
\end{algorithm}

\section{Failure case analysis}\label{sec:failureanalysis}
As in essentially all previous work (see, e.g.,~\cite{proos2003shor,haner2020improved,roetteler2017quantum,gouzien2023performance,litinski2023compute,babbush2026securing,schrottenloher2026optimized}), the elliptic curve point addition circuits described here are approximate, which means there exist inputs for which the circuits fail to compute the correct result for reasons from two different categories. Failures in the first category occur when the input elliptic curve points constitute an exceptional case where the underlying addition formula fails. Failures in the second category are due to an approximate modular arithmetic component failing to produce the exact intended result. This section shows that the overall number of failure cases occurring in our proposed randomized circuit for affine point addition is small enough to allow the full Shor algorithm to succeed with high probability. In particular, we show that for each pair of scalars $(k,l)$, the probability that it leads to a failure case is bounded by a constant $p_f$ independent of $k$ and $l$. 

\subsection{Exceptional cases for elliptic curve point addition}

As in most previous work, we use the affine point addition formula for the short Weierstrass model $E: y^2 = x^3 + ax + b$. Given $P_1 = (x_1, y_1)$ and $P_2 = (x_2, y_2)$, the point $P_3 = (x_3, y_3) = P_1 + P_2$ is computed as $x_3 = \lambda^2 - x_1 - x_2$ and $y_3 = (x_2 - x_3)\lambda - y_2$, where $\lambda = (y_2-y_1)/(x_2-x_1)$. This formula has exceptional cases. It cannot be used if $x_2 = x_1$, i.e., if $P_1 = \pm P_2$, and also if one of the points is the point at infinity $\Ocal$. Because quantum circuits can only implement reversible operations, the slope value $\lambda$ is often uncomputed from the results by re-computing $\lambda=(y_2+y_3)/(x_2-x_3)$. This implies another exceptional case coming from the condition $x_3 = x_2$. It follows that $y_3 = \pm y_2$, hence this happens when the sum is equal to $\pm P_2$, so either $P_1 = \Ocal$ or $P_1 = -[2]P_2$. In total, there are four exceptional cases for the point $P_1$ when we aim to add a point $P_2$, namely $P_1 \in \{\Ocal, P_2, -P_2, -[2]P_2\}$. Note that Roetteler et al. missed this fourth exceptional case in their analysis of the exceptional case count in~\cite[Section~4.2]{roetteler2017quantum}.

Let $A$ be the accumulator point represented by the quantum register that will hold the approximate superposition over $(k,l)$ of $f(k,l) = [k]P + [l]Q$ after the computation. As already suggested in \cite{proos2003shor}, we initialize it with a classically precomputed uniform random point $[a_0]P$, where $a_0$ is sampled uniformly at random with $a_0 \in \Z_r^* = \Z_r\setminus\{0\}$. Using this random offset allows us to avoid the exceptional case of adding a point to $\Ocal$ and to reason about probabilities over the uniform random starting point. After the quantum double-scalar multiplication, the accumulator register holds an approximate superposition over $k$ and $l$ of $f(k,l) + [a_0]P$. Since $a_0$ is independent of $(k,l)$, the constant shift has no impact on the Fourier-basis measurement that follows and does not have to be subtracted from $A$.
 
As described earlier, the double-scalar multiplication is computed via windowed addition of points looked up from precomputed tables of classical elliptic curve points. We fix a window size $w>0$. The scalars are decomposed into windows of size $w$ as $k = \sum_{i=0}^{L-1}k_i 2^{i\cdot w}$ and $l = \sum_{i=0}^{L-1}l_i 2^{i\cdot w}$ with $k_i, l_i\in \{0,1,\dots,2^w-1\}$. There are $2L$ separate tables of size $2^w$, one for each window index $0\le i<2L$ containing the precomputed classical multiples $[j 2^{iw}]P$ for all $j \in \{0,1,\dots,2^w-1\}$ when $0 \le i <L$, and $[j 2^{(i-L)w}]Q$ when $L\leq i < 2L$. All tables contain the point $\Ocal$ for the cases where one of the $k_i=0$ or one of the $l_i=0$. Usually, the special case when the lookup point is $\Ocal$ needs to be treated separately to ensure that the point addition circuit does not change the accumulator, leading to operations that are controlled on the condition $k_i=0$ (or $l_i=0$). To avoid these controlled operations, we add a fixed classical masking point $M_i$ to all points in the $i$-th table. During the classical table generation, the point is chosen uniformly at random via sampling a uniform random $\mu_i\in \Z_r$ and computing $M_i = [\mu_i]P$ if $0\le i < L$ (or $M_i = [\mu_i]Q$ if $L\le i < 2L$) under the additional condition that none of $[j2^{iw}]P + M_i$ (or $[j2^{(i-L)w}]Q + M_i$) are equal to $\Ocal$, i.e., we sample $\mu_i \in \Z_r\setminus \{-2^{iw}j\mid j \in \{0,1,\dots,2^w-1\}\}$ (or $\mu_i \in \Z_r\setminus \{-2^{(i-L)w}j\mid j \in \{0,1,\dots,2^w-1\}\}$). Since we assume that $2^w<r$, there exists a point $M_i$ for each table such that none of the shifted points is $\Ocal$. To summarize, the classical table $T_i$ for looking up the point to be added into the accumulator for the $i$-th window is given as
\[
T_i =
\begin{cases}
\{[j2^{iw}+\mu_i]P \mid j \in [0,2^w)\}, \text{ if } 0\le i < L, \\
\{[j2^{(i-L)w}+\mu_i]Q \mid j \in [0,2^w)\}, \text{ if } L \le i < 2L.
\end{cases}
\]
Using these tables results in another additive shift for the oracle computation. Due to random initialization and table masking, after the double scalar multiplication, the accumulator register holds an approximation to the superposition over all $(k,l)$ of
\begin{align*}
A & = [k]P + [l]Q + ([a_0+\mu_P]P + [\mu_Q]Q), \text{ where }\\
& \mu_P = \sum_{i=0}^{L-1}\mu_i \text{ and } \mu_Q = \sum_{i=L}^{2L-1}\mu_i.
\end{align*} 
The shift $[a_0+\mu_P]P + [\mu_Q]Q$ is independent of $k$ and $l$ and has no impact on the subsequent phase estimation step in Shor's algorithm. Therefore, it does not need to be subtracted or uncomputed. Let $\mu = (\mu_0,\mu_1,\dots,\mu_{2L-1})$ and define 
\begin{align*}
f_{a_0,\mu}(k,l) & = f(k,l) + [a_0+\mu_P]P + [\mu_Q]Q\\
& = [k]P + [l]Q + [a_0+\mu_P]P + [\mu_Q]Q.
\end{align*}
For detailed, integer-valued pseudo-code representations of these operations, see Algorithms~\ref{alg:randomized-window-setup}, \ref{alg:windowed-double-scalar}, and \ref{alg:point-add-dialog-binary-gcd}.

Next, we analyze when a pair $(k,l)$ leads to an exceptional case during the windowed point addition. Let $A_i$ denote the state of the accumulator point during an ideal, correct computation before the $i$-th windowed point addition, i.e. $A_0=[a_0]P$, $A_1 = [a_0+k_0+\mu_0]P$, $A_L=[a_0+k+\mu_P]P$, $A_{L+1}=[a_0+k+\mu_P]P + [l_0+\mu_L]Q$, etc. Note that for $i>0$, the accumulator $A_i$ can only assume the value $\Ocal$ in a correct computation if, in the preceding step, $A_{i-1}$ was equal to the negative of the added table point, a case already covered as an exceptional case for the earlier step. Moreover, our random offset ensures that $A_0\neq\Ocal$. This means that an exceptional point in an addition step that adds the table point $C_i$ to $A_i$ comes from one of the three cases $A_i\in \{C_i, -C_i, -[2]C_i\}$.

A pair $(k,l)$ leads to an exceptional point addition if there exists an index $0\le i< 2L$ such that $A_i\in \{C_i, -C_i, -[2]C_i\}$, where $C_i = [k_i2^{iw} + \mu_i]P$ if $0\le i<L$ and $C_i = [l_{i-L}2^{(i-L)w} + \mu_i]Q$ if $L\le i<2L$. Define the set of all $a_0\in\Z_r^*$ that lead to an exceptional case for given $(k,l)$ as follows:
\[
\Xcal_{w, \mu}(k,l) = \{a_0 \in \Z_r^* \mid \exists i <2L, A_i \in \{C_i, -C_i, -[2]C_i\}\}.
\]
Because the map $a_0 \mapsto [a_0]P$ is a bijection $\Z_r \rightarrow \langle P \rangle$, each of $C_i, -C_i, -[2]C_i$ can only be reached by $A_i$ for fixed $(k,l)$ starting from a single value for $a_0$. Therefore, we have
\[
|\Xcal_{w, \mu}(k,l)| \le 3(2L) = 6L.
\]
Let $\bar f_{w,a_0,\mu}(k,l)$ denote the output of the $w$-windowed affine-addition circuit for $f_{a_0,\mu}(k,l)$ when every modular-arithmetic subroutine is exact. If $a_0\notin\mathcal X_{w,\mu}(k,l)$, every addition along the ideal trajectory is non-exceptional. Correctness of the affine addition circuit on non-exceptional inputs and induction over the $2L$ additions therefore give $\bar f_{w,a_0,\mu}(k,l)=f_{a_0,\mu}(k,l)$. It follows that $\{a_0 \in \Z_r^* \mid \bar f_{w,a_0,\mu}(k,l) \neq f_{a_0,\mu}(k,l)\}\subseteq \Xcal_{w, \mu}(k,l)$.

\begin{lemma}
Let $n, w\in \N$, $1\le w <n$, $2^w < r$, $L = \lceil n/w \rceil$, $\mu = (\mu_0,\dots,\mu_{2L-1}) \in \Z_r^{2L}$ a fixed mask vector chosen as described above. Let $(k,l) \in [0,2^n) \times [0,2^n)$ and let $\Xcal_{w, \mu}(k,l)$ be the set of all $a_0\in \Z_r^*$ that lead to an exceptional case in a quantum circuit for computing $f_{a_0,\mu}(k,l)$ using windowed point addition with masked precomputed tables and a non-zero initialization point $[a_0]P$ as described above. Then, over uniform random $a_0\in \Z_r^*$, we have
\begin{align*}
\Pr_{a_0\in \Z_r^*}(\bar f_{w,a_0,\mu}(k,l) \neq f_{a_0,\mu}(k,l)) & \le \Pr_{a_0\in \Z_r^*}(a_0\in \Xcal_{w, \mu}(k,l))\\
& = \frac{|\Xcal_{w,\mu}(k,l)|}{r-1} \le \frac{6L}{r-1}.
\end{align*}
\end{lemma}
For our concrete parameters of $\bitcoincurve$, this bound is negligibly small, smaller than $2^{-248}$.

In practice, we implement only $29<2L=32$ windows (the first using a simple table lookup, 28 using windowed point additions), but the expression above still holds as a conservative upper bound for the failure probability due to exceptional cases.

\subsection{Approximate modular arithmetic failures}
The second category of failures occurs because the modular arithmetic operations that are used in the affine point addition circuit are approximate operations themselves~\cite{schrottenloher2026optimized}. Again, there exist inputs for which the circuits fail to produce the expected results.

The modular arithmetic circuits presented here make use of the special shape of the pseudo-Mersenne base field prime $p$. Several simplifications make modular arithmetic circuits more efficient but possibly lead to failures. For example, integer comparison with $p$ for modular reduction is replaced by comparison of the most significant bits and the correction addition that depends on the outcome of the comparison truncates the carry propagation. These optimizations are applied to the modular reduction operations in modular addition, subtraction, doubling, and halving circuits~\cite{schrottenloher2026optimized}. Similarly, truncated comparisons during phase corrections after measurement-based carry uncomputation can lead to phase failures even if the resulting values are correct.
Furthermore, Dialog-based in-place multiplication~\cite{khattar2025verifiable,schrottenloher2026optimized} uses a non-worst-case bound for both the number of iterations and the register sizes used as the algorithm progresses~\cite{proos2003shor}. We begin by discussing the optimizations in modular reduction.

\subsubsection{Approximate modular reduction}
An element in $\F_p$ is usually represented by its unique representative in $\{0,1,\dots,p-1\}$ (which we call the least non-negative or canonical representative). Integer arithmetic operations on such representatives may lead to values outside that range and a modular reduction is applied to reduce such a value back to its canonical representation.
For example, taking two input operands from $\F_p$ and adding their standard integer representatives results in a value $0 \le u < 2p-1$. Modular reduction then subtracts $p$ if $u \ge p$, which means $u \bmod p = u$, if $u < p$, and $u \bmod p = u - p$ if $u\ge p$. Thus, there are two computational components, namely, computing the integer comparison $u < p$ and depending on its outcome an integer subtraction $u - p$. 

We work in the specific setting of a pseudo-Mersenne prime $p$ of bitlength $n$. Let $c\in \N$ be odd, $c\ll 2^n$ such that $p = 2^n - c$. That $p$ has $n$ bits implies $2c\le2^n$. For the curve $\bitcoincurve$, we have $n=256$ and $c=2^{32} + 2^9 + 2^8 + 2^7 + 2^6 + 2^4 + 1 = 2^{32} + 977$. Subtracting such a $p$ results in $u - p = u - (2^n - c) = u + c - 2^n$ and can be implemented by computing the addition $u + c$ and discarding the $(n+1)$-th bit. Note that, if $u \ge p$, then $u + c \ge 2^n$ and its $(n+1)$-th bit is always $1$.

The two optimizations above, which were also used in previous work~\cite{schrottenloher2026optimized}, lead to approximate modular reduction circuits that we use in different situations in our circuits. First, the full integer comparison $u<p$ is replaced by a comparison $u<2^n$, which amounts to simply observing the $(n+1)$-th bit of $u$. Instead of computing the integer comparison into a result bit that acts as the control for the modular subtraction, the $(n+1)$-th bit directly serves as the control bit. This operation differs from exact integer comparison if
\[
u \in [p, 2^n).
\]
This means that for $c$ of the possible results of an addition, the optimized reduction fails to produce a canonical representative. Note that the result can still be represented with $n$ bits and is correct modulo $p$ but is by $p$ larger than the canonical representative. This behavior does not necessarily lead to a direct arithmetic failure; however, it might lead to failure cases when the larger result is used in further computations. We therefore treat non-canonical representatives for outputs of any operation as a failure case.

Second, for the addition $u + c$, we truncate carry propagation after $\kappa \le n$ bits for $c < 2^\kappa$ as in~\cite{schrottenloher2026optimized}. This operation differs from exact addition with full carry propagation if the carry propagates further than $\kappa$ bits, i.e.,
\[
u \bmod 2^\kappa \in [2^\kappa - c, 2^\kappa).
\]   
This failure case, where the result is wrong modulo $p$, occurs for a proportion of $c/2^\kappa$ of input values, in which case the result after reduction is $2^\kappa$ smaller than the correct one.

Because of the simplified comparison and since our circuits work on $n$-bit quantum registers, we allow intermediate values modulo $p$ to be represented by integers in $[0,2^n)$. We assume throughout that $2c \le 2^\kappa \le 2^n$. For a non-negative integer $u$, define
\begin{align*}
q_u & = \lfloor u/2^n \rfloor,\quad r_u = u - q_u2^n \in [0,2^n),\\
w_u & = 
\begin{cases}
1, \text{ if }\ (r_u \bmod 2^\kappa) \ge 2^\kappa - c,\\
0, \text{ otherwise.}
\end{cases}
\end{align*}
If $u\in [0,2^{n+1})$, $q_u \in \{0,1\}$ is a single bit and the approximate reduction as described above then returns
\[
\red_\kappa(u) = r_u + q_uc - q_uw_u2^\kappa,
\]
where $q_u$ is the $(n+1)$-th bit, the correction is triggered if it is $1$, in which case $c$ is added to the first $n$ bits of $u$ and the result is $2^\kappa$ smaller than the correct value if the carry when adding $c$ propagates out of the first $\kappa$ bits. Since the $(n+1)$-th bit was discarded and the carry never propagates beyond the $\kappa$-th bit, $\red_\kappa(u)\in [0,2^n)$. Using $2^n \equiv c \bmod p$ and taking modulo $p$ gives
\begin{equation}
\red_\kappa(u) \equiv u - q_u w_u 2^\kappa \pmod p,
\end{equation}
which means that the approximate reduction $\red_\kappa$ returns a correct integer representative of $u$ if and only if either no correction occurs ($q_u = 0$) or the carry from the addition of $c$ does not propagate beyond $\kappa$ bits ($w_u = 0$). The above two failure cases imply incorrect result values. Another type of failure leads to an incorrect phase.

\subsubsection{Approximate phase repair}
To uncompute certain bits that are computed in our circuits such as a carry bit $b$ depending on a register $\ket{x}$, we use measurement in the $X$ basis, which changes the phase by a factor $(-1)^{mb}$, where $m$ is the measurement outcome. This means that a phase correction is required if $m=1$ and $b=1$. Our circuits use approximate recomputation of the bit $b$ for example by truncating carry propagation or comparing only the higher order bits. If the approximation is incorrect, a phase failure occurs because the circuit fails to correct the phase flip.

We now proceed to deduce bounds on the failure probabilities for the various components in the approximate affine point addition circuit. We start with the simpler modular operations such as addition, subtraction, and negation, then discuss the modular subtraction of a square, and finally, the in-place modular division and multiplication circuits.

\subsubsection{Addition, subtraction, and negation}

\paragraph*{Addition.}
The modular addition circuit we use~\cite{schrottenloher2026optimized} is an approximation of the operation $\ket{u}\ket{v}\mapsto\ket{u}\ket{(u+v)\bmod p}$. The inputs $u,v\in[0,p)$ are canonical integer representatives of finite field elements. Let $s=u+v \in [0,2p-1) \subseteq [0, 2^{n+1})$ be their integer sum. The circuit applies the approximate reduction described above and returns
\[
  \red_\kappa(s)
  =r_s+q_sc-q_sw_s2^\kappa
  =s-q_sp-q_sw_s2^\kappa.
\]
The circuit can fail to produce the correct result value in two different ways due to the applied approximations. The first is a failure due to the simplified comparison, where the most significant bit serves directly as the comparison bit. If $p\le s<2^n$, no correction is applied. The result is correct modulo $p$ but remains larger than $p$.

The second failure occurs if the simplified comparison is correct, it causes a correction but the truncated carry propagation modifies the result because the carry propagates further than $\kappa$ bits. This happens for $q_sw_s=1$ and the result is $2^\kappa$ too small.

To count failure pairs $(u,v)\in[0,p)^2$, i.e., pairs that lead to the above two failures, fix an integer value $s\in [p,2p-1)$ (no failure occurs for $s\in [0,p)$). For such $s$, the number of $(u,v)$ with sum $s = u+v$ is $2p-1-s$. Summing over all relevant values for $s$ that lead to a simplified comparison failure yields their overall number as
\[
  \sum_{s=p}^{2^n-1}(2p-1-s)
  =\frac{c(2p-c-1)}{2}.
\]
Failures due to truncated carry propagation occur when $q_sw_s=1$ and $r_s\bmod2^\kappa\in[2^\kappa-c,2^\kappa)$. The relevant sums are of the form $s=2^n+j2^\kappa-1-b$, $1\le j\le 2^{n-\kappa}-1$, $0\le b<c$. This time, the number of pairs $(u,v)\in[0,p)^2$ satisfying $u+v=s$ for fixed $s$ is $2p-1-s =2^n-2c-j2^\kappa+b$, and again, summing over the relevant values yields their overall count as
\begin{align*}
  &\sum_{j=1}^{2^{n-\kappa}-1}\sum_{b=0}^{c-1}
    \left(2^n-2c-j2^\kappa+b\right)\\
  & = (2^{n-\kappa}-1)c(2^n-2c) -c2^\kappa
     \frac{2^{n-\kappa}}{2}(2^{n-\kappa}-1)\\
  & \qquad+(2^{n-\kappa}-1)\frac{c(c-1)}{2}\\
  &= c(2^{n-\kappa}-1)\frac{2^n-3c-1}{2}
  = c(2^{n-\kappa}-1)\frac{p-2c-1}{2}.
\end{align*}

The two cases fall into disjoint events because for the first, the simplified comparison does not lead to a correction ($q_s = 0$), whereas for the second, it does ($q_s = 1$).

A third failure case is a $(-1)$-phase that is not corrected due to a failure in the approximate recomputation of the bit $q_s$ after an $X$-measurement. This can happen even if the value of the representative is correctly computed. We are only interested in this case because it is disjoint from the previous two, where all value failures have been accounted for. This means that $w_s=0$, i.e., $\red_\kappa(s) = s-q_sp = r_s + q_sc$ and the reduction yields the canonical representative.

The bit $q_s$ can be precisely recovered using the fact that $q_s = 0$ if and only if $\red_\kappa(s) \ge u$. The comparison used here is approximate, only taking into account the $\delta$ most significant bits of $\red_\kappa(s)$ and $u$. Therefore, a failed comparison computation can occur if $q_s=1$ while $\red_\kappa(s) = s - p <u$, but the top $\delta$ bits of both are equal. To count the number of pairs $(u,v)\in [0,p)^2$ for which this failure arises, note that $s-p = u+v-p \in [c,p)$ and $\lfloor (s-p)/2^{n-\delta} \rfloor = \lfloor u/2^{n-\delta} \rfloor$. We obtain an upper bound by ignoring the condition $w_s=0$. We count failure pairs $(u,v)$ that satisfy $c \le u+v-p < u <p$ by counting ordered pairs $(z = u+v-p, u)$ in $[c,p)$ that agree on their most significant $n-\delta$ bits, i.e., pairs that lie in an interval $I_j = [j2^{n-\delta}, (j+1)2^{n-\delta}) \cap [c,p)$. Let $|I_j| = c_j$. The $I_j$ partition $[c,p)$, so $\sum_j c_j = p-c$ and each $c_j \le 2^{n-\delta}$. The number of ordered pairs in $I_j$ is $\binom{c_j}{2}$ and summing them yields
\[
\sum_j \binom{c_j}{2} = \sum_j \frac{c_j(c_j-1)}{2} \le (2^{n-\delta}-1)\frac{p-c}{2}.
\]

The total number of failure pairs covering all three cases is therefore bounded by
\begin{align*}
N_\add &\le \frac{c(2p-c-1)}{2} + c(2^{n-\kappa}-1)\frac{p-2c-1}{2}\\
& \quad + (2^{n-\delta}-1)\frac{p-c}{2}\\
& = 2^{n-1}c\left(1 + \frac{p-2c-1}{2^\kappa}\right)+ (2^{n-\delta}-1)\frac{p-c}{2}.
\end{align*}
For a uniform distribution of $(u,v) \in \F_p\times \F_p$, the failure probability due to simplified comparison, truncated carry propagation, or truncated phase correction is
\[
  \frac{N_\add}{p^2} < \frac{c}{p}+\frac{c}{2^{\kappa+1}}+ \frac{1+c/p}{2^{\delta+1}}.
\]

\paragraph*{Inputs during affine point addition.}
Inputs to the approximate addition in our elliptic curve point addition are not uniform random but intermediate results of our elliptic curve point addition.

Let $(k,l)$ be a fixed pair of scalars and $0\le i\le 2L-1$. Assume no failures in the circuit up to the $i$-th point addition step. Let $A=(x_1,y_1)$ be the accumulator point and $T=(x_2,y_2)$ the $i$-th lookup point. Approximate addition is used three times in the point addition circuit: to compute $(x_1-x_2) + R \bmod p$, $3x_2 + (x_1-x_2 + R)\bmod p$, and $(R+x_2-x_3) + (-R)\bmod p$, where $R$ is a uniform random mask sampled fresh for each point addition step.

Note that we sample $a_0\in\Z_r^*$ uniformly at random for each run of the double-scalar multiplication, and $\mu_i\in\Z_r\setminus\{-j2^{iw}:0\le j<2^w\}$ and $R\in\F_p$ are uniform random samples independent of $a_0$ and each other for the $i$-th windowed point addition step. The sample space of triples $(a_0, \mu_i, R)$ has size $(r-1)(r-2^w)p$. 

Since we assume no failures up to this point, and with fixed previous random mask values, the map $(a_0,\mu_i)\mapsto(A,T)$ is injective. We can thus count the number of triples that lead to a given pair of inputs as follows.

For an input $(R, x_1-x_2)$ to the first addition to be equal to a fixed pair $(u,v)\in \F_p \times \F_p$, it has to satisfy $R = u$ and $x_1 - x_2 = v$, which gives $p$ possible triples. The two $x$-coordinates belong to at most two elliptic curve points each, meaning there are at most $4p$ pairs $(A,T)$ corresponding to the given addition input pair. The same bound holds for the third addition. For the addition with inputs $3x_2$ and $x_1-x_2+R$, $3x_2 = u$ fixes $x_2$, yielding at most two curve points. There are at most $r-1$ possible points $A$, each of which then fixes $x_1$ and thus $R$. This means that there are at most $2(r-1)$ points per addition input.

Adding these three counts and multiplying their sum by the number of addition failure pairs $N_{\mathrm{add}}$ gives a bound on the number of triples $(a_0, \mu_i, R)$ leading to one of the two approximation failures. Dividing by the total number of triples yields an upper bound on the addition failure probability $p_\add$ during a single elliptic curve point addition as
\begin{align*}
 p_\add &\le
 \frac{(8p+2(r-1))N_{\mathrm{add}}}
      {p(r-1)(r-2^w)} \le
 \frac{10N_{\mathrm{add}}}
      {(r-1)(r-2^w)}\\
\end{align*}
Note that we used $r-1< p$, which is not necessarily true for a general elliptic curve, but holds for $\bitcoincurve$.

\begin{lemma}\label{lem:padd}
Let $n,w,p,c,r,\kappa, \delta$ be the parameters as described in this section. Let $a_0$, $\mu_i$, and $R$ be the random mask values as above. The probability $p_\add$ over random $(a_0,\mu_i,R)$ that either a simplified comparison failure or a carry propagation failure occurs during one of the three approximate additions in the $i$-th point addition step is bounded by 
\begin{align*}
 p_\add <
 \frac{10\cdot 2^{n-1}}{(r-1)(r-2^w)}
 \left(c(1 + \frac{p-2c}{2^\kappa})  + \frac{p-c}{2^{\delta}}\right).
\end{align*}
\end{lemma}

For the concrete parameters of $\bitcoincurve$, we have $n=256$ and $c=2^{32}+977$. Our circuits use $\kappa=65$, $w=16$, and $\delta=32$ such that
\[
p_{\add} <
1.7463\cdot10^{-9}
\approx2^{-29.0931}
<2^{-29.093}.
\]

\paragraph*{Subtraction.}
Our modular subtraction circuit is an approximation of $\ket{u}\ket{v}\mapsto\ket{u}\ket{(v-u)\bmod p}$ for $u,v\in[0,p)$. It is computed using the approximate addition, braced by two modular complements that implement the operation $\ket{x} \mapsto \ket{p-1-x}$. More precisely, the circuit realizes
\begin{align*}
\ket{u}\ket{v} &\mapsto \ket{u}\ket{p-1-v}\\
&\mapsto \ket{u}\ket{((p-1-v)+u) \bmod p}\\
&\mapsto \ket{u}\ket{p-1-((p-1-v+u) \bmod p)}\\
&= \ket{u}\ket{(v-u) \bmod p}.
\end{align*}
Since we already obtained a bound on  the failure probability due to approximate addition, it remains to bound the failure probability due to the approximate modular complement. The operation $\ket{x} \mapsto \ket{p-1-x}$ is implemented by a bitwise complement followed by a truncated $\kappa$-bit addition of $2^\kappa-c$. The same result could be computed by a truncated addition of $c$ to $x$ followed by a bitwise complement. Without failures, this computes the modular complement because $p-1-x = 2^n-c-1-x = ((2^n-1) - x) - c = (2^n-1) - (x+c)$, again using the pseudo-Mersenne shape of $p$. The only additional source of failures is the truncated carry propagation leading to an incorrect value if $x \bmod 2^\kappa \ge 2^\kappa - c$. Counting failure inputs similarly to the analogous case for the modular addition, we obtain the number of inputs leading to a complement failure as
\[
N_\cmpl = c(2^{n-\kappa}-1),
\]
which means that for a uniform input, the failure probability due to truncated carry propagation in the modular complement is $N_\cmpl/p < c/2^\kappa$. 

\paragraph*{Inputs during affine point addition.}
Like in the analysis of approximate addition, we fix scalars $(k,l)$ and $0\le i\le 2L-1$. Approximate subtraction occurs four times, namely to compute $x_1 - x_2$, $y_1 - y_2$, $(y_3+y_2) - y_2$, and $(x_2-x_3) - x_2$.

We argue very similarly to the addition case. Given $x_1, x_2$, there are at most $2$ elliptic curve points with each $x$-coordinate, hence at most $4$. For a pair of $y$-coordinates, there can be $3$ points each due to the corresponding $x$-coordinate satisfying a cubic equation, at most $9$ points in total. Again, due to bijections mapping the pairs of points before and after the point addition and pairs of values through the ideal subtraction, the same bounds apply to the last two calls as to the first two.

Hence, we can conclude that the probability $p_{\sub}$ of modular subtraction approximation failures during a single elliptic-curve point addition is
\begin{align*}
  p_{\sub}
  &\le\frac{(4+9+9+4)N_{\sub}}
              {(r-1)(r-2^w)}.
\end{align*}
Here, $N_\sub \le N_\add + 2pN_\cmpl$ is the number of input pairs that lead to a failure in the approximate subtraction circuit.

\begin{lemma}\label{lem:psub}
Let $n,w,p,c,r,\kappa,\delta$ be the parameters as described in this section. Let $a_0$ and $\mu_i$ be the random mask values as above. The probability $p_\sub$ over random $(a_0,\mu_i)$ that either a borrow propagation failure or a phase failure occurs during one of the four approximate subtractions in the $i$-th point addition step is bounded by 
\begin{align*}
p_\sub
  &\le \frac{13}{5}p_\add + \frac{52 p c (2^{n-\kappa}-1)}{(r-1)(r-2^w)}.
\end{align*}
\end{lemma}

Evaluating the expression with the concrete $\bitcoincurve$ parameters and implementation parameters $\kappa=65$, $w=16$, and $\delta=32$ gives
\[
  p_{\sub} < 1.0593796008\times10^{-8}
\approx2^{-26.4922}
<2^{-26.492}.
\]

\paragraph*{Negation.}
Finally, we analyze an approximation to $\ket{u}\mapsto\ket{(-u)\bmod p}$. Modular negation is computed by a bitwise complement, which gives $2^n-1-u$ and then we subtract $c-1$ to obtain $2^n-c-u = p-u$. The subtraction is approximate by truncating it to $\kappa$ bits. This leads to an incorrect value if $u\bmod2^\kappa\ge2^\kappa-c+1$. The integer value $p-u$ is the correct canonical representative of the modular negative of $u$ except if $u=0$. Therefore, the circuit swaps the values of $0$ and $p$ at the end, leaving all other values intact.

Again, we have two failure cases. The first is due to truncated subtraction of $c-1$ to $\kappa$ bits, happening if $u\bmod2^\kappa\ge2^\kappa-c+1$, and thus is induced by $(2^{n-\kappa}-1)(c-1)$ of the canonical inputs.

The second failure is due to a shortened comparison window for the swap of $0$ and $p$ to the $\tau$ most-significant bits of the register. The swap is triggered when all these bits equal the least-significant bit. The circuit assumes that $\kappa\le n-\tau$ and $\tau\ge2$, otherwise it computes exact negation. Clearly, this comparison implementation has false positives, namely on non-zero even integers that are smaller than $2^{n-\tau}$ and on odd integers in $[2^n-2^{n-\tau},p)$. The total number of these failure cases is
\[
  \left(2^{n-\tau-1}-1\right)
  +\frac{2^{n-\tau}-c-1}{2}
  =2^{n-\tau}-\frac{c+3}{2}.
\]
However, we have already counted some of these as failure cases, namely those that have the form $u=j2^\kappa+2^\kappa-c+h$, $0\le j<2^{n-\kappa}-1$, $1\le h<c$ which lead to an intended intermediate value $(2^{n-\kappa}-j-1)2^\kappa-h$. Counting the blocks in these intervals, the overlap has size
\[
  (2^{n-\tau-\kappa+1}-1)\frac{c-1}{2}.
\]
Since a failure due to the truncated subtraction cannot be corrected by a short swap failure, the total number of failure inputs for approximate negation is
\begin{align*}
  N_{\mathrm{neg}}
  &=(2^{n-\kappa}-1)(c-1)
    +2^{n-\tau}-\frac{c+3}{2}\\
  &\qquad
    -(2^{n-\tau-\kappa+1}-1)\frac{c-1}{2}\\
  &=2^{n-\tau}
    +(c-1)(2^{n-\kappa}-2^{n-\tau-\kappa})-c-1,
\end{align*}
which means that the failure
probability for uniform input satisfies
\[
  \frac{N_{\mathrm{neg}}}{p}
  <\frac{c-1}{2^\kappa}+\frac{1+c/p}{2^\tau}.
\]

\paragraph*{Inputs during affine point addition.}
Modular negation occurs once to compute $-x_3 \bmod p$. Fixing an input $u$ therefore fixes $x_3$, so there are at most $2$ elliptic curve points leading to this $x$-coordinate. The table lookup point still can assume any of its $r-2^w$ possible values, leading to $2(r-2^w)$ possible curve points leading to input $x_3$. With $(r-1)(r-2^w)$ possible point pairs, the probability $p_{\mathrm{neg}}$ of a failure due to approximate negation during a single elliptic-curve point addition therefore is bounded by
\[
  p_{\mathrm{neg}} \le \frac{2N_{\mathrm{neg}}}{r-1}.
\]
\begin{lemma}\label{lem:pneg}
Let $n,w,p,c,r,\kappa$ be the parameters as described in this section. Let $a_0$ and $\mu_i$ be the random mask values as above. The probability $p_{\mathrm{neg}}$ over random $(a_0,\mu_i)$ that either a truncated subtraction failure or a shortened comparison failure occurs during the approximate negation in the $i$-th point addition step is bounded by 
\[
  p_{\mathrm{neg}} \le \frac{2(2^{n-\tau} +(c-1)(2^{n-\kappa}-2^{n-\tau-\kappa})-c-1)}{r-1}.
\]
\end{lemma}
Setting concrete parameters in our implementation for $\bitcoincurve$ with $\tau=32$ results in
\[
  p_{\mathrm{neg}} < 6.9850 \cdot 10^{-10} \approx 2^{-30.4150} < 2^{-30.415}.
\]

\subsubsection{Subtracting a square} 
Consider the subtraction of a square as described in Section~\ref{sec:squaring}. It acts on two registers and subtracts the square of the value in the first from that in the second modulo $p$. It maps $\ket{x}\ket{y}$ to $\ket{x}\ket{(y-x^2)\bmod p}$. We use a uniform random mask $R\in \F_p$ that is added to the target register, which also acts as the accumulator register in the algorithm. The square subtraction circuit then maps
\[
\ket{x}\ket{(y+R)\bmod p}\mapsto \ket{x}\ket{(y-x^2+R)\bmod p}
\]
and a subtraction of $R$ in the second register gives the desired result. 

To understand the possible failure cases due to approximate arithmetic, we discuss the specific components as shown in Appendix~\ref{sec:algorithms}. The main operation is detailed for its integer-valued representation in Algorithm~\ref{alg:karatsuba-square-sub}. It uses the components $\AddBMultiple_\pm^q$ (Algorithm~\ref{alg:add-b-multiple}), $\AddCMultiple_-^q$ (Algorithm~\ref{alg:add-c-multiple}), $\AddSlice_\sigma^q$ (Algorithm~\ref{alg:add-slice}), and $\PhaseGE^q$ (Algorithm~\ref{alg:phasege}). Notation and parameter properties are introduced in Appendix~\ref{sec:algorithms}.

\paragraph*{Approximation failures in $\AddSlice$.} The basic component used in the other algorithms is $\AddSlice$ (Algorithm~\ref{alg:add-slice}). It operates on two quantum registers, a source register and a target register, holding values $s$ and $v$. It updates the target by adding $\sigma2^t s_{[o,o+w)}$. The parameter $\delta$ provides a way to trade efficiency for accuracy: it determines the number of bits used in the phase comparison, while the parameter $\rho\in\Z_{\ge0}$ is the carry-padding width for local slice additions. The parameters must satisfy either $t+w+\rho<n$ or $t+w=n$ with $\delta\le w$. There are four possible sources of approximation failures in $\AddSlice$.

A failure can occur during the complement computation of $(p-1)-v$ in the case of negative sign $\sigma=-1$. While the bitwise complement operation is exact, the preceding addition of $c$ may drop a propagating carry. The corresponding failure preimage set is
\[
\widetilde E_{\mathrm{compl}} = \{v \in [0,2^n) \mid v_{[0,\kappa)}\ge 2^\kappa - c\}
\]
and has cardinality $c\,2^{n-\kappa}$.
Complement computation only happens if $\sigma=-1$, but then it occurs twice, once when entering and once when leaving the complemented representation.

Another failure occurs for $t+w+\rho<n$ if in the addition of the source slice $s' = s_{[o,o+w)}$ into $v$, the carry propagates more than $w + \rho$ bits. The failure preimage set is
\[
\widetilde E_{\mathrm{carry}}(s,o,t,w)
  =\left\{v\in [0,2^n)\mid 
    v_{[t,t+w+\rho)}\ge2^{w+\rho}-s'
    \right\}
\]
with cardinality $s'2^{n-w-\rho}\le(2^w-1)2^{n-w-\rho}<2^{n-\rho}$.

If $t+w=n$, a carry that propagates more than $w$ bits means that the sum is at least $2^n$. This carry is captured in an additional qubit, and an approximate reduction modulo $p$ is applied by adding $c$ with carry propagation truncated to $\kappa$ bits. This failure case is captured by
\begin{align*}
\widetilde E_{\mathrm{wrap}}(s,o,t,w)
  =\{ & v<2^n\ \mid 
    v+2^ts'\ge2^n,\\
    & ((v+2^ts')\bmod2^n)_{[0,\kappa)}\ge2^\kappa-c
    \}.
\end{align*}
Translation by $2^ts_{[o,o+w)}$ modulo $2^n$ is a bijection. Dropping the carry condition therefore upper bounds the cardinality of this set by the usual truncated carry-propagation bound $c2^{n-\kappa}$.

The fourth failure case in $\AddSlice$ is a failure due to an approximate phase correction in the case $t+w=n$. The extra carry bit can be uncomputed by an $X$-measurement and phase correction. Let $h=\ind[v+2^ts_{[o,o+w)}\ge2^n]$ and
\[
  v^+=((v+2^ts_{[o,o+w)})\bmod2^n)+hc.
\]
If no failure occurs due to $\widetilde E_{\mathrm{wrap}}(s,o,t,w)$, the addition of $hc$ has no carry out of the low $\kappa$ bits, so $v^+<2^n$. When $m=1$ is measured, the exact carry correction contributes the phase $(-1)^q$. Together with a full-width call to $\PhaseGE$, it gives
\[
  (-1)^q(-1)^{q\ind[(v^+)_{[t,n)}\ge s_{[o,o+w)}]}
  =(-1)^{q\ind[(v^+)_{[t,n)}< s_{[o,o+w)}]},
\]
which is the required carry-dependent phase. The condition using $\ge$ implemented by $\PhaseGE$ is consistent with reconstructing the carry through the condition using $<$. The implementation supplies only the most-significant $\delta$ bits to $\PhaseGE$, and the resulting reconstruction can fail for two reasons. First, when $h=1$, the corresponding full-width comparison reconstructs $h$ exactly if and only if $v<p$; hence a non-canonical input $v\ge p$ causes a failure. Second, the comparison on the most-significant $\delta$ bits can disagree with the full-width comparison. The failure preimage sets for those events are
\[
\widetilde E_{\mathrm{ph,rep}}(s,o,t,w)
  =\left\{v\in [0,2^n)\ \middle|\ v\ge p,\quad h=1\right\},
\]
with cardinality at most $c$
and (with $s'=s_{[o,o+w)}$)
\begin{align*}
\widetilde E_{\mathrm{ph,trunc}}(s,o,t,w)
  =\{& v<2^n\ \mid
    \ind[(v^+)_{[t,n)}\ge s']\\
    & \ne
      \ind[(v^+)_{[n-\delta,n)}\ge s'_{[w-\delta,w)}]\},
\end{align*}
and $|\widetilde E_{\mathrm{ph,trunc}}(s,o,t,w)|\le2^{n-\delta}$.
The bound on the second cardinality follows because if the two comparisons disagree, the $\delta$ most significant bits of the values must be equal, fixing $\delta$ of the bits of $v^+$.

The operation $\AddBMultiple$ implements an addition or subtraction by a $B$-multiple of a register using a $B$-adic decomposition. For an integer $s$ with $0\le s<2^{n_s}$, where $n_s\in\{n,n+2\}$, write $s=s_L+Bs_H$ with $0\le s_L<B$. Then
\[
  Bs=Bs_L+B^2s_H\equiv Bs_L+cs_H\pmod p.
\]
The $Bs_L$ term is applied by one wrapping call to $\AddSlice$, and five local calls to compute $cs_H$ using the fixed list $T_c$, which is a list containing a signed binary representation of $c$. For $\bitcoincurve$, we have $T_c=((32,+1),(10,+1),(6,-1),(4,+1),(0,+1))$,
such that $c=\sum_{(d,\epsilon)\in T_c}\epsilon2^d=2^{32}+2^{10}-2^6+2^4+1$.
There are no further target-dependent failure sites outside these six $\AddSlice$ calls.

The operation $\AddCMultiple$ adds or subtracts a multiple of $c$ using the signed binary representation list $T_c$.
For $0\le s<2^n$, each element in $T_c$ with $d>0$ satisfies $2^ds=2^d s_{[0,n-d)}+2^n s_{[n-d,n)}$; if $d=0$, there is no decomposition. The decomposition into low and high parts at the boundary $n$ allows the addition of the high parts to be postponed and later corrected by combining them into $h_s=\sum_{(d,\epsilon)\in T_c}\epsilon s_{[n-d,n)}$. This contributes $2^nh_s\equiv ch_s\pmod p$. For $\bitcoincurve$, where $n=256$, this is
\[
  h_s=s_{[224,256)}+s_{[246,256)}-s_{[250,256)}+s_{[252,256)},
\]
and $0\le h_s\le c-3$. Hence $0\le ch_s<2^{65}$, so the high correction fits in the fixed 65-bit auxiliary register. The implementation constructs $ch_s$ by exact reversible signed additions modulo $2^{65}$. Intermediate values may get reduced modulo $2^{65}$, but the final residue is the ordinary integer $ch_s$. Reversing the same sequence therefore clears the auxiliary qubits.
$\AddCMultiple$ makes five calls to $\AddSlice$ for the direct shifted terms and one call for the high correction $ch_s$. The exact construction and uncomputation of $ch_s$ introduce no additional target-dependent failure sites, so all such sites occur in these six calls to $\AddSlice$. 

For $\bitcoincurve$, the square-subtract trace contains three calls to $\AddBMultiple$, each containing one call to $\AddSlice$ with three wrapping sites and five $\AddSlice$ calls with three occurrences of the local carry failure sites.
There is one call to the negative version of $\AddCMultiple$, which has five $\AddSlice$ calls that can wrap and one local call for the upper bits correction. The negative terms in the representation of $c$ contribute ten complement failure sites and five phase corrections. Overall, there are $16$ local
$\AddSlice$ calls, $9$ wrapping calls, and $13$ negative calls, giving $26$
complement sites. 

Intermediate results in the algorithm can lie in $[0,2^n)$ and thus be non-canonical representations of elements modulo $p$. However, under a failure-free execution on canonical inputs, the outputs lie again in $[0,p)$ and are thus canonical. This is clear if the control is $q=0$ as the operation is the identity. For $q=1$, the last operation of the algorithm that acts on the output register is a modular complement implemented by an addition of $c$ followed by a bitwise complement in the last line of the $\AddSlice$ algorithm, which is called with negative sign $\sigma=-1$ from $\AddBMultiple$. On a failure-free execution, the truncated addition of $c$ into the $v$ register right before the complement succeeds, which means that $v \bmod 2^\kappa < 2^\kappa - c$ because no carry has propagated out of the lowest $\kappa$ bits. Since $v<2^n$, we have $v+c < 2^n$, thus $v<p$ and the bitwise complement correctly computes $(2^n-1) - (v+c) = p - 1 - v \in [0,p)$. Therefore, any failure that would lead to a non-canonical output is already accounted for through previous failures.

Overall, under the helper conditions above, and assuming a uniformly
masked target and exact integer squaring operations, residue projection and a union bound
give bounds for the probabilities $p_{\mathrm{val}}$ of complement, wrapping and carry failures and $p_{\mathrm{phase}}$ of phase correction errors as 
\begin{align*}
  p_{\mathrm{val}}
  & \le \left(1+\frac{c}{p}\right)
      \left(\frac{16}{2^\rho}+\frac{35c}{2^\kappa}\right),\\
  p_{\mathrm{phase}}
  & \le \left(1+\frac{c}{p}\right)
      \left(\frac{9}{2^\delta}+\frac{9c}{2^{256}}\right).
\end{align*}
Consequently, the coherent failure probability $p_{\mathrm{sq}}$ for the square-subtract operation satisfies
\begin{align*}
  p_{\mathrm{sq}}
  \le \left(1+\frac{c}{p}\right)
      \left(\frac{16}{2^\rho}+\frac{35c}{2^\kappa}
            +\frac{9}{2^\delta}+\frac{9c}{2^{256}}\right).
\end{align*}
For $\rho=\delta=32$ and $\kappa=65$, this gives
$p_{\mathrm{sq}}<2^{-26.59}$.

\subsubsection{Approximate in-place multiplication and division}\label{sec:mc-sim-bound}

Another component needed in the point addition circuit is an in-place modular multiplication and its inverse, which functions as an in-place modular division, as described in Schrottenloher's recent work~\cite{schrottenloher2026optimized}. It is used to compute and uncompute the slope $\lambda$. While we have proved failure bounds for the modular squaring operation and the coordinate offset computations in the previous section, we obtain a bound on the failure probability for the gcd-based in-place multiplication by a Monte Carlo simulation instead. We sampled $1{,}000{,}000$ uniformly random pairs $(u,v) \in \F_p^*\times \F_p$, simulated our approximate circuits to compute $\ket{u}\ket{v}\mapsto \ket{u}\ket{(uv)\bmod p}$ and $\ket{u}\ket{v}\mapsto \ket{u}\ket{(u^{-1}v)\bmod p}$, and counted the number of inputs for which the circuit failed (wrong output, residual phase, deallocation of a nonzero ancilla, etc.). The simulation recorded $K=21$ failures for both operations.
The circuit parameters we used were mostly identical to the ones used by Schrottenloher~\cite{schrottenloher2026optimized}. We list the parameters in \cref{tab:binary-gcd-simulation-parameters}.

\begin{table}[ht]
    \centering
    \caption{The parameters used in the Monte Carlo simulations of the in-place multiplication circuit. Most of these parameters were described and chosen as in~\cite{schrottenloher2026optimized}.}
    \label{tab:binary-gcd-simulation-parameters}
    \begin{tabular}{@{}l@{\hspace{2em}}r@{}}
        \toprule
        Parameter & Value \\
        \midrule
        Iteration count & $1.413n+2.4\sqrt{n}$ \\
        Register size at iteration $i$ & $n-i/1.413+2.3\sqrt{n}$ \\
        Rounds with $x+y=p$ & $37$\\
        Comparison bits in gcd & $40+2.3\sqrt{n}$ \\
        Short modular corrections & $+32$ bits ($65$ total) \\
        Phase-fix comparisons & $32$ bits \\
        \bottomrule
    \end{tabular}
\end{table}

The simulation allows us to make the following statement using the Clopper-Pearson~\cite{clopper1934use} bound: With confidence level at least $1-2^{-129}$, the probability that the in-place multiplication procedure using the binary gcd algorithm with the above parameters fails under uniform random inputs $(u,v) \in \F_p^*\times \F_p$ is bounded by $0.000149304$. In turn, this allows us to conclude via union bound that with confidence level at least $1-2^{-128}$, both in-place multiplications succeed with probability at least $1-2\cdot 0.000149304$ for uniform random inputs from $\F_p^*\times \F_p$.

Since the actual inputs to the in-place multiplication in the circuit are differences of elliptic curve point coordinates, their distribution is not equal to the uniform distribution used in the simulation. This introduces a factor of approximately $4$ into the bound inequality, which we will show next.

The actual input to the first in-place modular division algorithm is the pair of differences $(x_1-x_2, y_1-y_2)$. Here $A_i = (x_1, y_1)$ is the current accumulator point after $i$ point addition steps, i.e., $A_i = (x_1, y_1) = A_0 + S_i$, where $S_i$ is the point accumulated into $A_0$ after $i$ steps. The point $(x_2, y_2)$ is the point looked up in the $i$-th window from the $i$-th lookup table $T_i$.  

Because $a_0\sample \Z_r^*$ is sampled uniformly at random and the map $A_0\mapsto A_i$ as a shift by the point $S_i$ is injective, and we only need to take into account first-time failures of trajectories that are correct up to this point, the point $A_i$ can be one of $r-1$ points, namely in the set $E(\F_p)\setminus \{\Ocal\} + S_i$. 

The point $(x_2,y_2)$ is of the form $[j2^{iw} + \mu_i]P$ or $[j2^{(i-L)w} + \mu_i]Q$. Since $\mu_i$ is sampled uniformly at random from $\Z_r\setminus \{-2^{iw}j\mid j \in \{0,1,\dots,2^w-1\}\}$ or $\Z_r\setminus \{-2^{(i-L)w}j\mid j \in \{0,1,\dots,2^w-1\}\}$, the point can be any point in $E(\F_p)$ except for the set of points of size $2^w$ given by the exceptions above.

The set of possible input pairs $((x_1,y_1), (x_2, y_2))$ therefore has size $(r-1)(r-2^w)$, and we need to
relate the distribution of input pairs $(x_1-x_2, y_1-y_2)$ to the uniform random distribution of pairs $(u,v) \in \F_p^*\times \F_p$ that was used in our Monte Carlo simulation.

The $u=0$ case is already covered by our exceptional point analysis, so we may assume $u\neq 0$.
Fix a specific pair $(u,v)$ with $u\neq 0$.
A given difference of curve point coordinates is equal to $(u,v)$ if and only if $x_1 = x_2 + u$ and $y_1 = y_2 + v$.

Using that both points satisfy the curve equation, we obtain
\begin{align*}
  y_2^2 + 2y_2v + v^2 & = y_1^2 = x_1^3 + b\\
                      & = x_2^3 + 3x_2^2u + 3x_2u^2 + u^3 + b, 
\end{align*} 
which, together with $y_2^2 = x_2^3 + b$, gives
\[
  2vy_2 + v^2 = 3x_2^2u + 3x_2u^2 + u^3.
\]
If $v\ne 0$, $y_2$ can be expressed by a term quadratic in $x_2$, and substituting this expression into the curve equation shows that $x_2$ satisfies a polynomial equation of degree $4$. Hence, there are at most $4$ possible values for $x_2$, yielding at most $4$ solutions for pairs $(x_2, y_2)$.

If $v=0$, then $y_1 = y_2$ and $x_2$ satisfies a quadratic polynomial equation, yielding at most $2$ solutions for $x_2$ and again at most $4$ pairs in total.

The same bound of at most $4$ possible curve point pairs yielding a given pair of differences applies to the second call of the binary-gcd-based circuit, the in-place modular multiplication $(u,v)\rightarrow (u,uv \pmod p)$. After an execution of an ideal in-place multiplication, the outputs are $x_2-x_3$ and $y_2 + y_3 = y_2- (-y_3)$. These are again coordinate differences, where $(x_3,y_3) = (x_1, y_1) + (x_2,y_2)$ are the coordinates of the (negative of the) accumulator point $A_{i+1}$ for the next iteration and the lookup point $P_i$, and so the factor-of-four bound holds at the point in the circuit.
Now, the mapping $(A_i, P_i)$ to $(P_i, -A_{i+1})$ is a bijection, as is the ideal, non-approximate in-place multiplication for non-zero $u$ ($u=0$ corresponds to an exceptional case already accounted for). As a result, there is a one-to-one correspondence between the multiplier inputs that fail and the corresponding (ideal) outputs, each of which corresponds to at most four curve point pairs.

Let $q_{\divi}$ and $q_{\mul}$ be the failure probabilities of the in-place division and multiplication circuits for pairs $(u,v)$ sampled uniformly at random. Hence there are at most $p(p-1)q_\divi$ and $p(p-1) q_\mul$ failure pairs in $\F_p^*\times\F_p$. Each can be reached by at most $4$ pairs of points under the difference map. The probability $p_{gcd,\divi}$ over uniform random samples of $a_0$ and $\mu_i$ that an input point pair leads to a failure in the division algorithm is at most
\[
  p_{gcd,\divi}\le \frac{4p(p-1)}{(r-1)(r-2^w)}q_\divi,
\]
and analogously for the in-place multiplication
\[
  p_{gcd,\mul}\le \frac{4p(p-1)}{(r-1)(r-2^w)}q_\mul.
\]
Overall, the probability that a failure occurs in one of the in-place binary gcd circuits over the random uniform choices of $a_0$ and $\mu_i$ is bounded by $ \approx 4(q_\divi + q_\mul)$:
\[
  p_{gcd,\divi} + p_{gcd,\mul} \le \frac{4p(p-1)}{(r-1)(r-2^w)}(q_\divi + q_\mul).
\]

\subsection{Failure probability bound}
Using the above failure probabilities for all the components in the computation of $\widetilde f_\omega(k,l)$, we obtain the overall failure bound $p_f$ as follows.

We start by bounding the failure probability $p_{\mathrm{ecadd}}$ for a single affine elliptic curve point addition using a union bound over all possible failure types for the separate components. We thus find that with probability $1-2^{-128}$,
\begin{align*}
  p_{\mathrm{ecadd}} & \leq p_{\sub} + p_{gcd,\divi} + p_{\mathrm{add}} + p_{\mathrm{sq}} + p_{gcd,\mul} + p_{\mathrm{neg}}\\
  & \leq 0.001195.
\end{align*}
Then, via union bound on the $28$ point additions, we have that
\[
  p_f \leq 28 p_{\mathrm{ecadd}} + \frac{6L}{r-1} \leq 0.03346.
\]

\subsection{Algorithm success probability with approximate arithmetic}
\label{proof:masked-oracle-success}

In this section, we derive a bound on the success probability of Shor's algorithm for solving the ECDLP when using approximate double scalar multiplication with failure bound $p_f$. Our bound takes into account wrong outputs as well as phase errors due to approximate phase-fix predicates.

Recall the failure set
\begin{align*}
B_\omega = \{(k,l)\in & [0,2^n)\times [0,2^n) \mid \\
& f_\omega(k,l) \neq \widetilde f_\omega(k,l) \vee s_\omega(k,l) = -1\}
\end{align*}
for the masked approximate implementation $\tilde f_\omega(k,l)$, i.e., the set of inputs $(k,l)$ for which the masked approximate implementation does not agree with the ideal masked value $f_\omega(k,l)$ or it agrees but the implementation produces a residual $(-1)$-phase.

For the analysis, we assume that we perform the known translation from $f_\omega(k,l)$ to $f(k,l)$ right after applying the oracle, i.e., we assume that the random point offsets of the masked implementation are removed. Because the projector onto successful outcomes of the algorithm (see the definition of $M$ below) only acts on the input register, this shift does not affect outcome probabilities, and thus does not have to be implemented in practice. However, this assumption simplifies the analysis significantly.
For brevity, we will write $x\in X:=[0,2^n)\times [0,2^n)$ instead of $(k,l)$ and $f(x)$ instead of $f(k,l)$.

Let $N=2^{2n}=|X|$ and the ideal state before the QFT-measurement
\[
    \ket{\Psi}=\frac1{\sqrt{N}}\sum_{x\in X}\ket{x}\ket{f(x)}.
\]
The corresponding approximate state is
\[
    \ket{\tilde\Psi_\omega}=\frac1{\sqrt{N}}\sum_{x\in X}s_\omega(x)\ket{x}\ket{\tilde f_\omega(x)},
\]
where $s_\omega(x)\in\{-1,1\}$ tracks the phase error for label $x$. We thus have $s_\omega(x)=1$ and $\tilde f_\omega(x)=f(x)$ for all $x\notin B_\omega$.

Let the projector onto the successful outcome set $S$ after the inverse QFT be $\sum_{y\in S}\ket{y}\bra{y}$. Before the inverse QFT, this corresponds to the projector
\[
    M=\left(\text{QFT}\sum_{y\in S}\ket{y}\bra{y}\text{QFT}^\dagger\right) \otimes \mathbbm{1},
\]
and we have that the success probabilities
\[
    P=\braket{\Psi|M|\Psi}\;\text{ and }\;\tilde P=\mathbb E_\omega[\braket{\tilde\Psi_\omega|M|\tilde\Psi_\omega}].
\]
We now write the difference between the ideal and the expected state $\mathbb{E}_\omega[\ket{\tilde\Psi_\omega}]$ as
\[
 \ket{\Delta}:=\mathbb{E}_\omega[\ket{\tilde\Psi_\omega}] - \ket{\Psi} = \frac{1}{\sqrt N}\sum_{x\in X}\ket{x}\ket{d_x},
\]
with $\ket{d_x}:=\mathbb E_\omega[1_{x\in B_\omega}(s_\omega(x)\ket{\tilde f_\omega(x)}-\ket{f(x)})]$.

Recall that $p_f$ is such that for all $x\in X$,
\[
    \Pr_\omega[x\in B_\omega]\leq p_f.
\]
Therefore, $\|\ket{d_x}\|\leq 2p_f$ from the triangle inequality for expectations and
\[
    \|\mathbb{E}_\omega[\ket{\tilde\Psi_\omega}] - \ket{\Psi}\|^2=\frac{1}{N}\sum_{x\in X}\|\ket{d_x}\|^2 \leq \frac{4}{N}\sum_{x\in X} p_f^2 = 4p_f^2.
\]
Using Jensen's inequality and the fact that $M$ is a projector independent of $\omega$, we have
\begin{align*}
    \tilde P & = \mathbb E_\omega[\|M\ket{\tilde\Psi_\omega}\|^2]\geq \left\|M\mathbb E_\omega[\ket{\tilde\Psi_\omega}]\right\|^2\\
    & \geq \left(\max\{0,\|M\ket{\Psi}\| - \|M \ket{\Delta}\|\}\right)^2 \\
    & \geq \left(\max\{0,\sqrt{P} - 2p_f\}\right)^2.
\end{align*}
Given a lower bound $P_0$ with $P_0\leq P$, we thus have
\[
  \tilde P \geq \left(\max\{0,\sqrt{P_0} - 2p_f\}\right)^2.
\]
\newpage

\clearpage
\onecolumngrid

\section{End-to-end integrated routing}
\label{sec:compiler-appendix}

\begin{figure}[!ht]
      \centering
      \includegraphics[
            width=\textwidth,
            trim=0 579bp 0 0,
            clip
      ]{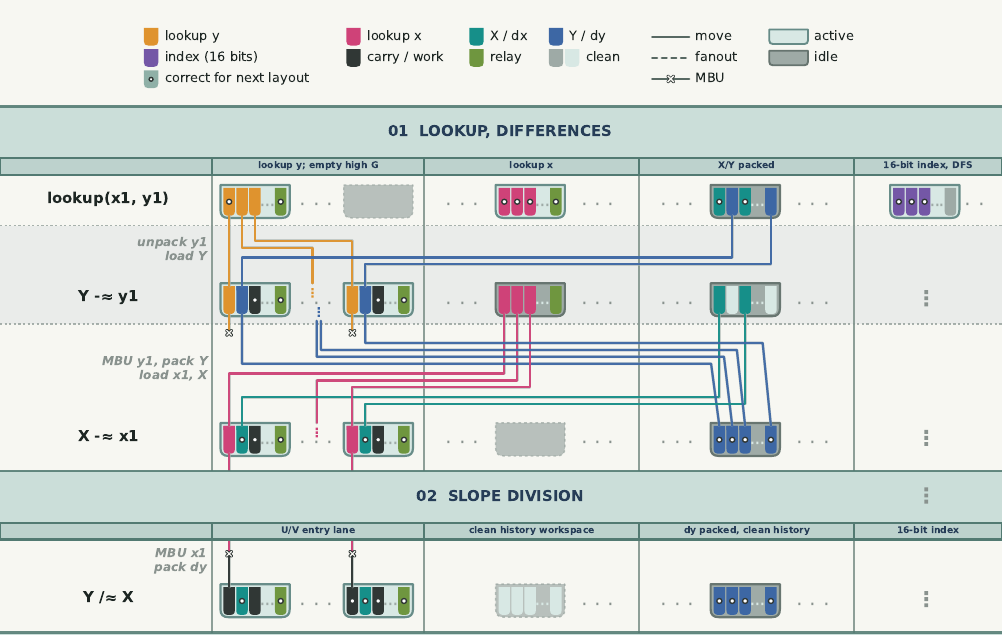}
      \caption{Routing of the point-addition algorithm, stages 1--2.}
      \label{fig:point-add-route}
\end{figure}

\clearpage

\begin{figure}[!ht]
      \ContinuedFloat
      \centering
      \includegraphics[
            width=\textwidth,
            trim=0 0 0 304bp,
            clip
      ]{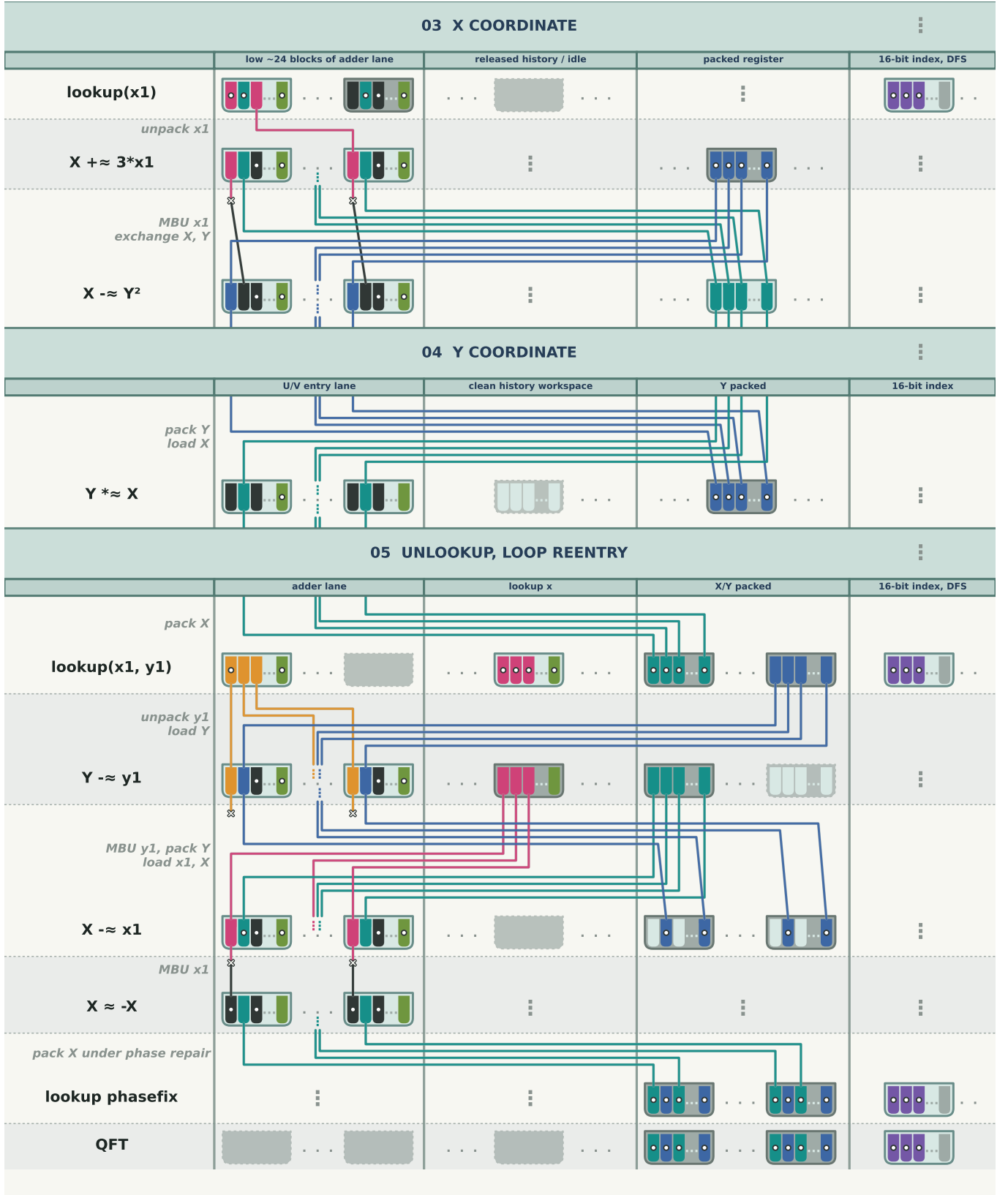}
      \caption{Routing of the point-addition algorithm, continued: stages 3--5.
      The windowed point add ends in the same configuration as it begins, allowing
      for the routing to loop cleanly.}
\end{figure}

\clearpage
\twocolumngrid

\section{Circuit synthesis}
\label{sec:compiler-circuit-synthesis}

Circuit synthesis turns the verified hierarchy of specs and defs into an
architecture-legal, executable schedule. The entry point takes a root spec and
its parameters, a target architecture, an optional layout plan, and synthesis
options; it returns the measurement schedule together with its metrics and a
serialized schedule artifact from which all downstream reporting is derived.

The architecture object declares a group of memory blocks and one or more
groups of factory blocks. The memory type fixes block size and shape, which
operations are directly legal within one block or between two blocks, and how
many concurrent operations and logical measurements a block supports. Each
factory group names a resource-state kind, a block count, the number of
states per block, and the preparation and cleanup recipes for that factory
type---for the ECC pipeline, \CCZ{} factories. Legality checks, placement, and
resource allocation all key off this object.

\subsection{Synthesis tree and frontier lowering}

Synthesis expands the root into a synthesis tree that records the decomposing
hierarchy, the def selected at each node, qubit and classical-bit bindings, and
classical conditions on components. The frontier lowering loop repeatedly selects a frontier
node, chooses a def if the node is not directly legal on the architecture,
allocates and binds the def's ancillae, and expands the node---until every leaf
is architecture-legal.
Routing is handled by placement and by local
measurement-based gadgets rather than a separate global \SWAP{} network: an
operation that is not legal in the current geometry is decomposed into
components such as logical measurements, controlled loads, \CNOT{} fanout, or
PcP gadgets (though we managed to entirely avoid naively decomposing PcPs in
this work).

Schedule-dependent resource binding is deferred during this refinement:
factory ancillae are virtual qubits in the tree,
so defs and legality checks remain coherent while the
physical factory assignment is chosen later, at scheduling time.
Run-local facts---qubit ownership, ordering dependencies---are
recorded in synthesis-side state and never leak into the reusable
specs and defs. Two further mechanisms keep synthesis tractable at
application scale: a compact representation of repeated structure, and
memoized replay of previously lowered subtrees.

\subsection{Compact repetition and deferred parameters}

Repeated algorithmic structure is kept as a compact repeated subcircuit rather
than expanded into a separate copy for every iteration. Its compile-time bounds
determine how many times the body executes, while an iteration variable may
select the operand qubit, register slice, constant, or implementation variant
used on each pass. Repetitions may be nested or traversed in reverse, and
inverting a repeated computation reverses both the iteration order and the
enclosed operations. This representation is used for repeated shifts,
multiplier digit sweeps, lookup traversal, modular-inversion rounds, and the
repeated work units above point addition.

The repeated body remains compact while the compiler selects implementations,
checks architectural legality, estimates costs, and records the synthesis
state. When different iterations genuinely require different circuit shapes,
the alternatives are selected by the iteration variable; changing register
widths are a common example.
Only when the final operation stream is emitted
does the compiler substitute concrete operands and iteration-local classical
results for every pass. This late expansion reduces compilation work and memory
even though the executable schedule ultimately contains every operation.

\subsection{Memoized lowering and subtree replay}

The second mechanism is a dynamic-programming-style replay cache.
After a subtree has been lowered completely, the compiler records its expansion
under a structural key containing the source spec type and normalized
parameters, the chosen def, and compatible ancilla bindings. These immutable
tuples are hashable, so a repeated semantic component can locate an earlier
lowering with a dictionary lookup rather than repeating recursive def
selection, child-support tests, allotment, and layout decisions. Separate
run-local caches remember preferred defs and architecture-support results.

On a replay hit, the compiler reproduces the cached child structure under the
new occurrence and remaps occurrence-local classical identities. The replayed
state includes the selected decomposition, ancilla bindings, sibling topology,
factory-loan annotations, and qubit-ownership classes; it is not merely a copy
of a flattened gate string. A candidate is rejected when its completed subtree
or resource context is not safe to reuse, in which case ordinary expansion is
performed.
Rewinding or editing the live tree invalidates the affected cached
state. These restrictions make hashing an optimization only: enabling or
disabling replay must not change the emitted program.

This reuse happens during hierarchical synthesis, before measurement
scheduling: the scheduler receives the bound serial stream and constructs
concrete layers for every occurrence, so the replay speedup is a lowering
speedup. Scheduler-side scalability comes instead from the packed layer
representation and the optional late resolution of classical conditions
described next.

\subsection{Measurement scheduling}

Beyond the general routine described in the main text, the scheduler applies a
resource-pool model to state injections.
Each \CCZ{} factory use is split into an injection
(three $ZZ$ logical measurements, each coupling one \CCZ{} state qubit to its memory target)
and factory cleanup measurements. The
scheduler tracks the number of available resource states per factory kind,
decrementing on injection and incrementing on cleanup, and when a cleanup is
delayed the propagating phase fix is tracked as a conditional Pauli byproduct.
Only after scheduling are virtual factory indices assigned to physical factory
blocks, pairing each injection with its matching cleanup. Internally, layers
are stored as block-local Pauli bitmasks and flat instruction buffers rather
than per-gate objects---the representation change that makes scheduling
application-scale circuits tractable.

\subsection{Layout plans}
\label{sec:compiler-layout-plans}

A layout plan is an authored, scoped placement policy for one spec
implementation. It can force the def choice for its spec (and, through nested
child plans, for the children) and binds semantic ports and ancillae to
placement rules: serial strips through memory blocks, interleaved cycles, or
direct block-and-slot coordinates. It can also reserve pools of helper qubits
for Pauli-controlled-Pauli (PcP) gadgets. The plan is where
design-level decisions enter: which registers are colocated, which resource
pools exist, and which defs are forced. Below the planned shape the compiler
still applies its local policies, but changing the plan changes the compiled
tree, the schedule pressure, and the resulting metrics.

\section{Architecture resource totals}
\label{sec:architecture-resource-totals}

The compiled-circuit resource counts used throughout this section are
summarized in \cref{tab:compiled-circuit-resource-counts}.

We note that the total number of Toffoli injections
is different from the reported estimate in \cref{part:point_addition}, where we used the accounting from previous work.

This is for three reasons: (1) as discussed in \cref{sec:compiler-components-routing}, we always execute conditional phase comparators
in our compiled circuits as opposed to only counting operations if they are executed (as in previous work~\cite{babbush2026securing,schrottenloher2026optimized});
(2) our phase-fix lookup is implemented suboptimally, i.e., we have not
implemented the phaseup operation from~\cite{gidney2025factor}; and (3) the counts in \cref{tab:compiled-circuit-resource-counts}
include the initial lookup replacing the first point addition.

\begin{table}[ht]
    \centering
    \caption{Compiled-circuit resource counts, assuming conditional operations are always executed.}
    \label{tab:compiled-circuit-resource-counts}
    \footnotesize
    \setlength{\tabcolsep}{4pt}
    \begin{tabular}{@{}lcr@{}}
        \toprule
        Resource & Variable & Total \\
        \midrule
        Measurement-layer depth
            & \(\measurementLayerDepthVariable\)
            & \(\measurementLayerDepthTotal\) \\
        Toffoli injections
            & \(\toffoliInjectionsVariable\)
            & \(\toffoliInjectionsTotal\) \\
        Non-Toffoli logical measurements
            & \(\nonToffoliLogicalMeasurementsVariable\)
            & \(\nonToffoliLogicalMeasurementsTotal\) \\
        Peak measurement parallelism
            & \(\peakMeasurementParallelismVariable\)
            & \(\peakMeasurementParallelismTotal\) \\
        \bottomrule
    \end{tabular}
\end{table}

\subsection{Footprint}
\label{subsec:footprint}

Assuming an independent per-ion loss rate of
\(p_{\mathrm{loss}}=10^{-7}\) per POC, a code block containing its \(n\) data
ions and \(n\) syndrome ancillas requires an average reloading rate
\begin{equation}
    R_{\mathrm{reload}}=(2n)p_{\mathrm{loss}}.
\end{equation}
The resulting per-block rates for the three code types used in the architecture
are given in \cref{tab:code-block-reloading-rates}.

\begin{table}[ht]
    \centering
    \caption{Per-block ion-reloading rates at
    \(p_{\mathrm{loss}}=10^{-7}\) per ion per POC.}
    \label{tab:code-block-reloading-rates}
    \begin{tabular}{@{}lrrc@{}}
        \toprule
        Code & \(n\) & Physical qubits \((2n)\) & Expected losses per POC \\
        \midrule
        Q102 & 102 & 204 & \(2.04\times10^{-5}\) \\
        Q66  &  66 & 132 & \(1.32\times10^{-5}\) \\
        C6   &   6 &  12 & \(1.20\times10^{-6}\) \\
        \bottomrule
    \end{tabular}
\end{table}

\paragraph*{Algorithm-wide exhaustion budget.}
For a reservoir of capacity \(r\), we call the period between consecutive
reservoir refills a \emph{refill interval}, and let \(X\) be the number of
replacement ions demanded during that interval.  We define \emph{exhaustion}
as the event \(X>r\), in which the replacement demand exceeds the available
stock before the next refill.  We assume that exhaustion of any local or global
reservoir counts as a failed run and choose the reservoir capacities to bound
the total run-wide exhaustion probability below the target
\begin{equation}
    \varepsilon_{\mathrm{res}}=10^{-3}.
    \label{eq:reservoir-exhaustion-budget}
\end{equation}
Meeting this target ensures that reservoir exhaustion contributes at most
\(0.1\) percentage points to the failure probability of an attempt.

For the runtime bound derived in \cref{subsec:runtime}, the device
operates for at most \(\measurementLayerDepthVariable=
\measurementLayerDepthTotal\) measurement layers, as listed in
\cref{tab:compiled-circuit-resource-counts}.  Thus,
\begin{align}
    T_{\mathrm{run}}^{(\mathrm{POC})}
    &=\measurementLayerDepthVariable\times T_{\mathrm{inj}} \notag\\
    &=\measurementLayerDepthTotal\times147.48 \notag\\
    &=1.10980\times10^{10}\ \mathrm{POCs}.
    \label{eq:runtime-pocs-reservoir}
\end{align}
The device contains \(69+4=73\) Q102 blocks, where the four additional blocks
are the logical CliNR blocks.  The four magic factories contain four Q66 blocks
and eight C6 blocks, with two dedicated C6 blocks per Eastinthillation factory.
Using these component allocations gives the following
total numbers of refill intervals across all reservoirs:
\begin{align}
    N_{\mathrm{Q102}}
      &=73T_{\mathrm{run}}^{(\mathrm{POC})}/27
        =3.00057\times10^{10}, \\
    N_{\mathrm{Q66}}
      &=4T_{\mathrm{run}}^{(\mathrm{POC})}/17
        =2.61129\times10^{9}, \\
    N_{\mathrm{C6}}
      &=8T_{\mathrm{run}}^{(\mathrm{POC})}/5.4
        =1.64415\times10^{10}, \\
    N_{\mathrm{global}}
      &=T_{\mathrm{run}}^{(\mathrm{POC})}/5000
        =2.21960\times10^{6}.
    \label{eq:reservoir-refill-interval-counts}
\end{align}
If \(q_i\) is the exhaustion probability in one refill interval for reservoir
class \(i\), the union bound gives
\begin{equation}
    p_{\mathrm{res}}
    \leq \sum_i N_i q_i.
    \label{eq:algorithm-reservoir-union-bound}
\end{equation}
We choose the smallest local and global reservoir capacities for which this
bound remains below \(\varepsilon_{\mathrm{res}}\).  The corresponding
sizing calculations follow.

To size each local reservoir, we assume that it can be refilled between SECs
and model the number of losses in one SEC as a Poisson random variable with
mean
\begin{equation}
    \lambda=(2n)p_{\mathrm{loss}}T_{\mathrm{SEC}}.
\end{equation}
We choose the smallest number of ions \(r\) in each code block's local
reservoir that is consistent with the run-wide exhaustion budget, using
\(q_i(r)=\Pr[\operatorname{Poisson}(\lambda)>r]\).  The resulting values
are shown in \cref{tab:local-reservoir-sizes}.  With one
    fewer ion, the Q102, Q66, and C6 contributions to the run-wide union bound
    would be \(0.835\), \(0.00492\), and \(0.345\), respectively---all above the
\(10^{-3}\) exhaustion budget.

\begin{table}[ht]
    \centering
    \caption{Local ion-reservoir sizes, assuming only one refill per SEC.}
    \label{tab:local-reservoir-sizes}
    \footnotesize
    \setlength{\tabcolsep}{4pt}
    \begin{tabular}{@{}lcccc@{}}
        \toprule
        Code & \(T_{\mathrm{SEC}}\) (POCs) & \(\lambda\) & \(r\) & \(q_i(r)\) \\
        \midrule
        Q102 & 27.0 & \(5.508\times10^{-4}\) & 3 & \(3.83\times10^{-15}\) \\
        Q66  & 17.0 & \(2.244\times10^{-4}\) & 3 & \(1.06\times10^{-16}\) \\
        C6   &  5.4 & \(6.480\times10^{-6}\) & 2 & \(4.53\times10^{-17}\) \\
        \bottomrule
    \end{tabular}
\end{table}

Unlike the loss model of~\cite{tripier2026walkingcat}, the swap-loss model
produces no cascading losses, so the global reservoir only needs to buffer the
independent steady-state loss stream from the device components, including the
cat-state bundles and Bell-state resources.  Based on the loading-rate comparison below,
we assume that the reservoir is replenished once per second, or every
\(T_{\mathrm{refill}}=5000\) POCs.  For a candidate global
capacity \(R_{\mathrm{global}}\), the mean number of losses in this interval is
\begin{equation}
    \lambda_{\mathrm{global}}
    =\bigl(N_{\mathrm{components}}+R_{\mathrm{global}}\bigr)
     p_{\mathrm{loss}}T_{\mathrm{refill}}.
\end{equation}
The peak component allocations contain \(19{,}363\) qubits before the global
reservoir is added.  Applying our run-wide reservoir-exhaustion budget
gives a 34-ion global reservoir, as summarized in
\cref{tab:global-reservoir-size}.  Its exhaustion probability per refill
interval is \(2.77\times10^{-10}\).  With a 33-ion buffer, the run-wide union
bound would instead be \(2.36\times10^{-3}\), so 34 is the smallest capacity
consistent with \cref{eq:reservoir-exhaustion-budget}.
The runtime conversion in \cref{subsec:runtime} gives
\(\tau_{\mathrm{POC}}=200\,\mu\mathrm{s}\).  The mean reloading rate needed
to maintain the global reservoir is therefore
\begin{equation}
    R_{\mathrm{load}}^{(\mathrm{min})}
    =\frac{1.9397\times10^{-3}}{\tau_{\mathrm{POC}}}
    \approx9.70\ \text{ions/s}.
\end{equation}
This rate is modest compared with demonstrated loading capabilities across the
Walking Cat architecture's targeted qubit modalities.  A trapped-ion
experiment loaded about ten barium ions per ablation shot at its highest tested
power.  At the required mean rate, this yield corresponds to about \(0.970\)
ablation shots per second~\cite{wall2026barium}.  In a neutral-atom system,
continuous reloading has produced \(30{,}000\) initialized qubits per
second~\cite{chiu2025continuous}.

\begin{table}[ht]
    \centering
    \caption{Global reservoir sizing for independent steady-state ion loss.
    Here \(\rho_{\mathrm{loss}}\) is the expected number of losses per POC and
    \(q_{\mathrm{global}}\) is the exhaustion probability per refill interval.}
    \label{tab:global-reservoir-size}
    \footnotesize
    \setlength{\tabcolsep}{2.5pt}
    \begin{tabular}{@{}lcccc@{}}
        \toprule
        Device & Total qubits & \(\rho_{\mathrm{loss}}\)
                      & \(R_{\mathrm{global}}\) & \(q_{\mathrm{global}}\) \\
        \midrule
        \(\bitcoincurve\) & \(19{,}397\) & \(1.9397\times10^{-3}\) & 34
                & \(2.77\times10^{-10}\) \\
        \bottomrule
    \end{tabular}
\end{table}

Substitution into \cref{eq:algorithm-reservoir-union-bound} gives
\begin{equation}
    p_{\mathrm{res}}
    \leq7.32\times10^{-4}.
    \label{eq:algorithm-reservoir-exhaustion-bound}
\end{equation}
Therefore, even if every exhaustion event aborts the computation, the
probability of completing a run without reservoir exhaustion is at least
\(99.927\%\).

\paragraph*{Memory and logical CliNR blocks.}
Both the memory blocks and the logical CliNR blocks use Q102.  Each block
contains \(n=102\) data ions, the same number of syndrome ancillas, and the
three-ion local reservoir from \cref{tab:local-reservoir-sizes}, giving the
per-block footprint
\begin{equation}
    n_{\mathrm{phys}}^{\mathrm{Q102}}
    =2n+r_{\mathrm{Q102}}
    =2(102)+3
    =207.
\end{equation}
So, the 69 memory blocks contribute \(69\times207=14{,}283\) physical
qubits.  The four logical CliNR blocks contribute a further
\(4\times207=828\) physical qubits, so the Q102 blocks occupy \(15{,}111\)
physical qubits in total.

\paragraph*{Cat-state register counts.}
We first determine the number of cat-state verification rounds used by both
code types.  For a target verification error \(\varepsilon=10^{-11}\) and
physical error rate \(p=10^{-4}\), the Walking Cat architecture's
verification-round formula~\cite{tripier2026walkingcat} gives
\begin{equation}
\begin{aligned}
    m
    &=\left\lceil\frac{\log\varepsilon}{2\log(2p)}\right\rceil
      =\lceil1.49\rceil=2, \\
    (2p)^2&=4\times10^{-8}>\varepsilon,
    \qquad
    (2p)^4=1.6\times10^{-15}<\varepsilon.
\end{aligned}
\end{equation}
Thus, two verification rounds are sufficient and minimal under this bound. We
use the Walking Cat transport timing \(\eta=1/20\) below.

For Q102, \(n=102\), \(T_{\mathrm{SEC}}=27\) POCs from
\cref{tab:architecture-qec-codes}, and \(\wcat=30\).  One long edge accommodates
\(s(102,30)=\lfloor102/30\rfloor=3\) factories.  The three parallel logical
measurements therefore use three adjacent factories along the top edge,
occupying 90 sites; no factories are required along the bottom or either short
edge.  Every factory has circulation distance
\(P_{\mathrm{loop}}(102)=212\) transport steps, and
\begin{equation}
\begin{aligned}
    T_{\mathrm{prep}}(30,2)
        &=14.8\ \mathrm{POCs}, \\
    T_{\mathrm{prep}}(30,2)+\eta P_{\mathrm{loop}}(102)
        &=25.4\ \mathrm{POCs}.
\end{aligned}
\end{equation}
The total time is below the 27-POC SEC, so
\cref{prop:cat-state-rotation} gives \(r=1\) for all three measurements.

For Q66, \(n=66\), \(T_{\mathrm{SEC}}=17\) POCs, and \(\wcat=16\).  One long
edge accommodates \(s(66,16)=\lfloor66/16\rfloor=4\) factories, so the three
parallel logical measurements use three adjacent factories along one long
edge, occupying 48 sites.  Every factory has circulation distance
\(P_{\mathrm{loop}}(66)=140\) transport steps, and
\begin{equation}
\begin{aligned}
    T_{\mathrm{prep}}(16,2)
        &=13.45\ \mathrm{POCs}, \\
    T_{\mathrm{prep}}(16,2)+\eta P_{\mathrm{loop}}(66)
        &=20.45\ \mathrm{POCs}.
\end{aligned}
\end{equation}
The total time exceeds one 17-POC SEC but fits within two, so
\cref{prop:cat-state-rotation} gives \(r=2\) registers for each measurement,
or six weight-16 cat-state registers in total.

\paragraph*{Cat-state bundle pairs and Bell-state resources.}
Let \(k_{\mathrm{cat}}\) denote the maximum number of logical measurements that
the device must support in parallel.  The reuse schedule of
\cref{sec:reusing-cat-states} assigns two mobile cat-state bundles to each
measurement stream.  For Q102, each bundle contains one weight-\(30\) cat-state
register and one weight-\(30\) verification bank; the weight \(30\) is set by
the widest logical measurement made in a Q102 block.  Together, the two bundles
allocated to one measurement stream form a cat-state bundle pair.  Each
bundle occupies \(2\times30=60\) physical qubits, so each bundle pair occupies
\(120\) physical qubits.  The resulting footprint is therefore
\begin{equation}
    n_{\mathrm{phys}}^{\mathrm{cat}}(k_{\mathrm{cat}})
    =2k_{\mathrm{cat}}(2\times30)
    =120k_{\mathrm{cat}}.
\end{equation}
The peak compiled-circuit requirement is
\(\peakMeasurementParallelismVariable=
\peakMeasurementParallelismTotal\) simultaneous logical measurements, as
listed in \cref{tab:compiled-circuit-resource-counts}.  The cat-state allocation
is therefore
\begin{equation}
    n_{\mathrm{phys}}^{\mathrm{cat}}
    =120\peakMeasurementParallelismVariable
    =120(\peakMeasurementParallelismTotal)=2{,}880.
\end{equation}

In each cat-state production round, every LM2 measurement involving
memory blocks consumes
\(k_{\mathrm{Bell}}\) Bell pairs to stitch the two cat states used in that
round, where \(k_{\mathrm{Bell}}\) is the number of pairs required to verify the
stitched cat state with undetected-error probability at most \(\varepsilon\).
Following the Walking Cat stitching
protocol~\cite{tripier2026walkingcat}, each Bell pair is prepared
with a small unitary circuit.
One qubit from each pair is routed to each of the two target cat-state
factories.  At the two factories,
we measure a joint \(ZZ\)-stabilizer parity check between each Bell-pair qubit and one qubit of the local cat state.
This procedure stitches the two cat states into a larger,
joint cat state.  Each Bell pair supplies one stabilizer measurement of the joint cat state, so the
\(k_{\mathrm{Bell}}\) pairs provide \(k_{\mathrm{Bell}}\) redundant stitch
checks.  Writing \(p'\leq4p\) for the error rate of one stitch check,
including Bell-pair and transport faults, an incorrect parity can induce
a high-weight \(X\) error on one half of the stitched cat state, but the stitch
is rejected unless all \(k_{\mathrm{Bell}}\) stitch checks return the same wrong
value.
At \(p=10^{-4}\), choosing \(k_{\mathrm{Bell}}=4\) gives
\((p')^{k_{\mathrm{Bell}}}\leq(4p)^4=2.56\times10^{-14}\), below the target
precision \(\varepsilon=10^{-11}\). We provision one Bell-state
bundle for each of the 12 simultaneous LM2 measurements at peak demand.  Each
bundle therefore contains \(2k_{\mathrm{Bell}}=8\) physical qubits, including
while those qubits are in flight.  The four independent pairs can be prepared in
parallel in two POCs, well within the 13.45--14.8-POC cat-state production
interval; we assume that pre-positioning the bundles similarly hides their
transport latency, so no additional Bell pairs are required for pipelining.
Therefore, our bell-state allocation is:
\begin{equation}
    n_{\mathrm{phys}}^{\mathrm{Bell,LM2}}
      =12\times k_{\mathrm{Bell}}\times2
      =12\times4\times2=96.
\end{equation}

\paragraph*{Magic factories.}
The throughput analysis in \cref{sec:ccz-factory} specifies a total of four
Eastinthillation factories.  Each factory contains one Q66 block, two dedicated
C6 blocks, a local reservoir provisioned with the
three Q66 replacement ions and two replacement ions for each C6 block from
\cref{tab:local-reservoir-sizes}, and three weight-16 cat-state factories.
Including syndrome ancillas, each Q66 block
and its two C6 blocks occupy \(2(66)\) and
\(2\times2(6)\) physical qubits, respectively, while each
cat-state factory contains two weight-16 registers and one bank of 16
verification ancillas, occupying \(3\times16\) physical qubits.  Finally, we
count one additional six-qubit C6 Bell-state source for each C6 block, as
required by the measurement construction in \cref{sec:ccz-factory}.  Thus, the
footprint of one magic factory is
\begin{equation}
\begin{aligned}
    n_{\mathrm{phys}}^{\mathrm{magic}}
    &=2(66)+2\times2(6)+(3+2\times2) \\
    &\quad{}+3(3\times16)+2(6) \\
    &=319,
\end{aligned}
\end{equation}
and the four factories occupy
\begin{equation}
    N_{\mathrm{phys}}^{\mathrm{magic}}
    =4n_{\mathrm{phys}}^{\mathrm{magic}}
    =4\times319
    =1{,}276
\end{equation}
physical qubits in total.

Combining the peak allocation of each component type with the global reservoir
gives the device footprint
\begin{align}
    N_{\mathrm{phys}}
    &=69(207)+24(120)+12(8) \notag\\
    &\quad{}+4(207)+4(319)+34=19{,}397.
\end{align}

\subsection{Runtime}
\label{subsec:runtime}

The compiled algorithm has measurement-layer depth
\(\measurementLayerDepthVariable=\measurementLayerDepthTotal\), as listed in
\cref{tab:compiled-circuit-resource-counts}. To obtain an upper bound on the
runtime, we assign every measurement layer the mean duration of a \CCZ{}
injection. Among the logical measurements used in the compiled elliptic-curve
point-addition circuits, a \CCZ{} injection has the longest mean duration. It
also determines the critical path whenever it appears in a parallel layer,
because its outcome is needed to resolve subsequent Clifford corrections on
the data blocks. The three independent Viterbi processes in a \CCZ{} injection
have a joint average stopping time of \(5.46\) Q102 SECs, as derived in
\cref{eq:parallel-ccz-viterbi-duration}. Since one Q102 SEC takes 27 POCs, this
gives \(T_{\mathrm{inj}}=147.48\) POCs, or approximately \(0.02950\) seconds, when
evaluated using the unrounded duration. The
resulting expected-runtime upper bound per attempt is therefore
\begin{align}
    T_{\mathrm{run}}
    &\lesssim \measurementLayerDepthVariable\times0.0294958~\mathrm{s} \notag\\
    &=\measurementLayerDepthTotal\times0.0294958~\mathrm{s} \notag\\
    &=2{,}219{,}600~\mathrm{s}
      \approx25.7~\mathrm{days}.
\end{align}

\subsection{Logical failure probability}
\label{subsec:logical-failure-probability}

We include five contributions to the algorithm-wide logical failure
probability: memory errors in the Q102 blocks, errors in logical measurements,
errors in the \CCZ{} states consumed by Toffoli gates, errors in the \(T\) states
consumed by the iQFT, and reservoir exhaustion.

\paragraph*{Memory errors.}
As an upper bound, we assume that all
\(\measurementLayerDepthVariable=\measurementLayerDepthTotal\) measurement
layers occur on all \(69\) Q102 memory blocks in the device.  But, in actuality,
only a small subset of qubits needs to
hold quantum information continuously for this whole time in the application, so
our upper bound is quite conservative.  We find
\begin{align}
    N_{\mathrm{SEC}}
    &=\measurementLayerDepthVariable\times\tau_{\CCZ}^{\mathrm{avg}}\times69 \notag\\
    &=\measurementLayerDepthTotal\times\tau_{\CCZ}^{\mathrm{avg}}\times69 \notag\\
    &=2.83615\times10^{10}.
\end{align}
At the Q102 logical error rate
\(p_{\mathrm{L,SEC}}^{\mathrm{Q102}}=9.34\times10^{-12}\) per SEC from
\cref{tab:architecture-qec-codes}, this corresponds to a memory-failure
probability of
\begin{equation}
    p_{\mathrm{mem}}
    =1-\left(1-p_{\mathrm{L,SEC}}^{\mathrm{Q102}}\right)^{N_{\mathrm{SEC}}}
    =0.2327.
\end{equation}
The four Q102 logical CliNR blocks hold encoded states only for the duration
of frame clearing, rather than for the full algorithm.  Their exposure is
therefore included in the non-Toffoli logical-measurement count below, not in
the long-lived memory-failure total.

Although the calculated memory-failure probability is relatively high, our
estimate is substantially inflated by the assumption that all 69 Q102 blocks
continuously hold quantum information throughout the algorithm.  The
probability can also be readily reduced by increasing the code distance or
tuning the decoder, as discussed for the memory-error bottleneck in the magic
factory.

\paragraph*{Logical measurement errors.}
For a weight-\(30\) Viterbi measurement at physical error rate \(p=10^{-4}\),
the single-round bit-flip model gives
\(p_{\mathrm{flip},30}=2.1\times30\times10^{-4}=6.3\times10^{-3}\).
A vote margin of five therefore has logical error probability
\begin{equation}
    p_{\mathrm{Vit},30}
    =\left[
       1+\left(\frac{1-p_{\mathrm{flip},30}}
                     {p_{\mathrm{flip},30}}\right)^5
      \right]^{-1}
    =1.02\times10^{-11}.
\end{equation}
For \(\nonToffoliLogicalMeasurementsVariable=
\nonToffoliLogicalMeasurementsTotal\) such non-Toffoli logical measurements,
as listed in \cref{tab:compiled-circuit-resource-counts}, their combined
failure probability is
\(1-(1-p_{\mathrm{Vit},30})^{\nonToffoliLogicalMeasurementsVariable}
=0.00378\).

\paragraph*{Eastinthillation factory errors.}
The Eastinthillation factory has logical error rate
\(p_{\mathrm{L}}^{\CCZ}\lesssim9.36\times10^{-10}\) per accepted state
from \cref{eq:eastin-factory-logical-error-rate}.  Using one accepted state for
each of \(\toffoliInjectionsVariable=\toffoliInjectionsTotal\) Toffoli
injections from \cref{tab:compiled-circuit-resource-counts}, the cumulative
failure probability of the Eastinthillation factories is
\begin{equation}
    p_{\CCZ}
    =1-\left(1-p_{\mathrm{L}}^{\CCZ}\right)^{\toffoliInjectionsVariable}
    \lesssim0.0371.
\end{equation}

\paragraph*{iQFT-stage errors.}
The \(5{,}684\) accepted \(T\)-state pairs and \(464\) synthesized rotations
used for the iQFT contribute
\begin{equation}
    p_{\mathrm{iQFT}}\leq7.64\times10^{-5},
\end{equation}
as calculated in \cref{eq:iqft-t-state-failure}.

\paragraph*{Reservoir exhaustion.} The device's reservoir-exhaustion
probability is at most \(7.32\times10^{-4}\).

Assuming these five failure mechanisms are statistically independent, their
contributions combine to give
\begin{align}
    p_{\mathrm{fail,alg}}
    &\leq 1-(1-0.2327)(1-0.00378) \notag\\
    &\quad{}\times(1-0.0371)(1-7.64\times10^{-5}) \notag\\
    &\quad{}\times(1-7.32\times10^{-4}) \notag\\
    &=0.2646\approx26\%.
\end{align}

\section{Swap-loss sensitivity}
\label{sec:swap-loss-sensitivity}

We choose \(\pswap=0.5\) as a conservative order-one swap
probability. The precise choice has little effect on the reported logical error
rate: with the swap-resilient LDU, the simulations show no discernible dependence
on \(\pswap\) over the tested range \([0,0.75]\), within statistical
uncertainty, as shown in \cref{fig:q102-swap-loss-sensitivity}.

\begin{figure}[ht]
    \centering
    \includegraphics[width=0.70\linewidth]{%
        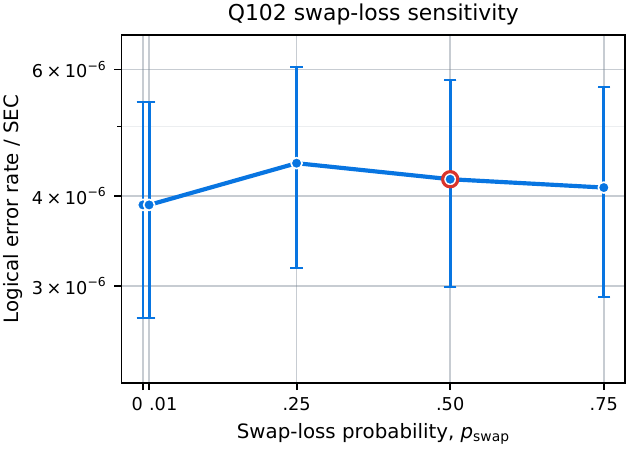}
    \caption{Swap-loss sensitivity of the \([[102,22,9]]\) code. The logical
    error rate per SEC is shown as a function of \(\pswap\) at
    \(p=10^{-3}\), \(\pleak=10^{-4}\), and
    \(\ploss=10^{-6}\). Each point uses \(10^6\) shots of nine
    SECs, and the error bars are exact 95\% binomial confidence intervals
    converted to per-SEC rates. The red ring marks the default simulation setting
    \(\pswap=0.5\).}
    \label{fig:q102-swap-loss-sensitivity}
\end{figure}

\clearpage
\onecolumngrid
\section{Q66 representatives for the Eastinthillation factory}
\label{app:gb66-representative-certificates}

This appendix gives physical Pauli representatives on the Q66 factory block
for every measurement layer in
\cref{fig:eastin-synthillation-factory-measurements}.

\begingroup

\scriptsize

\setlength{\tabcolsep}{3pt}

\tablecaption{Q66 representatives for the RY0--RY3 layers.\label{tab:gb66-factory-y-certificates}}
\tablehead{%
    \toprule
    Class & Target & Logical & \(w_{66}\) & Physical Pauli on qubits \(0,\ldots,65\) \\
    \midrule}
\tabletail{%
    \midrule
    \multicolumn{5}{r}{Continued on next page} \\
    \bottomrule}
\tablelasttail{\bottomrule}
\begin{supertabular}{@{}rllr@{\hspace{0.7em}}l@{}}
    0 & \texttt{L3/Y0} & \texttt{YIII} & 16 & \texttt{IYIYYIIZIIIIIIIYIIIIIIIXXIIIIIIIIIZYIIIIIYIIIIZZXZIIIIIIIZIIIIIIYI} \\
     & \texttt{L2/Y2} & \texttt{IIYI} & 15 & \texttt{XIIIIYIIIIXXIIIIIIIIZYZIIIIIXIIIIXIIZZYIIIIIIYIIIIIYIIIIIIIIIIIZII} \\
    \addlinespace[2pt]
\end{supertabular}

\tablecaption{Q66 representatives for the Acceptance layer.\label{tab:gb66-acceptance-certificates}}
\tablehead{%
    \toprule
    Class & Target & Logical & \(w_{66}\) & Physical Pauli on qubits \(0,\ldots,65\) \\
    \midrule}
\tabletail{%
    \midrule
    \multicolumn{5}{r}{Continued on next page} \\
    \bottomrule}
\tablelasttail{\bottomrule}
\begin{supertabular}{@{}rllr@{\hspace{0.7em}}l@{}}
    0 & \texttt{acceptance} & \texttt{ZIII} & 11 & \texttt{IIIIIIIIIIIIIIIIIIIZIIIIZIIIIIZIIIIIZIIIIIIZZIIIIIIIZIIZIZIIIZIZII} \\
    \addlinespace[2pt]
    1 & \texttt{acceptance} & \texttt{XIXI} & 10 & \texttt{IIIIIIIIIIIIIIIIIIXIIXXIIXXIIIXIIIIIIIIIIIIIIIIIIIIXIIIIXXIIIXIIII} \\
    \addlinespace[2pt]
    2 & \texttt{acceptance} & \texttt{YIYI} & 14 & \texttt{IIIIIIIIIIIXIIIIYZIIIIIIZIIYYIIZIIIIIIIIIIIIXIIYYIIXIIIIIIXYIIIIZI} \\
    \addlinespace[2pt]
    3 & \texttt{acceptance} & \texttt{IIZI} & 10 & \texttt{IIIIIIIIIIIIIIIIIIIIIZIIIIIIIIIZIIIIIIIZIIIIIIZZIIIIZZIZIIIIIIZZII} \\
    \addlinespace[2pt]
    4 & \texttt{acceptance} & \texttt{ZZII} & 11 & \texttt{IIIIIIIIIIIIIIIIIIZIZIIIIZIIIZIIZIIZIIIIIIZZIIIIIIIIIIIIIIIZIZIIIZ} \\
    \addlinespace[2pt]
    5 & \texttt{acceptance} & \texttt{XZXI} & 15 & \texttt{IXIIIIIIZYIIZIYIIIIIIIIIIIIIIIIYIIIXYYIIYIIIIXZIIIIIIIIZIIIIIIIZXI} \\
    \addlinespace[2pt]
    6 & \texttt{acceptance} & \texttt{YZYI} & 16 & \texttt{IIIIIIIIXIIIIXIIIIIIZIIIYYIIIIIIZIIIIZIIIXIIYYIIXIIIIZIYYIZIIIIIIZ} \\
    \addlinespace[2pt]
    7 & \texttt{acceptance} & \texttt{IZZI} & 11 & \texttt{IIIIIIIIIIIIIIIIIIIIIZIIIIZIIIIIZIIIIIZIIIIIIZZIIIIIIIZIIZIZIIIZIZ} \\
    \addlinespace[2pt]
    8 & \texttt{acceptance} & \texttt{ZIIZ} & 10 & \texttt{IIIIIIIIIIIIZIIIIIZIIIIIIZIIIIZIIIIZIIIIIIZIIIIIIIIIZZIIIIZIIIIZII} \\
    \addlinespace[2pt]
    9 & \texttt{acceptance} & \texttt{XIXZ} & 15 & \texttt{IIIIXIIIIIIIZYYXIIIIIIIYZIIIIIIIYIIXXZIXIIIZIIIIIIIIIIYIIIIIIIIIIY} \\
    \addlinespace[2pt]
    10 & \texttt{acceptance} & \texttt{YIYZ} & 16 & \texttt{IYIZIIIIZXIIIIYIIIIIIIIIZIIIIIIYIIZXYXIIXZIIIYZIIIIIIIIIIIIIIIIIYI} \\
    \addlinespace[2pt]
    11 & \texttt{acceptance} & \texttt{IIZZ} & 10 & \texttt{IIIIIIIIIIIIIIIIIIIIZIIIIIIIIIZIIIIIIIZIIIIIIZZIIIIZZIZIIIIIIZZIII} \\
    \addlinespace[2pt]
    12 & \texttt{acceptance} & \texttt{ZZIZ} & 10 & \texttt{IIIIIIIIIIIIIZIIIIIZIIIIIIZIIIIZIIIIZIIIIIIZIIIIIIIIIZZIIIIZIIIIZI} \\
    \addlinespace[2pt]
    13 & \texttt{acceptance} & \texttt{XZXZ} & 14 & \texttt{IIIIIIIIIIIIIIIIIIXIYIIIIIYIZZZIYIIIIIIIIIIIIIIIIIIYIXXXIYIIIIIYIZ} \\
    \addlinespace[2pt]
    14 & \texttt{acceptance} & \texttt{YZYZ} & 14 & \texttt{IIIIIIIIIIIIIIIIIIIIIIIIXZYIYYIYXIIIIIIIIIIIIIIIIIIIIIIIZYIYYIYXZI} \\
    \addlinespace[2pt]
    15 & \texttt{acceptance} & \texttt{IZZZ} & 11 & \texttt{IIIIIIIIIIIIIIIIZIZIIIIZIIIZIIZIIZIIIIIIZZIIIIIIIIIIIIIIIZIZIIIZII} \\
    \addlinespace[2pt]
\end{supertabular}

\tablecaption{Q66 representatives for the \CCZ{} injection layer.\label{tab:gb66-injection-certificates}}
\tablehead{%
    \toprule
    Class & Target & Logical & \(w_{66}\) & Physical Pauli on qubits \(0,\ldots,65\) \\
    \midrule}
\tabletail{%
    \midrule
    \multicolumn{5}{r}{Continued on next page} \\
    \bottomrule}
\tablelasttail{\bottomrule}
\begin{supertabular}{@{}rllr@{\hspace{0.7em}}l@{}}
    0 & \texttt{MZ0} & \texttt{IIIZ} & 15 & \texttt{IIZZIIZIIIIIIIIIIZIIIIIZZIZZIIZIIIZIIIIZZIIIIIIZIIIIIIIZIIIIIIIIZI} \\
     & \texttt{MZ1} & \texttt{IZII} & 11 & \texttt{IZIIZIIZIIZIIZIIZIIZIIZIIZIIZIIZIIIIIIIIIIIIIIIIIIIIIIIIIIIIIIIIII} \\
     & \texttt{MZ2} & \texttt{ZIII} & 12 & \texttt{ZIIIIIIIZIIZIIIIIIIIZIIIIIIIIIIIIIIIIIZIIZZIIZIIIIZIIZIIIZIIIIIIIZ} \\
    \addlinespace[2pt]
    1 & \texttt{MZ0} & \texttt{IIIZ} & 10 & \texttt{IIIIZIZZIIIIIIIZIZZZIIIIIIIIIIZIIIIIIIIIIIZIIIIIIIIIIZIIIIIIIIIIII} \\
     & \texttt{MZ1} & \texttt{IZII} & 10 & \texttt{IIIIIIIIIIIIIIZIIIIIZIIIIIIZIIIIZIIIIZIIIIIIZIIIIIIIIIZZIIIIZIIIIZ} \\
     & \texttt{MZ2} & \texttt{XIXI} & 10 & \texttt{IIIIIIIIIIXIIXIIIIIIIXIIIIXIIIIIIIIIIIIIIIIXIIIXXXIIXIIIIXIIIIIIII} \\
    \addlinespace[2pt]
    2 & \texttt{MZ0} & \texttt{IIIZ} & 10 & \texttt{IIIIIIIIIIIIIIZIIIIIIZIIIIIZIIIZIIIIIIZZIIIIIIIZIIIIIIIIIIZZIZIIII} \\
     & \texttt{MZ1} & \texttt{IZII} & 15 & \texttt{IIIIIZIIZIIIZIIIIZZIZIZIZZIIZIZIIIIIIZIIIIIIZIIIIIZIIIIIIIIIIIIIZI} \\
     & \texttt{MZ2} & \texttt{YIYI} & 15 & \texttt{XIYIYIXXIIIZIZIYIIIIIIIIIIIIIIIIIYIYXIIIIZIIIZYIZIIIIIIIIIIIIIIIII} \\
    \addlinespace[2pt]
    3 & \texttt{MZ0} & \texttt{IIIZ} & 10 & \texttt{IIIIIIZIZZZIIIIIIIIIIZIIIIIIZIZZIZIIIIIIIIIIZIIIIIIIIIIIIIIIIIIIII} \\
     & \texttt{MZ1} & \texttt{IZII} & 10 & \texttt{IIIIIZIIIIIZIIIIIIZIIIIZIIIIIIIIIIIZIIIIIIIIIZZIIIIZIIIIZIIIIZIIII} \\
     & \texttt{MZ2} & \texttt{IIZI} & 10 & \texttt{ZZIIIIIIIIIIIIIIIZIIIIIIIIIIIIIIZIIIIIIIIIZIIIIIZZZIIIIIIZIIIIIZII} \\
    \addlinespace[2pt]
    4 & \texttt{MZ0} & \texttt{IIIZ} & 10 & \texttt{ZIIIZIIIIIIIIIIIIIIIZIIIIIIZIIIIIIZIIIIIIIIIZZIIIIIIIZIIIIIIIIIIZZ} \\
     & \texttt{MZ1} & \texttt{IZII} & 10 & \texttt{IIZIIIIIIZIIIIZIIIIIIIIIIIIIIZIIIIIIZZIIIIZIIIIZIIIIZIIIIIIZIIIIII} \\
     & \texttt{MZ2} & \texttt{ZZII} & 11 & \texttt{IIIIIZIZIIZIIIIIIIIIIIIIZIIIIIIZIIIIIIIZIZIZIIIIZZIIIIIIIZIIIIIIII} \\
    \addlinespace[2pt]
    5 & \texttt{MZ0} & \texttt{IIIZ} & 10 & \texttt{IIIIIIIIIZIIIIIIZIZZIIIIIIIZIZZZIIIIIIIIIIIIIIIIIIIIIIZIIIIIIIIIIZ} \\
     & \texttt{MZ1} & \texttt{IZII} & 10 & \texttt{ZIIIIZIIIIIIIIIIIIIIZIIIIIZIIIIIIZIIIIZIIIIZIIIIIIZIIIIIIIIIZZIIII} \\
     & \texttt{MZ2} & \texttt{XZXI} & 15 & \texttt{IIZYIIZIYIIIIIIIIIIIIIIIIYIIXIIIIIYIIIIXZIIIIIIIIZIIIIIIIZXIIIXYYI} \\
    \addlinespace[2pt]
    6 & \texttt{MZ0} & \texttt{IIIZ} & 10 & \texttt{IIIIIIIIIZIIIIIIZIZZIIIIIIIZIZZZIIIIIIIIIIIIIIIIIIIIIIZIIIIIIIIIIZ} \\
     & \texttt{MZ1} & \texttt{IZII} & 10 & \texttt{ZIIIIZIIIIIIIIIIIIIIZIIIIIZIIIIIIZIIIIZIIIIZIIIIIIZIIIIIIIIIZZIIII} \\
     & \texttt{MZ2} & \texttt{YZYI} & 16 & \texttt{IIIIIIXZYIYIYYZIIIIIIXIIIIIIIIIIIIIIZIIYIYXIIIIIIIIIXIIZIIIIIIZZII} \\
    \addlinespace[2pt]
    7 & \texttt{MZ0} & \texttt{IIIZ} & 13 & \texttt{IIIZIIZIIIIIIIIIIIIIZZIZIIZIIIIIZZIIIIIZIIIZIIIIIIZIZIZIIIIIIIIIII} \\
     & \texttt{MZ1} & \texttt{IZII} & 12 & \texttt{IZIIIIIIZIZZZIIZIZIIIIIIIIIZIIIIIIZZIIIIIZIIIIIIIZIIIIIIIIIIIIIIII} \\
     & \texttt{MZ2} & \texttt{IZZI} & 13 & \texttt{IIIIZIIIIIIIIIIIIIZZIIZIIZIIZIIZIIIIIIIIIIZIZZIIIIIZIIIIIZIIIIIIZI} \\
    \addlinespace[2pt]
    8 & \texttt{MZ0} & \texttt{IIIZ} & 11 & \texttt{ZIIIIIIIIIIIIIIIIIZIIIIIIIIZIIIIIIZZIIIIIIIZIIIIIZZZIIIIIIIIIIIIZZ} \\
     & \texttt{MZ1} & \texttt{IZII} & 11 & \texttt{IZIIZIIZIIZIIZIIZIIZIIZIIZIIZIIZIIIIIIIIIIIIIIIIIIIIIIIIIIIIIIIIII} \\
     & \texttt{MZ2} & \texttt{ZIIZ} & 11 & \texttt{IIZIIZIIZIIZIIZIIZIIZIIZIIZIIZIIZIIIIIIIIIIIIIIIIIIIIIIIIIIIIIIIII} \\
    \addlinespace[2pt]
    9 & \texttt{MZ0} & \texttt{IIIZ} & 10 & \texttt{IIIIIIIIIIIIIIIIZIIIIIIIIIZIIIIIIIZIIIIIIZZIIIIZZIZIIIIIIZZIIIIIII} \\
     & \texttt{MZ1} & \texttt{IZII} & 12 & \texttt{IIIIIIIIIIIIIZZIIIIIIZZZZIIZIZIIIIIIIZIIIIIIIIIIIIIZIIZIIIIIIZIIII} \\
     & \texttt{MZ2} & \texttt{XIXZ} & 15 & \texttt{IYIXYIIIIIIIYIIIIIIIIIIIIIIIIIXIXYIIIIIIIIIYZZIIIIIIZIIIIIIZZIIXIY} \\
    \addlinespace[2pt]
    10 & \texttt{MZ0} & \texttt{IIIZ} & 10 & \texttt{IIIIIIIIIIIIIIIIZIIIIIIIIIZIIIIIIIZIIIIIIZZIIIIZZIZIIIIIIZZIIIIIII} \\
     & \texttt{MZ1} & \texttt{IZII} & 10 & \texttt{IIIIIZIIIIIZIIIIIIZIIIIZIIIIIIIIIIIZIIIIIIIIIZZIIIIZIIIIZIIIIZIIII} \\
     & \texttt{MZ2} & \texttt{YIYZ} & 16 & \texttt{YZIIYIIIIIIIZZYIIIIIIIYIZIIIIIZIZXIIXIXIIIIYXIIIIIIIIYIIIIIIIIIIII} \\
    \addlinespace[2pt]
    11 & \texttt{MZ0} & \texttt{IIIZ} & 15 & \texttt{IIIIIZIIIIIZZIIZIIZIIZIIIIIZIZIIZIIZIIIIIIIZIZIIIIIIIZIZIIIIIIIIIZ} \\
     & \texttt{MZ1} & \texttt{IZII} & 11 & \texttt{IZIIZIIZIIZIIZIIZIIZIIZIIZIIZIIZIIIIIIIIIIIIIIIIIIIIIIIIIIIIIIIIII} \\
     & \texttt{MZ2} & \texttt{IIZZ} & 12 & \texttt{IIIIIIZIZIIIIIIIIIIIIIIIZIIIIIZIIIIIIZZZIIIIIIIIZZIIIIIIIZIIIZZIII} \\
    \addlinespace[2pt]
    12 & \texttt{MZ0} & \texttt{IIIZ} & 10 & \texttt{IIIIIIZIZZZIIIIIIIIIIZIIIIIIZIZZIZIIIIIIIIIIZIIIIIIIIIIIIIIIIIIIII} \\
     & \texttt{MZ1} & \texttt{IZII} & 10 & \texttt{IIIIIZIIIIIZIIIIIIZIIIIZIIIIIIIIIIIZIIIIIIIIIZZIIIIZIIIIZIIIIZIIII} \\
     & \texttt{MZ2} & \texttt{ZZIZ} & 10 & \texttt{IZIIIIIZIIIIIIZIIIIZIIIIIIIIIIIIIIIIIIIIIZZIIIIZIIIIZIIIIZIIIIIIZI} \\
    \addlinespace[2pt]
    13 & \texttt{MZ0} & \texttt{IIIZ} & 10 & \texttt{IIIIIIZIIIIIZIIIZIIIIIIIIIIIIIIIZIIIIIIIIIIZZIZIIIIIIIIIZZIIIIIIIZ} \\
     & \texttt{MZ1} & \texttt{IZII} & 10 & \texttt{IIZIIIIIIZIIIIZIIIIIIIIIIIIIIZIIIIIIZZIIIIZIIIIZIIIIZIIIIIIZIIIIII} \\
     & \texttt{MZ2} & \texttt{XZXZ} & 15 & \texttt{XIIIIIIIIIIIIIIIIYIIIIYIXYZYIIIZIIIIIIIIIIIIIIIIZIXIIYXIIIXIIZIIYI} \\
    \addlinespace[2pt]
    14 & \texttt{MZ0} & \texttt{IIIZ} & 10 & \texttt{IIIIIIIIIZIIIIIIZIZZIIIIIIIZIZZZIIIIIIIIIIIIIIIIIIIIIIZIIIIIIIIIIZ} \\
     & \texttt{MZ1} & \texttt{IZII} & 10 & \texttt{ZIIIIZIIIIIIIIIIIIIIZIIIIIZIIIIIIZIIIIZIIIIZIIIIIIZIIIIIIIIIZZIIII} \\
     & \texttt{MZ2} & \texttt{YZYZ} & 14 & \texttt{IIIIIIYIYIYYIXXZIIIIIIIIIIIIIIIIIIIIIZIYIYYIXYIIZIIIIIIIIIIIIIIIII} \\
    \addlinespace[2pt]
    15 & \texttt{MZ0} & \texttt{IIIZ} & 10 & \texttt{IIIIIIZIZZZIIIIIIIIIIZIIIIIIZIZZIZIIIIIIIIIIZIIIIIIIIIIIIIIIIIIIII} \\
     & \texttt{MZ1} & \texttt{IZII} & 10 & \texttt{IIIIIZIIIIIIZIIIIZIIIIIIIIIIIIIIZIIIIIIZZIIIIZIIIIZIIIIZIIIIIIZIII} \\
     & \texttt{MZ2} & \texttt{IZZZ} & 11 & \texttt{IIIIIIIIIIIIIZIZIIIIZIIIZIIZIIIIIIIIIZZIIIIIIIIIIIIIIIZIZIIIZIIZII} \\
    \addlinespace[2pt]
\end{supertabular}

\tablecaption{Q66 representatives for the protected phase-fixup singleton layer.\label{tab:gb66-final-singleton-certificates}}
\tablehead{%
    \toprule
    Class & Target & Logical & \(w_{66}\) & Physical Pauli on qubits \(0,\ldots,65\) \\
    \midrule}
\tabletail{%
    \midrule
    \multicolumn{5}{r}{Continued on next page} \\
    \bottomrule}
\tablelasttail{\bottomrule}
\begin{supertabular}{@{}rllr@{\hspace{0.7em}}l@{}}
    0 & \texttt{S} & \texttt{IIXZ} & 15 & \texttt{IIIIIIIXIIIIIYIIXXIIIIIYIIIIIIIXIIIIIIZYIXIIYYIZIIIIIIZYIIIIIIIIZI} \\
    \addlinespace[2pt]
    1 & \texttt{S} & \texttt{IXIZ} & 14 & \texttt{IIIIIIIIIZIXIIIIIXIIIIIIYYYIIIIIXIZIIIIIYYYIIIIIIZIIIIIZIXIIIIIIII} \\
    \addlinespace[2pt]
    2 & \texttt{S} & \texttt{IZXI} & 15 & \texttt{YIIIIYIIIIXXIIIIIIIIIXIIIIZZYIIIIYIIIIYIIIIZIXIIIIZYIIIIIIZIIIIIII} \\
    \addlinespace[2pt]
    3 & \texttt{S} & \texttt{IZXZ} & 15 & \texttt{IXZIIIIYIIYYZIIIZXIIIIIIIXIIIIIIIYIXIIYXIIIIIIZIIYIIIIIIIIIIIIIIII} \\
    \addlinespace[2pt]
    4 & \texttt{S} & \texttt{XIYY} & 15 & \texttt{IIIIIIIIXIIIIIYIIXXIIIIIYIIIIIIIXIIIIIIZYIXIIYYIZIIIIIIZYIIIIIIIIZ} \\
    \addlinespace[2pt]
    5 & \texttt{S} & \texttt{XIZI} & 14 & \texttt{IIIIIIIIZIXIIIIIXIIIIIIYYYIIIIIXIZIIIIIYYYIIIIIIZIIIIIZIXIIIIIIIII} \\
    \addlinespace[2pt]
    6 & \texttt{S} & \texttt{XIZZ} & 16 & \texttt{IYYIIIIXIIYIYZZIIIIIIIIIIYIIIIIIZXIIIZIIXIXXIIIIYXIIIIIIIIIIIIIIII} \\
    \addlinespace[2pt]
    7 & \texttt{S} & \texttt{XYYI} & 16 & \texttt{XIYIYIIZIZIIIIIZIZZYIIIIIIIXYIIIIXIIYIYIIIIIIYIIIIIIYIIIIIIIIIIIII} \\
    \addlinespace[2pt]
    8 & \texttt{S} & \texttt{XYYZ} & 16 & \texttt{IIXZIIYIXIIIIIIIIIIIIIIIXYIZIIIIYZIYIIXZIIYIIIIXIYZIIIIIIIIIIIIIII} \\
    \addlinespace[2pt]
    9 & \texttt{S} & \texttt{XZYY} & 15 & \texttt{IIIYIZIZYYIIIIIIIIZXIIIIIIYIIIIXIZIIYIIIIIIYIIIIIYIZIIIIIIIIIIIIYI} \\
    \addlinespace[2pt]
    10 & \texttt{S} & \texttt{XZZI} & 16 & \texttt{IXXIZIIIIYIYXIZIZYYIIIIIIXIIZIIIIYIIZIIIIZIYIIIIIIIIIIIIIIIIIIIIII} \\
    \addlinespace[2pt]
    11 & \texttt{S} & \texttt{XZZZ} & 16 & \texttt{IYIXXIIIIIZIIIIYIIIIIIIYXIIIIIZIZIIYIIIIIXIZIIZIYIIIIIIIIIIIIIIIYZ} \\
    \addlinespace[2pt]
    12 & \texttt{S} & \texttt{YIXY} & 15 & \texttt{IIIIYIIIIYIIIIIIIXIIIIIIYXIIIIZIXIIIIYIIIIYIIIIYXXIIIIZIIIIIIIIIZZ} \\
    \addlinespace[2pt]
    13 & \texttt{S} & \texttt{YYXI} & 16 & \texttt{XIYIYIIZIIIIIIIIIIIXIZIIIIIXYIIIIXIIYIYIZIIIIYZIIIIIXIZIIIIIIIIIII} \\
    \addlinespace[2pt]
    14 & \texttt{S} & \texttt{YYXZ} & 15 & \texttt{IIIIIYIIXIYIIIIIIIIIIIIXIYIIIIIYIYXIXIYIIIIZIIIIYIIIIIIZXIIIIIIIIY} \\
    \addlinespace[2pt]
    15 & \texttt{S} & \texttt{YZXY} & 14 & \texttt{IIIIIYIIIIYIIIIXIIIIXXIIIZIIIIIYIZIIIIZIIIIYIIIIYIIIIIIYIIIIIXIIIZ} \\
    \addlinespace[2pt]
    16 & \texttt{S} & \texttt{ZXZZ} & 16 & \texttt{ZIYIXXIIIIIZIIIIYIIIIIIIYXIIIIIZIZIIYIIIIIXIZIIZIYIIIIIIIIIIIIIIIY} \\
    \addlinespace[2pt]
    17 & \texttt{S} & \texttt{ZZZX} & 16 & \texttt{IIYIIIXIIIIIZXIZIIYIIYIIIIIIXIIIIIIIIIIYIIYIIXIZXIIIIIZIIIYIIIIIIZ} \\
    \addlinespace[2pt]
\end{supertabular}

\clearpage

\tablecaption{Q66 representatives for the terminal phase-fixup DMG triple.\label{tab:gb66-final-dmg-triple-certificates}}
\tablehead{%
    \toprule
    Class & Target & Logical & \(w_{66}\) & Physical Pauli on qubits \(0,\ldots,65\) \\
    \midrule}
\tabletail{%
    \midrule
    \multicolumn{5}{r}{Continued on next page} \\
    \bottomrule}
\tablelasttail{\bottomrule}
\begin{supertabular}{@{}rllr@{\hspace{0.7em}}l@{}}
    0 & \texttt{MX0} & \texttt{IIIX} & 14 & \texttt{IXIIIIIIIIXIXXXIIXIIXXIIXIIIXIIIIXIIIIIIIIIIXIIXIIIIIIIIIIIIIIIIIX} \\
     & \texttt{MX1} & \texttt{IXII} & 15 & \texttt{IIIIXXIIXIIIIIIIIIIIIIXXIXXIIXXIXIIIIIXXIIIIIXIIIIIIIIIIIXIIIIXIII} \\
     & \texttt{MX2} & \texttt{IIXI} & 11 & \texttt{IIIIIIIIIIIIIIIIIIIIIIIIIIIIIIIIIIXIIXIIXIIXIIXIIXIIXIIXIIXIIXIIXI} \\
    \addlinespace[2pt]
    1 & \texttt{MX0} & \texttt{IIIX} & 12 & \texttt{IIIIIXXIIIXXIXIIIIIIIIIIXIIXIIIIXXIXIIIIIIIIIIIIIIIXIIIIIXIIIIIIII} \\
     & \texttt{MX1} & \texttt{YYXI} & 18 & \texttt{YZZIXIIIIIIIIIIIIXIZYYIIIYXIXIIZIIIIYIXIIIIIZXIIIIYIIIIIZIIIIIIIII} \\
     & \texttt{MX2} & \texttt{IIXI} & 11 & \texttt{IIIIIIIIIIIIIIIIIIIIIIIIIIIIIIIIIIXIIXIIXIIXIIXIIXIIXIIXIIXIIXIIXI} \\
    \addlinespace[2pt]
    2 & \texttt{MX0} & \texttt{IIIX} & 11 & \texttt{XXIIIIIIXIIIIIXIXIIIXIXIIXIIIIIIIIIIIIIIIIIIIIIXIIIIIXIIIIXIIIIIII} \\
     & \texttt{MX1} & \texttt{IXII} & 12 & \texttt{IIIIIIXXIIIIIIIXIIXIIXIIIIIXIIIIIIIXIIXIIIIIIXXIIIIXIIIIIIIXIIIIII} \\
     & \texttt{MX2} & \texttt{XIZI} & 14 & \texttt{IIIIXIIIIIIYYYIIIIIXIIIIIIIIIZIXIIIIZIIIIIZIXIIIIIIIIIZIIIIIYYYIII} \\
    \addlinespace[2pt]
    3 & \texttt{MX0} & \texttt{IIIX} & 11 & \texttt{XIIXIIIIXIIIIIXXIIIIIIIIIIIIIXIIIXIIIIIIIIXIXIIIIIIIXXIIIIIIIIIIII} \\
     & \texttt{MX1} & \texttt{XYYI} & 17 & \texttt{IIIIXIIIIIIIXIIIIYZIYIZZYIIIIIIIXIIIIXIZZIIIIIYIIYIYIIIIXIIIIIIIIZ} \\
     & \texttt{MX2} & \texttt{XZZI} & 16 & \texttt{IIIIIIYIIXXZIZIIIIIIIIIIIYYIIIXIIIZIIIYIIYIIIZIYZIIIIIIIIYIIIIIYII} \\
    \addlinespace[2pt]
    4 & \texttt{MX0} & \texttt{YIXY} & 15 & \texttt{IIIYIIIIXXIIIIIIIIIXIIIIZZYIIIIYIIIIYIIIIZIXIIIIZYIIIIIIZIIIIIIIYI} \\
     & \texttt{MX1} & \texttt{IXII} & 10 & \texttt{IIIIIXIIIIIIXXXIIIIIXIIIIIIIIIIIXIIIIIIIIIIIIXIIIIIIIIIIIIIIIXXXII} \\
     & \texttt{MX2} & \texttt{IIXI} & 10 & \texttt{IXIIIIXIIIIXIIIIXXIIIIIIIIIXIIIIIIIIIIIXIIIIXIIIIIIXIIIIIXIIIIIIII} \\
    \addlinespace[2pt]
    5 & \texttt{MX0} & \texttt{YZXY} & 14 & \texttt{IZIIIIIYIIIIIIYIIIIYIIIIXIIIIXXIIIIIIXIIIZZIIIIZIIIIYIIIIYIIIIIIYI} \\
     & \texttt{MX1} & \texttt{YYXZ} & 15 & \texttt{IIXIYIIIIIYIIIIIIYIIXIYIIIIIIIIIIIZXIIIIIIIIYYXIXIYIIIIZIIIIYIIIII} \\
     & \texttt{MX2} & \texttt{IIXI} & 10 & \texttt{IIIIIIIIIXIIIIIIXIIIIXIIIIXIIIIXXXIIIIIXIIIIIIIIIIIIIIXIIIIXIIIIII} \\
    \addlinespace[2pt]
    6 & \texttt{MX0} & \texttt{XIYY} & 15 & \texttt{IIXIIIIIYIIXXIIIIIYIIIIIIIXIIIIIIZYIXIIYYIZIIIIIIZYIIIIIIIIZIIIIII} \\
     & \texttt{MX1} & \texttt{IXII} & 10 & \texttt{XIIIIIIXIXXIIIIXXIIIIIIXIIIIIIIIXIIIIIIIIXIIIIIIIIIIIIIIIIIIIIIIXI} \\
     & \texttt{MX2} & \texttt{XIZZ} & 16 & \texttt{IIIIXIIIIIIIIYXIIIIXIYIIIIIIIIZXIIIIIIZIIIIIZYZIIIIIYZIIIIIIXYYIII} \\
    \addlinespace[2pt]
    7 & \texttt{MX0} & \texttt{XZYY} & 15 & \texttt{IIIYIZIZYYIIIIIIIIZXIIIIIIYIIIIXIZIIYIIIIIIYIIIIIYIZIIIIIIIIIIIIYI} \\
     & \texttt{MX1} & \texttt{XYYZ} & 16 & \texttt{IIIIIIIIIIIIXYIZIIIIYIIXZIIYIXIIIIIXIYZIIIIIIIIIIIIIIIZIYIIXZIIYII} \\
     & \texttt{MX2} & \texttt{XZZZ} & 18 & \texttt{IIIIIIYIIIIZIIZIIYIIIZIIIIIIIIIIZIIIIIIZYXYIYZIZIIZIYYIYIIIIIIXIII} \\
    \addlinespace[2pt]
\end{supertabular}

\tablecaption{Q66 representatives for the terminal phase-fixup DMG pair.\label{tab:gb66-final-dmg-pair-certificates}}
\tablehead{%
    \toprule
    Class & Target & Logical & \(w_{66}\) & Physical Pauli on qubits \(0,\ldots,65\) \\
    \midrule}
\tabletail{%
    \midrule
    \multicolumn{5}{r}{Continued on next page} \\
    \bottomrule}
\tablelasttail{\bottomrule}
\begin{supertabular}{@{}rllr@{\hspace{0.7em}}l@{}}
    0 & \texttt{MX0} & \texttt{IIIX} & 11 & \texttt{XXIIIIIIXIIIIIXIXIIIXIXIIXIIIIIIIIIIIIIIIIIIIIIXIIIIIXIIIIXIIIIIII} \\
     & \texttt{MX2} & \texttt{XIZI} & 14 & \texttt{IIIIXIIIIIIYYYIIIIIXIIIIIIIIIZIXIIIIZIIIIIZIXIIIIIIIIIZIIIIIYYYIII} \\
    \addlinespace[2pt]
    1 & \texttt{MX0} & \texttt{IIIX} & 11 & \texttt{XIIXIIIIXIIIIIXXIIIIIIIIIIIIIXIIIXIIIIIIIIXIXIIIIIIIXXIIIIIIIIIIII} \\
     & \texttt{MX1} & \texttt{XYYI} & 17 & \texttt{IIIIXIIIIIIIXIIIIYZIYIZZYIIIIIIIXIIIIXIZZIIIIIYIIYIYIIIIXIIIIIIIIZ} \\
    \addlinespace[2pt]
    2 & \texttt{P0} & \texttt{IIIX} & 14 & \texttt{IXIIIIIIIIXIXXXIIXIIXXIIXIIIXIIIIXIIIIIIIIIIXIIXIIIIIIIIIIIIIIIIIX} \\
     & \texttt{P1} & \texttt{ZXZI} & 17 & \texttt{IIIZIYIIIZIYIIIZXIIIIIIIIIIXIIIIIIYXYIIIZIIZIZIIZXIIIIIIIXXIIIIIII} \\
    \addlinespace[2pt]
    3 & \texttt{P0} & \texttt{IIXI} & 12 & \texttt{IIIXXIIXIIIXIIIIIXXIIIIIIXIIIIIXIXIIIIXIXIIIIIIIIIIIIIIIIIIIIIIIXI} \\
     & \texttt{P1} & \texttt{IZIX} & 20 & \texttt{IXXIIXYIIIIIIYIIIIIIYIIZZIIZIIZIIIXIXIIIIIIXIIZZIIIIIIIIIIIIZYYYIX} \\
    \addlinespace[2pt]
    4 & \texttt{MX1} & \texttt{YYXI} & 18 & \texttt{YZZIXIIIIIIIIIIIIXIZYYIIIYXIXIIZIIIIYIXIIIIIZXIIIIYIIIIIZIIIIIIIII} \\
     & \texttt{MX2} & \texttt{IIXI} & 11 & \texttt{IIIIIIIIIIIIIIIIIIIIIIIIIIIIIIIIIIXIIXIIXIIXIIXIIXIIXIIXIIXIIXIIXI} \\
    \addlinespace[2pt]
    5 & \texttt{MX0} & \texttt{YZXY} & 14 & \texttt{IZIIIIIYIIIIIIYIIIIYIIIIXIIIIXXIIIIIIXIIIZZIIIIZIIIIYIIIIYIIIIIIYI} \\
     & \texttt{MX2} & \texttt{IIXI} & 10 & \texttt{IIIIIIIIIXIIIIIIXIIIIXIIIIXIIIIXXXIIIIIXIIIIIIIIIIIIIIXIIIIXIIIIII} \\
    \addlinespace[2pt]
    6 & \texttt{P0} & \texttt{IIXZ} & 19 & \texttt{IIZIIIIIIIIIIIIIXIZIIIIIIIZIIIIXIZZXXXIIXIIYXIIIIYYYIIIIIYYIIIIIXI} \\
     & \texttt{P1} & \texttt{IXIZ} & 16 & \texttt{ZIIIIXIIZIIIIIXXIIIIIYYYIIIIIIXIYIIIIIIIIZIIIIYIIIIIZIXIIIIIIIXYII} \\
    \addlinespace[2pt]
    7 & \texttt{P0} & \texttt{IIXZ} & 18 & \texttt{ZIIIIIIIIZIIIIIZIIIIIIIIIIIIIIIIIZXIIXIIXIIXIIYZIYZIXIIXIIYIIXIIYZ} \\
     & \texttt{P1} & \texttt{YYXI} & 18 & \texttt{IIIXIZIYIIIZZIIIIZIIYIIXYIIIXXIYIIIIIIIYIXZIIIIIXIIIIYIIIZIIIIIIII} \\
    \addlinespace[2pt]
    8 & \texttt{P0} & \texttt{IIXZ} & 17 & \texttt{IIIIIIIXIIXIYXXIIIXYIYZIIIYIIXYIYYYYIIIIIIIIIIIIIIIIIIIIIIIIIIIIYI} \\
     & \texttt{P1} & \texttt{YYXZ} & 18 & \texttt{IIXIZYIIIYIYIIIIIIIIIIIIIIIXYIIIIIIIIIXZIXIZZXIIIIXZXIIIIIZIIIIIIZ} \\
    \addlinespace[2pt]
    9 & \texttt{MX0} & \texttt{XIYY} & 15 & \texttt{IIXIIIIIYIIXXIIIIIYIIIIIIIXIIIIIIZYIXIIYYIZIIIIIIZYIIIIIIIIZIIIIII} \\
     & \texttt{MX1} & \texttt{IXII} & 10 & \texttt{XIIIIIIXIXXIIIIXXIIIIIIXIIIIIIIIXIIIIIIIIXIIIIIIIIIIIIIIIIIIIIIIXI} \\
    \addlinespace[2pt]
    10 & \texttt{MX1} & \texttt{IXII} & 12 & \texttt{IIIIIIXXIIIIIIIXIIXIIXIIIIIXIIIIIIIXIIXIIIIIIXXIIIIXIIIIIIIXIIIIII} \\
     & \texttt{MX2} & \texttt{XIZI} & 14 & \texttt{IIIIXIIIIIIYYYIIIIIXIIIIIIIIIZIXIIIIZIIIIIZIXIIIIIIIIIZIIIIIYYYIII} \\
    \addlinespace[2pt]
    11 & \texttt{P0} & \texttt{IXII} & 11 & \texttt{XXXIIIIIIIIIXXIIIIXXIIIIIIXIIIIIIXIIIIIIIXIIXIIIIIIIIIIIIIIIIIIIII} \\
     & \texttt{P1} & \texttt{ZIZX} & 18 & \texttt{IIIXIIIYIZIZIIYIIIIIIIXIIIIXIZYIIIIYIIIIYIIYIXIYIIIIIXXIIIIIYZIIII} \\
    \addlinespace[2pt]
    12 & \texttt{B1} & \texttt{IXIZ} & 14 & \texttt{IIIIIIZIXIIIIIXIIIIIIYYYIIIIIXIIIIIIIYYYIIIIIIZIIIIIZIXIIIIIIIIIZI} \\
     & \texttt{B2} & \texttt{XIZI} & 14 & \texttt{IIZIXIIIIIXIIIIIIYYYIIIIIXIIIIIIIYYYIIIIIIZIIIIIZIXIIIIIIIIIZIIIII} \\
    \addlinespace[2pt]
    13 & \texttt{P0} & \texttt{IXIZ} & 18 & \texttt{XIIXIXYIZZZIXIIIIIIIXZIIIIIIYXYZXZIIIIIIIIIIYIIIIIIIIIIIIIIIIIIIXI} \\
     & \texttt{P1} & \texttt{XIZZ} & 16 & \texttt{IIXIIIIZIIIIIIIIIIIIIIZZYIIIIIIIIIIIIXYZYIIIIIIYIIIIIZIIIIXXXIXIIZ} \\
    \addlinespace[2pt]
    14 & \texttt{P0} & \texttt{IZIX} & 18 & \texttt{IIIIIIIIIIIXIIIIIIZIYXIIIIIIIXIIIZIIIXXYIYIIIXYZZIIYIYIIIIZIIIIIIZ} \\
     & \texttt{P1} & \texttt{IZXI} & 18 & \texttt{IXIXIIIIIIIIZZIIIIIYIIIZXIIIXIIXIIIYYIIIIIYYIIIIIZIIZIIIZXIIIYIIII} \\
    \addlinespace[2pt]
    15 & \texttt{P0} & \texttt{IZIX} & 16 & \texttt{XIIYIIIIXZIIIIYXIIIIIIIIIIIIZYIIIYIZIZIIIIYIXIIZIIIIYXIIIIIIIIIIII} \\
     & \texttt{P1} & \texttt{XIZI} & 14 & \texttt{IIIIIYYYIIIIIXIIIIIIIIIZIXIIIIIXIIIIZIXIIIIIIIIIZIIIIIYYYIIIIIIZII} \\
    \addlinespace[2pt]
    16 & \texttt{P0} & \texttt{IZIX} & 20 & \texttt{XIIZZZZXIYXZIIIXIIIXIIXIIIIIZYIIIYIIYIIXIIIYIIIIIIXIYIIIIIIIIIIIII} \\
     & \texttt{P1} & \texttt{XZZI} & 17 & \texttt{IIIIIIIIIIIIIIIIIIIIIYIIIIZIIIIIYIXXIXZIIIIIXZZIIIIIIIZXXYIYIIIZIZ} \\
    \addlinespace[2pt]
    17 & \texttt{P0} & \texttt{IZXI} & 22 & \texttt{XZXXZIIZIIYIIZIIZIXYXXYXXYIXZXIZIIIIIIIIIIIIIIIIIIIIIIIIIIIIIIIIXI} \\
     & \texttt{P1} & \texttt{YIXY} & 22 & \texttt{IIIIIZIIIIIIIIIIIIIIIIIIIIIIIIIIIZZXZIXIZXZIYZZYZZYZIXIIXIIYIIXIIX} \\
    \addlinespace[2pt]
    18 & \texttt{P0} & \texttt{IZXI} & 17 & \texttt{XIIYZZYXIZIZIIXIIIIIXXIIIIIIYZIIIYIIYIIIIIIYIIIIIIIIZIIIIIIIIIIIII} \\
     & \texttt{P1} & \texttt{YZXY} & 19 & \texttt{IYIIIIIIZIIIIIIZIIIIIIXIIIIXIIYIIIYIIIXXIXIIIIIIYZIIIZIXZZIIYIIYIX} \\
    \addlinespace[2pt]
    19 & \texttt{MX0} & \texttt{XIYY} & 15 & \texttt{IIXIIIIIYIIXXIIIIIYIIIIIIIXIIIIIIZYIXIIYYIZIIIIIIZYIIIIIIIIZIIIIII} \\
     & \texttt{MX2} & \texttt{XIZZ} & 16 & \texttt{IIIIXIIIIIIIIYXIIIIXIYIIIIIIIIZXIIIIIIZIIIIIZYZIIIIIYZIIIIIIXYYIII} \\
    \addlinespace[2pt]
    20 & \texttt{P0} & \texttt{XIYY} & 15 & \texttt{IIIYIIIIIIIXIIIIIIIIXIIIIIYIIXXIIIZYIIIIIIIIZIIIIIIZYIXIIYYIZIIIII} \\
     & \texttt{P1} & \texttt{ZXZI} & 15 & \texttt{IIIIIIXYYIIIIIIZIIYYIIZIXXIIIIIIXIIIIIZYIIIIIIIXIIYIIIIZIIIIIIIIII} \\
    \addlinespace[2pt]
    21 & \texttt{P0} & \texttt{XIZX} & 15 & \texttt{XIIIIIYIIXXIIIIIYIIIIIIIXIIIIIIIIIXIIYYIZIIIIIIZYIIIIIIIIZIIIIIIZY} \\
     & \texttt{P1} & \texttt{ZXZZ} & 16 & \texttt{IIIXIXIIIIIYIIIIIIZXIYYIIZIYYIIIIIIYYIIIIIIIZIIIIIIIZYIIIIIIZIIIII} \\
    \addlinespace[2pt]
    22 & \texttt{P0} & \texttt{XIZZ} & 16 & \texttt{IZIIIIIIXIIIZIIIIIIIZIIIIIIIZIXIIYIYZIIIIIIYYIXIIIIZIYIIIIIIIIIIXY} \\
     & \texttt{P1} & \texttt{XZYY} & 15 & \texttt{IIIIXIIIIYIZIZYYIIIIIIIIZXIIIIIIYIIIIYIZIIYIIIIIIYIIIIIYIZIIIIIIII} \\
    \addlinespace[2pt]
    23 & \texttt{MX1} & \texttt{XYYI} & 17 & \texttt{IIIIXIIIIIIIXIIIIYZIYIZZYIIIIIIIXIIIIXIZZIIIIIYIIYIYIIIIXIIIIIIIIZ} \\
     & \texttt{MX2} & \texttt{XZZI} & 16 & \texttt{IIIIIIYIIXXZIZIIIIIIIIIIIYYIIIXIIIZIIIYIIYIIIZIYZIIIIIIIIYIIIIIYII} \\
    \addlinespace[2pt]
    24 & \texttt{P0} & \texttt{XYYI} & 16 & \texttt{XIYIYIIZIZIIIIIZIZZYIIIIIIIXYIIIIXIIYIYIIIIIIYIIIIIIYIIIIIIIIIIIII} \\
     & \texttt{P1} & \texttt{XZZX} & 16 & \texttt{IYIIIIXIXIYIIIIIIIIIIIIZYIIIIXIIZIXZIIIIYIIZIIIYXIIIIIIIIIIIIZYIII} \\
    \addlinespace[2pt]
    25 & \texttt{P0} & \texttt{XYYI} & 16 & \texttt{XIYIYIIZIZIIIIIZIZZYIIIIIIIXYIIIIXIIYIYIIIIIIYIIIIIIYIIIIIIIIIIIII} \\
     & \texttt{P1} & \texttt{XZZZ} & 16 & \texttt{IZIIIIZIIIIIIIZIIIIIIYIZIZIIIIIIXIYXIYIZIIIIXIIIIIIIIIIYYXIYIIIIII} \\
    \addlinespace[2pt]
    26 & \texttt{P0} & \texttt{XYYI} & 16 & \texttt{ZYIIIIIIIXYIIIIXIYIYIIZIZIIIIIZIZIYIIIIIIIIIIIIIXIIYIYIIIIIIYIIIII} \\
     & \texttt{P1} & \texttt{ZIZX} & 17 & \texttt{IIIIIIIIIIIIYIIIIIZIXXIIIIIYIIIIIZIYIIYYYZYIIYIIIYIIIIIIIZIYIIIIZI} \\
    \addlinespace[2pt]
    27 & \texttt{MX0} & \texttt{XZYY} & 15 & \texttt{IIIYIZIZYYIIIIIIIIZXIIIIIIYIIIIXIZIIYIIIIIIYIIIIIYIZIIIIIIIIIIIIYI} \\
     & \texttt{MX1} & \texttt{XYYZ} & 16 & \texttt{IIIIIIIIIIIIXYIZIIIIYIIXZIIYIXIIIIIXIYZIIIIIIIIIIIIIIIZIYIIXZIIYII} \\
    \addlinespace[2pt]
    28 & \texttt{P0} & \texttt{XYYZ} & 16 & \texttt{YIYIXZIIIIIIIIIIIIIYIIIIIIZXXIIIIXZIYIYIIIIZZXIIIIIIYIIIIIIIIIIIII} \\
     & \texttt{P1} & \texttt{XZZI} & 16 & \texttt{IIIIIIIXXIIIIIIXIYYIZIZYYIIIIIIXIIIIIIIYZIIIIIZIIYIZIIIIIIIIIIIIIZ} \\
    \addlinespace[2pt]
    29 & \texttt{MX1} & \texttt{XYYZ} & 16 & \texttt{IIIIIIIIIIIIXYIZIIIIYIIXZIIYIXIIIIIXIYZIIIIIIIIIIIIIIIZIYIIXZIIYII} \\
     & \texttt{MX2} & \texttt{XZZZ} & 18 & \texttt{IIIIIIYIIIIZIIZIIYIIIZIIIIIIIIIIZIIIIIIZYXYIYZIZIIZIYYIYIIIIIIXIII} \\
    \addlinespace[2pt]
    30 & \texttt{P0} & \texttt{XYYZ} & 16 & \texttt{IIIYIYIIZIIIIIIIIIIIIXYIIIIIIYIIXIIXIZIYIZIIXIYZIIIIIIIZIIIIIIIIIX} \\
     & \texttt{P1} & \texttt{ZIYY} & 16 & \texttt{YIXIZIIZIIIIIIIIIXIIIIIIIXYIZIIXIIYIYIIIZIIYIIIIIIYZIIIIIIIIIIIIXI} \\
    \addlinespace[2pt]
    31 & \texttt{MX0} & \texttt{XZYY} & 15 & \texttt{IIIYIZIZYYIIIIIIIIZXIIIIIIYIIIIXIZIIYIIIIIIYIIIIIYIZIIIIIIIIIIIIYI} \\
     & \texttt{MX2} & \texttt{XZZZ} & 18 & \texttt{IIIIIIYIIIIZIIZIIYIIIZIIIIIIIIIIZIIIIIIZYXYIYZIZIIZIYYIYIIIIIIXIII} \\
    \addlinespace[2pt]
    32 & \texttt{P0} & \texttt{XZYY} & 18 & \texttt{IIIZIIYIIYIIIIIIIIXIIIIIIIIXIIZIIZIYIIYIXIYXXYIIIIIXYZIIIIIYIIIIII} \\
     & \texttt{P1} & \texttt{YYXZ} & 19 & \texttt{IIXIIIIIYIXIIXIZIZIIIIYIIIYIXZIYIIIIIIIIIXIIIIXIIZIIIIXIIYYIIIXZII} \\
    \addlinespace[2pt]
    33 & \texttt{P0} & \texttt{YIIY} & 20 & \texttt{IIYIXYIIZIIZXIZIIZIIYXXYXXYXIYIXZXIIIIIIIIIIIIIIIIIIIIIIIIIIIIIIII} \\
     & \texttt{P1} & \texttt{ZXZZ} & 18 & \texttt{IZIIIIIYIIZIIZIIIIXIIIIIIIIIIIIIIIIIIZIZIXYXIXZIIZIIIXXIYZIIIIIYII} \\
    \addlinespace[2pt]
    34 & \texttt{P0} & \texttt{YIXY} & 15 & \texttt{IIIYIIIIIIIXIIIIIIYXIIIIZIXIIIIYIIIIYIIIIYXXIIIIZIIIIIIIIIZZIIIIYI} \\
     & \texttt{P1} & \texttt{ZXZI} & 15 & \texttt{YYIIZIXXIIIIIIXIIIIIIXYYIIIIIIZIIIIIIZIIIIIIIIIIIIIIIZYIIIIIIIXIIY} \\
    \addlinespace[2pt]
    35 & \texttt{P0} & \texttt{YYXI} & 17 & \texttt{IYIIIIIIIZIIIIYIIIIIIIIIZIIIIIZIIXXYIYIIYZIZYIIXIIIIIIYIIIIIIIIIYX} \\
     & \texttt{P1} & \texttt{ZIYY} & 16 & \texttt{IIIIXYIZIIIIYIIXZIIYIXIIIIIIIIIIIIIIIIIIIIIIIIZIYIIXZIIYIIIIXIYZII} \\
    \addlinespace[2pt]
    36 & \texttt{P0} & \texttt{YYXI} & 16 & \texttt{IIYIZXIIIZIIIIIYIIIIIIIIXYIIIIIIYIIYIZIXIIYIIIYYIXIIIIIIIIIIIIIIIZ} \\
     & \texttt{P1} & \texttt{ZIZX} & 16 & \texttt{IIIZIIYYIIZIIIXIIIIXIIIIIIXIYYIXIIIIIIZIIIIZIIIIIIIIXIIIIIIYYIIIXI} \\
    \addlinespace[2pt]
    37 & \texttt{P0} & \texttt{YYXZ} & 15 & \texttt{IIIIIYIIXIYIIIIIIIIIIIIXIYIIIIIYIYXIXIYIIIIZIIIIYIIIIIIZXIIIIIIIIY} \\
     & \texttt{P1} & \texttt{YZIY} & 15 & \texttt{IYXIIIIZIXIIIIYIIIIYIIIIIIIXIIIIIIIIIIIIIZZIIIIYIIIIYIIIIYXXIIIIZI} \\
    \addlinespace[2pt]
    38 & \texttt{MX0} & \texttt{YZXY} & 14 & \texttt{IZIIIIIYIIIIIIYIIIIYIIIIXIIIIXXIIIIIIXIIIZZIIIIZIIIIYIIIIYIIIIIIYI} \\
     & \texttt{MX1} & \texttt{YYXZ} & 15 & \texttt{IIXIYIIIIIYIIIIIIYIIXIYIIIIIIIIIIIZXIIIIIIIIYYXIXIYIIIIZIIIIYIIIII} \\
    \addlinespace[2pt]
    39 & \texttt{P0} & \texttt{ZIYY} & 16 & \texttt{YIXZIIIIIIIIIIIIIYIIIIIIZXXIIIIYIIYIYIIIIZZXIIIIIIYIIIIIIIIIIIIIXZ} \\
     & \texttt{P1} & \texttt{ZXZZ} & 16 & \texttt{IIIIIYIIIIIIZXIYYIIZIYYIIIIIIIXIXIIIIIZIIIIIIIZYIIIIIIZIIIIIIIYYII} \\
    \addlinespace[2pt]
    40 & \texttt{P0} & \texttt{ZIZX} & 20 & \texttt{IIIZXIIIIIIZZXIXXYIIYIIXYIIYIIIYIIIYYIIIIZIIIIIXIIYIIZIIIIIIIIIIZI} \\
     & \texttt{P1} & \texttt{ZXZI} & 18 & \texttt{IIIIIIIIXIIIIIYIIIIIIYIIIYIIZYIIXIIIIXYYZIIIIIIIZIIIIIYXXIIIZIIYIZ} \\
    \addlinespace[2pt]
\end{supertabular}

\endgroup

\clearpage

\end{document}